\documentclass[11pt,a4paper, reqno]{amsart}

\usepackage{color} \usepackage[T1]{fontenc} \usepackage[utf8]{inputenc}
\usepackage{textcomp} \usepackage{amsthm} \usepackage{amsmath}
\usepackage{amssymb} \usepackage{esint} \usepackage{bbm}
\usepackage[colorlinks=false,hidelinks]{hyperref}
\usepackage[english]{babel} \usepackage{xcolor} \usepackage{tikz}
\usetikzlibrary{matrix,arrows} \usepackage{tikz-cd}
\usepackage[font=scriptsize,labelfont=bf]{caption}
\usepackage[square,sort&compress,comma,numbers]{natbib}

\usepackage{amscd}

\usepackage{amsfonts} \usepackage{amsthm} \usepackage{mathrsfs}
\usepackage{dsfont}

\usepackage{mathtools}

\usepackage{enumerate} \usepackage{bm} \usepackage{bbm} \usepackage{verbatim}

\usepackage{enumitem} \newlist{abbrv}{itemize}{1}
\setlist[abbrv,1]{label=,labelwidth=1.2in,align=parleft,itemsep=0.1\baselineskip,leftmargin=!}

\numberwithin{equation}{section}

\newcommand{\Fock}{\mathcal F} \newcommand{\tr}{\text{Tr}}
\newcommand{\abs}[1]{\left| #1 \right|}

\newcommand{\scp}[2]{\big\langle #1 , #2 \big\rangle} \newcommand{\bra}[1]{\langle #1 |}
\newcommand{\ket}[1]{| #1 \rangle} \newcommand{\norm}[1]{\| #1 \| }

\renewcommand{\Re}{\mathrm{Re}} \renewcommand{\Im}{\mathrm{Im}}

\newcommand{\id}{\mathds{1}} 
 \newcommand{\loc}{\mathrm{loc}}

\newcommand{\h}[1]{\mathfrak h_{#1}} \newcommand{\2}{{L^2}}
\newcommand{\Lp}[1]{L^{#1}} \newcommand{\Sob}[1]{H^{#1}}

 \newcommand{\R}{\mathbb{R}} \newcommand{\C}{\mathbb{C}}
\newcommand{\cF}{\Fock} \newcommand{\cN}{\mathcal{N}}

\newcommand{\RRR}{\mathbb{R}}  \newcommand{\NNN}{\mathbb{N}}

\DeclareMathOperator*{\slim}{\mathrm{s--lim}}

\newtheorem{theorem}{Theorem}[section] \newtheorem{lemma}[theorem]{Lemma}
\newtheorem{def:lemma}[theorem]{Definition and Lemma}
\newtheorem{corollary}[theorem]{Corollary}
\newtheorem{proposition}[theorem]{Proposition}

\theoremstyle{definition} 

\theoremstyle{remark} \newtheorem{remark}{Remark}[section]

\begin{document}

\title{Renormalized Bogoliubov Theory for the Nelson Model}

\author[M. Falconi]{Marco Falconi} \address{Marco Falconi, Dipartimento di
  Matematica, Politecnico di Milano, Piazza Leonardo da Vinci 32, 20133
  Milano, Italy.}  \email{marco.falconi@polimi.it}

\author[J. Lampart]{Jonas Lampart} \address{Jonas Lampart, CNRS \& LICB (UMR
  6303), Universit\'e de Bourgogne Franche--Comt\'e, 9 Av.~A.~ Savary, 21078
  Dijon Cedex, France.}  \email{lampart@math.cnrs.fr}
\author[N. Leopold]{Nikolai Leopold} \address{Nikolai Leopold, University of
  Basel, Department of Mathematics and Computer Science, Spiegelgasse 1, 4051
  Basel, Switzerland.}  \email{nikolai.leopold@unibas.ch}
\author[D. Mitrouskas]{David Mitrouskas} \address{David Mitrouskas, Institute
  of Science and Technology Austria (ISTA), Am Campus 1, 3400 Klosterneuburg,
  Austria.} \email{mitrouskas@ist.ac.at}

\begin{abstract}
  We consider the time evolution of the renormalized Nelson model, which
  describes $N$ bosons linearly coupled to a quantized scalar field, in the
  mean-field limit of many particles $N\gg 1$ with coupling constant
  proportional to $N^{-1/2}$. First, we show that initial states exhibiting
  Bose--Einstein condensation for the particles and approximating a coherent
  state for the quantum field retain their structure under the many-body time
  evolution.~Concretely, the dynamics of the reduced densities are
  approximated by solutions of two coupled PDEs, the
  Schr\"odinger--Klein--Gordon equations. Second, we construct a renormalized
  Bogoliubov evolution that describes the quantum fluctuations around the
  Schr\"odinger--Klein--Gordon equations. This evolution is used to extend the
  approximation of the evolved many-body state to the full norm topology. In
  summary, we provide a comprehensive analysis of the Nelson model that
  reveals the role of renormalization in the mean-field Bogoliubov theory.
\end{abstract}

\subjclass[2020]{Primary:
  \href{https://mathscinet.ams.org/msc/msc2020.html?t=81Vxx&btn=Current}{81V73}. Secondary:
  \href{https://mathscinet.ams.org/msc/msc2020.html?t=81Txx&btn=Current}{81T16}.}

\maketitle

\frenchspacing


\tableofcontents


\section{Introduction and main results}

We study the effective behavior of a large number of bosonic particles in
weak interaction with a quantized scalar field. Microscopically, such a
system is described by the Nelson Hamiltonian. This was first introduced in
1964 in the mathematical physics literature by E.\ Nelson~\cite{nelson}, and
provides an example of rigorous renormalization in quantum field theory: The
formal Hamiltonian needs to be corrected by the divergent self-energy of the
particles to obtain a self-adjoint operator and associated unitary dynamics.
Renormalization plays a crucial role not only in (mathematical) physics, it
also led, perhaps unexpectedly, to groundbreaking advances in pure and
applied mathematics, from stochastic and nonlinear partial differential
equations to dynamical systems and geometry (see
\cite{bourgain,burq,hairer-regularity,mcmullen1994} for some celebrated
examples). A deeper understanding of renormalization is thus of great
relevance for both mathematics and physics. In this paper, we clarify the
role played by renormalization in the mean-field Bogoliubov theory for the
Nelson model.

What came to be known as Bogoliubov theory was introduced in the 1940s as a
heuristic approach to the analysis of excitations in the condensed Bose gas
\cite{Bogoliubov1947}. After the successful creation of Bose--Einstein
condensates in laboratory experiments during the 1990s, these ideas have
regained significant interest from the mathematical physics community. This
led to the development of a larger research endeavor aimed at providing a
rigorous justification of Bogoliubov's approach, starting from the many-body
Schr\"odinger theory. For the low-energy excitation spectrum of large bosonic
systems, Bogoliubov theory was first justified by Seiringer
\cite{Seiringer11} and Grech and Seiringer \cite{GrechS13}. Regarding the
time evolution of excitations in many-particle systems, pioneering results
were obtained by Ginibre and Velo \cite{GV1979a,GV1979b} and Grillakis,
Machedon and Margetis \cite{GMM2011,GMM2010}. Over the past decade, there has
been substantial progress in developing refined methods for the derivation of
Bogoliubov's theory and in extending the analysis to cover more singular
interactions, e.g. \cite{BBCS19,Lewin:2015a,LNSS2015,brennecke:2017} and we
refer to Section \ref{sect:comparison} for a detailed overview. In the
present paper, we continue this effort by establishing Bogoliubov's
approximation for the time evolution in the Nelson model.  What sets this
problem qualitatively apart from previous works is the need for
renormalization of the underlying many-body Hamiltonian at finite particle
number and the corresponding Bogoliubov evolution. Our results shed new light
on the behavior of such systems and we hope that they pave the way for
further investigation of the interplay between many-body effects and singular
particle-field interactions. It is worth mentioning that such systems are of
continued relevance in physics. To cite a recent example from condensed
matter theory, the interplay of many-body effects and singular particle-field
interactions is crucial in quantum fluids of light \cite{frerot23,
  carusotto13, fontaine18}. There, a renormalized Bogoliubov theory of the
condensate-phonon interaction is necessary to explain the properties of
concrete polariton-exciton condensates \cite{frerot23}.

Now in more detail, we investigate the dynamical evolution of the Nelson
model in a mean-field limit in which the number of particles, denoted by $N$,
becomes large while the coupling to the scalar field is proportional to $1/
\sqrt{N}$.  Our first result concerns the dynamics of the one-particle and
one-field-mode reduced density matrices. We assume that, at the initial time,
the particles exhibit Bose--Einstein condensation and the quantum field is
approximately in a coherent state.  We prove that the time evolution of the
reduced density matrices of such initial states can be described by a
condensate wave function and a classical scalar field that solve a system of
two coupled PDEs, the Schr\"odinger--Klein--Gordon equations, with errors that
tend to zero as $N\to \infty$. The renormalization on the microscopic level,
interestingly, does not appear in the mean-field equations. This was observed
earlier by Z.\ Ammari and one of the authors in \cite{AF2017}, who proved a
similar statement using semiclassical techniques and without quantitative
bounds.  In this article we employ different techniques that yield an
explicit rate of convergence for initial states satisfying an energy
condition.

Our second result provides an approximation in norm of the time evolved
many-body state, by a state obtained from the Schrödinger--Klein--Gordon
equations and a quadratic Bogoliubov type evolution modeling the quantum
fluctuations around the mean-field dynamics. On the level of the quantum
fluctuations, it is important to take the renormalization into account and
construct a renormalized Bogoliubov evolution, whose construction may be of
interest on its own.


\subsection{The Nelson Hamiltonian}
\label{sec:nelson-hamiltonian}

We consider the massive Nelson model in the mean-field regime. It describes a
system of $N$ non-relativistic bosonic particles that are linearly coupled to
a scalar quantum field, whose states are elements of the Hilbert space
\begin{align}
  \mathcal{H}_N = \bigotimes_{\rm sym}^N L^2(\mathbb R^3) \otimes \mathcal{F}
\end{align}
where $\mathcal{F} = \C \Omega \oplus \bigoplus_{n = 1}^{\infty} \bigotimes_{\rm sym}^n L^2(\mathbb R^3)$ is the
bosonic Fock space over $L^2(\mathbb{R}^3)$ with vacuum state $\Omega$. The state of the
system evolves according to the Schr\"odinger equation
\begin{align}
  \label{eq:Schroedinger equation}
  i \partial_t \Psi_{N}(t) = H_N \Psi_{N}(t) .
\end{align}
Formally, the Nelson Hamiltonian $H_N$ is given by the expression
\begin{align}
  \label{eq:Nelson Hamiltonian formal definition}
  \sum_{j=1}^N  \left[ - \Delta_j
    + N^{-1/2}
    \int_{\mathbb{R}^3} dk \, \omega^{- 1/2}(k)
    \left( e^{- i k x_j } a^*_k  + e^{ ikx_j}  a_k  \right)  \right] +  \text{d} \Gamma_a(\omega) ,
\end{align}
where $x_1,\ldots,x_N $ denote the variables of the particles, $\omega (k) = \sqrt{k^2
  + 1}$ and $\text{d} \Gamma_a(\omega) = \int_{\mathbb{R}^3} dk \, \omega(k) a_k^* a_k$ is the second
quantization of the multiplication operator $\omega$, describing the energy of the
quantum field.  The annihilation and creation operators are defined by the
distribution-valued expressions
\begin{align}
  \begin{split}
    \left( a_k \Psi_N \right)^{(n)}(X_N,K_n)
    &= \sqrt{n+1} \Psi_N^{(n+1)}(X_N, k, K_n) ,
    \\
    \left( a_k^* \Psi_N \right)^{(n)}(X_N,K_n)
    &= n^{-\frac{1}{2}} \sum_{j=1}^n \delta(k - k_j) \Psi_N^{(n-1)}(X_N, K_n\setminus k_j )
  \end{split}
\end{align}
with $\Psi^{(n)}_N \in \bigotimes_{\rm sym}^N L^2(\mathbb R^3) \otimes \bigotimes_{\rm sym}^n L^2(\mathbb
R^3)$ and $X_N = (x_1, \ldots, x_N)$, $K_n=(k_1, \dots, k_n)$.  They satisfy the
canonical commutation relations
\begin{align}
  \label{eq: canonical commutation relation}
  [a_k, a^*_l ] &=  \delta(k-l), \quad
  [a_k, a_l ] = 
  [a^*_k, a^*_l ] = 0 .
\end{align}

This definition of $H_N$ is only formal, since no domain has been specified.
The quadratic form associated to the expression is ill defined on the
form-domain of the non-interacting Hamiltonian, and while it may be defined
on more regular states, this makes it unbounded from below and not closable.
However, this problem can be remedied by renormalization~\cite{nelson}:
Denote by $H_{N}^\Lambda$ the version of~\eqref{eq:Nelson Hamiltonian formal
  definition} with $\omega^{-1/2}$ replaced by $\omega^{-1/2}\id_{|k|\leq
  \Lambda}$ in the interaction, then there exists a diverging family of
numbers $E^\Lambda$ and a self-adjoint operator $H_N$, $D(H_N)$ so that
\begin{equation}\label{eq:def:ren:Nelson}
  e^{-i t H_N  } =\slim_{\Lambda\to \infty} e^{-i t  H_N^\Lambda }e^{-i t E^\Lambda}.
\end{equation}
We take this as the definition of the Nelson Hamiltonian $H_N$, and remark
that, due to the coupling constant $N^{-1/2}$, the numbers $E^\Lambda$ can be
chosen independent of $N$.  The operator $H_N$ can be characterized further
by applying a dressing transformation, see~\cite{nelson, GW2018} and
Lemma~\ref{lem:H:D:representation}, or by an alternative approach related to
boundary conditions~\cite{LaSch19}.  It is important to note the effect of
the renormalization on the domain of $H_N$. While the operators $H_N^\Lambda$
all share the domain of the free Hamiltonian, it holds that $D(H_N^{1/2})\cap
H^1(\R^{3N})\otimes \Fock=\{0\}$, i.e., even the form domain of $H_N$ is
completely different from that of the free Hamiltonian~\cite{GW2018,
  LaSch19}.

An important role in our analysis is played by a unitary transformation,
usually called the dressing transformation, that relates the renormalized
Nelson Hamiltonian to an operator whose quadratic form is more explicit and,
importantly, comparable to the one of the free Hamiltonian.  Although our
main results below can be stated without any reference to this
transformation, it is crucial in their proofs.

\subsection{Main results}
\label{sect:main-results}

In this section we state our results on the approximation for the Nelson time
evolution, first on the level of reduced densities by the mean-field
equations, and then in norm by a renormalized Bogoliubov evolution.\\

\noindent\textbf{Mean-field approximation.}
We are interested in the evolution of many-body states in which the particles
form a Bose--Einstein condensate and the field is approximately in a coherent
state. To be more precise, let us define the unitary Weyl operator
\begin{align}
  \label{eq: Weyl operator}
  W(f) = \exp \left(  \int_{\mathbb{R}^3} dk \, \big( f(k) a^*_k - \overline{f(k)} a_k \big) \right).
\end{align}
The initial states we have in mind are of the form
\begin{align}
  \label{eq:many-body initial states}
  \Psi_{N} \approx  u^{\otimes N}\otimes W(\sqrt N \alpha) \Omega
\end{align}
with $\Omega$ being the vacuum in $\mathcal{F}$ and $u, \alpha \in
L^2(\mathbb{R}^3)$. We will show that this product-like structure is
preserved during the time evolution and that
\begin{align}
  \label{eq:many-body state time evolved approximation introduction}
  e^{-i t H_N } \Psi_{N} \approx u_t^{\otimes N}\otimes W(\sqrt N \alpha_t) \Omega,
\end{align}
where $(u_t, \alpha_t) \in L^2(\mathbb{R}^3) \oplus L^2(\mathbb{R}^3)$ solve
the Schr\"odinger--Klein--Gordon (SKG) equations
\begin{align}
  \label{eq:Schroedinger-Klein-Gordon equations intro}
  \begin{cases}
    \begin{aligned}
      i \partial_t u_t(x) & \, =\, \big( - \Delta + \phi_{\alpha_t}(x) - \tfrac{1}{2} \scp{u_t}{ \phi_{\alpha_t} u_t}_{L^2} \big) u_t(x) \\[1.5mm]
      i \partial_t \alpha_t(k) &\, =\, \omega(k) \alpha_t(k) + \scp{u_t}{G_{(\cdot)}(k) u_t}_{L^2} \\[1.5mm]
    \end{aligned}
  \end{cases}
\end{align}
where
\begin{align}
  \label{eq:definition of G-x}
  G_x(k) &= \frac{1}{\sqrt{\omega(k)}} e^{- i k x}
  \quad \text{and} \quad
  \phi_{\alpha} (x) = 2 \Re\scp{G_x}{\alpha}_{L^2(\mathbb{R}^3)} .
\end{align}
These equations are the Hamiltonian equations of the energy
\begin{align}
  \label{eq:energy functional Schroedinger-Klein-Gordon equations regular}
  \mathcal{E}(u, \alpha) & = \scp{u}{\left( - \Delta + \phi_{\alpha} \right) u}_{L^2(\mathbb{R}^3)} + \scp{\alpha}{\omega \alpha}_{L^2(\mathbb{R}^3)} .
\end{align}

We denote the flow of solutions to the SKG equations by
$\mathfrak{s}[t](u,\alpha) = (u_t,\alpha_t)$, that is $(u_t,\alpha_t)$ solves
\eqref{eq:Schroedinger-Klein-Gordon equations intro} with initial conditions
$(u_t,\alpha_t)|_{t=0} = (u,\alpha)$. In our main results, we will use that
$\mathfrak s$ is well defined on $H^3(\mathbb R^3) \oplus \mathfrak h_{5/2}$
where $H^3$ denotes the $L^2$-Sobolev space and $\mathfrak h_{5/2}$ is the
weighted space with norm $\|\alpha\|_{\mathfrak{h}_{5/2}} = \| \omega^{5/2}
\alpha\| _{L^2(\mathbb{R}^3)} $. Moreover, the SKG flow conserves the energy
$\mathcal E$ and the $L^2$-norm of $u_t$. These statements follow from the
more general well-posedness result summarized in Proposition \ref{prop:skg
  well posed}.

Our first result states that the one-particle reduced density matrices of a
state $e^{-i t H_N } \Psi_N$ are close to those of the product
state~\eqref{eq:many-body state time evolved approximation introduction} if
this holds at the initial time and if the energy expectation per particle of
$\Psi_N$ is close to the initial mean-field energy. A convenient measure for
the convergence of reduced densities is given by the functional
\begin{align}\label{eq:definition of beta}
  \beta \left[ \Psi_N, ( u , \alpha ) \right]
  &= \scp{\Psi_N}{(q_u)_1 \Psi_N}_{\mathcal{H}_N} 
  +  N^{-1} \|\mathcal{N}_a^{1/2} W^* (\sqrt{N} \alpha)\Psi_N\|^2_{\mathcal{H}_N} ,
\end{align} 
where $(q_u)_1$ denotes the orthogonal projection $q_u = \id - |u\rangle
\langle u |$ on $L^2(\mathbb{R}^3)$ acting on the first particle's variable
$x_1$, and $\mathcal N_a = \int_{\mathbb R^3} dk\, a_k^* a_k$ is the number
operator on $\mathcal F$.  The functional $\beta$ counts the number of
particles in states orthogonal to $u$ and the number of field modes outside
of the coherent state $W(\sqrt{N}\alpha)\Omega$, both relative to the total
number $N$ of particles.  In particular, $\beta[\Psi_N, ( u, \alpha ) ]=0$
for states of the product form~\eqref{eq:many-body initial states}.

\begin{theorem}
  \label{theorem:final reduced density matrices} Let $(u, \alpha) \in
  H^3(\mathbb R^3) \oplus \mathfrak h_{5/2} $ with
  $\norm{u}_{L^2(\mathbb{R}^3)} =1 $ and let $\mathfrak s[t](u,\alpha) =
  (u_t,\alpha_t)$ denote the solution of \eqref{eq:Schroedinger-Klein-Gordon
    equations intro} with initial data $(u,\alpha)$. Then, there exists a
  constant $C>0$ such that for all $N\ge 1$, $t\in \mathbb R$ and $\Psi_N\in
  D(H_N^{1/2})$ with $\norm{\Psi_N}_{\mathcal{H}_N} =1$,
  \begin{multline*}
    \beta \left[ e^{-i t H_N }\Psi_N , \mathfrak s[t] (u ,\alpha ) \right]  \le
    e^{C R (t) } \Big(
    \big| N^{-1} \scp{\Psi_N}{H_N \Psi_N}_{\mathcal{H}_N} - \mathcal{E}(u , \alpha) \big|   \notag\\
    + \max_{ j =1,2} \big( \beta \left[ \Psi_N, ( u , \alpha ) \right] + N^{-1} \big)^{1/j}  \Big)
  \end{multline*}
  with $R(t) = 1 + \int_0^{|t|} (\|u_s\|_{H^3(\mathbb{R}^3)}^{10} +
  \|\alpha_s\|_{\mathfrak{h}_{5/2}}^{10}) ds$, and $\mathcal E$ defined by
  \eqref{eq:energy functional Schroedinger-Klein-Gordon equations regular}.
\end{theorem}

The proof of the theorem is given in Section \ref{sect:final reduced
  density}. We note that $R(t)$ does not grow faster than polynomially in
time by Proposition \ref{prop:skg well posed}.  Initial states of interest
are of course those for which the right-hand side is small.  It is important
to mention that this is not the case if $\Psi_N$ is exactly of the product
form~\eqref{eq:many-body initial states}. Indeed, due to the singular nature
of the Nelson Hamiltonian such states are not in the form domain of $H_N$,
see~\cite{GW2018,LaSch19}. Next, we provide an example of initial states in
the form domain of the Nelson Hamiltonian that are close to product
states. To this end, we modify the large momenta by means of a dressing
transformation.
\begin{proposition}
  \label{proposition:class of suitable initial states}
  For $(u,\alpha)\in L^2(\R^3) \oplus L^2(\R^3)$, $K\geq 0$, let $B_{K,x} =
  (k^2+\omega)^{-1}G_{K,x} \id_{| k | \ge K}$ and define
  \begin{equation}
    \Psi_{N,K}:=\prod_{j=1}^N W^* ( N^{-1/2}  B_{K,x_j} )\big( u^{\otimes N} \otimes W(\sqrt{N} \alpha) \Omega \big).\notag
  \end{equation}
  There exists a constant $C >0$ such that for all $(u, \alpha) \in
  H^1(\mathbb{R}^3) \oplus \mathfrak{h}_1$ with $\norm{u}_{L^2(\mathbb{R}^3)}
  = 1$, $K > 0$ and $N\geq1$
  \begin{align*}
    \beta[\Psi_{N,K}, (u , \alpha ) ] &\leq C  K^{-1} \left( 1 + \norm{\alpha}_{L^2(\mathbb{R}^3)} \right)
    \\
    \big| N^{-1} \scp{\Psi_{N,K}}{H_N \Psi_{N,K}} - \mathcal{E}(u, \alpha) \big|   & \leq C  \left(   K^{-1} +  \tfrac{1+ \ln K}{N}  \right) \big( \norm{u}_{H^1(\mathbb R^3)}^2 + \norm{\alpha}_{\mathfrak{h}_1}^2 \big) .
  \end{align*}
\end{proposition}
The proof is given in Appendix~\ref{sect:initial states prop}. Note that the
dressing transformation used above converges strongly to the identity as
$K\to \infty$, and thus $\Psi_{N,K}$ for large $K$ is close to a product
state also in the norm topology. For initial states $\Psi_{K,N}$ with $K=N$,
Theorem \ref{theorem:final reduced density matrices} simplifies to the
following form.
\begin{corollary}\label{cor:mean-field}
  Let $(u, \alpha) \in H^3(\mathbb R^3) \oplus \mathfrak h_{5/2} $ with
  $\norm{u}_{L^2(\mathbb{R}^3)} =1 $ and let $\Psi_{N,N}$ be the state
  defined in Proposition~\ref{proposition:class of suitable initial states}
  for $K=N$.  Then, there exist constants $C>0$, $\delta>0$ so that
  \begin{align*}
    & \beta \left[ e^{-i  t H_N }\Psi_{N,N} , \mathfrak s[t](u , \alpha)  \right]  \le 
    e^{C (1+ | t |^\delta)} N^{-1/2}.
  \end{align*}
\end{corollary}
\begin{proof}
  The corollary is a direct consequence of Theorem~\ref{theorem:final reduced
    density matrices}, Proposition \ref{proposition:class of suitable initial
    states} and the bound on the solutions to the SKG equations of
  Proposition~\ref{prop:skg well posed}.
\end{proof}

\begin{remark}[Convergence of reduced densities]\label{rem:reduced:densities}
  We briefly explain how Theorem~\ref{theorem:final reduced density matrices}
  relates to the approximation of reduced densities. To this end, recall the
  definition of the reduced one-particle density matrix for the bosons,
  \begin{align}
    \gamma_{\Psi_N}^{(1,0)} =  N \, \tr_{2, \ldots, N} \otimes \tr_{\mathcal{F}} \big( \ket{\Psi_N} \bra{\Psi_N} \big)
  \end{align}
  where $\tr_{2, \ldots, N}$ is the partial trace w.r.t. $(x_2,\ldots,x_N)$,
  and the definition of the reduced density of the field, given in terms of
  its integral kernel
  \begin{align}
    \gamma_{\Psi_N}^{(0,1)}(k, l ) =  \scp{\Psi_N}{a_{ l  }^* a_k \Psi_N}_{\mathcal H_N}.
  \end{align}
  Their distance, measured in trace norm, to the density operators obtained
  from the solutions of the SKG equations can be controlled by $\beta$ from
  \eqref{eq:definition of beta} via the inequalities
  \cite[Lem. VII.2]{LP2018}
  \begin{subequations}
    \begin{align}
      \tr \big| \gamma_{\Psi_N}^{(1,0)} -  N\ket{u } \bra{u}\big|
      &\leq N \sqrt{ 8 \scp{\Psi_N}{ (q_u)_1 \Psi_N}_{\mathcal{H}_N} },\\
      \tr \big| \gamma_{\Psi_N}^{(0,1)} - N \ket{\alpha } \bra{\alpha}\big|
      &\leq 3 \| \mathcal{N}_a^{1/2} W^* (\sqrt{N} \alpha)\Psi_N\|^2_{\mathcal{H}_N }  \\
      &\qquad + 6 \norm{\alpha}_{L^2(\mathbb R^3) } \| \mathcal{N}_a^{1/2} W^* (\sqrt{N} \alpha)\Psi_N\|_{\mathcal{H}_N}\notag .
    \end{align}
  \end{subequations}
  For $\Psi_N(t)=e^{-i t H_N }\Psi_N$ and $(u_t$, $\alpha_t)$, we
  consequently get a bound on the difference of reduced densities in terms of
  the right-hand side of the bound in Theorem~\ref{theorem:final reduced
    density matrices}.  That is, for suitable initial states (such as those
  of Corollary~\ref{cor:mean-field}) the average boson behaves like $u_t$ and
  there are on average $N$ field modes behaving like $\alpha_t$.
\end{remark}
\vspace{2mm}

\noindent\textbf{Bogoliubov approximation.} In the following we define the renormalized Nelson--Bogoliubov evolution and explain its role in approximating the fluctuations in the Nelson dynamics.
We gather the quantum fluctuations around the condensate with wave function
$u \in L^2(\mathbb R^3)$ and the coherent state associated with the field
$\sqrt N \alpha \in L^2(\mathbb{R}^3)$ in an element $\chi \in
\mathcal{F}\otimes \mathcal{F}$ that is orthogonal to $u$ in every variable
$x_1, \dots, x_N$. That is,
\begin{equation}\label{eq:excitation space intro}
  \chi \in \bigoplus_{k=0}^\infty \bigotimes_{\rm sym}^k \{u \}^\perp  \otimes \mathcal{F} =: \mathcal{F}_{\perp u} \otimes \mathcal{F}. 
\end{equation}
For any $\Psi_N \in \mathcal{H}_N $ one obtains such a $\chi := X_{u,\alpha}
\Psi_N$ using a variant of the excitation map introduced in~\cite{LNSS2015},
see Section~\ref{sect:excitation space} for details.  To describe the inverse
of this map, let $\chi^{(k)}$ denote the component of $\chi$ in the $k$-th
summand of~\eqref{eq:excitation space intro}. If $\chi^{(k)} =0$ for $k>N$,
we can reconstruct $\Psi_N$ as
\begin{equation}
  \label{eq:Psi_N excitation}
  \Psi_N = W(\sqrt N \alpha ) \sum_{k=0}^N u ^{\otimes N- k} \otimes_{\rm s} \chi^{(k)} = X_{u,\alpha}^*\chi ,
\end{equation}
where the symmetric tensor product has to be understood as the tensor product
of the subspaces of $L^2(\R^3)$, $\mathrm{span}(u)$ and $\{u \}^\perp$, so
that each summand yields an element of $ \bigotimes_{\rm sym}^N L^2(\R^3)
\otimes \mathcal{F}$, on which the Weyl operator $W(\sqrt{N}\alpha)$ acts in
the second tensor factor.\footnote{More precisely, we set $ u ^{\otimes N- k}
  \otimes_{\rm s} \chi^{(k)} := P_{\rm sym}^N ( u ^{\otimes N- k} \otimes
  \chi^{(k)})$, where the tensor product is taken w.r.t. the spaces
  $\text{Span}( u )^{\otimes N-k}$ and $(\{ u\}^\perp )^{ \otimes k} \otimes
  \cF $ and where $P^N_{\rm sym}$ is the orthogonal projection onto the
  symmetric subspace of $\bigotimes^N L^2(\mathbb R^3)$, while it acts as the
  identity on $\mathcal F$.} Note that the product state~\eqref{eq:many-body
  initial states} would correspond to $\chi =\Omega\otimes \Omega$.

Now let $\Psi_N(t)$ and $(u_t, \alpha_t)$ be solutions of
\eqref{eq:Schroedinger equation} and \eqref{eq:Schroedinger-Klein-Gordon
  equations intro} with suitable initial conditions $\Psi_N$ and
$(u,\alpha)$, and consider the fluctuation vector $\chi(t)$ satisfying
$\Psi_N(t) = X_{u_t,\alpha_t}^*\chi(t)$. This vector is an element of
$\cF_{\perp u_t} \otimes \cF$, which we can naturally identify with a
subspace of the double Fock space
\begin{equation}
  \mathcal{F}\otimes \mathcal{F} \cong \bigoplus_{n=0}^\infty \bigotimes_{\rm sym}^n \big( L^2(\R^3)\oplus L^2(\R^3)\big) .
\end{equation}
We will denote the creation and annihilation operator on the first factor in
$\cF \otimes \cF$, associated with excitations of the bosons, by $b^*$ and
$b$.

The dynamics of the fluctuations $\chi(t)$ will be approximated by a time
dependent Bogoliubov transformation. Roughly speaking, Bogoliubov
transformations on $\cF \otimes \cF$ are unitary maps that are determined (up
to a phase) by a map on $L^2(\mathbb R^3) \oplus L^2(\mathbb R^3)$. This
property makes Bogoliubov transformations much simpler in terms of
complexity. To be more precise, define the joint creation operator of the
excitations and the field by
\begin{equation}
  c^* (f\oplus g)=b^*(f) + a^*(g),
\end{equation}
and the annihilation operator as its adjoint. A Bogoliubov transformation on
$\mathcal{F}\otimes\mathcal{F}$ is a unitary map $\mathbb{U}$ with the
property that
\begin{equation}\label{eq:bogoliubov:transformation}
  \mathbb{U}^* c^*(f\oplus g) \mathbb{U} = c^*(\mathfrak{u}( f\oplus g) ) + c( \mathfrak{v}(\overline {f\oplus g} ) ) 
\end{equation}
for some bounded linear maps $\mathfrak{u},\mathfrak{v}: L^2(\R^3) \oplus
L^2(\R^3)\to L^2(\R^3)\oplus L^2(\R^3)$.  In other words, conjugation of
$a^*$, $b^*$ by $\mathbb{U}$ maps these to linear combinations of $a, a^*, b,
b^*$ with modified arguments. For a more detailed introduction of Bogoliubov
transformations and related concepts, we refer to \cite{JPS2007} and
\cite[Sec. 4]{Bossmann2019}.

The generators of Bogoliubov transformations are (formally) Hamiltonians
quadratic in the creation and annihilation operators, $a,a^*, b, b^*$. For
the Nelson model with ultraviolet cutoff, such a quadratic generator can be
obtained from the full Hamiltonian~\cite{Falconi21} following the
approximation ideas of Bogoliubov~\cite{Bogoliubov1947}. The
Nelson--Bogoliubov Hamiltonian with cutoff $\Lambda \in (0,\infty)$ is
defined by
\begin{align}
  \label{eq:reg:Nelson:Bog}
  \mathbb H^\Lambda_{u,\alpha}(t) &  = \int_{\mathbb{R}^3} dx\, b^{*}_x  h_{\alpha_t} b_x    + \int_{\mathbb{R}^3} dk\, \omega(k) a^{*}_k a_k  
  \\
  &\, +   \int_{\mathbb{R}^6} dx  dk  \Big(\big(q_{u_t}  G^\Lambda_{(\cdot)}(k) u_t\big)(x)\, a^{*}_k b_x^*  + \big(q_{u_t}  \overline{G^\Lambda_{(\cdot)}(k)} u_t\big)(x)\, a_{k} b_x^* \Big) +\text{h.c.} ,\notag
\end{align}
where $G_x^{\Lambda}(k) = G_x (k) \id_{\abs{k} \leq \Lambda}$ and
$h_{\alpha_t} = - \Delta + \phi_{\alpha_t} - \frac{1}{2} \scp{u_t}{
  \phi_{\alpha_t} u_t} $ with $G_x(k) $ and $\phi_{\alpha_t}$ given by
\eqref{eq:definition of G-x}, and where $\text{h.c.}$ denotes the hermitian
conjugate of the preceding term. Let $\mathbb{U}^\Lambda(t)$ be the unique
unitary propagator (with initial time $t=0$) on the double Fock space
$\cF\otimes\cF$ associated with $\mathbb H_{u,\alpha}^\Lambda(t)$. For a
discussion of its existence, we refer to \cite[Thm.~4.1]{Falconi21} or
Proposition \ref{prop:theta-Bog} below with $\theta =0$.  Our next result
states the existence of a renormalized Nelson--Bogoliubov time evolution in
the limit $\Lambda\to \infty$.  We note that the numbers $E^\Lambda$ below,
as given explicitly in Proposition \ref{prop:id:Bog:ren}, have the same
asymptotic behavior as those in the renormalization of the Nelson
Hamiltonian~\eqref{eq:def:ren:Nelson}.

\begin{theorem}\label{theorem:limit:unitary} Let $(u,\alpha) \in H^3\oplus
  \mathfrak h_{5/2}$ with $\norm{u}_{L^2(\mathbb{R}^3)} =1$ and let
  $(u_t,\alpha_t)$ denote the solution of \eqref{eq:Schroedinger-Klein-Gordon
    equations intro} with initial datum $(u,\alpha)$. There exists a family
  $E^\Lambda$ (with $E^\Lambda \to \infty$ as $\Lambda \to \infty$) such that
  \begin{align*}
    \mathbb{U} (t) := \slim_{\Lambda\to \infty} \mathbb{U}^\Lambda(t ) e^{- i t E^\Lambda }
  \end{align*}
  exists for all $ t \in \mathbb R$. Moreover, $\mathbb{U}(t )$ has the
  following properties:
  \begin{enumerate}
  \item[\textnormal{(i)}] $\mathbb{U}(t)$ is unitary and strongly continuous
    in $t$,
  \item[\textnormal{(ii)}] $\mathbb{U}(t) ( \mathcal{F}_{\perp u} \otimes
    \mathcal{F} ) \subseteq \mathcal{F}_{\perp u_t} \otimes \mathcal{F}$,
  \item[\textnormal{(iii)}] $\mathbb{U}(t)$ is a Bogoliubov transformation on
    $\cF \otimes \cF$.
  \end{enumerate}
\end{theorem}

Our second main result is a norm approximation of the dynamics generated by
the Nelson Hamiltonian. It states that the fluctuations around the condensate
wave function $u_t \in L^2(\mathbb R^3)$ and the coherent state associated
with the field mode $\sqrt N \alpha_t \in L^2(\mathbb{R}^3) $ are effectively
described by the renormalized Nelson--Bogoliubov evolution introduced in
Theorem \ref{theorem:limit:unitary}. Together with the fact that $\mathbb
U(t)$ is a Bogoliubov transformation, this implies an approximation of the
Nelson evolution in terms of a transformation of the form
\eqref{eq:bogoliubov:transformation}.

\begin{theorem}\label{thm:norm:approximation} Let $(u,\alpha) \in H^3(\mathbb
  R^3)\oplus \mathfrak h_{5/2}$ with $\norm{u}_{L^2(\mathbb{R}^3)} =1$ and
  let $(u_t,\alpha_t)$ denote the solution to
  \eqref{eq:Schroedinger-Klein-Gordon equations intro} for initial data $(u
  ,\alpha )$.
  There exists a quadratic form $\delta \ge 1$ whose domain is dense in
  $\cF_{\perp u}\otimes \cF$ and a constant $C>0$ so that for all $t\in
  \mathbb R$, $N\ge 1$ and $\chi \in \mathcal{F}_{\perp u} \otimes
  \mathcal{F}$ with $\norm{\chi}_{\cF\otimes \cF} =1$, we have for
  $\Psi_N=X_{u,\alpha}^*\chi$ given by~\eqref{eq:Psi_N excitation}
  \begin{equation*}\label{eq:main:bound:norm}
    \Big\| e^{-i t H_N } \Psi_N  - W(\sqrt N \alpha_t) \sum_{k=0}^N u_t^{\otimes N-k} \otimes_{\rm s} (\mathbb{U}(t)\chi)^{(k)} \Big\|_{\mathcal{H}_N} \\
    \le e^{C R(t)}  \delta(\chi)^{1/2} \frac{\sqrt {\ln N}}{N^{1/4}} ,
  \end{equation*}
  where $R(t) = 1+ \int_0^{|t|} (\| u_s\|_{H^3(\mathbb{R}^3)}^{10} + \|
  \alpha_s \|_{\mathfrak h_{5/2}}^{10}) ds $.
\end{theorem}

This theorem is proved in Section~\ref{sect:proof:thm:norm}.

\begin{remark}\label{rem:delta}
  The quadratic form $\delta$ is constructed explicitly using a unitary
  Bogoliubov transformation $\mathbb W $ that implements a dressing on the
  level of the fluctuation vector $\chi$. With this transformation, which is
  introduced in Proposition \ref{prop:Gross-Bog:evolution} as $\mathbb
  W^\infty_{u,\alpha}(1)$, $\delta$ is given by
  \begin{align}
    \delta(\chi) =  \norm{\chi}^2_{\cF\otimes \cF} + \scp{ \mathbb W \chi }{ ( \mathcal N^3 + \textnormal{d}\Gamma_b(-\Delta) + \textnormal{d} \Gamma_a(\omega))  \mathbb W  \chi}_{\mathcal{F} \otimes \mathcal{F}}  \, ,
  \end{align}
  i.e., it measures the expectation of the third moment of the number of
  excitations and the energy of the excitations and the field after dressing
  with $\mathbb W$.  In Proposition \ref{prop:Gross-Bog:evolution}, we also
  show that $\mathbb W$ preserves the domain of $\cN^{3/2}$. This is
  relevant, as it implies that the norm of the initial state $\Psi_N$ in
  \eqref{eq:main:bound:norm} approaches one as $N\to \infty$. More precisely,
  it implies that $\| X^*_{u,\alpha}\chi \|_{\mathcal H_N} \geq \| \chi
  \|_{\cF\otimes \cF}-CN^{-3}\delta(\chi)$ for some $C>0$. On the other hand,
  let us note that we do not expect that $\mathbb W$ preserves the norm of
  $\textnormal{d}\Gamma_b(-\Delta)$.
\end{remark}

The following corollary extends the norm approximation of
Theorem~\ref{thm:norm:approximation} to all initial states with finite number
of excitations, in the sense that $X_{u,\alpha}\Psi_N$ has a well-defined
limit as $N\to \infty$.

\begin{corollary}
  Let $(u,\alpha) \in H^3(\mathbb R^3)\oplus \mathfrak h_{5/2}$ with
  $\norm{u}_{L^2(\mathbb{R}^3)}=1$ and let $(u_t,\alpha_t)$ denote the
  solution to \eqref{eq:Schroedinger-Klein-Gordon equations intro} for
  initial data $(u ,\alpha )$.  Let $X_{u_t, \alpha_t}^*$ be the adjoint of
  the excitation map, given by~\eqref{eq:Psi_N excitation}.  Let $\chi \in
  \mathcal{F}_{\perp u}\otimes \mathcal{F}$ with $\|\chi \|_{\mathcal{F}
    \otimes \mathcal{F}}=1$ and let $\Psi_N$, $N\geq 1$, be such that
  $\lim_{N\to \infty} \norm{\Psi_N -X_{u,\alpha} ^*\chi}_{\mathcal{H}_N}
  =0$. Then, for all $T>0$
  \begin{equation*}
    \lim_{N\to \infty} \|  e^{-i t H_N } \Psi_N  - X_{u_t,\alpha_t}^* \mathbb{U}(t)\chi \|_{\mathcal{H}_N} =0
  \end{equation*}
  uniformly for $|t|\leq T$.
\end{corollary}
\begin{proof}
  This follows from the density of the domain of the quadratic form $\delta$
  by an approximation argument, using that $e^{-i t H_N} $, $\mathbb{U}(t)$
  and $X_{u,\alpha}$ are isometries. To be precise, let $\varepsilon>0$ and
  choose $N$ large enough so that $\Psi_N -X^*_{u,\alpha}\chi$ has norm less
  than $\varepsilon$.
  Then choose $\chi_{\varepsilon}\in \cF_{\perp u}\otimes \cF$ in the domain
  of the form $\delta$ with $\|\chi-\chi_{\varepsilon}\|_{\cF\otimes
    \cF}<\varepsilon$. The difference of $e^{-i tH_N}
  X_{u,\alpha}^*\chi_\varepsilon$ and $X_{u_t, \alpha_t}^*
  \mathbb{U}(t)\chi_\varepsilon$ converges to zero by
  Theorem~\ref{thm:norm:approximation}, so its norm is smaller than
  $\varepsilon$ for $N$ sufficiently large. Using unitarity of $e^{-i t
    H_N}$, $\mathbb U (t)$, and $\| X_{u,\alpha}^*\|= 1$, this implies that
  \begin{align}
    \|  e^{-i t H_N } \Psi_N  - X_{u_t,\alpha_t}^* \mathbb{U}(t)\chi \|_{\mathcal{H}_N}
    < 4\varepsilon,
  \end{align}
  which proves the claim.
\end{proof}

\subsection{Comparison with the literature}
\label{sect:comparison}
The broader subject of this work, the justification of the time-dependent
mean-field and Bogoliubov approximations, has been addressed extensively in
the literature, mainly in the context of the Bose gas with two-body pair
interaction. The situation of particles coupled to a quantum field has been
explored to a much lesser extent. Below we give a brief overview of the
literature on this topic and other works related to this study.

The first works on the mean-field approximation of reduced densities for the
many-body Bose gas with two-body interaction date back to the 1970s and 1980s
by Hepp, Ginibre, Velo, and Spohn \cite{Hepp74,GV1979b,GV1979a,Spohn80}. The
question was revived in the early 2000s \cite{Bardos00,ES01} and within the
next years, new techniques were developed to obtain explicit rates of
convergence \cite{Chen2011,pickl,RS2009} and to cover more singular two-body
potentials, in particular those converging to a Dirac-delta potential
\cite{ESY07,ESY10,KP2010,Pickl2015,BOS2015}.  Since then, this topic
continues to be actively studied and we recommend \cite{Narpiokowski23,
  BPS16, Golse16} for a comprehensive survey of recent works. Fluctuations
around the time-dependent mean-field equations were considered first in
\cite{GV1979b,GV1979a,GMM2011,GMM2010}. Since \cite{Lewin:2015a,LNSS2015},
this subject has gained increased interest which led to further extensions
and refinements in the derivation of the Bogoliubov approximation, see
e.g. \cite{DFPP,nam:2015,nam:2016,brennecke:2017,namnap_review,Boccato:2016,mpp,soffer,Chong2016,NS2020,Kuz2017,
  namnap_low_dim,GM2017}. Higher-order corrections to Bogoliubov theory have
been obtained in \cite{BPPA20,Bossmann2019, Ginbre_Velo_expansion1,
  Ginbre_Velo_expansion2}. Let us note that Bogoliubov theory plays a crucial
role also in the description of the excitation spectrum of large bosonic
systems. While this has been extensively studied for bosons with two-body
potentials \cite{Nam-Triay,Nam-Seiringer,HST22,GrechS13,Seiringer11,BBCS19,
  derezinski14,Boccato:2020,LNSS2015,Pizzo:iii}, we are not aware of any
results concerning the spectral properties of many bosons coupled to a
quantum field.

The derivation of the SKG equations starting from the renormalized Nelson
model, in the same limit as considered in this work, has been addressed
previously by \cite{AF2017}. Using techniques from semiclassical theory, the
authors demonstrate that the Wigner measure associated with the many-body
dynamics evolves in the limit $N\to \infty$ in accordance with the
push-forward of a Wigner measure under the SKG flow. Since convergence of the
Wigner measure implies weak-$\ast$ convergence for the reduced densities,
this statement is comparable to Theorem \ref{theorem:final reduced density
  matrices} of the present work. Unlike the approach taken in \cite{AF2017},
which provides a limit result without explicit error estimates, our method
allows us to determine an explicit rate of convergence for the reduced
densities. On the other hand, the results in \cite{AF2017} apply to a wider
class of initial states. Regarding our second result, the construction of the
renormalized Nelson--Bogoliubov Hamiltonian and the norm approximation, we
are not aware of any prior work that has addressed this problem.

More results have been obtained for models with regular particle-field
interactions (i.e., without need for renormalization): For the regularized
Nelson model with ultraviolet cutoff, derivations of the corresponding
mean-field dynamics were obtained in \cite{AF2014, falconi, LP2018} and the
validity of the Bogoliubov approximation as well as higher-order corrections
was established in \cite{Falconi21}. In addition, the regularized Nelson
model was studied also in a many-fermion limit that is closely linked to a
semiclassical limit \cite{LP2019}. Other particle-field systems, such as the
Fr\"ohlich model and the Pauli--Fierz Hamiltonian, have been studied in the
scaling regime of the present article too, see \cite{LMS2021,LP2020} for the
mean-field approximation and \cite{L2022} for an approximation of the
Fr\"ohlich dynamics in norm. The dressed Nelson Hamiltonan, which will play a
crucial role in our analysis (see Section \ref{sec:outline}), has similar
regularity properties to the Fr\"ohlich Hamiltonian, as both are given in
terms of perturbations of the non-interacting quadratic form. However, the
dressed Nelson Hamiltonian has a more complicated structure than the
Fr\"ohlich Hamiltonian since it is not linear in creation and annihilation
operators. This makes the analysis of the time evolution more involved
already on the level of the dressed Hamiltonian.

Within the broader scope of deriving effective equations from particle-field
models, it is worth noting the following works. The subject of
\cite{CCFO2021, CF2018, CFO2019, GNV06} is a partially classical limit of a
class of models (covering the regularized Nelson, and the Fr\"ohlich and
Pauli--Fierz models), where a fixed number of particles is weakly coupled to
a quantum scalar field with high occupation number. For the Fr\"ohlich
Hamiltonian specifically, the time evolution has been actively studied also
in the strong coupling regime \cite{FG2017, FS2014, G2017,LMRSS2021,
  LRSS2019, M2021}. While the resulting effective equations are of similar
form as the SKG equations, the strong coupling limit is accompanied by a
separation of time scales between the particle and the field, a feature that
is absent in the mean-field limit. The papers \cite{davies,
  hiroshima1998,teufel} focus on the derivation of effective pair particle
potentials arising from the particle-field interaction, in suitable
weak-coupling and adiabatic limits.

Finally, for an overview of results on the renormalized Nelson model not
directly linked to the derivation of effective equations, we refer to the
discussion of \cite{M2018}.

\subsection{Outline of the proofs}\label{sec:outline}

\textbf{General idea.} The Hamiltonian expressed formally in~\eqref{eq:Nelson
  Hamiltonian formal definition} can be represented in terms of an operator
$H_N^\mathrm{D}$ with more regular and explicitly given quadratic form,
conjugated with a unitary dressing transformation $W^\mathrm{D}$, see
Lemma~\ref{lem:H:D:representation}. By unitarity of $W^{\rm D}$, this allows
us to relate the Nelson dynamics to the dynamics generated by the dressed
Hamiltonian via
\begin{align}\label{eq:relation:dressed:undressed:into}
  e^{- i t H_N} =  ( W^{\rm D})^*  e^{- i t H_{N}^{\rm{D}}} W^{\rm{D}}.
\end{align}

The general strategy of the proof is to analyze the mean-field and norm
approximations of the dressed time evolution $e^{- i t H_{N}^{\rm{D}}}$, and
then connect the corresponding dressed mean-field and Bogoliubov evolutions
to the original (undressed) ones. To accomplish this, we will introduce
approximations of the dressing transformation that relate the dressed and
undressed effective evolutions, in analogy to the relation shown in
\eqref{eq:relation:dressed:undressed:into}. Denoting the dressed mean-field
flow and the dressed Bogoliubov evolution as $\mathfrak s^{\rm D}[t]$ and
$\mathbb U^{\rm D}(t)$, respectively, and the approximations of the dressing
transformation by $\mathfrak D$ and $\mathbb W$, then the connection between
the effective evolutions can be expressed as
\begin{align}\label{eq:effective:flows:intro}
  \mathfrak s [t] = \mathfrak D^{-1} \circ  \mathfrak s^{\rm D}[t] \circ \mathfrak D ,\qquad   \mathbb U(t) =  \mathbb W^* \mathbb U^{\rm D}(t) \mathbb W.
\end{align}
To determine $\mathfrak D$ and $\mathbb W$, we view $W^{\rm D}=W^{\rm D}(1)$
as the special case of a quantum evolution operator $W^{\rm D}(\theta)$ with
``time'' $\theta$ and examine its mean-field and Bogoliubov
approximations. The motivation for this stems from the observation that
$W^{\rm D}(\theta)$ is a unitary group that is generated by a field operator
resembling the interaction term in \eqref{eq:Nelson Hamiltonian formal
  definition}, but with a square-integrable form factor replacing
$\omega^{-1/2}$. While the mean-field flow approximation $\mathfrak
D[\theta]$ has been proposed and studied previously in \cite{AF2017}, we
extend this idea to the level of the Bogoliubov approximation. One of the
difficulties that arises in this context is that the effective dressing
$\mathbb W(\theta)$ will be generated by a non-autonomous equation.

In the proof of Theorem \ref{theorem:limit:unitary}, we establish an identity
similar to~\eqref{eq:effective:flows:intro} but for $\mathbb
U^\Lambda(t)e^{-i t E_\Lambda} $ and with $\Lambda$-dependent versions of
$\mathbb W$ and $\mathbb U^{\rm D}(t)$, and then use that the cutoff can be
removed for the conjugated dressed evolution.  \medskip

\noindent \textbf{Mean-field approximation.} In order to derive
Theorem~\ref{theorem:final reduced density matrices}, we consider the dressed
dynamics $e^{ - i t H_N^\mathrm{D} }$ applied to the dressed initial state
$W^{\rm D}\Psi_N$, and compare it with the corresponding mean-field equations
introduced in~\eqref{eq:classical:transformed:equations}, whose flow is
denoted by $\mathfrak s^{\rm D}[t]$. This is the content of Theorem
\ref{thm:gross-transformed dynamics reduced density matrices}, which gives an
analogous statement to Theorem~\ref{theorem:final reduced density matrices}
but for the dressed evolutions. The proof of
Theorem~\ref{thm:gross-transformed dynamics reduced density matrices} relies
on the use of the excitation map and estimates on the generator of the
fluctuation dynamics. To relate the two theorems, i.e., to pass from the
approximation of $e^{-i t H_N^\mathrm{D} }$ to that of $e^{-i t H_N}$, we
then expand on the idea that the dressing $W^{\rm D}$ itself can be
approximated by the mean-field dressing transformation $\mathfrak D$. This is
the subject of Lemmas~\ref{lemma:reduced densities dressing flow} and
\ref{lem:beta:gamma:relation}, with the latter providing the reason for the
required energy condition in Theorem \ref{theorem:final reduced density
  matrices}. Since $\mathfrak D$ interpolates between the dressed and the
undressed mean-field evolutions, i.e. $ \mathfrak s[t] = \mathfrak D^{-1}
\circ \mathfrak s^{\rm D}[t] \circ \mathfrak D $, this allows us to translate
the approximation result of the dressed dynamics to the desired result on the
undressed ones.  The explained strategy is summarized in the commutative
diagram of Figure~\ref{fig:1}.
\begin{figure}[t]
  \centering
  \begin{tikzcd}[row sep=0.7cm, column sep=0.2cm, arrows={crossing over}]
    &&&&&&&\mathfrak D(u,\alpha)
    \arrow[from=3-8,shorten=1mm,mapsto]\arrow[to=1-16,shorten=2mm, mapsto]
    &&&&&&&&\mathfrak s^{\rm D}[t]\circ \mathfrak D(u,\alpha) \arrow[to=3-16,
    ,shorten=1mm, mapsto, "\mathfrak D^{-1}" near start]  &\\
    &W^{\rm D}\Psi_N \arrow[to=1-8,shorten=1mm,"N\to
    \infty"]\arrow[to=2-15,,shorten=2mm,mapsto]&&&&&&&&&&&&& e^{-itH_N^{\rm D}}W^{\rm D} \Psi_N  \arrow[to=1-16,shorten=1mm,"N \to \infty " ]  &&\\
    &&&&&&& (u,\alpha) \arrow[to=3-16,shorten=2mm,mapsto] &&&&&&&&\mathfrak s[t](u,\alpha) &\\
    &{\color{white}{a}}{\Psi_N}{\color{white}{a}}
    \arrow[to=2-2,shorten=1mm,mapsto] \arrow[to=3-8,shorten=1mm,
    "\mspace{25mu}N \to \infty " right, near start]
    \arrow[to=4-15,shorten=2mm,mapsto]&&&&&&&&&&&&&
    \arrow[from=2-15,shorten=1mm,mapsto, "{(W^{\rm D})^*}" near start]e^{-it
      H_N
    }\Psi_N \arrow[to=3-16,shorten=1mm,"\mspace{70mu}N\to\infty" below] &&\\
  \end{tikzcd}
  \vspace{-8mm}
  \caption{\tiny{Diagram schematizing Theorem~\ref{theorem:final reduced
        density matrices} and its relation with the dressed counterpart,
      Theorem~\ref{thm:gross-transformed dynamics reduced density
        matrices}. Each microscopic state on the front face is close to the
      associated mean-field state in the limit $N\to \infty$ (lying on the
      face beyond), as measured by the functional $\beta$. The arrows are
      kept unlabeled if the associated map is obvious.}}
  \label{fig:1}
\end{figure}
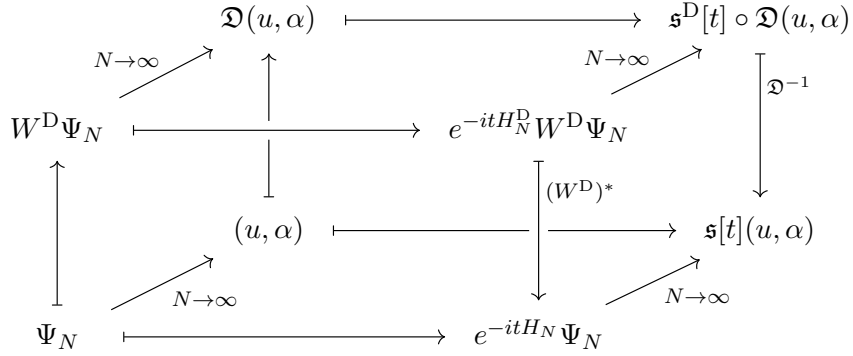

Since the proof of Theorem \ref{theorem:final reduced density matrices} is
based on Theorem~\ref{thm:gross-transformed dynamics reduced density
  matrices} for the dressed Hamiltonian and this diagram, the natural initial
condition would be a many-body state of the form
\begin{align}\label{eq:initial dressed product state}
  \Psi_N = (W^{\rm D})^*  (u^{\rm D})^{\otimes N}\otimes W(\sqrt N \alpha^{\rm D}) \Omega.
\end{align}
However, such states are not close to product states of the
form~\eqref{eq:many-body initial states} when measured by the full Hilbert
space norm (by Lemma~\ref{lemma:reduced densities dressing flow}, they are
close to such states when measured by $\beta)$. In Corollary
\ref{cor:mean-field} we show that Theorem~\ref{theorem:final reduced density
  matrices} also applies to states that are close to product states. These
are obtained from Proposition~\ref{proposition:class of suitable initial
  states}, which shows that it is sufficient to apply the dressing to momenta
larger than $K=N$ to solve the problem that exact product
states~\eqref{eq:many-body initial states} are not in the form domain of the
Nelson Hamiltonian.

\medskip

\noindent \textbf{Bogoliubov approximation.} The general strategy of the norm
approximation resembles the one of the mean-field approximation. That is, we
start again from the analysis of the dressed dynamics, where in analogy to
the discussion preceding Theorem \ref{thm:norm:approximation}, we now examine
the fluctuation vector $ \chi^{\rm D}(t) = X_{\mathfrak s^{\rm
    D}[t](u,\alpha) } e^{-i t H_N^{\rm D}} X_{u,\alpha}^* \chi $ associated
with the dressed mean-field flow and compare it with the effective evolution
$ \mathbb U^{\rm D}(t) \chi$. Here, $\mathbb U^{\rm D}(t)$ is obtained from
the Bogoliubov approximation of the dressed dynamics. The statement analogous
to Theorem~\ref{thm:norm:approximation} for the dressed dynamics is given in
Theorem~\ref{thm:norm:approximation:dressed:dynamics}, whose proof is based
on estimates on the difference of the generator of the dressed fluctuations
and the quadratic generator of $\mathbb U^{\rm D}(t)$.  To establish a
connection of this result to the one for the undressed dynamics, we use
\eqref{eq:relation:dressed:undressed:into} together with a norm approximation
for the dressing transformation $W^\mathrm{D}$. To this end, we elevate the
mean-field approximation $\mathfrak D$ to the level of the fluctuations by
implementing a Bogoliubov transformation $\mathbb{W}$. For the definition of
$\mathbb W$, we need to introduce a non-autonomous flow of Bogoliubov
transformations on $\cF\otimes \cF$, which is the content of Proposition
\ref{prop:Gross-Bog:evolution}. Lemma \ref{lemma:U:Lambda:approx:fluc} then
demonstrates that $\mathbb W$ indeed offers a norm approximation of $W^{\rm
  D}$, in the sense that $ X_{\mathfrak D (u,\alpha) } W^{\rm D}
X_{u,\alpha}^* \chi \approx \mathbb W \chi $ as $N\to \infty$ for suitable
$\chi \in \cF \otimes \cF$. As a final step, we argue that $\mathbb W$
interpolates between the dressed and the undressed Bogoliubov evolutions,
i.e. that $\mathbb U(t) = \mathbb W^* \mathbb U^{\rm D}(t) \mathbb W $, with
the precise version of this identity given in \eqref{eq:renormalized-Bog}.

Altogether, the argument for the norm approximation is based on the following
sequence of identities and approximations (which amount to the arrows in a
diagram similar to Figure~\ref{fig:1}) \allowdisplaybreaks
\begin{align}
  e^{-i t H_N}  X_{u,\alpha}^* \chi   & \ \  \overset{\eqref{eq:relation:dressed:undressed:into}}{=} \ \ (W^{\rm D})^* e^{-i t H_N^\mathrm{D} } W^{\rm D} X_{u,\alpha}^* \chi \notag\\[0mm]
  &\  \  \overset{{\color{white}{blank}}}{\approx } \ \ (W^{\rm D})^* e^{-i t H_N^\mathrm{D} } X^*_{\mathfrak D(u,\alpha)} \mathbb W \chi  \notag\\[0mm]
  &\  \  \overset{{\color{white}{blank}}}{\approx }  \ \ (W^{\rm D})^*  X_{\mathfrak s^{\rm D}[t]\circ \mathfrak D(u,\alpha) }^*  \mathbb U^{\rm D}(t) \mathbb W \chi \notag\\[0mm]
  &\  \ \overset{  \eqref{eq:effective:flows:intro} }{=} \ \ (W^{\rm D})^*  X_{\mathfrak D \circ \mathfrak s [t] (u,\alpha) }^*  \mathbb U^{\rm D}(t) \mathbb W \chi \notag\\[0mm]
  &\  \  \overset{{\color{white}{blank}}}{\approx }   \ \ X_{\mathfrak s[t](u,\alpha) }^*  \mathbb W^*   \mathbb U^{\rm D}(t) \mathbb W \chi \ \ \overset{\eqref{eq:effective:flows:intro}}{=}\ \  X_{\mathfrak s[t](u,\alpha) }^*    \mathbb U (t)  \chi.
\end{align}

\noindent \textbf{Renormalization.} For the purpose of the norm
approximation, we could take \eqref{eq:effective:flows:intro} as the
definition of $\mathbb{U}(t)$ in Theorem~\ref{thm:norm:approximation}.
However, we also want to elucidate the relation of this evolution to the one
that can be formally derived by applying the quadratic Bogoliubov
approximation to \eqref{eq:Nelson Hamiltonian formal definition}. This
relation is given by Theorem~\ref{theorem:limit:unitary}. To this end, one
needs to introduce an ultraviolet cutoff $\Lambda$, since it is not clear
that the Nelson--Bogoliubov Hamiltonian~\eqref{eq:reg:Nelson:Bog} defines a
self-adjoint operator for $\Lambda=\infty$. With a cutoff $\Lambda$, one
might expect that the Bogoliubov approximation of the Nelson dynamics is
given exactly by the Bogoliubov approximation of the dressed Hamiltonian,
conjugated with the approximation of the dressing. However, the Bogoliubov
evolution is fixed only up to a phase, so the identity may only hold for an
appropriate choice of such a phase. This is the content of
Proposition~\ref{prop:id:Bog:ren}, where we show that the correct choice of
phase $e^{-i t E^\Lambda }$ is such that $E^\Lambda \to \infty$ as
$\Lambda\to \infty$. The phase serves to renormalize the Bogoliubov
evolution, which is in complete analogy to the renormalization of $H_N$, as
stated in \eqref{eq:def:ren:Nelson}.\medskip


\section{Preliminaries}
\label{sect:preliminaries}

\subsection{Fock spaces and excitation map}
\label{sect:excitation space}

We recall the Fock space $\cF = \bigoplus_{k=0}^\infty \bigotimes_{\rm sym}^k
L^2(\mathbb R^3)$ and define the (truncated) Fock spaces for the excitations
of the particles,
\begin{align}
  \label{eq: definition excitation Fock space particles}
  \cF^{(k)}_{\perp u} = \bigotimes_{\rm sym}^k \{ u\}^\perp , \quad \cF_{\perp u}^{\le N}  = \bigoplus_{k=0}^N \cF^{(k)}_{\perp u} ,\quad \cF_{\perp u}  = \bigoplus_{k=0}^\infty \cF^{(k)}_{\perp u}
\end{align}
for $\{ u \}^\perp = \{ \varphi \in L^2 (\mathbb R^3) : \langle \varphi , u
\rangle = 0\}$. The relevant double Fock spaces for the Nelson model are
\begin{align}\label{eq:double:Fock:spaces}
  \cF_{\perp u}^{\le N} \otimes \cF, \qquad \cF_{\perp u} \otimes \cF \quad \text{and} \qquad \cF \otimes \cF,
\end{align}
where the first factor always refers to the excitations of the $N$ bosonic
particles, while the second factor describes the excitations of the quantum
field. If the context is not unambiguous, we shall write $\cF= \cF_b$ for the
particles and $\cF = \cF_a$ for the quantum field. In the order of
\eqref{eq:double:Fock:spaces}, we refer to the Fock spaces as {truncated
  excitation Fock space}, {excitation Fock space} and {double Fock
  space}. Moreover, we denote by $b_x$, $b_x^*$, $\mathcal{N}_b$ and $a_k$,
$a_k^*$, $\mathcal{N}_a$ the annihilation, creation and number operators on
$\cF_b$ and $\cF_a$, respectively.  For $f \in L^2(\mathbb{R}^3)$ let
\begin{subequations}
  \begin{align}
    a(f) &= \int_{\mathbb{R}^3} dk \, \overline{f(k)} a_k ,
    \quad 
    a^*(f) = \int_{\mathbb{R}^3} dk \, f(k) a_k^* ,
    \\
    b(f) &= \int_{\mathbb{R}^3} dx \, \overline{f(x)} b_x ,
    \quad 
    b^*(f) = \int_{\mathbb{R}^3} dx \, f(x) b_x^*
  \end{align}
\end{subequations}
denote the bosonic annihilation and creation operators.  For self-adjoint
$T$, $D(T)$ on $L^3(\R^3)$ we denote by $\mathrm{d} \Gamma(T)$ the
self-adjoint second quantization of $T$ on $\cF$.  Depending on which factor
of $\cF\otimes \cF$ this acts, we write
\begin{align}
  \mathrm{d} \Gamma_a(T) &= \int_{\mathbb{R}^6} dx  dy \, T(k,l) a_k^* a_l
  \quad \text{and} \quad
  \mathrm{d} \Gamma_b(T) = \int_{\mathbb{R}^6} dx  dy \, T(x,y) b_x^* b_y,
\end{align}
where $T(\cdot,\cdot)$ is the Schwarz kernel of $T$.  With this, we introduce
the notation
\begin{align}
  \label{eq:total kinetic energy and number of particles operator definition}
  \mathcal{N} &= \mathcal{N}_b + \mathcal{N}_a,
  \qquad
  \mathbb{T} = \mathrm{d} \Gamma_a \left( - \Delta \right)
  + \mathrm{d} \Gamma_a(\omega) .  
\end{align}
We define the field operator by
\begin{equation}\label{eq:field-op}
  \hat \Phi(f)=a(f) +a^*(f).
\end{equation}
Using this definition, the Weyl operator introduced in~\eqref{eq: Weyl
  operator} can be expressed as
\begin{equation}
  W(f) = e^{ -i \hat\Phi(if)}.
\end{equation}
It satisfies
\begin{align}
  \label{eq: Weyl operators product}
  W^{-1}(f)= W(-f), \quad W(f) W(g)   = W(f+g) e^{-i \Im \langle f , g \rangle },
\end{align}
as well as the shift property
\begin{align}
  \label{eq: Weyl operators shift property}
  W^*(f) a_k W(f) = a_k + f(k).
\end{align}

As an important tool in our analysis, we introduce a variant of the
\textit{excitation map} introduced in \cite{Lewin:2015a,LNSS2015}. In the
context of the Nelson model, the excitation map factors out a condensate with
wave function $u$ and a coherent state with field mode $\sqrt N \alpha $. For
$u, \alpha \in L^2(\mathbb{R}^3)$ with $\norm{u}_{L^2} = 1$, it is defined as
the map
\begin{equation}
  X_{ u , \alpha} :\mathcal{H}_N \rightarrow \cF_{\perp u} \otimes \cF
\end{equation}
with $\Psi_{N} \mapsto ( \chi^{(k)} )_{k=0}^N$ given by
\begin{align}
  \label{eq:action of the unitary}
  \chi^{(k)} = \binom{N}{k}^{1/2} \prod_{i=1}^k (q_{u})_i
  \scp{u^{\otimes (N-k)}}{W^* ( \sqrt{N} \alpha ) \Psi_N}_{L^2(\mathbb{R}^{3(N-k)})}
  \; \in \mathcal{F}_{\perp u}^{(k)} \otimes \mathcal{F},
\end{align}
where $(q_{u})_i$ is the orthogonal projection $q_u = \id - |u\rangle \langle
u | $ acting on the $i$th particle coordinate $x_i$. Here, the partial inner
product is taken w.r.t. the particle coordinates $x_{k+1}, \ldots, x_N$. The
adjoint of $X_{u,\alpha}$ is given by \eqref{eq:Psi_N excitation}, and it
holds
\begin{align}\label{eq:X-isometry}
  X_{u,\alpha}^*X_{u,\alpha}=\id_{\mathcal{H}_N}, \qquad  X_{u,\alpha}X_{u,\alpha}^*=\id_{\mathcal{F}_{\perp u}^{\leq N} \otimes \mathcal{F}},
\end{align}
in particular $X_{ u , \alpha} :\mathcal{H}_N \rightarrow \mathcal{F}_{\perp
  u}^{\leq N} \otimes \mathcal{F}$ is unitary.  Written in terms of creation
and annihilation operators, the excitation map acts as
\begin{align}
  X_{u,\alpha} \Psi_N =\bigg( \bigoplus_{k=0}^N q_u^{\otimes k} \frac{b(u)^{ N- k } }{\sqrt{(N-k )!} } \bigg) \otimes W ^*(\sqrt N \alpha) \Psi_N .
\end{align}
This leads to the following useful relations \cite{Lewin:2015a, LNSS2015}.
\begin{lemma}
  \label{lemma:properties U}
  As identities on $\cF_{\perp u}\otimes \cF$, we have for all $f,g \in
  \{u\}^\perp $
  \begin{align*}
    X_{ u , \alpha} b^{*}(u ) b( u ) X_{u , \alpha}^*  & = [N - \mathcal{N}_b]_+  ,
    \\
    X_{ u , \alpha} b^{*}(f) b( u ) X_{ u , \alpha}^*  & = b^{*}(f) \left[ N - \mathcal{N}_b \right]^{1/2}_+  ,
    \\
    X_{ u , \alpha} b^{*}( u ) b(f) X_{ u , \alpha}^*  & =  \left[N - \mathcal{N}_b \right]^{1/2}_+ b(f)  ,
    \\
    X_{ u , \alpha} b^{*}(f) b(g) X_{ u , \alpha}^*  & =  b^{*}(f) b(g)  ,
  \end{align*}
  where $[a]_+=\max\{a,0\}$. Moreover, for all $h \in L^2(\mathbb{R}^3)$,
  \begin{align*}
    X_{ u , \alpha} a(h) X_{ u , \alpha}^* = a(h) + \sqrt{N} \scp{h}{\alpha}.
  \end{align*}
\end{lemma}

To measure the distance between reduced densities, we introduced the
functional $\beta$ in \eqref{eq:definition of beta}. In Section
\ref{chap:mean-field}, we need a second functional given by
\begin{align}\label{eq:gamma:functional:def}
  & \gamma \left[ \Psi_N, ( u ,\alpha ) \right]  =  \norm{ \nabla_1 (q_u)_1  \Psi_N}^2  +   N^{-1}\| \text{d} \Gamma_a(\omega)^{1/2} W^*(\sqrt{N} \alpha)\Psi_N\|^2.
\end{align}
Using Lemma \ref{lemma:properties U}, they can be expressed in terms of the
excitation map as
\begin{subequations}
  \begin{align}
    \beta[  \Psi_N, (u,\alpha) ] & = \tfrac{1}{N} \scp{X_{u,\alpha} \Psi_N  }{\mathcal N X_{u,\alpha} \Psi_N}_{\mathcal{H}_N} , \label{eq:beta:second:quant} \\[1mm]
    \gamma[  \Psi_N ,  (u,\alpha) ] &  = \tfrac{1}{N} \scp{X_{u,\alpha} \Psi_N  }{\mathbb T X_{u,\alpha} \Psi_N }_{\mathcal{H}_N}  \label{eq:gamma:second:quant}.
  \end{align}
\end{subequations}


\subsection{Notation}
We recall the Hilbert space
\begin{equation}
  \mathcal{H}_N= \bigotimes_{\rm sym}^N L^2(\mathbb R^3) \otimes \mathcal{F}\subset L^2(\mathbb R^{3N})\otimes \cF,
\end{equation}
and the definitions $\omega(k) = \sqrt{k^2 +1}$ and
\begin{align}
  \label{eq:definition of G-x notation section}
  G_x(k) &= \frac{1}{\sqrt{\omega(k)}} e^{- i k x}
\end{align}
and introduce the function
\begin{align} \label{eq:def:B(k)} B_{x}(k) &= \frac{G_x(k)}{k^2+\omega(k)}
  =\frac{e^{- i k x}}{\sqrt{\omega(k)} \left( k^2 + \omega(k) \right) } .
\end{align}

\noindent Moreover, we adopt the following notation conventions.

\begin{itemize}
\item For a normed space $X$ we denote its topological dual by $X'$.
\item For normed spaces $X,Y$ we denote the norm of a linear map $A:X\to Y$
  by $\|A\|_{X\to Y}$.
\item $H^s(\mathbb{R}^3)$ with $s \in \mathbb{R}$ denotes the non-homogeneous
  $L^2$-Sobolev space.
\item For the norms on $L^2(\mathbb R^3)$ and $H^s(\mathbb R^3)$ we write
  $\norm{u}_{L^2}$ and $\norm{u}_{H^s}$.
\item $\mathfrak{h}_s$ with $s \in \mathbb{R}$ is the weighted $L^2$-space
  with norm $\norm{\alpha}_{\mathfrak{h}_{s}} = \norm{ \omega^s
    \alpha}_{L^2}$.
\item $\| \cdot \| $ denotes the norms of $\mathcal{H}_N$ and $\mathcal{F}
  \otimes \mathcal{F}$, depending on the context.
\item $\dot u_t$ denotes the time-derivative of a function $t\mapsto u_t$.
\item We do not specify the domain of integration if it is equal to
  $\R^{3n}$.
\item For the Fourier transform of $u\in L^2(\mathbb R^3)$, we use the
  convention that
  \begin{align}
    \widehat{u}(k) = (2 \pi)^{-3/2} \int  dx \, e^{- ik x} u(x).
  \end{align}
\item For a quadratic form with domain $Q(A)$ associated with $A:Q(A)\to
  Q(A)'$, we write $A+\text{h.c.}$ for the quadratic form
  \begin{equation}
    \langle \psi, (A+\text{h.c.})\psi\rangle = 2\text{Re}\langle \psi, A \psi\rangle, \qquad \psi \in Q(A).
  \end{equation}
\item The letter $C$ denotes a generic constant, whose value may change
  within a sequence of inequalities. For example, in $X \leq C Y \leq C Z$
  the two occurrences of $C$ may represent different numbers.
\end{itemize}

\subsection{The SKG equations}

The next statement recaps the well-posedness theory of the non-linear SKG
equations,
\begin{align}
  \label{eq:Schroedinger-Klein-Gordon equations regular}
  \begin{cases}
    \begin{aligned}
      i \partial_t u_t(x) & \, =\, \big( - \Delta + \phi_{\alpha_t}(x) - \tfrac{1}{2} \scp{u_t}{ \phi_{\alpha_t} u_t} \big) u_t(x) \\[1.5mm]
      i \partial_t \alpha_t(k) &\, =\, \omega(k) \alpha_t(k) + \scp{u_t}{G_{(\cdot)}(k) u_t} \\[1.5mm]
      (u_t,\alpha_t)|_{t=0}& \,=(u,\alpha).
    \end{aligned}
  \end{cases}
\end{align}
\begin{proposition}
  \label{prop:skg well posed}
  For any $s\geq 0$ the Cauchy problem \eqref{eq:Schroedinger-Klein-Gordon
    equations regular} is globally well-posed in $H^s(\mathbb{R}^3) \oplus
  \mathfrak{h}_{s-1/2}$. The solutions satisfy $\| u_t \|_{L^2} =
  \|u\|_{L^2}$ and, for $s\geq 1$, $\mathcal E(u_t,\alpha_t) = \mathcal E (
  u,\alpha ) $ with $\mathcal E$ defined by \eqref{eq:energy functional
    Schroedinger-Klein-Gordon equations regular}.

  In addition, for any integer $ s\geq 1$, there exists $\delta$ such that
  for any $M>0$ there exists $C$ so that for all $t\in \mathbb{R}$ and
  $\lVert (u,\alpha) \rVert_{H^s \oplus \mathfrak{h}_{s-1/2}}\leq M$, the
  solution $(u_t,\alpha_t)$ satisfies
  \begin{equation*}
    \lVert u_t \rVert_{H^s} + \| \alpha_t\|_{\mathfrak{h}_{s-1/2}} \leq
    \begin{cases}\, 
      C \quad & \text{if} \ s= 1 \; ,
      \\
      \, C(1+\lvert t  \rvert )^\delta
      \quad & \text{otherwise}
      \;.
    \end{cases}
  \end{equation*}
\end{proposition}

\begin{proof}
  The well-posedness together with the conservation properties is a special
  case of \cite[Thm.1.4]{pecher}; see also~\cite{bachelot,CHT2018}. The
  estimate on the norms for $s=1$ follows from the conservation properties
  and an application of the bound \eqref{eq:bound:G:alpha:u}.
  The polynomial-in-time bounds for $s >1$ can be proved by an iterative
  argument, adapting the approach of \cite{FG2017} for the Landau--Pekar
  equations, see also~\cite{gruenrock} for a different approach.
\end{proof}


\section{The mean-field approximation}\label{chap:mean-field}

In this section we study the approximation of $e^{-i t H_N }$ on the level of
reduced densities and prove Theorem~\ref{theorem:final reduced density
  matrices}.  We first consider the evolution generated by the Nelson
Hamiltonian after dressing it with a suitable unitary transformation. Even
though the dressed Hamiltonian $H_N^{\rm D}$ has a more complicated structure
than the original Hamiltonian, it contains less singular interaction terms
and is thus better suited for our analysis. In fact, $H_N^{\rm D}$ is a
perturbation of the non-interacting Hamiltonian in the sense of quadratic
forms.

In order to approximate the evolution $e^{-i t H_N^\mathrm{D} }$, it is
necessary to replace the mean-field equations by their dressed variant, as
already observed in~\cite{AF2017}.  The approximation result for the reduced
densities of $e^{-i t H_N^\mathrm{D} }\Psi_N$ using the dressed mean-field
equations, with a statement analogous to that of Theorem \ref{theorem:final
  reduced density matrices}, is given in Theorem~\ref{thm:gross-transformed
  dynamics reduced density matrices}. We then study the dressing
transformations, and state in Lemma~\ref{lemma:reduced densities dressing
  flow} that they can be approximated by a mean-field dressing flow in a
similar way. The proof of Theorem~\ref{theorem:final reduced density
  matrices} is given by combining these results in Section~\ref{sect:final
  reduced density}.

\subsection{Dressed dynamics on the microscopic level}

With $B_x(k)$ given by \eqref{eq:def:B(k)} and the field operator $\hat \Phi$
from~\eqref{eq:field-op}, we consider the family of unitary dressing
transformations
\begin{align}
  \label{eq:definition of the quantum dressing group:theta}
  W^{\rm D} (\theta)
  =  \prod_{j=1}^N \exp\bigg( \frac{-i\theta}{\sqrt N } \hat \Phi(i B_{x_j}) \bigg)
\end{align}
and set $W^\mathrm{D}:=W^\mathrm{D}(1)$. This transformation goes back to
Gross and Nelson \cite{nelson}, and the following lemma recalls a well-known
relation between the renormalized Nelson Hamiltonian and the dressed Nelson
Hamiltonian.

\begin{lemma} \label{lem:H:D:representation} Consider the symmetric quadratic
  form defined on the form domain of $\mathrm{d} \Gamma_a(\omega) +
  \sum_{i=1}^N (-\Delta_i)$,
  \begin{multline*}
    H_{N }^{ \rm D} = \textnormal{d} \Gamma_a(\omega) + \sum_{i=1}^N(-\Delta_i)  + \frac{1}{\sqrt N}\sum_{i=1}^N  \hat{A}_{x_i}  + \frac{1}{N} \sum_{1\le i<j\le N }    V  (x_i-x_j)
    \\ + \frac{1}{N} \sum_{i=1}^N \Big(a(kB_{x_i})^2 + 2a^*(kB_{x_i})a(kB_{x_i}) +a^*(kB_{x_i})^2 \Big),
  \end{multline*}
  where 
  \begin{subequations}
    \begin{align}
      \hat A_x &=-2 \big(i\nabla_x \cdot a(k B _x) + a^*(kB _x) \cdot  i\nabla_x\big),\\[1mm]
%
%
      V (x ) & = -4 \Re\langle G _x, B _0 \rangle + 2   \Re \langle \omega B _x, B_0 \rangle.\label{eq:definition of V}
    \end{align}
  \end{subequations}
  There exists a unique self-adjoint operator $H_N^\mathrm{D}$,
  $D(H_N^\mathrm{D})$ whose quadratic form coincides with the above, and we
  have
  \begin{equation*}
    H_N= (W^\mathrm{D})^*H_N^\mathrm{D} W^\mathrm{D}.
  \end{equation*}
\end{lemma}
\begin{proof}
  Formally, this follows from the definition of the Weyl operators and a
  direct computation.  The precise statement is a corollary to the original
  construction of the renormalized Nelson Hamiltonian~\cite{nelson}, refined
  in~\cite{GW2018}. There, one considers the operator $H_{N,K}^\mathrm{D}$
  related to a dressing transformation with an infrared cutoff $K$ (as in
  Proposition~\ref{proposition:class of suitable initial states}). This is
  used to bound the interaction terms relative to the form of the
  non-interacting operator with bound less than one, for $K$ sufficiently
  large (see~\cite[Thm. 3.3]{GW2018}). Transforming $H_{N,K}^\mathrm{D}$ with
  the dressing transformation on momenta below $K$ gives the formula
  above. This does not change the form domain by~\cite[Thm. 4.1,
  Lem. C.4]{GW2018}.
\end{proof}

\subsection{Mean-field approximation of the dressed
  dynamics\label{sec:dressed-dyn-classical}}

Given the dressed Nelson Hamiltonian $H_N^\mathrm{D}$, one can derive an
associated mean-field energy by projecting onto states of the product-like
form~\eqref{eq:many-body initial states} with given one-particle functions
$u,\alpha$.  The dressed mean-field equations, the Hamiltonian equations
associated with this energy, take the form
\begin{align}\label{eq:SKG dressed}
  \begin{cases}
    \begin{aligned}
      i \partial_t u_t (x)
      &=  h_{u_t,\alpha_t} u_t(x)
      \\[1.5mm]
      i \partial_t \alpha_t(k) &  = \omega (k)\alpha_t + 2 \langle u_t, k B_{(\cdot)}(k) (-i\nabla +  F_{\alpha_t } )u_t \rangle \\[1.5mm]
      (u_0,\alpha_0 ) & = (u,\alpha)
    \end{aligned}
  \end{cases}
\end{align}
where
\begin{subequations}
  \begin{align}\label{eq:def:dressed:meanfield:Ham}
    h_{u,\alpha } & = -\Delta +   A_{\alpha } +   (F_{\alpha})^2 + V \ast|u|^2  -\mu_{u,\alpha } ,   \\[1mm]
    A_{\alpha,x} & =   2(- i\nabla_x)  \scp{kB_x }{\alpha} +  2 \overline{\scp{kB_x}{\alpha}} (- i\nabla_x) , \label{eq:def:A_alpha}\\[1mm]
    F_\alpha(x) & =  2 \Re \scp{k B _x}{\alpha} , \label{eq:def:F_alpha} \\[1mm]
    \mu_{u,\alpha }&=  \tfrac{1}{2} \scp{u }{ V  \ast |u|^2 u } +   \Re \scp{\alpha }{ f_{u} } +  \Re \scp{\alpha }{g_{u,\alpha} } , \label{eq:def:dressed:mu}\\[2mm]
    f_{u   }(k) & = 2   \scp{u}{ k B_{(\cdot)}(k)   (-i \nabla  ) u} , \label{eq:def:f}\\[1.5mm]
    g_{u ,\alpha }(k) & = 2   \scp{ u }{ k B_{(\cdot)}(k) F_{\alpha } u}  . \label{eq:def:g}
  \end{align}
\end{subequations}
We denote the associated flow by $\mathfrak {s}^{\rm D}[t](u,\alpha) =
(u_t,\alpha_t)$, that is $(u_t,\alpha_t)$ solves \eqref{eq:SKG dressed} with
initial conditions $(u_t,\alpha_t)|_{t=0} = (u,\alpha)$ (existence of this
flow is the special case $\theta=1$ of
Lemma~\ref{lem:theta-mean-field-flow}).

In the next statement we compare the evolution generated by $H_N^{\rm D}$
with the dressed mean-field flow $\mathfrak s^{\rm D}[t] $. To this end, we
recall \eqref{eq:definition of beta} for the definition of the functional
$\beta$ and
\begin{align}
  & \gamma \left[ \Psi_N, ( u ,\alpha ) \right]  =  \norm{ \nabla_1 (q_u)_1  \Psi_N}^2  +   N^{-1} \|\text{d} \Gamma_a(\omega)^{1/2} W^*(\sqrt{N} \alpha)\Psi_N\|^2.
\end{align}
Essentially, $\gamma$ is the mean kinetic energy of particles outside the
condensate state $u$ and field modes outside of the coherent state
$W(\sqrt{N} \alpha) \Omega$. Also note that by Lemma
\ref{lem:H:D:representation} we have $ e^{- i t H_N } \Psi_N = (W^{\rm D})^*
e^{- i t H_N^{\rm D} } W^{\rm D} \Psi_N$, which explains why we now consider
initial states of the form $W^{\rm D} \Psi_N $.

The following statement is the main result of this section.
\begin{theorem}
  \label{thm:gross-transformed dynamics reduced density matrices}
  Let $(u,\alpha) \in H^3 \oplus \mathfrak h_{5/2}$ with $\norm{u}_\2 = 1$
  and let $ \mathfrak s^{\rm D}[t](u,\alpha) = ( u_t , \alpha_t ) $ denote
  the solution to \eqref{eq:classical:transformed:equations} for initial
  conditions $(u,\alpha)$. There exists a constant $C>0$ such that for all
  $\Psi_N \in D(H_N^{1/2})$ with $\norm{\Psi_N}=1$, $N\ge 1$, and $t \in \R$,
  we have
  \begin{multline*}
    \beta \big[ e^{- i t H_N^{\rm D} } W^{\rm D} \Psi_N , \mathfrak s^{\rm D}[t](u,\alpha)  \big]\\[1mm]
    \leq  e^{C R^{\rm D}(t)} \big( \beta \big[ W^{\rm D} \Psi_N , ( u , \alpha ) \big]
    + \gamma \big[ W^{\rm D} \Psi_N , (u , \alpha ) \big]
    + N^{-1} \big)
  \end{multline*}
  where $R^{\rm D}(t) = 1 + \int_0^{|t|} \| u_s\|^2_{\Sob{3}}
  (1+\norm{\alpha_s}_{\h{3/2}})^2 ds $.
\end{theorem}

To prepare the proof of the theorem, we introduce the fluctuation generator
associated with $e^{- i t H_N^{\rm D} }$ and $\mathfrak s^{\rm
  D}[t]$. Recalling the definition of the excitation map \eqref{eq:action of
  the unitary}, and fixing $\Psi_N \in \mathcal H_N$, $(u,\alpha)\in
H^3\oplus \mathfrak h_{5/2}$, we consider the fluctuation vector
\begin{align}\label{eq:fluctuation:vector}
  \chi^{\rm D}(t) := X_{\mathfrak s^{\rm D}[t](u,\alpha)} e^{-i t H_N^{\rm D} } W^{\rm D} \Psi_N.
\end{align}
A simple computation shows that $\chi^{\rm D}(t)$ satisfies the equation
$\label{eq:fluctuation:equation} i\partial_t \chi^{\rm D}(t) = H^{ {\rm D},
  \le N}_{u,\alpha}(t) \chi^{\rm D}(t)$ with
\begin{align}
  \label{eq:fluctuation:generator:2}
  H^{{\rm D}, \le N }_{u,\alpha}(t)  & = i \dot X_{\mathfrak s^{\rm D}[t](u,\alpha)} ( X_{\mathfrak s^{\rm D}[t](u,\alpha)} )^* + X_{\mathfrak s^{\rm D}[t](u,\alpha)} H_N^{\rm D}  ( X_{\mathfrak s^{\rm D}[t](u,\alpha)} ) ^*. 
\end{align} 
Note that $H^{{\rm D}, \le N }_{u,\alpha}(t) $ maps $ \cF^{\le N}_{\perp u_t}
\otimes \cF$ into $ \cF^{\le N} \otimes \cF$, but for convenience, we will
write it as the restriction of a symmetric operator $H^{\rm D}_{u,\alpha}(t)
: \cF \otimes \cF \rightarrow \cF \otimes \cF $, that is
\begin{align}\label{eq:generator:restriction}
  H^{{\rm D}, \le N}_{u,\alpha}(t) = H_{u,\alpha}^{\rm D} (t) \restriction \cF^{\le N}_{\perp u_t} \otimes \cF.
\end{align}
The explicit expression for $H_{u,\alpha}^{\rm D}(t)$ is given in Section
\ref{sec:dressed:fluc:generator}.  The fluctuation vector $\chi^{\rm D}(t)$
then satisfies the Schr\"odinger type equation
\begin{align}\label{eq:fluctuation equation}
  \begin{cases}
    \begin{aligned}
      i \partial_t \chi^{\rm D}(t) &=  H^{\rm D}_{u,\alpha}(t) \chi^{\rm D}(t) \\[1mm]
      \chi^{\rm D}(0) & = X _{u,\alpha} W^{\rm D}  \Psi_N.
    \end{aligned}
  \end{cases}
\end{align}
In Lemma \ref{lem:bounds:fluc:generator} we shall prove the following bounds
\begin{subequations}
  \begin{align}
    \pm \big( H^{\rm D}_{u,\alpha}(t) - \mathbb T \big) & \le \tfrac{1}{2}\mathbb T + C  (\mathcal N+ 1) (1 + \tfrac{1}{N}\mathcal N_b )^2  ,  \label{eq:key:bound:1:L(t):new}\\[1.5mm]
    \pm i [\mathcal N ,  H^{\rm D}_{u,\alpha}(t)  ]  & \le  \tfrac{1}{2} \mathbb T + C  (\mathcal N+ 1) (1 + \tfrac{1}{N}\mathcal N_b )^2   , \label{eq:key:bound:3:L(t):new}\\[1.5mm]
    \pm \tfrac{d}{dt}  H^{\rm D}_{u,\alpha}(t)  & \le  \tfrac{1}{2} \mathbb T  + C \rho(t)  (\mathcal N+ 1) (1 + \tfrac{1}{N}\mathcal N_b )^2 ,\label{eq:key:bound:4:L(t):new}
  \end{align}
\end{subequations}
where $\mathcal{N}$ and $\mathbb{T}$ are defined as in \eqref{eq:total
  kinetic energy and number of particles operator definition}, $\rho(t) =
\norm{u_t}_{\Sob{3}}^2 (1+\norm{\alpha_t}_{\h{3/2}})^2$ and $(u_t,\alpha_t) =
\mathfrak s^{\rm
  D}[t](u,\alpha)$. 
Equipped with these estimates we can now come to the proof of Theorem
\ref{thm:gross-transformed dynamics reduced density matrices}, whose strategy
is inspired by \cite{BS2019,NS2020,L2022}.

\begin{proof}[Proof of Theorem \ref{thm:gross-transformed dynamics reduced
    density matrices}] Consider $(u_t,\alpha_t) = \mathfrak s^{\rm
    D}[t](u,\alpha)$ and the fluctuation vector $\chi^{\rm D}(t)$ given by
  \eqref{eq:fluctuation:vector} for initial states $(u,\alpha)$ and $\Psi_N$
  as stated in the hypothesis.  Note that by definition, $\chi^{\rm D}(t) \in
  \cF^{\le N}_{\perp u_t} \otimes \cF$. Relations
  \eqref{eq:beta:second:quant} and \eqref{eq:gamma:second:quant} imply that
  \begin{subequations}
    \begin{align}
      \beta[ e^{-i t H_N^{\rm D} } W^{\rm D}  \Psi_N, \mathfrak s^{\rm D}[t](u,\alpha) ] & = \tfrac{1}{N} \scp{\chi^{\rm D} (t)}{\mathcal N \chi^{\rm D} (t)}\label{eq:beta:second:quant:D} \\[1mm]
      \gamma[e^{-i t H_N^{\rm D} } W^{\rm D}  \Psi_N , \mathfrak s^{\rm D}[t](u,\alpha) ] &  = \tfrac{1}{N} \scp{\chi^{\rm D} (t)}{\mathbb T \chi^{\rm D} (t)} . \label{eq:gamma:second:quant:D}
    \end{align}
  \end{subequations}
  From \eqref{eq:key:bound:1:L(t):new} and $\id_{\mathcal N_b \le N}
  \chi^{\rm D}(t) = \chi^{\rm D}(t)$ it follows that
  \begin{align}
    & N  \beta[ e^{-i t H_N^{\rm D} } W^{\rm D} \Psi_N ,  \mathfrak s^{\rm D}[t](u,\alpha)   ]    \le \scp{\chi^{\rm D} (t)}{ ( \mathbb T + \mathcal N  ) \chi^{\rm D} (t)} \notag \\[1mm] 
    & \hspace{2cm} \le 2 \scp{\chi^{\rm D} (t)}{ H^{\rm D}_{u,\alpha}(t)  \chi^{\rm D} (t)} + C  \scp{\chi^{\rm D} (t)}{( \mathcal N +1 ) \chi^{\rm D} (t)} =: f(t).
  \end{align}
  We proceed by estimating the time-derivative of $f(t)$ in order to conclude
  via Gr\"onwall's inequality. Using \eqref{eq:fluctuation equation} and with
  the aid of
  \eqref{eq:key:bound:1:L(t):new}--\eqref{eq:key:bound:4:L(t):new}, one
  computes
  \begin{align}
    \big|\dot f(t )  \big| & =  \big| 2 \scp{\chi^{\rm D} (t)}{ (\tfrac{d}{dt} H_{u,\alpha}^{\rm D}(t)  ) \chi^{\rm D} (t)} + C \scp{\chi^{\rm D} (t)}{ i [ H^{\rm D}_{u,\alpha} (t)  , \mathcal N ] \chi^{\rm D} (t)} \big| \notag\\[1mm]
    & \le C \scp{ \chi^{\rm D}  (t) }{ \mathbb T \chi^{\rm D} (t) } + C  \rho(t) \scp{ \chi^{\rm D}   (t) }{ (\mathcal N + 1) \chi^{\rm D} (t) } \notag \\[1mm]
    & \le C \scp{ \chi^{\rm D} (t) }{ H_{u,\alpha}^{\rm D}(t)  \chi^{\rm D} (t) } + C \rho(t) \scp{ \chi^{\rm D}  (t) }{ (\mathcal N + 1) \chi^{\rm D} (t) } \notag\\[1.5mm]
    & \le C \rho(t)  f(t),
  \end{align}
  where we used that $\rho(t) \ge 1 $. Gr\"onwall's inequality thus implies $
  f(t) \le e^{C \int _0^{|t|} \rho(s) ds } f(0)$ and using again
  \eqref{eq:key:bound:1:L(t):new} together with
  \eqref{eq:beta:second:quant:D}, \eqref{eq:gamma:second:quant:D}, we arrive
  at
  \begin{align}
    & \beta[e^{-i t H_N^{\rm D} } W^{\rm D} \Psi_{N},   \mathfrak s^{\rm D}[t](u,\alpha)    ]   \le  C e^{C \int _0^{|t|}  \rho(s) ds } \tfrac{1}{N}\scp{\chi^{\rm D} (0)}{(\mathbb T + \mathcal N + 1) \chi^{\rm D} (0)} \notag\\[1mm]
    & \hspace{1.5cm} =  e^{C R^\mathrm{D}(t) } \big( \beta[ W^{\rm D} \Psi_N , ( u,\alpha) ]  + \gamma[ W^{\rm D}  \Psi_N , ( u,\alpha )] + N^{-1} \big).
  \end{align}
  This completes the proof of Theorem \ref{thm:gross-transformed dynamics
    reduced density matrices}.
\end{proof}

\subsection{Mean-field approximation of the dressing transformation}

The dressing transformation $W^\mathrm{D}$ is generated by an operator that
looks like the interaction term of the Nelson Hamiltonian, but has the
regular form factor $iB_x$. There are thus mean-field equations associated to
the dynamics $\theta \mapsto W^\mathrm{D}(\theta)$, \eqref{eq:definition of
  the quantum dressing group:theta}, as in the case of the (dressed) Nelson
Hamiltonian.  The mean-field equations corresponding to the dressing
transformation are given by
\begin{align}
  \label{eq:classical dressing equations}
  \begin{cases}
    \begin{aligned}
      i \partial_{\theta} u^{\theta}(x) &= \tau_{u^\theta,\alpha^\theta}(x) u^{\theta}(x)
      \\[1mm]
      \partial_{\theta} \alpha^{\theta}(k) &= B_0(k)  \widehat{|u^\theta|^2}(k)  \\[1mm]
      (u^\theta,\alpha^\theta)|_{\theta =0}  & = (u,\alpha)
    \end{aligned}
  \end{cases}
\end{align}
where we introduced
\begin{align}\label{eq:def:tau:Lambda}
  \tau _{u,\alpha}(x) = \widetilde{\phi}_{\alpha}(x) - \frac{1}{2}
  \scp{u}{\widetilde{\phi}_{\alpha} u}, \quad
  \widetilde{\phi}_{\alpha}(x) = 2 \Re \scp{i B_x}{\alpha}.
\end{align}
We denote by $\mathfrak D[\theta]$ the flow corresponding to this equation,
i.e.
\begin{equation}\label{eq:dressing:flow:theta}
  \mathfrak D[\theta](u,\alpha)=(u^\theta, \alpha^\theta),
\end{equation}
where $(u^\theta, \alpha^\theta)$ is the solution to~\eqref{eq:classical
  dressing equations} with initial condition $(u,\alpha)$.  Being the flow of
an autonomous system of equations, we have $\mathfrak D[\theta] \circ
\mathfrak D[-\theta] =1$. For $\theta=1$ we use the shorthand
$\mathfrak{D}:=\mathfrak{D}[1]$.

In fact, $\mathfrak{D}[\theta]$ can be determined explicitly
following~\cite[Lem. III.11]{AF2017}.  Since $\tau_{u^\theta,\alpha^\theta}$
is real, the solution satisfies $|u^\theta|^2 = |u |^2$, and then the
equation for $\alpha$ can be solved for
\begin{equation}
  \alpha^\theta(k)=\alpha(k) + \theta B_0(k)\widehat{|u|^2}(k).
\end{equation}
Since $B_0$ is an even function,
\begin{equation}
  \Re \scp{i B_x}{B_0 \widehat{|u|^2}} = \Im \int dy \widehat{B_0^2}(y-x) |u|^2(y) = 0.
\end{equation}
Hence, we have $\tilde \phi_{\alpha^\theta}=\tilde \phi_\alpha$, and one can
simplify the equations using $ \tau_{u^\theta,\alpha^\theta} =
\tau_{u,\alpha}$.  The system of ordinary differential
equations~\eqref{eq:classical dressing equations} for each $(x,k)$ is then
solved explicitly by
\begin{equation}\label{eq:explicit:solution:dressing}
  (u^\theta, \alpha^\theta):=\mathfrak{D}[\theta](u,\alpha)=\Big(e^{-i\theta \tau_{u,\alpha}} u ,\alpha + \theta B_0(\cdot) \widehat{|u|^2}\Big).
\end{equation}
The flow $\mathfrak D[\theta]$ preserves the relevant spaces of Cauchy data
for the SKG equations (cf. Proposition~\ref{prop:skg well posed}).

\begin{lemma}\label{lem:mf-dressing-regularity}
  Let $n \geq1$ be an integer and $n-1<s<n+1$. There exists $C$ so that for
  all $(u, \alpha)\in H^n(\R^3)\oplus \mathfrak{h}_s$ and $|\theta|\leq 1$
  the functions $(u^\theta, \alpha^\theta)=\mathfrak{D}[\theta](u,\alpha)$
  satisfy
  \begin{align*}
    \|u^\theta\|_{H^n(\R^3)} & \leq C \|u\|_{H^n} \|\alpha\|_{\mathfrak{h}_s}^n \\
    \|\alpha^\theta\|_{\mathfrak{h}_s} &\leq C(\|u\|_{H^n}^2+\|\alpha\|_{\mathfrak{h}_s}).
  \end{align*}
\end{lemma}
\begin{proof}
  We use the explicit form~\eqref{eq:dressing:flow:theta} of $\mathfrak
  D[\theta]$ and some straightforward estimates (some of which will be proved
  in Section~\ref{sec:preliminary:estimates}).  We have for $m \leq n$
  \begin{equation}
    e^{i \theta \tau_{u,\alpha}} \nabla^m u^\theta = (\nabla + i\theta \nabla \tau_{u,\alpha})^m u.
  \end{equation}
  With $\nabla \tau_{u,\alpha}(x) =-2 \Re \langle k B_x, \alpha \rangle$ we
  thus have
  \begin{align}
    \| \nabla^m u^\theta\|_{L^2} &\leq C \sum_{\ell=0}^m \sum_{j_1 +\dots + j_\ell \leq m}\Big\| \Big(\prod_{i=0}^{\ell} \langle |k|^{j_i} B_0, |\alpha| \rangle  \Big) \nabla^{m-\sum j_i} u\Big\|_{L^2} \notag\\
    &\leq C \|u\|_{H^m} \|\alpha\|_{\mathfrak{h}_s}^m,
  \end{align}
  by Lemma~\ref{lem:bounds:alpha:phi} (where we used that $s>n-1$).  For
  $\alpha^\theta$ we have
  \begin{equation}
    \| \alpha^\theta \|_{\mathfrak{h}_s} \leq \| \alpha\|_{\mathfrak{h}_s}  + |\theta| \| B_0(k) \widehat{|u|^2}\|_{\mathfrak{h}_s} \leq  \| \alpha\|_{\mathfrak{h}_s}  + |\theta| \| |u|^2\|_{H^{s-5/2}}.
  \end{equation}
  For $n=1$ we simply bound the last term by $\| u\|_{L^4}^2 \leq C
  \|u\|_{H^1}^2$, and for $n\geq 2$ we use that
  \begin{equation}
    \|u^2\|_{H^{s-5/2}}\leq  \|u^2\|_{H^n} \leq C\|u\|_{H^n}^2.
  \end{equation}
  This proves the claim.
\end{proof}

To connect the statements of Theorems \ref{thm:gross-transformed dynamics
  reduced density matrices} and \ref{theorem:final reduced density matrices},
we make use of the fact that the dressing flow $\mathfrak D$ interpolates
between the SKG flow \eqref{eq:Schroedinger-Klein-Gordon equations regular}
and the dressed mean-field flow \eqref{eq:SKG dressed}. This is a direct
consequence of Lemma \ref{lem:theta-mean-field-flow} for $\theta =1$.

\begin{lemma} \label{lemma:commuting:flow:x} For all $t\in \mathbb R$, we
  have
  \begin{align*}
    \mathfrak s^{\rm D}[t] \circ \mathfrak D = \mathfrak D \circ \mathfrak s[t].
  \end{align*}
\end{lemma}

The next lemma is an analogue to Theorem~\ref{thm:gross-transformed dynamics
  reduced density matrices} for the dynamics $W^\mathrm{D}(\theta)$ in
``time'' $\theta$.

\begin{lemma}
  \label{lemma:reduced densities dressing flow}
  There exists a constant $C>0$ such that for all $(u,\alpha) \in H^1(\mathbb
  R^3) \oplus \mathfrak h_0$ with $\norm{u}_\2 = 1 $, $\Psi_N \in \mathcal
  H_N$ with $\norm{\Psi_N} = 1$, $N\ge 1$, and $|\theta| \le 1$, we have
  \begin{align*}
    \beta \big[W^{\rm D}(\theta)\Psi_N, \mathfrak D [\theta] ( u , \alpha ) \big] & \le e^{C(\| u \|_{H^1}^2 + \| \alpha\|_{\mathfrak h_0} ) } \big( \beta \big[   \Psi_N,  (u,\alpha) ] + N^{-1}\big) .
  \end{align*}
\end{lemma}

To prove the lemma, we introduce, in close analogy to the discussion after
Theorem \ref{thm:gross-transformed dynamics reduced density matrices}, the
fluctuation generator associated with $W^{\rm D}(\theta)$ and $\mathfrak
D[\theta]$. For $\Psi_N\in \mathcal H_N$ and $(u,\alpha) \in H^1(\mathbb R^3)
\oplus \mathfrak h_0$, we consider the fluctuation vector
\begin{align}
  \label{eq:fluctuatoin:dynamics dressing dynamics}
  \zeta (\theta) := {X}_{\mathfrak D[\theta](u,\alpha) } W^{\rm D} (\theta)  \Psi_N.
\end{align}
A short computation shows that $ i \partial_{\theta} \zeta(\theta) = D^{\le
  N} _{u,\alpha} (\theta) \zeta_N(\theta)$ with
\begin{multline}
  D_{u,\alpha}^{\le N}(\theta)  =  X_{\mathfrak D[\theta](u,\alpha)  } \bigg( \tfrac{1}{\sqrt N}\sum_{j=1}^N \hat \Phi ( i B_{x_j} ) \bigg) (X_{ \mathfrak D[\theta](u,\alpha)  })^*
  \\+  i \left( \partial_{\theta} X_{\mathfrak D[\theta](u,\alpha) } \right) (X_{\mathfrak D[\theta](u,\alpha) })^* .
\end{multline}
As before, it is convenient to write $D^{\le N}_{u,\alpha} (\theta) =
D_{u,\alpha}(\theta) \restriction \cF_{\perp u}^{\le N} \otimes \cF$ as the
restriction of a symmetric operator $D_{u,\alpha}(\theta) : \cF \otimes \cF
\to \cF \otimes \cF$. After calculating $D_{u,\alpha}^{\le N}(\theta) $ by a
straightforward application of Lemma \ref{lemma:properties U}, we make the
choice
\begin{align} \label{eq:def:fluc:gen:dressing:trafo}
  D_{u,\alpha}(\theta) & = \textnormal{d}\Gamma_b  \big( \tau_{u,\alpha} \big) + \bigg( \int dx dk \,  \kappa_{u^\theta}( k,x )  a_k^* b_x^*     \big[ 1  - \tfrac{\mathcal N_b}{N}]_+^{1/2}  + \textnormal{ h.c.} \bigg) \notag \\
  &\quad -  \bigg( \int dx dk \,  \kappa_{u^\theta}( - k , x ) a_{ k} b_x^*  \big[ 1  - \tfrac{\mathcal N_b}{N}]_+^{1/2}  + \textnormal{ h.c.} \bigg)
  \nonumber \\
  &\quad + N^{-1/2} \int dx \, b_x^* \Big( q_{u^\theta} \hat \Phi (i B_x) q_{u^\theta} - \scp{u^{\theta}}{\hat \Phi (i B_{(\cdot)}) u^{\theta}}
  \Big) b_x,
\end{align}
where $\tau_{u,\alpha}$ is defined by \eqref{eq:def:tau:Lambda} and
\begin{equation}\label{eq:def:kappa:new}
  \kappa_{u}(k,x)=\big( q_u i B_{(\cdot)}(k)u \big)(x).
\end{equation}

\begin{proof}[Proof of Lemma \ref{lemma:reduced densities dressing flow}]
  Consider the fluctuation vector $\zeta( \theta ) $ given by
  \eqref{eq:fluctuatoin:dynamics dressing dynamics}. Using
  \eqref{eq:beta:second:quant}, we can express the relevant $\beta$
  functional as
  \begin{align}
    \beta[ W^{\rm D} (\theta) \Psi_N , \mathfrak D[\theta](u,\alpha)  ] & = \tfrac{1}{N} \scp{\zeta (\theta)}{\mathcal N \zeta(\theta)}.\label{eq:beta:second:quant:theta:new} 
  \end{align}
  We use $i\partial_\theta \zeta(\theta) = D_{u,\alpha}(\theta)
  \zeta(\theta)$, with $D_{u,\alpha}(\theta)$ given by
  \eqref{eq:def:fluc:gen:dressing:trafo}, in combination with the commutator
  bound (which is stated precisely and proved in Lemma
  \ref{lem:bounds:Bog:generator:Gross})
  \begin{align}
    \pm i [\mathcal N ,  D_{u,\alpha}(\theta) ]  & \le  C( \| u\|_{H^1}^2 + \| \alpha \|_{\mathfrak h_0} )    (\mathcal N+ 1) (1 + (\tfrac{1}{N}\mathcal N_b)^{1/2}).   \label{eq:key:bound:3:D(t):new}
  \end{align}
  Together with $\zeta(\theta) = \id_{\mathcal N_b \le N} \zeta(\theta)$,
  this implies
  \begin{align}
    \big| \tfrac{d}{d \theta}  \scp{\zeta(\theta)}{  \mathcal{N}  \zeta(\theta)} \big| &  \le C ( \| u\|_{H^1}^2 + \| \alpha \|_{\mathfrak h_0} )    \scp{ \zeta(\theta) }{(\mathcal N+1)  \zeta(\theta) }
  \end{align}
  and by applying Gr\"onwall's inequality, we obtain for $|\theta | \le 1$
  \begin{align}
    \scp{\zeta (\theta ) }{(\mathcal N+1) \zeta (\theta)}  \le  e^{C|\theta| ( \| u\|_{H^1}^2+ \| \alpha \|_{\mathfrak h_0} ) } \scp{\zeta(0)}{(\mathcal N +1)  \zeta(0)}.
  \end{align}
  In combination with \eqref{eq:beta:second:quant:theta:new} and the fact
  that $\| \zeta(0) \| =1$, we can derive the desired bound.
\end{proof}

As a final preparation for the proof of Theorem \ref{theorem:final reduced
  density matrices}, the following lemma gives an upper bound for the
functional $\gamma$, defined in~\eqref{eq:gamma:functional:def}, when
evaluated for the dressed states $W^{\rm D}\Psi_N$ and $\mathfrak
D(u,\alpha)$ in terms of the energy difference of the microscopic and
mean-field models evaluated in the states $\Psi_N$ and $(u,\alpha)$, without
dressing.

\begin{lemma}\label{lem:beta:gamma:relation} Let $(u,\alpha) \in H^3(\mathbb
  R^3) \oplus \mathfrak h_{5/2}$ with $\norm{u}_\2=1$ and $\mathcal E$ be
  given by \eqref{eq:energy functional Schroedinger-Klein-Gordon equations
    regular}. There exists a constant $C>0$ such that for all $\Psi_N \in D
  (H_N^{1/2})$ with $\norm{\Psi_N} = 1$, $N\ge 1$, we have
  \begin{multline*}\
    \gamma \big[ W^{\rm D} \Psi_N ,   \mathfrak D(u , \alpha )  \big]
    \leq C \Big(\abs{N^{-1} \scp{\Psi_N}{H_N \Psi_N} - \mathcal{E}( u , \alpha)} \\
    +\max_{j=1,2} ( \beta \left[ \Psi_N, ( u ,  \alpha ) \right] + N^{-1} )^{j/2}
    \Big) .
  \end{multline*}
\end{lemma}

For the proof of Lemma~\ref{lem:beta:gamma:relation}, which is given in
Section \ref{section:estimates-dressing}, it is important to note that the
strategy used to prove Theorem \ref{thm:gross-transformed dynamics reduced
  density matrices} does not work. The reason for this is that the operator
$\mathbb T$ is not dominated by the generator $D_{u,\alpha}(\theta)$, and
hence we do not have analogous estimates to
\eqref{eq:key:bound:1:L(t):new}--\eqref{eq:key:bound:4:L(t):new} at our
disposal. Therefore we rely on a different type of energy estimates,
motivated by ideas from \cite{KP2010}.

\subsection{Proof of Theorem \ref{theorem:final reduced density
    matrices}}\label{sect:final reduced density}

Combining the results of the previous two sections, we can prove our first
main theorem.

\begin{proof}[Proof of Theorem \ref{theorem:final reduced density matrices}]
  For $(u,\alpha)\in H^3(\mathbb R^3) \oplus \mathfrak h_{5/2}$, denote by
  \begin{align*}
    (u_t,\alpha_t) = \mathfrak s[t](u,\alpha), \quad (u^{\mathrm{D}},\alpha^{\mathrm{D}}) = \mathfrak D(u,\alpha),
    \quad  (u^{\mathrm{D}}_t,  \alpha^{\mathrm{D}}_t)  =  \mathfrak s^{\rm D}[t] \circ \mathfrak D(u,\alpha)
  \end{align*}
  the solutions of the SKG equations \eqref{eq:Schroedinger-Klein-Gordon
    equations regular}, and the dressed SKG equations \eqref{eq:SKG dressed}
  with initial data transformed by the mean-field dressing
  transformation~\eqref{eq:dressing:flow:theta}. Also recall that by Lemma
  \ref{lemma:commuting:flow:x}
  \begin{align}\label{eq:D1:sD}
    \mathfrak D(u_t,\alpha_t) = \mathfrak s^{\rm D}[t](u^{\mathrm{D}},\alpha^{\mathrm{D}}).
  \end{align}
  Since $W^{\rm D} (-1) W^{\rm D}(1) =\id $ and $\mathfrak D[-1]\circ
  \mathfrak D=\id$, we can use Lemma \ref{lemma:reduced densities dressing
    flow} for $\theta =-1$ and with $\Psi_N \to W^{\rm D} e^{-i t H_N }
  \Psi_N$, and $(u,\alpha) \to \mathfrak D(u,\alpha)$, to get
  \begin{multline}\label{eq:beta:proof:remark}
    \beta[ e^{-i t  H_N }\Psi_N, (u_t, \alpha_t)]
    \\ \leq e^{C(\| u^{\rm D}_t \|_{H^1}^2 + \| \alpha_t^{\rm D}\|_{\mathfrak h_0} ) }  \left( \beta \big[ W^{\rm D} e^{- i t H_N }   \Psi_N , \mathfrak D(u_t , \alpha_t) \big]
      + N^{-1} \right) .
  \end{multline}
  The exponential factor is uniformly bounded in $t$, since by Lemma
  \ref{lem:mf-dressing-regularity} and Proposition \ref{prop:skg well posed}
  \begin{subequations}
    \begin{align}\label{eq:bounds:uD:alpha:D}
      \| u^{\rm D}_t \|_{H^1}^2 & \le C   \| u_t \|_{H^1}^2 \| \alpha_t \|^2_{\mathfrak h_{1/2}} \le C \,,  \\
      \| \alpha^{\rm D}_t \|_{\mathfrak h_0}  &\le C(  \| u_t \|_{H^1}^2 +  \| \alpha_t \|_{\mathfrak h_{1/2}} )\le C .
    \end{align}
  \end{subequations}
  In view of $H_N = (W^{\rm D})^* H_N^{\rm D} W^{\rm D}$ (see Lemma
  \ref{lem:H:D:representation}) and \eqref{eq:D1:sD}, we can proceed using
  Theorem \ref{thm:gross-transformed dynamics reduced density matrices} to
  estimate
  \begin{multline}
    \beta \big[ e^{- i t H_N^{\rm D} } W^{\rm D}  \Psi_N ,  \mathfrak s^{\rm D}[t](u^{\mathrm{D}} , \alpha^{\mathrm{D}}) \big] \\[1mm]
    \leq   e^{C R^{\rm D}(t) } \big( \beta \big[ W^{\rm D}  \Psi_N , (u^{\mathrm{D}}, \alpha^{\mathrm{D}}) \big]
    + \gamma \big[  W^{\rm D}  \Psi_N , (u^{\mathrm{D}} , \alpha^{\mathrm{D}} )  \big]
    + N^{-1} \big),
  \end{multline}
  where $R^{\mathrm{D}}(t) = 1 + \int_0^{|t|}
  \norm{u^{\mathrm{D}}_s}_{H^{3}}^2 (1+
  \norm{\alpha^{\mathrm{D}}_s}_{\h{3/2}})^2 ds $. The $\beta$ functional on
  the right side is estimated with the aid of Lemma \ref{lemma:reduced
    densities dressing flow},
  \begin{align}
    \beta \big[  W^{\rm D}  \Psi_N , \mathfrak D (u,\alpha)  \big]
    &\leq  C\big(  \beta \big[ \Psi_N , (u , \alpha) ] + N^{-1} \big).
  \end{align}
  The $\gamma$ functional is bounded by Lemma \ref{lem:beta:gamma:relation},
  which yields altogether
  \begin{multline}
    \beta[ e^{-i t  H_N }\Psi_N, (u_t, \alpha_t)]
    \leq  e^{C R^{\rm D}(t) } \Big(   \abs{N^{-1} \scp{\Psi_N}{H_N \Psi_N} - \mathcal{E}( u , \alpha)} \\ + \max_{j=1,2}\big(\beta \big[ \Psi_N , (u , \alpha ) \big] +N^{-1} \big] \big)^{j/2} \Big).
  \end{multline}
  It remains to relate the time-dependent pre-factor to the solution of the
  SKG equation.  Using Lemma~\ref{lem:mf-dressing-regularity} and the fact
  that $(u_s^\mathrm{D}, \alpha^{\mathrm{D}}_s)=\mathfrak{D}(u_s, \alpha_s)$
  we have
  \begin{align}
    \|u^{\mathrm{D}}_s\|_{H^{3}} (1+ \norm{\alpha^{\mathrm{D}}_s}_{\mathfrak{h}_{3/2}}) &\leq C \|u_s\|_{H^3} \|\alpha_s\|_{\mathfrak{h}_{5/2}}^3 (1+ \|\alpha_s\|_{\mathfrak{h}_{3/2}} + \|u_s\|_{H^1}^2) .
 %
  \end{align}
  By Proposition \ref{prop:skg well posed} we have $\|u_s\|_{H^1}^2\leq
  C$. Young's inequality then yields
  \begin{equation}\label{eq:RD(t):bound}
    R^{\rm D}(t)  \leq  1+ C\int_0^{|t|}  (\|u_s\|_{H^3}^{10} + \|\alpha_s\|_{\mathfrak{h}_{5/2}}^{10}) ds ,
  \end{equation}
  and this proves the claim.
\end{proof}


\section{Bogoliubov theory and the norm approximation}\label{chap:Bogoliubov}

In this section we study the Bogoliubov approximation of $e^{-i t H_N }$ and
prove Theorem \ref{thm:norm:approximation}.  Our strategy is similar to the
case of the mean-field approximation discussed in
Section~\ref{chap:mean-field}. We start by introducing the dressed Bogoliubov
Hamiltonian $\mathbb{H}^\mathrm{D}_{u,\alpha}(t)$ and the associated Fock
space evolution $\mathbb U^{\rm D}_{u,\alpha}(t)$, which describe the
fluctuations around the mean-field solution for the dressed dynamics. In
Theorem~\ref{thm:norm:approximation:dressed:dynamics} we provide a statement
analogous to Theorem~\ref{thm:norm:approximation} for the dressed case.  In
Section~\ref{sec:dressing:fluctuations} we then study the norm approximation
of the dressing transformation, which is given again in terms of a suitable
Bogoliubov type evolution. This evolution describes the fluctuations with
respect to the mean-field dressing $\mathfrak D[\theta]$. We use the norm
approximation of the dressing to relate the statements of Theorem
\ref{thm:norm:approximation} and
\ref{thm:norm:approximation:dressed:dynamics}. For that purpose, it is
crucial to observe that the Bogoliubov approximation of $W^{\rm D}(\theta)$
in fact interpolates between the dressed and undressed Bogoliubov
evolutions. This is stated in Proposition \ref{prop:id:Bog:ren:new}, whose
proof is given in Section \ref{sec:renormalization}.

\subsection{Norm approximation of the dressed
  dynamics}\label{sec:dressed:fluctuations}

We introduce the Bogoliubov evolution describing the fluctuations associated
with the dressed Nelson Hamiltonian $H_N^{\rm D}$ and the dressed mean-field
equations \eqref{eq:classical:transformed:equations}. To this end, we
consider the quadratic approximation of the fluctuation generator
$H_{u,\alpha}^{\rm D}(t)$ introduced in \eqref{eq:fluctuation:generator:2}
(see Lemma \ref{lem:fluctuation:generator} for the explicit form of $H^{\rm
  D}_{u,\alpha}(t)$). For $(u,\alpha) \in H^3(\mathbb R^3) \oplus \mathfrak
h_{5/2}$ and $(u_t,\alpha_t) = \mathfrak s^{\rm D} [t](u,\alpha)$, we
introduce the quadratic operator acting on $\cF \otimes \cF$ given by
\begin{align}
  \label{eq:HD:Bog definition}
  \mathbb{H}_{u,\alpha}^{\rm D}(t) & = \textnormal{d}\Gamma_b(  h_{t} )+   \mathbb K^{(1)}_{u_t} + \big( \mathbb K^{(2)}_{u_t}  + \text{h.c.} \big)  +  \textnormal{d}\Gamma_a(  \omega  )   \\
  & \ +   \int dx  dk \Big( \big( q_{u_t} L_{ \alpha_t}(k) u_t\big)(x) a^{*}_k b_x^*  +  \big(q_{u_t} L_{ \alpha_t}(k)^* u_t\big)(x) a_{k} b_x^* \Big)
  + \text{h.c.}
  \notag \\
  & \ +  \int dk dl \Big(  -2  M_{u_t}(k,-l) a^*_k a_l + M_{u_t}(k,l) a_k^* a_l^* +   M_{u_t}(-k,-l) a_k a_l \Big)\notag  %
\end{align}
with $h_{t }=h_{u_t,\alpha_t }$ as defined in
\eqref{eq:def:dressed:meanfield:Ham},
\begin{subequations}
  \begin{align} \label{eq:L(k):operator}
    \big(L_{ \alpha}(k)u\big)(x) &=  2 k B _{ x  }(k) \big((-i \nabla +  F_\alpha(x))u\big)(x), \\
    M_u(k,l) & = \scp{u}{k B_{(\cdot)}(k) \cdot l B_{(\cdot)}(l)
      u}, \label{eq:def:M(k,l)}
  \end{align}
  with $F_\alpha$ given by \eqref{eq:def:F_alpha}, and
  \begin{align}
    \mathbb K^{(1)}_{ u} & = \int dx dy \,  K_{   u }^{(1)}(x ,y ) b^*_{x} b_{y} ,\quad
    \mathbb K^{(2)}_{ u } = \frac{1}{2}\int dx dy \, K_{ u }^{(2)}(x ,y) b^*_{x} b^*_{y} \label{eq:def:K(2)}
  \end{align}
  where
  \begin{align}
    K_{ u}^{(1)}= q_u \widetilde K^{(1)}_{ u}q_u , \qquad   \widetilde K_{  u}^{(1)}  (x,y) & = u(x) V (x-y)  \overline{u(y)}, \\
    K_{ u}^{(2)} = (q_u \otimes q_u) \widetilde K^{(2)}_{  u} , \qquad   \widetilde K_{  u}^{(2)}  (x,y) & = u(x) V (x-y) u(y)
  \end{align}
  with $q_u = 1 - |u \rangle \langle u|$ and $V(x)$ defined in
  \eqref{eq:definition of V}.
\end{subequations}

The next proposition on the evolution generated by the operator $\mathbb
H^{\rm D}_{u,\alpha}(t)$ is the special case $\theta=1$, $\Lambda=\infty$ of
Proposition \ref{prop:theta-Bog} below.

\begin{proposition}\label{prop:Bog-D}
  Let $(u,\alpha)\in H^3(\R^3)\oplus \mathfrak{h}_{5/2}$ with
  $\norm{u}_{L^2}=1$ and let $(u_t,\alpha_t ) = \mathfrak s^{\rm
    D}[t](u,\alpha)$ be given by \eqref{eq:SKG dressed}. For every $\Psi \in
  D(( \mathcal{N} + \mathbb{T})^{1/2})$ there exists a unique solution to the
  Cauchy problem
  \begin{equation*}
    \label{eq:U_1:evolution}
    \begin{cases}
      \begin{aligned}
        i\partial_t \Psi(t) & =  \mathbb H^{\rm D}_{u,\alpha}(t)\Psi(t) \\
        \Psi(0)&= \Psi_0
      \end{aligned}
    \end{cases}
  \end{equation*}
  such that $\Psi \in C(\R,\cF \otimes \cF)\cap L_{\rm{\loc}}^{\infty} (\R,
  D(( \mathcal{N} + \mathbb{T})^{1/2}))$. The solution map $\Psi_0 \mapsto
  \Psi(t)$ extends to a unitary $\mathbb{U}_{u,\alpha}^{\rm D}(t)$ on $\cF
  \otimes \cF$ satisfying $\mathbb U^{\rm D}_{u,\alpha}(t) ( \cF_{\perp u}
  \otimes \cF ) \subseteq \cF_{\perp u_t} \otimes \cF$. Moreover, for every
  $\ell \in \mathbb N$ there is a constant $C(\ell)$ such that for all $t\in
  \mathbb R$,
  \begin{align*}
    \mathbb U^{\rm D}_{u,\alpha}(t)^* (\mathbb T + \mathcal N + 1 ) \mathbb U^{\rm D}_{u,\alpha}(t)  & \le  e^{C(1) R^{\rm D}(t)} (\mathbb T + \mathcal N +1), \\[0mm]
    \mathbb U_{u,\alpha}^{\rm D}(t)^* (\mathcal N+1)^\ell \, \mathbb U_{u,\alpha}^{\rm D} (t) & \le e^{C(\ell)  R^{\rm D}(t)} \, (\mathcal N+1)^\ell .
  \end{align*}
  in the sense of quadratic forms on $\cF \otimes \cF$, with $R^{\rm D}(t) =
  1 + \int_0^{|t|} \norm{u_s}_{H^3}^2 (1+ \norm{\alpha_s}_{\mathfrak
    h_{3/2}})^2 ds $.
\end{proposition}

The next theorem is our main statement of this section, making precise that
the evolution $\mathbb U^{\rm D}_{u,\alpha}(t)$ describes the fluctuations of
the dressed Nelson dynamics around the dressed mean-field solutions.

\begin{theorem}\label{thm:norm:approximation:dressed:dynamics}
  Let $(u,\alpha) \in H^3(\mathbb R^3) \oplus \mathfrak h_{5/2}$ with
  $\norm{u}_\2 = 1 $, and $\mathbb U^{\rm D}_{u,\alpha}(t)$ the unitary
  defined in Proposition \ref{prop:Bog-D}. There exists a constant $ C>0$
  such that for all $\chi \in \cF_{\perp u} \otimes \cF$ with
  $\norm{\chi}=1$, $N\ge 1$ and $t\in \mathbb R$,
  \begin{align*}
    \Big\| e^{-i t H_N^{\rm D} } X_{u,\alpha}^* \chi  - X_{\mathfrak s^{\rm D}[t](u,\alpha)}^* \mathbb U^{\rm D}_{u,\alpha} (t)\chi  \Big\| \le e^{ C R^{\rm D}(t)}     \delta^{\rm D}(\chi)^{1/2} \frac{ \sqrt{\log N}}{N^{1/4}} ,
  \end{align*}
  where $R^{\rm D}(t) = 1 + \int_0^{|t|} \norm{ u_s }_{ H^3 }^2 ( 1 +
  \norm{\alpha_s}_{ \mathfrak{h}_{3/2}} )^2 ds $ with $(u_s, \alpha_s) =
  \mathfrak s^{\rm D}[s](u,\alpha)$ and
  \begin{equation*}
    \delta^{\rm D}(\chi) = \norm{ (1+\cN ^3  + \textnormal{d}\Gamma_b(-\Delta) + \textnormal{d} \Gamma_a ( \omega) )^{1/2} \chi}^2.
  \end{equation*}
\end{theorem}

The proof of the theorem relies on a technical bound on the difference of
$H_{u,\alpha}^{\rm D}(t)$ and its quadratic approximation $\mathbb
H_{u,\alpha}^{\rm D}(t)$ in terms of the operators $\mathcal N$ and $\mathbb
T$. In details, there exists a constant $C>0$, such that for all $\chi \in
\cF^{\le N} \otimes \cF$, $\phi\in \cF \otimes \cF$,
\begin{align}\label{eq:HD:difference:Bound}
  & \big| \scp{\chi }{\left( H^{\rm D}_{u,\alpha}(t) - \mathbb{H}_{u,\alpha}^{\rm D}(t) \right)  \phi} \big| \notag\\[1mm]
  &\hspace{1cm} \le C \rho(t) \frac{ \ln N}{N^{1/2} } \norm{\left( \mathcal{N} + \mathbb{T} + 1 \right)^{1/2} \chi }
  \norm{\left( \mathcal{N}^3 + \mathbb{T} + 1 \right)^{1/2} \phi}
\end{align}
with $\rho(t) = \norm{u_t}_{\Sob{3}}^2 (1+\norm{\alpha_t}_{\h{3/2}} )^2$. The
precise statement and its proof are given in Lemma
\ref{lem:fluctuationgen:BogHam:difference}. Note that we choose to distribute
the higher moments of $\mathcal{N}$ unequally in
\eqref{eq:HD:difference:Bound} because below we will rely on the estimates
provided by Theorem \ref{thm:gross-transformed dynamics reduced density
  matrices} and Proposition \ref{prop:Bog-D} that control the higher moments
of $\mathcal{N}$ during the Bogoliubov dynamics but only its first moment
during the many-body evolution.

\begin{proof}[Proof of Theorem \ref{thm:norm:approximation:dressed:dynamics}]
  Using that the excitation map is isometric, we can write the norm
  difference as
  \begin{align}
    \big\| e^{-i t H_N^{\rm D} } X_{u,\alpha}^* \chi  - X_{\mathfrak s^{\rm D}[t](u,\alpha)}^* \mathbb U^{\rm D}_{u,\alpha}(t)  \big\| =  \big\| \chi^{\rm D}(t) - \id_{\mathcal N_b \le N} \mathbb U^{\rm D}_{u,\alpha}(t) \chi \big\|,
  \end{align}
  where $\chi^{\rm D}(t) = X_{ \mathfrak s^{\rm D}[t](u,\alpha) } e^{-i t
    H_N^{\rm D} } X_{u,\alpha}^*\chi $ and by \eqref{eq:X-isometry},
  $\chi^{\rm D}(0)=\id_{\mathcal N_b \le N} \chi$.  Since $\id_{\mathcal N_b
    \le N} \chi^{\rm D}(t) = \chi^{\rm D}(t)$ we can omit the projection on
  the right-hand side by increasing the value of the norm. Recalling
  \eqref{eq:fluctuation equation}, we now use $i\partial_t \chi^{\rm D}(t) =
  H^{\rm D}_{u,\alpha} (t) \chi^{\rm D}(t)$ and Proposition \ref{prop:Bog-D}
  to obtain
  \begin{align}
    \tfrac{d}{dt} \norm{\chi^{\rm D}(t) - \mathbb U_{u,\alpha}^{\rm D} (t) \chi }^2 = 2 \Im \scp{ \chi^{\rm D}(t) }{( H_{u,\alpha}^{\rm D}(t) - \mathbb H_{u,\alpha}^{\rm D} (t))  \mathbb U_{u,\alpha}^{\rm D} (t) \chi  }.
  \end{align}
  Using again $\id_{\mathcal N_b \le N} \chi^{\rm D}(t) = \chi^{\rm D}(t)$,
  we can apply \eqref{eq:HD:difference:Bound} to bound the right-hand side,
  so that
  \begin{multline}
    \tfrac{d}{dt} \norm{\chi^{\rm D}(t)  -  \mathbb U_{u,\alpha}^{\rm D}(t) \chi }^2  \\
    \le C \rho(t)  \norm{\left( \mathcal{N} + \mathbb{T} + 1 \right)^{1/2} \chi^{\rm D}(t) }
    \norm{\left( \mathcal{N}^3 + \mathbb{T} + 1 \right)^{1/2} \mathbb U_{u,\alpha}^{\rm D}(t) \chi }\frac{\ln N}{\sqrt{N}}.
  \end{multline}
  By Proposition \ref{prop:Bog-D}, we have
  \begin{align}
    \norm{( \mathcal{N}^3 + \mathbb{T} + 1 )^{1/2} \mathbb U_{u,\alpha}^{\rm D}(t) \chi }  \le e^{CR^\mathrm{D}(t) } \norm{ (\mathcal N^3 +\mathbb T +1 )^{1/2} \chi }
  \end{align}
  and by the same argument as in the proof of Theorem
  \ref{thm:gross-transformed dynamics reduced density matrices},
  \begin{align}
    \scp{\chi^{\rm D}(t)}{( \mathcal N + \mathbb T +1 ) \chi^{\rm D}(t) } \le e^{CR^\mathrm{D}(t)} \scp{\chi^{\rm D}(0) } {(\mathcal N + \mathbb T +1 ) \chi^{\rm D}(0) }.
  \end{align}
  Since $\|\chi^{\rm D}(0) - \chi\|\leq N^{-1} \|\cN^{1/2}\chi\|^2$, this
  implies the desired bound
  \begin{align}
    \norm{\chi^{\rm D}(t) - \mathbb U_{u,\alpha}^{\rm D}(t) \chi }^2 
    &\leq  e^{C R^\mathrm{D}(t)} \norm{ ( \mathcal{N}^3 + \mathbb T + 1 )^{1/2} \chi}^2N^{-1/2} \ln N.
  \end{align}
  which concludes the proof of Theorem
  \ref{thm:norm:approximation:dressed:dynamics}.
\end{proof}

\subsection{Norm approximation of the dressing
  transformation}\label{sec:dressing:fluctuations}

We now consider the dressing transformation on the level of the
fluctuations. The effective dressing transformation is used for two
purposes. First, it allows for a norm approximation of $W^{\rm D}(\theta)$
and, second, it provides an interpolation between the undressed and dressed
Bogoliubov evolutions. Since the undressed and the dressed Bogoliubov
Hamiltonians are both quadratic operators on $\cF \otimes \cF$, see
\eqref{eq:reg:Nelson:Bog} and \eqref{eq:HD:Bog definition}, it is natural to
choose the dressing transformation that interpolates between the two itself
as an evolution generated by a quadratic operator. The right candidate for
this is the Bogoliubov type approximation of the microscopic dressing $W^{\rm
  D}(\theta)$ associated with the mean-field flow $\mathfrak D[\theta]$ (more
precisely, the evolution generated by the quadratic approximation of the
fluctuation generator \eqref{eq:def:fluc:gen:dressing:trafo}).

For $(u,\alpha)\in L^2(\mathbb R^3) \oplus L^2(\mathbb R^3)$,
$(u^\theta,\alpha^\theta) = \mathfrak D [\theta](u,\alpha)$, $\Lambda \in
\mathbb R_+ \cup \{ \infty\}$, consider the quadratic operator
\begin{align}\label{eq:Gross-Bog-generator}
  \mathbb D^\Lambda_{u ,\alpha} (\theta)  =  \textnormal{d}\Gamma_b(  \tau_{u, \alpha}) + \bigg(
  \int dx d k\,\big( \kappa_{u^\theta}^\Lambda(k,x) a^*_k   b_x^* - \kappa^\Lambda_{u^\theta}(-k,x)  a_k    b^*_x   \big) + \text{h.c.}   \bigg)
\end{align}
with $\tau_{u, \alpha}$ defined in \eqref{eq:def:tau:Lambda} and
\begin{equation}\label{eq:def:kappa}
  \kappa^\Lambda_{u}(k,x)=\big( q_u iB^\Lambda_{(\cdot)}(k)u\big)(x), \quad B^\Lambda_x (k) = \id_{|k|\le \Lambda } B_x(k).
\end{equation}
Since the unitary generated by $\mathbb{D}_{u,\alpha}^\Lambda$ will play an
important role in the renormalization of the Nelson--Bogoliubov Hamiltonian
in Section \ref{sec:renormalization}, we introduce the kernel
$\kappa^\Lambda_u$ with a cutoff $\Lambda\in \R_+\cup\{\infty\}$.
The next proposition states the existence and some important properties of
this unitary evolution.

\begin{proposition}\label{prop:Gross-Bog:evolution}
  Let $(u,\alpha)\in H^1(\R^3)\oplus L^2(\mathbb{R}^3)$ with
  $\norm{u}_{L^2}=1$ and let $(u^\theta,\alpha^\theta) = \mathfrak
  D[\theta](u,\alpha)$ denote the solution to \eqref{eq:classical dressing
    equations} with initial datum $(u,\alpha)$. For every $\Lambda \in
  \mathbb R_+ \cup \{ \infty\}$ and $\Psi_0 \in D(\cN^{1/2})$ there exists a
  unique solution to the Cauchy problem
  \begin{equation*}
    \begin{cases}
      \begin{aligned}
        i\partial_\theta \Psi(\theta)
        &= \mathbb{D}_{u,\alpha}^\Lambda(\theta)\Psi(\theta) \\
        \Psi(0)&=\Psi_0
      \end{aligned}
    \end{cases}
  \end{equation*}
  such that $\Psi \in C(\R, \mathcal{F} \otimes \mathcal{F}) \cap
  L_{\rm{loc}}^{\infty}(\R, D(\cN^{1/2}))$.  The solution map $\Psi_0 \mapsto
  \Psi(\theta)$ defines a unitary $\mathbb{W}_{u ,\alpha }^\Lambda(\theta)$
  on $\cF \otimes \cF $ with the following properties
  \begin{itemize}
  \item[\textnormal{(i)}] $\mathbb W^\Lambda_{u ,\alpha }(\theta) (\cF_{\perp
      u} \otimes \cF) = \cF_{\perp u^\theta}\otimes \cF$.
  \item[\textnormal{(ii)}] $(u,\alpha,\Lambda) \mapsto \mathbb W_{u ,\alpha
    }^\Lambda(\theta)$ is strongly continuous.
  \item[\textnormal{(iii)}] $\mathbb W^\Lambda_{u ,\alpha }(\theta)$ is a
    Bogoliubov transformation.
  \item[\textnormal{(iv)}] For every $\ell \in \mathbb N$ there exists a
    constant $C(\ell )>0$ such that for all $(u,\alpha) \in H^1(\mathbb R^3)
    \oplus L^2(\mathbb{R}^3)$ with $\| u \|_{\2}=1$, $\Lambda \in \mathbb R_+
    \cup \{ \infty\}$ and $|\theta| \le 1$,
    \begin{subequations}
      \begin{align}
        \mathbb W^\Lambda_{u ,\alpha }(\theta)^* (\mathcal N+1)^\ell \, \mathbb W^\Lambda_{u ,\alpha }(\theta) \le e^{C(\ell )( \| u\|_{H^1}^2 + \| \alpha \|_{L^2}) } (\mathcal N+1)^\ell \label{eq:bound:W*NW} \\
        \mathbb W^\Lambda_{u ,\alpha }(\theta) (\mathcal N+1)^\ell \, \mathbb W^\Lambda_{u ,\alpha }(\theta) ^*  \le e^{C(\ell )( \| u\|_{H^1}^2 + \| \alpha \|_{L^2}) } (\mathcal N+1)^\ell \label{eq:bound:WNW*}
      \end{align}
    \end{subequations}
    in the sense of quadratic forms on $\cF \otimes \cF$.
  \end{itemize}
\end{proposition}
\begin{proof}
  Existence and uniqueness of the solution follow from \cite[Theorem
  8]{Lewin:2015a} in combination with the bounds from
  Lemma~\ref{lem:bounds:Bog:generator:Gross}. As a consequence, there exists
  a two-parameter flow $\mathbb{W}_{u,\alpha}^\Lambda(\theta,\theta')$ that
  for $\theta,\theta',\theta''\in \R$ satisfies
  $\mathbb{W}_{u,\alpha}^\Lambda(\theta,\theta'')\mathbb{W}_{u,\alpha}^\Lambda(\theta'',\theta')=\mathbb{W}_{u,\alpha}^\Lambda(\theta,\theta')$. Indeed,
  define
  $\mathbb{W}_{u,\alpha}^\Lambda(\theta'+\vartheta,\theta')\Psi(\theta')$ to
  be the solution of $i\partial_{\vartheta} \chi(\vartheta) =
  \mathbb{D}^\Lambda_{u,\alpha}(\theta'+\vartheta)\chi(\vartheta)$ with
  $\chi(0)=\Psi(\theta')$. Then the flow property follows from uniqueness of
  the solution, since
  $\mathbb{W}_{u,\alpha}^\Lambda(\theta,\theta')\Psi(\theta')$ and
  $\Psi(\theta)$ are both solutions that agree at $\theta=\theta'$.
  Since $\mathbb{D}^\Lambda_{u,\alpha}(\theta)$ defines a symmetric quadratic
  form on $D(\cN^{1/2})$, the flow maps
  $\mathbb{W}_{u,\alpha}^\Lambda(\theta,\theta')$ are unitary.

  To show the mapping property, consider the orthogonal projection
  $\Gamma(q_{u^\theta})$ defined by $\Gamma(q_{u^\theta}) \restriction
  \mathcal F^{(n) } \otimes \mathcal F = (q_{u^\theta})^{\otimes n} \otimes
  1$.  Since $\Gamma(q_{u^\theta})( \cF \otimes \cF) = \cF_{\perp u^\theta}
  \otimes \cF$, proving that for all $\Psi \in D(\cN^{1/2})$
  \begin{align}\label{eq:proof:propagation:property:W}
    i \tfrac{d}{d\theta} \norm{\Gamma(q_{u^\theta}) \Psi (\theta) }^2 = - \scp{\Psi (\theta) }{[\mathbb D^\Lambda_{u ,\alpha } (\theta) - \text{d}\Gamma(\tau_{u,\alpha}),\Gamma(q_{u^\theta}) ] \Psi (\theta) } = 0
  \end{align}
  will imply $\mathbb W^\Lambda_{u ,\alpha }(\theta) (\cF_{\perp u} \otimes
  \cF) \subset \cF_{\perp u^\theta}\otimes \cF$.  This holds, since
  $\kappa^\Lambda_{u^\theta}(k,\cdot) \in \text{Ran}(q_{u^\theta})$ and thus
  \begin{align}
    \Big[\int d k dx\, \kappa_{u^\theta}^\Lambda(k,x) a_k^* b_x^* , \Gamma (q_{u^\theta}) \Big] & = \int dk\, a_k^* \Big[ b^*(\kappa^\Lambda_{u^\theta}(k,\cdot)) , \Gamma (q_{u^\theta}) \Big] \\
    &=  \int dk\, a_k^* b^*\big((1-q_{u^\theta})\kappa^\Lambda_{u^\theta}(k,\cdot)\big) \Gamma (q_{u^\theta}) =0, \notag
  \end{align}
  with similar calculations for the other terms in the commutator.  Applying
  the same argument to $\mathbb W^\Lambda_{u ,\alpha }(\theta) ^*=\mathbb
  W^\Lambda_{u ,\alpha }(\theta',\theta)\vert_{\theta'=0}$ shows that
  $\mathbb W^\Lambda_{u ,\alpha }(\theta)^* (\cF_{\perp u^\theta} \otimes
  \cF) \subset \cF_{\perp u}\otimes \cF$, which gives the equality.

  We now prove continuity of $(u,\alpha,\Lambda) \mapsto \mathbb
  W_{u,\alpha}^\Lambda (\theta) \Psi$ for every $\Psi\in \cF \otimes \cF$,
  $\theta\in \mathbb R$.  Let $u', \alpha'\in L^2(\R^3)$, $\Lambda' \in
  \R_+\cup\{\infty\}$.  By uniqueness of the solution, Duhamel's formula
  \begin{align}\label{eq:W-Duhamel}
    & \scp{\Phi}{( 1 - \mathbb {W}_{u,\alpha}^{\Lambda} (\theta)^*   \mathbb{W}_{u',\alpha'}^{\Lambda'}(\theta)) \Psi} \notag\\
    &\quad \qquad  = - i \int_0^\theta d\eta \scp{\Phi}{   \mathbb W^\Lambda_{u,\alpha} (\eta)^*  \big( \mathbb D^{\Lambda}_{u,\alpha}(\eta) - \mathbb D^{\Lambda'}_{u',\alpha'} (\eta)\big) \mathbb{W}_{u',\alpha'}^{\Lambda'}(\eta)) \Psi}
  \end{align}
  holds for all $\Psi,\Phi\in D(\cN^{1/2})$. As $(u',\alpha',\Lambda')\to
  (u,\alpha, \Lambda)$, $\tau_{ u , \alpha }$ tends to $\tau_{ u' , \alpha'
  }$ in $L^\infty(\mathbb R^3)$ and $\kappa_{u'}^{\Lambda'}$ tends to
  $\kappa_u^\Lambda$ in $L^2(\R^3\times \R^3)$, as one easily verifies. Thus,
  $\mathbb{D}_{u,\alpha}^{\Lambda}(\eta) -
  \mathbb{D}_{u',\alpha'}^{\Lambda'}(\eta)$ tends to zero as quadratic form
  on $D(\cN^{1/2})$ and since $\mathbb{W}_{u,\alpha}^{\Lambda}(\eta)$
  preserves $D(\cN^{1/2})$ the difference in~\eqref{eq:W-Duhamel} tends to
  zero.  Since $1 - \mathbb {W}_{u,\alpha}^{\Lambda} (\theta)^*
  \mathbb{W}_{u',\alpha'}^{\Lambda'}(\theta)$ is uniformly bounded, we can
  extend the convergence to arbitrary $\Phi, \Psi \in \cF\otimes \cF$, so
  $\mathbb{W}_{u',\alpha'}^{\Lambda'}(\theta) \to \mathbb
  {W}_{u,\alpha}^{\Lambda} (\theta)$ in the weak operator topology. Since
  this is a family of unitary operators, this implies convergence in the
  strong operator topology.

  The fact that $\mathbb W^\Lambda_{u,\alpha}(\theta)$ is a Bogoliubov
  transformation can be deduced from the property that $\tau_{u,\alpha} \in
  L^2(\mathbb R^3)$ and $\kappa_u^\Lambda \in L^2(\mathbb R^3 \times \mathbb
  R^3)$, which holds for all $\Lambda \in \mathbb R_+ \cup \{ \infty \}$. A
  proof of this well-known implication is provided in Appendix
  \ref{app:Bogoliubov:transformations} (for a different proof see
  e.g. \cite[Lem 4.8]{Bossmann2019}).

  To demonstrate the final statement, we rely on the fact that $ \mathbb
  W^\Lambda_{u,\alpha}(\theta)$ is a Bogoliubov transformation. This implies
  the existence of bounded linear maps $\mathfrak u ,\mathfrak v : L^2\oplus
  L^2 \to L^2\oplus L^2$, $\mathfrak v \in \mathfrak S^2(L^2 \oplus L^2 )$,
  such that
  \begin{align}
    \mathbb W^\Lambda_{u,\alpha}(\theta)^* c^*(f\oplus g)  \mathbb  W^\Lambda_{u,\alpha}(\theta) = c^*(\mathfrak{u} ( f\oplus g) ) + c( \mathfrak v  \overline{( f\oplus g )} )
  \end{align}
  for all $f,g \in L^2$, where $c^*(f\oplus g) = b^*(f) + a^*(g)$. Using this
  relation, one can deduce (see \cite[Lem. 4.4]{Bossmann2019}) that for all
  $\ell \in \mathbb N$
  \begin{align}
    \mathbb W^\Lambda_{u,\alpha}(\theta)^* (\mathcal N+1)^\ell \,  
    \mathbb W^\Lambda_{u,\alpha}(\theta) \le \ell^\ell ( 1 + 2 \norm{\mathfrak v }_{\mathfrak S_2} +  \norm{\mathfrak u} )^\ell \, (\mathcal N+1)^\ell
  \end{align}
  as quadratic forms on $\cF\otimes \cF$. To bound the norms on the
  right-hand side, we rely on the fact that $ \norm{\mathfrak v }_{\mathfrak
    S_2}^2 = \norm{ \mathcal N^{1/2} \mathbb
    W^\Lambda_{u,\alpha}(\theta)\Omega }^2 \le e^{C ( \| u \|^2_{H^1} +
    \|\alpha \|_{L^2} ) }$, where the estimate for $ \norm{ \mathcal N^{1/2}
    \mathbb W^\Lambda_{u,\alpha}(\theta)\Omega }^2 $ is obtained by
  estimating the derivative w.r.t. $\theta$, using the bounds from Lemma
  \ref{lem:bounds:Bog:generator:Gross} and then applying Gr\"onwall's
  inequality. Together, this proves \eqref{eq:bound:W*NW}. Now using the fact
  that the inverse of a Bogoliubov transformation is again a Bogoliubov
  transformation, the proof of \eqref{eq:bound:WNW*} can be carried out in a
  similar fashion. In fact, it is easy to show that $ \mathbb
  W^\Lambda_{u,\alpha}(\theta) c^*(f\oplus g) \mathbb
  W^\Lambda_{u,\alpha}(\theta)^* = c^*(\mathfrak{u}^*( f\oplus g) ) - c(
  \overline{\mathfrak v^* ( \bar f \oplus \bar g )} )$ and thus we can follow
  the same steps as in the proof of \eqref{eq:bound:W*NW}.
\end{proof}

\begin{lemma}\label{lem:W-reg}
  Let $(u,\alpha)\in H^2(\R^3)\oplus L^2(\R^3)$ with $\norm{u}_{L^2}=1$, and
  $\Lambda \in (0,\infty)$. Then for all $\theta\in \R$ and $r\in [0,1]$ we
  have
  \begin{equation*}
    \mathbb{W}_{u,\alpha}^\Lambda(\theta)D(\mathbb{\mathbb T}^r)\subset D(\mathbb{\mathbb T}^r)
  \end{equation*}
  where $\mathbb T = \textnormal{d}\Gamma_b(-\Delta) + \textnormal
  d\Gamma_a(\omega)$.
\end{lemma}

\begin{proof}
  For $u\in H^2(\mathbb R^3)$ and finite $\Lambda$, the coefficients of the
  generator, $\tau_{u,\alpha}(x)$ and $\kappa^\Lambda_{u^\theta}(k,x)$, are
  elements of $H^2(\R^3)$ in $x$ and compactly supported in $k$. The
  commutator $[\mathbb{T},\mathbb{\mathbb D}_{u,\alpha}^\Lambda(\theta)]$ is
  thus $\mathbb{T}$-bounded. From this the claim follows from Gr\"onwall's
  inequality and interpolation.
\end{proof}

The next lemma shows that the dressing transformation $W^{\rm D}(\theta)$
introduced in \eqref{eq:definition of the quantum dressing group:theta} is
effectively described by the evolution $\mathbb W_{u,\alpha}^\infty(\theta) $
introduced in Proposition \ref{prop:Gross-Bog:evolution}. The proof of the
lemma follows a similar strategy as the proof of Theorem
\ref{thm:norm:approximation:dressed:dynamics}.

\begin{lemma}\label{lemma:U:Lambda:approx:fluc} There exists a constant $C>0$
  such that for all $(u ,\alpha )\in H^1(\mathbb R^3) \oplus
  L^2(\mathbb{R}^3)$ with $\norm{u}_{L^2}=1$, $\chi \in \mathcal \cF_{\perp
    u} \otimes \cF $ with $\norm{\chi }=1$, $N\ge 1$ and $|\theta|\le 1$ ,
  \begin{align*}
    \big\|  W^{\rm D}  (\theta) X_{u,\alpha}^*\chi  -X_{\mathfrak D[\theta](u,\alpha)}^* \mathbb W^\infty_{u,\alpha}(\theta) \chi \big\|
    & \le e^{C ( \| u \|_{H^1}^2 + \| \alpha \|_{ L^2} ) }   \norm{ (1 + \mathcal N)^{\frac32} \chi } N^{-\frac12}.
  \end{align*}
\end{lemma}

\begin{proof} Let $\zeta(\theta) = X_{\mathfrak D[\theta](u,\alpha)} W^{\rm
    D}(\theta) X_{u,\alpha}^* \chi$ with $\zeta(0) = \id_{\cN_b\leq N}\chi$
  and write the norm difference as
  \begin{align}
    \big\|  W^{\rm D}  (\theta) X_{u,\alpha}^*\chi  -X_{\mathfrak D[\theta](u,\alpha)}^* \mathbb W^\infty_{u,\alpha}(\theta) \chi \big\| & =  \| \zeta(\theta) - \id_{\mathcal N_b \le N} \mathbb W^\infty_{u,\alpha}(\theta) \chi \| .
  \end{align}
  Observe that dropping the projection $\id_{\mathcal N_b \le N}$ from the
  norm increases its value, as $\id_{\mathcal N_b \le N} \zeta(\theta) =
  \zeta(\theta)$. Using $i\partial_\theta \zeta(\theta) =
  D_{u,\alpha}(\theta) \zeta(\theta)$, with $D_{u,\alpha}(\theta) $ given by
  \eqref{eq:def:fluc:gen:dressing:trafo}, and Proposition
  \ref{prop:Gross-Bog:evolution}, we obtain
  \begin{multline}
    \tfrac{d}{d\theta} \norm{\zeta(\theta) - \mathbb W^\infty_{u,\alpha}(\theta) \chi }^2\\
    =  2 \Im \scp{\zeta(\theta) -  \mathbb W^\infty_{u,\alpha}(\theta) \chi  }{\left( D_{u,\alpha}(\theta) - \mathbb{D}^\infty_{u,\alpha}(\theta) \right)  \mathbb W^\infty_{u,\alpha}(\theta) \chi  }.
  \end{multline}
  To bound the right-side, we employ Lemma
  \ref{lem:bounds:Bog:generator:Gross}, which states the existence of a
  constant $C$, such that for all $(u,\alpha) \in H^1(\mathbb R^3) \oplus
  L^2(\mathbb{R}^3)$ with $\| u\|_{L^2} = 1 $, $\phi,\chi \in \cF \otimes
  \cF$,
  \begin{align}
    \big| \langle \phi , ( D_{u,\alpha}(\theta) - \mathbb D_{u,\alpha}^\infty(\theta) ) \chi \rangle \big| \le C  \| \phi \| \, \| (\mathcal N+1)^{3/2} \chi \| N^{-1/2} .
  \end{align}
  With this at hand, the proof is readily finished, as
  \begin{multline}
    \tfrac{d}{d\theta} \norm{\zeta(\theta) - \mathbb W^\infty_{u,\alpha}(\theta) \chi }^2 \\
    \le C \| \zeta(\theta) -  \mathbb W^\infty_{u,\alpha}(\theta) \chi \|\,  \| (\mathcal N+1)^{3/2}  \mathbb W^\infty_{u,\alpha}(\theta) \chi  \| N^{-1/2} ,
  \end{multline}
  so using that $\norm{\zeta(0) - \chi } = \norm{\id_{\cN_b>N}\chi }\leq
  N^{-3} \|\cN^{3/2}\chi \|^2$ and integrating leads to
  \begin{align}
    \norm{\zeta(\theta) - \mathbb W^\infty_{u,\alpha}(\theta) \chi } \le C  \norm{(\mathcal N+1)^{3/2} \mathbb W^\infty_{u,\alpha}(\theta) \chi } N^{-1/2}.
  \end{align}
  The desired result now follows from \eqref{eq:bound:W*NW}.
\end{proof}

In the next proposition we make precise how the dressing $\mathbb
W^\infty_{u,\alpha}(1)$ interpolates between the dressed and undressed
Bogoliubov evolutions. This is the analogous statement to Lemma
\ref{lemma:commuting:flow:x} for the mean-field flows on the level of the
fluctuations. The proof of the proposition, along with the proof of the
existence of $\mathbb U_{u,\alpha}(t)$ (Theorem \ref{theorem:limit:unitary}),
is the subject of Section \ref{sec:renormalization} (see
Proposition~\ref{prop:id:Bog:ren}).

\begin{proposition}\label{prop:id:Bog:ren:new}
  Let $(u,\alpha)\in H^3(\R^3)\oplus \mathfrak{h}_{5/2}$ with
  $\norm{u}_{L^2}=1$, and $ \mathbb U $, $\mathbb U^{\rm D} $, $\mathbb
  W^\infty $ the evolutions introduced in Theorem \ref{theorem:limit:unitary}
  and Propositions \ref{prop:Bog-D}, \ref{prop:Gross-Bog:evolution}. For
  every $t\in \mathbb R$, we have
  \begin{align}
    \mathbb U^{\rm D}_{\mathfrak{D}(u,\alpha) } (t) \,
    \mathbb W^\infty_{u,\alpha}(1) \, = \,  \mathbb W_{\mathfrak s [t](u,\alpha)}^\infty (1) \,  \mathbb U_{u,\alpha}(t) .\notag
  \end{align}
\end{proposition}
Observe that the actions of the mean-field flows, with respect to which we
are considering fluctuations, on both sides of the identity agree.  Indeed,
on the left $(u,\alpha)$ is first evolved to $\mathfrak{D}(u,\alpha)$, by the
definition of the generator $\mathbb{D}_{u,\alpha}^\infty(\theta)$, which is
taken as the initial reference state in the dressed Bogoliubov transformation
$\mathbb U^{\rm D}_{\mathfrak{D}(u,\alpha)}$, whose generator includes the
evolution $\mathfrak{s}^\mathrm{D}$. On the right side, the reference states
evolve according to $(u,\alpha) \mapsto \mathfrak s [t](u,\alpha) \mapsto
\mathfrak D[1]\circ \mathfrak s [t](u,\alpha) $, which equals
$\mathfrak{s}^\mathrm{D}[t]\circ \mathfrak{D}(u,\alpha)$ by Lemma
\ref{lemma:commuting:flow:x}.

The next statement is the last step in the preparation for the proof of
Theorem \ref{thm:norm:approximation}. It is an immediate consequence of
Proposition~\ref{prop:id:Bog:ren:new} and Propositions \ref{prop:Bog-D},
\ref{prop:Gross-Bog:evolution}.

\begin{corollary}\label{Cor:Number:operator:bounds:U} Let $\mathbb
  U_{u,\alpha}(t)$ be the unitary of Theorem \ref{theorem:limit:unitary} for
  $(u_t,\alpha_t) = \mathfrak s[t](u,\alpha)$ and $(u,\alpha)$ as stated in
  the hypothesis. For every $\ell \in \mathbb N$ there exists a constant
  $C(\ell)$ so that for all $|t|\ge 0$
  \begin{align*}
    \mathbb U_{u,\alpha} (t)^* (\mathcal N+1)^\ell \, \mathbb U_{u,\alpha}(t) \le  e^{C (\ell ) R(t)} \, (\mathcal N+1)^\ell
  \end{align*}
  as quadratic forms on $\cF\otimes \cF$, where $R(t) = 1+ \int_0^{|t|}
  (\|u_s\|_{H^3}^{10} + \|\alpha_s\|_{\mathfrak{h}_{5/2}}^{10}) ds$.
\end{corollary}
\begin{proof} Starting from the identity of
  Proposition~\ref{prop:id:Bog:ren:new} and acting on both sides with
  $\mathbb W^\infty_{\mathfrak s[t](u,\alpha)}(1)^*$, we can express $\mathbb
  U_{u,\alpha}(t)$ in terms of three Bogoliubov evolutions, for each of which
  the corresponding bounds on number operators have been established in
  \eqref{eq:bound:W*NW}, \eqref{eq:bound:WNW*} and
  Proposition~\ref{prop:Bog-D}. Collecting the right constant, this gives
  \begin{multline}
    \mathbb U_{u,\alpha} (t)^* (\mathcal N+1)^\ell \, \mathbb U_{u,\alpha} (t) \\
    \le e^{C (\ell) ( R^{\rm D}(t) + \| u_t^{\rm D} \|^2_{H^1} + \| \alpha_t^{\rm D} \|_{L^2} + \| u \|^2_{H^1} + \| \alpha\|_{L^2} ) }  (\mathcal N+1)^\ell
  \end{multline}
  with $R^{\rm D}(t) = 1 + \int_0^{|t|} \| u_s^{\rm D} \|_{H^3}^2 (1 + \|
  \alpha^{\rm D}_s \|_{\mathfrak h_{3/2}} )^2 ds$ where $(u^{\rm D}_t ,
  \alpha^{\rm D}_t) = \mathfrak s^{\rm D}[t] \circ \mathfrak
  D(u,\alpha)$. That the exponential factor is bounded by $e^{C(\ell) R(t)}$
  follows from \eqref{eq:bounds:uD:alpha:D} and \eqref{eq:RD(t):bound}.
\end{proof}

\subsection{Proof of Theorem~\ref{thm:norm:approximation}}
\label{sect:proof:thm:norm}

We are now ready to prove our second main result.

\begin{proof}[Proof of Theorem \ref{thm:norm:approximation}]
  Let $(u,\alpha)\in H^3(\R^3)\oplus \mathfrak{h}_{5/2}$ with
  $\norm{u}_{L^2}=1$, $\chi \in \cF_{\perp u}\otimes \cF$, as given by the
  hypothesis.  Let $\mathbb W_{u,\alpha}^\infty(1)$ be the Bogoliubov
  approximation of the dressing transformation of Proposition
  \ref{prop:Gross-Bog:evolution} and $\mathbb U_{u,\alpha}(t)$ be the
  Nelson--Bogoliubov dynamics given by Theorem \ref{theorem:limit:unitary},
  and $\mathbb U_{\mathfrak{D}(u,\alpha)}^{\rm D}(t)$ be the dressed
  Bogoliubov dynamics of Proposition \ref{prop:Bog-D}.

  Using \eqref{eq:relation:dressed:undressed:into}, we can write the norm
  difference that we want to estimate in terms of the dressed dynamics,
  \begin{align}
    \mathscr D &  = \big\| e^{-i t H_N^{\rm D} }W^{\rm D} X_{u,\alpha}^* \chi - W^{\rm D} X^*_{\mathfrak{s}[t](u,\alpha)} \mathbb U_{u,\alpha} (t) \chi \big\|.
  \end{align}
  We split this into the difference of the dressed dynamics and its
  Bogoliubov approximation, and the corresponding approximation of the
  dressing transformation, i.e., we estimate $\mathscr D \le \mathscr D_1 +
  \mathscr D_2 + \mathscr D_3$ with
  \begin{subequations}
    \begin{align}
      \mathscr D_1 & =  \big\| e^{-i t H_N^{\rm D}}  W^{\rm D} X_{u,\alpha}^*\chi  - e^{-i t H_N^{\rm D}} X_{\mathfrak{D}(u,\alpha)}^*\mathbb W_{u,\alpha}^\infty(1)\chi \big\|  \\[1mm]
      \mathscr D_2 & =  \big\| e^{-i t H_N^{\rm D}} X_{\mathfrak{D}(u,\alpha)}^*\mathbb W_{u,\alpha}^\infty(1)\chi
      - X_{\mathfrak{s}^\mathrm{D}[t]\circ \mathfrak{D}(u,\alpha)}^* \mathbb U_{\mathfrak{D}(u,\alpha)}^{\rm D}(t)   \mathbb W_{u,\alpha}^\infty(1)\chi \big\| \\[2mm]
      \mathscr D_3 & = \big\| X_{\mathfrak{s}^\mathrm{D}[t]\circ \mathfrak{D}(u,\alpha)}^* \mathbb U_{\mathfrak{D}(u,\alpha)}^{\rm D}(t)  \mathbb W_{u,\alpha}^\infty(1)\chi   -  W^{\rm D} X_{\mathfrak{s}[t](u,\alpha)}^* \mathbb U_{u,\alpha} (t)\chi \big\|.
    \end{align}
  \end{subequations}
  Recalling $\mathfrak D = \mathfrak D[1]$ and applying Lemma
  \ref{lemma:U:Lambda:approx:fluc} with $\theta=1$ immediately gives
  \begin{align}
    \mathscr D_1 =  \big\|  W^{\rm D} X_{u,\alpha}^*\chi  -  X_{\mathfrak{D}(u,\alpha)}^* \mathbb W_{u,\alpha} ^\infty (1)\chi \big\|  \le C  \norm{(\mathcal N+1)^{3/2} \chi} N^{-1/2}.
  \end{align}
  By Proposition \ref{prop:id:Bog:ren:new} and Lemma
  \ref{lemma:commuting:flow:x}, we have
  \begin{align}
    \mathscr D_3  &= \big\| X_{\mathfrak D \circ \mathfrak{s}[t](u,\alpha)}^* \mathbb{W}^\infty_{\mathfrak{s}[t](u,\alpha)}(1)\mathbb U_{u,\alpha} (t)  \chi   -  W^{\rm D} X_{\mathfrak{s}[t](u,\alpha)}^* \mathbb{U}_{u,\alpha} (t)\chi \big\| ,
  \end{align}
  and applying Lemma \ref{lemma:U:Lambda:approx:fluc} together with
  Proposition~\ref{prop:skg well posed} then yields (here $(u_t,
  \alpha_t)=\mathfrak{s}[t](u,\alpha)$)
  \begin{align}
    \mathscr D_3 &\le e^{C( \| u_t \|_{H^1}^2 + \| \alpha_t \|_{L^2} ) } \big\| (\mathcal N+1)^{3/2} \mathbb{U}_{u,\alpha} (t)\chi \big\| N^{-1/2} \notag\\
    &\le C \big\| (\mathcal N+1)^{3/2} \mathbb{U}_{u,\alpha}(t)\chi \big\| N^{-1/2} .
  \end{align}
  Using Corollary~\ref{Cor:Number:operator:bounds:U} we arrive at the desired
  bound $ \mathscr D_3 \le e^{C R(t)} \norm{(\mathcal N+1)^{3/2}
    \chi}N^{-1/2}$.  To estimate $\mathscr D_2$, we apply Theorem
  \ref{thm:norm:approximation:dressed:dynamics} (note that
  $\mathfrak{D}(u,\alpha) \in H^3 \oplus \mathfrak h_{5/2}$ by Lemma
  \ref{lem:mf-dressing-regularity} and $(u,\alpha) \in H^3 \oplus \mathfrak
  h_{5/2}$) so that
  \begin{equation}
    \mathscr D_2 \le e^{C R^{\rm D} (t)} \delta^{\rm D}( \mathbb W_{u,\alpha}^\infty(1) \chi )^{1/2}   \sqrt{\ln N} N^{-1/4}
  \end{equation}
  with $R^{\rm D}(t) = 1 + \int_0^{|t|} \norm{ u^{\rm D}_s }_{ H^3 }^2 ( 1 +
  \norm{\alpha^{\rm D}_s}_{ \mathfrak{h}_{3/2}} )^2 ds$ and $(u_t^\mathrm{D},
  \alpha_t^\mathrm{D})=\mathfrak{s}^\mathrm{D}[t]\mathfrak{D}(u,\alpha)=\mathfrak{D}(\mathfrak{s}[t](u,\alpha))$.
  In view of the bound on $R^{\rm D}(t)$ stated in \eqref{eq:RD(t):bound},
  this yields the claimed bound with
  \begin{equation}
    \delta(\chi)= \delta^\mathrm{D}(\mathbb W_{u,\alpha}^\infty(1) \chi).
  \end{equation}
  The domain of $\delta$ is dense in $\cF\otimes \cF$ since it is the image
  under the continuous map $\mathbb W_{u,\alpha}^\infty(1)^*$ of the dense
  set $D((\cN^3+\mathbb{T})^{1/2})$. The set $D((\cN^3+\mathbb{T})^{1/2})\cap
  \cF_{\perp u^\mathrm{D}}\otimes \cF$ is also dense in $\cF_{\perp
    u^\mathrm{D}}\otimes \cF$, as $\Gamma(q_{u^\mathrm{D}})$ is continuous
  and leaves $D((\cN^3+\mathbb{T})^{1/2})$ invariant for $u^\mathrm{D}\in
  H^1$. By the mapping property from
  Proposition~\ref{prop:Gross-Bog:evolution}, we have $\mathbb
  W_{u,\alpha}^\infty(1)^* \cF_{\perp u^\mathrm{D}}\otimes \cF = \cF_{\perp
    u}\otimes \cF$, so the image of $ D((\cN^3+\mathbb{T})^{1/2})$ is also
  dense in $\cF_{\perp u}\otimes \cF$. This completes the proof.
\end{proof}


\section{Renormalization of the Nelson--Bogoliubov evolution}
\label{sec:renormalization}

This section is dedicated to proving the existence of the renormalized
Nelson--Bogoliubov evolution, which is stated in Theorem
\ref{theorem:limit:unitary}. To accomplish this, we consider a family of
Bogoliubov Hamiltonians $\mathbb{H}_{u,\alpha,\theta}^\Lambda$ that
interpolate between $\mathbb{H}_{u,\alpha}^\Lambda$ and
$\mathbb{H}_{u,\alpha}^\mathrm{D}$ with ``dressing parameter''
$\theta\in[0,1]$ and UV cutoff
$\Lambda<\infty$. 
The key point is that for $\theta=1$ we can remove the cutoff, i.e.,
$\mathbb{H}_{u,\alpha,1}^\infty=\mathbb{H}_{u,\alpha}^\mathrm{D}$, and that
the Bogoliubov transformations $\mathbb W^\Lambda_{u,\alpha}(\theta)$
associated with the dressing interpolate between the different members of
$(\mathbb{H}_{u,\alpha,\theta}^\Lambda)_{\theta \in [0,1]}$. We exploit this
property to prove Theorem~\ref{theorem:limit:unitary} in
Section~\ref{sect:unitary} by defining $\mathbb{U}_{u,\alpha}(t)$ as the
Bogoliubov transformation generated by $\mathbb{H}_{u,\alpha}^\mathrm{D}$ and
transformed by the Bogoliubov approximation of the dressing. Along with the
existence of $\mathbb U_{u,\alpha}(t)$, this also proves Proposition
\ref{prop:id:Bog:ren:new}.

Even though the general strategy in this section is motivated by Nelson's
original approach for renormalizing the Nelson Hamiltonian, the argument is
more involved on the level of the fluctuations. This is because the dressing
transformation and the different Nelson--Bogoliubov evolutions are all
generated by non-autonomous equations.

\subsection{Bogoliubov Hamiltonians with $\theta\in
  [0,1]$}\label{sect:theta-Bog}

In order to define the Bogoliubov Hamiltonians
$\mathbb{H}_{u,\alpha}^\Lambda(t)$, we first introduce the partially dressed
mean-field flow
\begin{equation}\label{eq:s_theta-definition}
  \mathfrak{s}_\theta [t]=\mathfrak{D} [\theta]\circ \mathfrak{s} [t]\circ \mathfrak{D} [-\theta],
\end{equation}
where $\mathfrak{s}$ is the flow of the Schr\"odinger--Klein--Gordon
system~\eqref{eq:Schroedinger-Klein-Gordon equations regular}. For the next
statement recall the definitions below \eqref{eq:SKG dressed} and $
\phi_\alpha(x) = 2 \Re \scp{G_x}{\alpha}$.

\begin{lemma}\label{lem:theta-mean-field-flow}
  For $(u, \alpha)\in H^1(\R^3)\oplus \mathfrak{h}_{1/2}$, the function
  $(u_t, \alpha_t)=\mathfrak{s}_\theta [t](u,\alpha)$ is the unique solution
  to the Hamiltonian equations
  \begin{align}\label{eq:classical:transformed:equations}
    \begin{cases}
      \begin{aligned}
        i \partial_t u_t (x)
        &=  h_{u_t,\alpha_t, \theta} u_t(x)
        \\[1.5mm]
        i \partial_t \alpha_t(k) &  = \omega (k)\alpha_t +(1-\theta) \langle u_t, G_{(\cdot)} (k)u_t\rangle \\[1mm]
        & \qquad + 2 \theta  \langle u_t, k B_{(\cdot)}(k) (-i\nabla + \theta F_{\alpha_t })u_t \rangle
      \end{aligned}
    \end{cases}
  \end{align}
  with conserved energy
  \begin{align}\label{eq:theta:mf-energy}
    &\mathcal{E}_\theta(u, \alpha) \\
    &\quad =\big\langle u , \big( - \Delta + (1-\theta) \phi_\alpha + \theta A_{\alpha} + \theta^2 (F_\alpha)^2
    +\tfrac{1}{2} V_\theta * \abs{u}^2  \big) u \big\rangle
    + \scp{\alpha}{\omega \alpha}, \notag
  \end{align}
  where\begin{subequations}
    \begin{align}\label{eq:def:h:Lambda:theta:t}
      h_{u,\alpha,\theta} & = -\Delta + (1-\theta)\phi_{\alpha } + \theta A_{\alpha } + \theta^2 (F_{\alpha})^2 + V_\theta \ast|u|^2  -\mu_{u,\alpha,\theta}   \\[1mm]
      \mu_{u,\alpha,\theta}&= \tfrac{1-\theta}{2}  \scp{u}{ \phi_{\alpha} u}  + \tfrac{1}{2} \scp{u }{ V_\theta \ast |u|^2 u } + \theta \Re \scp{\alpha }{ f_{u} + \theta g_{u,\alpha}} \\[1mm]
      V_\theta  (x)  & = -4 \theta \Re \scp{G_{x} }{B_{0}}   + 2 \theta ^2  \Re \scp{B_{x} }{\omega B_{0}}. \label{eq:V_theta}
    \end{align}
  \end{subequations}
\end{lemma}

\begin{remark} Since, for $\theta=1$,
  Eqs. \eqref{eq:classical:transformed:equations} coincide with the dressed
  mean-field equations \eqref{eq:SKG dressed}, we have $s_1[t] = \mathfrak
  s^{\rm D}[t]$.  Lemma \ref{lem:theta-mean-field-flow} thus provides a proof
  of Lemma \ref{lemma:commuting:flow:x}.
\end{remark}

\begin{proof}[Proof of Lemma \ref{lem:theta-mean-field-flow}]
  One checks by direct calculation~\cite[Prop.III.12]{AF2017} that
  \begin{equation}\label{eq:proof:conversation:mf:energy}
    \mathcal E_\theta = \mathcal E_0 \circ \mathfrak D[-\theta ].
  \end{equation}
  Since $\mathfrak{s}$ is a Hamiltonian flow of $\mathcal{E}_0$ with respect
  to the symplectic form
  \begin{equation}
    \sigma((u,\alpha), (u',\alpha')) = 2 \Im (\langle u,u'\rangle_{L^2} + \langle \alpha ,\alpha'\rangle_{L^2}),
  \end{equation}
  and the equations~\eqref{eq:classical:transformed:equations} are the
  Hamiltonian equations for $\mathcal{E}_\theta$, the claim should follow by
  showing that $\mathfrak{D}$ acts by symplectic transformations. However,
  since the involved spaces are infinite-dimensional, we need to take care of
  some domain issues. These are addressed by~\cite[Lem.6.9]{DBGRN}, so we
  will check the hypothesis of this Lemma (the reader might find it helpful
  to consult the examples given in~\cite[Sect.6.5]{DBGRN}).
  \begin{enumerate}
  \item $(u,\alpha)\mapsto\mathfrak{D}[\theta](u,\alpha)$ must be
    differentiable with derivative continuous on
    $E=H^1(\R^3)\oplus\mathfrak{h}_{1/2}$, and symplectic.  The derivative of
    the flow is
    \begin{equation}\label{eq:dressing-Frechet-derivative}
      \partial_{(v,\eta)} \mathfrak{D}[\theta](u,\alpha) = \Big(e^{-i\theta \tau_{u,\alpha}} v - i\theta(\tilde \phi_\eta -2\Re \langle u,\tilde \phi_\alpha v\rangle)u^\theta, \gamma + \theta B_0 2\widehat{\Re \bar u  v}\Big).
    \end{equation}
    As a linear function of $(v,\eta)$, this is continuous on
    $H^1(\R^3)\oplus\mathfrak{h}_{1/2}$, for all $(u,\alpha)\in
    H^1(\R^3)\oplus\mathfrak{h}_{1/2}$, as follows from the bounds of
    Lemma~\ref{lem:bounds:alpha:phi}.  Using the
    formula~\eqref{eq:dressing-Frechet-derivative}, one also checks that
    \begin{align}
      \sigma(\partial_{(v,\eta)} \mathfrak{D} , \partial_{(v',\eta')} \mathfrak{D})=\sigma((v, \eta), (v', \eta'),
    \end{align}
    by using that $B_0$ is an even function and observing that the mixed
    terms in $v, \eta'$ cancel each other (see
    also~\cite[Prop. IV.1]{AF2017}), i.e., $\mathfrak{D} $ is symplectic in
    the sense of~\cite[Def.6.5]{DBGRN}.

  \item The domain $\mathscr{D}$ on which the Hamiltonian vector field (given
    by Equation~\eqref{eq:classical:transformed:equations}) is defined must
    be invariant under the flow $\mathfrak{D}$. We take
    $\mathscr{D}=H^3(\R^3)\oplus \mathfrak{h}_{5/2}$, so this follows from
    Lemma \ref{lem:mf-dressing-regularity}.
  \item The derivative of $\mathcal{E}_\theta$ on $\mathscr{D}$ should be
    compatible with the symplectic structure as
    in~\cite[Def. 6.5]{DBGRN}. Since in our case the symplectic map is simply
    $\mathscr{J}=2 i$ with range $H^1(\R^3)\oplus \mathfrak{h}_{1/2}\subset
    E'$, we must check that for $(u,\alpha)\in \mathscr{D}$, the right hand
    side of Equation~\eqref{eq:classical:transformed:equations} is an element
    of $H^1(\R^3)\oplus \mathfrak{h}_{1/2}$. This follows from the bounds of
    Lemmas~\ref{lem:bounds:alpha:phi} and
    \ref{lem:bounds:f:g}.
  \end{enumerate}
  Thus the hypothesis of~\cite[Lem. 6.9]{DBGRN} are satisfied and this
  implies the claim.
\end{proof}

Next, we introduce a family of $\theta$- and $t$-dependent quadratic
operators $\mathbb H_{u,\alpha,\theta}^\Lambda(\theta)$ on $\cF \otimes \cF$
that is associated with the mean-field flow $\mathfrak s_\theta
[t](u,\alpha)$ and defined such that
\begin{align}
  \mathbb H_{u,\alpha,0}^\Lambda(t) = \mathbb H^\Lambda_{u,\alpha}(t)  \qquad \text{and} \qquad \mathbb H_{u,\alpha,1}^\infty (t) = \mathbb H^{\rm D}_{u,\alpha}(t)
\end{align}
for the Nelson--Bogoliubov Hamiltonian \eqref{eq:reg:Nelson:Bog} and the
dressed Bogoliubov Hamiltonian \eqref{eq:HD:Bog definition}.  These are
essentially the Bogoliubov Hamiltonians associated with the partially dressed
Hamiltonians $W^\mathrm{D}(\theta) H_N W^\mathrm{D}(\theta)^*$.  It is
important to note that for $\theta\in [0,1)$ the operator needs to be defined
with a UV cutoff $\Lambda < \infty$, while for $\theta =1$ the definition
makes sense also for $\Lambda=\infty$.

Now more concretely, for $(u,\alpha)$ and $\theta \in [0,1]$, let $ (u_t,
\alpha_t) = \mathfrak{s}_\theta[t](u, \alpha)$ denote the solution to
\eqref{eq:classical:transformed:equations}, and define
\begin{align}
  \label{eq:theta:HG:Bog}
  &\mathbb{H}_{u,\alpha,\theta}^{\Lambda}(t)  = \textnormal{d}\Gamma_b ( h_{u_t,\alpha_t,\theta} ) +  \mathbb K^{(1),\Lambda}_{\theta,u_t} + \big( \mathbb K^{(2),\Lambda}_{\theta,u_t}  + \text{h.c.} \big)  +  \textnormal{d}\Gamma_a(\omega)
  \\[1mm]
  & \quad +   \int dx  dk\, \Big( \big(q_{u_t} L^\Lambda_{\theta,\alpha_t}(k) u_t\big)(x) a^{*}_k b_x^*  +  \big(q_{u_t} L^\Lambda_{\theta,\alpha_t}(k)^* u_t\big)(x) a_{k} b_x^* \Big)
  + \text{h.c.}
  \notag \\
  & \quad + \theta^2 \int dk dl\, \Big(  -2  M_{u_t}^\Lambda(k,-l) a^*_k a_l + M^\Lambda_{u_t}(k,l) a_k^* a_l^* +   M^\Lambda_{u_t}(-k,-l) a_k a_l \Big)  \notag
\end{align}
with $h_{u_t,\alpha_t, \theta}$ defined in \eqref{eq:def:h:Lambda:theta:t},
\begin{subequations}
  \begin{align} \label{eq:def:L:theta:Lambda}
    \big(L^\Lambda_{\theta,\alpha}(k)u\big)(x) &= (1-\theta)G^\Lambda_{x}(k) + 2 B^\Lambda_{x}(k) k \big((-i\theta\nabla +\theta^2 F_\alpha(x))u\big)(x) \\
    \label{eq:def:M:Lambda}
    M^\Lambda_u(k,l) & = \scp{u}{k B_{(\cdot)}^\Lambda(k) \cdot l B_{(\cdot)}^\Lambda(l) u},
  \end{align}
  and
  \begin{align}
    \mathbb K^{(1),\Lambda}_{\theta,u} & = \int dx dy \,  K_{ \theta, u }^{(1),\Lambda}(x ,y ) b^*_{x} b_{y}, \quad
    \mathbb K^{(2),\Lambda}_{\theta,u }  = \frac{1}{2}\int dx dy \, K_{\theta,u }^{(2),\Lambda}(x ,y) b^*_{x} b^*_{y}
  \end{align}
  with
  \begin{align}
    K_{\theta,u}^{(1),\Lambda}= q_u \widetilde K^{(1),\Lambda }_{\theta, u}q_u , \qquad   \widetilde K_{\theta, u}^{(1),\Lambda}  (x,y) & = u(x) V_\theta^\Lambda(x-y)  \overline{u(y)}, \\
    K_{\theta,u}^{(2),\Lambda} = (q_u \otimes q_u) \widetilde K^{(2),\Lambda}_{\theta, u} , \qquad   \widetilde K_{\theta, u}^{(2),\Lambda}  (x,y) & = u(x) V_\theta^\Lambda(x-y) u(y),\label{eq:def:K2:theta}
  \end{align}
\end{subequations}
where $V_\theta$ is defined in \eqref{eq:V_theta}.  For $\theta=1$ and
$\Lambda = \infty$ these definitions coincide with those from Section
\ref{sec:dressed:fluctuations}.  For $\theta\neq 1$, the cutoff
$\Lambda<\infty$ is necessary, since the term involving $G_x^\Lambda(k)$ in
\eqref{eq:def:L:theta:Lambda} does not yield a quadratic form on
$D(\mathbb{T}^{1/2})$ for $\Lambda=\infty$. All other terms are in fact
unproblematic for $\Lambda =\infty$ also if $\theta \neq 1$. For the purpose
of the proof of Proposition \ref{prop:id:Bog:ren}, where cancellations
between different terms are important, we work with the definition given
above.

The next proposition states the existence and suitable properties of the
unitary time evolution generated by $\mathbb H^\Lambda_{u,\alpha,\theta}(t)$.

\begin{proposition}\label{prop:theta-Bog}
  Let $(u,\alpha)\in H^3(\R^3)\oplus \mathfrak{h}_{5/2}$ with
  $\norm{u}_{L^2}=1$ and let $(u_t,\alpha_t ) = \mathfrak
  s_\theta[t](u,\alpha)$ be given by \eqref{eq:s_theta-definition}. Moreover,
  let $\theta=1$ and $\Lambda=\infty$ or $\theta \neq 1$ and $\Lambda \in
  (0,\infty)$. For every $\Psi \in D( ( \mathcal{N} + \mathbb{T})^{1/2})$
  there exists a unique solution to the Cauchy problem
  \begin{equation*}
    \begin{cases}
      \begin{aligned}
        i\partial_t \Psi(t) & =  \mathbb H^\Lambda_{u,\alpha,\theta}(t)\Psi(t) \\
        \Psi(0)&= \Psi_0
      \end{aligned}
    \end{cases}
  \end{equation*}
  such that $\Psi \in C(\R,\cF \otimes \cF)\cap L_{\rm{\loc}}^{\infty} (\R,
  D(( \mathcal{N} + \mathbb{T})^{1/2}))$.  The solution map $\Psi_0 \mapsto
  \Psi(t)$ extends to a unitary $\mathbb{U}_{u,\alpha,\theta}^\Lambda(t)$ on
  $\cF \otimes \cF$ satisfying $\mathbb U^\Lambda_{u,\alpha,\theta}(t) (
  \cF_{\perp u} \otimes \cF ) \subseteq \cF_{\perp u_t} \otimes
  \cF$. Moreover, for $\theta=1$ we have the following properties.
  \begin{itemize}
  \item[\textnormal{(i)}] There is a constant $C>0$ such that for all
    $\Lambda \in \mathbb R_+\cup \{ \infty\}$ and $t\in \mathbb R$
    \begin{align*}
      \mathbb U_{u,\alpha,1}^\Lambda (t)^* (\mathbb T + \mathcal N) \mathbb U_{u,\alpha,1}^\Lambda (t)  \le  e^{C R^{\rm D}(t)} (\mathbb T + \mathcal N +1)
    \end{align*}
    in the sense of quadratic forms on $\cF \otimes \cF$, where $R^{\rm D}(t)
    = 1 + \int_0^{|t|} \norm{u_s}_{H^3}^2 (1+ \norm{\alpha_s}_{\mathfrak
      h_{3/2}})^2 ds $.
  \item[\textnormal{(ii)}] For every $t\in \R$
    \begin{equation*}
      \mathbb U^{\infty}_{u,\alpha,1}(t) = \slim_{\Lambda\to \infty} \mathbb U^{\Lambda}_{u,\alpha,1} (t).
    \end{equation*}
  \item[\textnormal{(iii)}] $\mathbb U^\Lambda_{u,\alpha,1}(t)$ is a
    Bogoliubov transformation for all $t\in \mathbb R$, $\Lambda \in \mathbb
    R_+ \cup \{\infty\}$.
  \item[\textnormal{(iv)}] For every $\ell \in \mathbb N$ there is a constant
    $C(\ell)$ such that for all $t\in \mathbb R$ and $\Lambda \in \mathbb R_+
    \cup \{ \infty \}$
    \begin{align*}
      \mathbb U_{u,\alpha,1}^\Lambda (t)^* (\mathcal N+1)^\ell\, \mathbb U_{u,\alpha,1}^\Lambda(t) \le e^{C (\ell) R^{\rm D}(t)} \, (\mathcal N+1)^\ell
    \end{align*}
    in the sense of quadratic forms on $\cF \otimes \cF$, with $R^{\rm D}(t)$
    as in \textnormal{(i)}.
  \end{itemize}
\end{proposition}
\begin{proof}
  Existence and uniqueness of the dynamics and Property (i) follow
  from~\cite[Thm.8]{Lewin:2015a} and the bounds of
  Lemma~\ref{lem:bounds:dressed:Bog:theta}. Note that the existence of the
  unitary $\mathbb U_{u,\alpha,\theta}^\Lambda(t)$ follows by the same
  reasoning as in the proof of Proposition~\ref{prop:Gross-Bog:evolution} and
  that the mapping property is obtained by a similar argument as in
  \eqref{eq:proof:propagation:property:W}.

  To prove (ii) use Duhamel's formula for $\Psi, \Xi\in D((\mathcal N +
  \mathbb{\mathbb T})^{1/2})$ together with Lemma
  \ref{lem:H:Bog:Lambda:difference} to obtain
  \begin{align}
    & | \langle \Xi, (1  -\mathbb{U}^{\infty}_{u,\alpha,1} (t)^* \mathbb{U}^{\Lambda}_{u,\alpha,1}  (t)) \Psi \rangle|  \\[1.5mm]
    &  \leq    \int_0^{|t|} ds\, |\langle \Xi, \mathbb U_{u,\alpha,1}^\infty(s)^* (\mathbb{H}_{u,\alpha,1}^{\infty}(s)  -\mathbb{H}_{u,\alpha,1}^{\Lambda}(s))\mathbb{U}^\Lambda_{u,\alpha,1}(s)\Psi \rangle| \notag\\
    &  \leq \varepsilon_\Lambda  \int_0^{|t|} ds\, e^{C R^{\rm D}(s) } \norm{(\mathbb T+ \mathcal N + 1)^{\frac12}  \mathbb U_{u,\alpha,1}^\infty(s)   \Xi }\,  \norm{ (\mathbb T + \mathcal N+ 1 )^{\frac12} \mathbb{U}^\Lambda_{u,\alpha,1}(s) \Psi } \notag
  \end{align}
  where $R^{\rm D}(s) = \norm{u_s}_{H^3}^2 (1+ \norm{\alpha_s}_{\mathfrak
    h_{3/2}})^2$ and $\varepsilon_\Lambda \to 0$ as $\Lambda \to 0$. The
  right hand side thus converges to zero by Property (i). This implies strong
  convergence of $\mathbb{U}_{u,\alpha,1}^\Lambda$ to
  $\mathbb{U}_{u,\alpha,1}^\infty$ by unitarity of
  $\mathbb{U}_{u,\alpha,1}^\Lambda$, $\Lambda\in \R_+\cup\{\infty\}$, as
  argued in the proof of Proposition~\ref{prop:Gross-Bog:evolution}.

  In Appendix \ref{app:Bogoliubov:transformations}, we provide a proof that
  $\mathbb U_{u,\alpha,1}^\Lambda(t)$ is a Bogoliubov transformation. For
  finite $\Lambda$, this is essentially due to the square-integrability of
  the kernels in \eqref{eq:theta:HG:Bog} appearing in the terms with $b^*
  b^*$, $b^* a^*$ and $a^*a^*$. For $\Lambda = \infty$, however, the kernel
  corresponding to $a^*b^*$ fails to meet the Hilbert--Schmidt criterion. In
  this case, we establish the statement by proving a suitable approximation
  argument and utilizing the fact that $\mathbb U_{u,\alpha,1}^\Lambda (t)\to
  \mathbb U_{u,\alpha,1}^\infty (t)$ converges strongly.

  The final statement can be derived using the same logic as in the proof of
  Proposition \ref{prop:Gross-Bog:evolution}, in combination with
  Proposition~\ref{prop:theta-Bog}(i).
\end{proof}

\subsection{Dressing identity and proof of
  Theorem~\ref{theorem:limit:unitary}}\label{sect:unitary}

We can now prepare for the proof of Theorem~\ref{theorem:limit:unitary}. We
will start by making precise how the unitaries $\mathbb{W}_{u,\alpha}^\Lambda
(\theta) $ interpolate between the Bogoliubov dynamics
$\mathbb{U}_{u,\alpha,\theta}^\Lambda (t) $ for different $\theta$. For
$\theta=1$ we can then take the limit $\Lambda\to \infty $ of
$\mathbb{U}_{u,\alpha,1}^\Lambda(t)$ to obtain information on the behavior of
$\mathbb{U}^\Lambda_{u,\alpha,0}(t)$ as $\Lambda\to \infty$.

Let $(u,\alpha) \in H^3(\mathbb R^3) \oplus \mathfrak h_{5/2}$ with $\| u\|
_{L^2}=1$, and consider the evolutions
\begin{equation} \label{eq:abbreviations:generators:2} \mathbb
  U^\Lambda_{\mathfrak{D} [\theta](u,\alpha), \theta}(t) \qquad \text{and}
  \qquad \mathbb W_{\mathfrak s_0[t](u,\alpha)}^\Lambda (\theta),
\end{equation}
with generators $ \mathbb H^\Lambda_{\mathfrak
  D[\theta](u,\alpha),\theta}(t)$ and $\mathbb D_{\mathfrak
  s_0[t](u,\alpha)}^\Lambda(\theta)$, respectively (see
Propositions~\ref{prop:theta-Bog} and \ref{prop:Gross-Bog:evolution}). It is
important to keep in mind that the subscripts refer to the {initial
  condition} of the mean-field flow, which is used to define the generator of
each evolution.

The next result shows that the two flows in
\eqref{eq:abbreviations:generators:2} commute, up to a global phase. The
additional phase is due to the fact that we wrote $\mathbb H_{\mathfrak{D}
  [\theta](u,\alpha), \theta}^\Lambda(t)$ in normal order, which is not
preserved by the transformations. As discussed below, this identity is the
key ingredient for the renormalization of the Nelson--Bogoliubov Hamiltonian.

\begin{proposition}\label{prop:id:Bog:ren}
  Let $(u,\alpha)\in H^3(\R^3)\oplus \mathfrak{h}_{5/2}$ with
  $\norm{u}_{L^2}=1$.  If we denote
  \begin{equation*}
    E^\Lambda_\theta= (2\theta-\theta^2) \langle G^\Lambda_0, B^\Lambda_0\rangle + \tfrac12 \langle u,V_\theta^\Lambda \ast |u|^2 u\rangle
  \end{equation*}
  then for all $ t , \theta \in \mathbb R$ and $\Lambda \in (0,\infty)$ we
  have the identity
  \begin{align*}
    \mathbb U^\Lambda_{u,\alpha,0} (t) e^{-i t E^\Lambda_\theta } = \mathbb W_{\mathfrak s_0[t](u,\alpha)}^\Lambda (\theta) ^* \, \mathbb U^\Lambda_{\mathfrak{D} [\theta](u,\alpha), \theta } (t) \,
    \mathbb W^\Lambda_{u,\alpha}(\theta) .
  \end{align*}
\end{proposition}

This proposition follows from the uniqueness of both sides by comparing their
derivatives. The lengthy calculation is given in
Section~\ref{sec:proof:prop:commutin:Bog:evolution}.  Assuming this for now,
we can prove Theorem \ref{theorem:limit:unitary} and Proposition
\ref{prop:id:Bog:ren:new}.

\begin{proof}[Proofs of Theorem~\ref{theorem:limit:unitary} and Proposition \ref{prop:id:Bog:ren:new}]
  By strong continuity of $\mathbb W_{u,\alpha}^\Lambda(\theta)$ in $\Lambda$
  (Lemma~\ref{prop:Gross-Bog:evolution}) and of
  $\mathbb{U}^\Lambda_{u,\alpha,1}(t)$ (Proposition~\ref{prop:theta-Bog}), we
  have
  \begin{multline}
    \slim_{\Lambda\to \infty}  \mathbb W_{\mathfrak s_0[t](u,\alpha)}^\Lambda (1)^* \, \mathbb U^\Lambda_{\mathfrak{D} (u,\alpha), 1} (t) \,
    \mathbb W^\Lambda_{u,\alpha}(0)  \\
    =  \mathbb W_{\mathfrak s_0[t](u,\alpha)}^\infty (1)^* \, \mathbb U^\infty_{\mathfrak{D} (u,\alpha), 1} (t) \,
    \mathbb W^\infty_{u,\alpha}(0)
  \end{multline}
  By Proposition~\ref{prop:id:Bog:ren} with $\theta=1$ and $\mathbb
  U_{u,\alpha}^\Lambda(t) = \mathbb U_{u,\alpha,0}^\Lambda(t) $ this shows
  that
  \begin{equation}\label{eq:renormalized-Bog}
    \mathbb{U}_{u,\alpha}(t)=\slim_{\Lambda\to \infty}\mathbb U^\Lambda_{u,\alpha,0} (t) e^{-i t E^\Lambda_1 }= \mathbb W_{\mathfrak s_0[t](u,\alpha)}^\infty (1)^* \, \mathbb U^\infty_{\mathfrak{D} (u,\alpha), 1} (t) \,
    \mathbb W^\infty_{u,\alpha}(0) .
  \end{equation}
  Strong continuity of $t\mapsto \mathbb{U}_{u,\alpha}(t)$ follows since the
  right hand side is strongly continuous, as for $\Psi \in \cF \otimes \cF$
  the maps
  \begin{equation}
    t\mapsto  \mathbb W_{\mathfrak s_0[t](u,\alpha)}^\infty (1)^* \Psi, \, \qquad t\mapsto \mathbb U^\infty_{\mathfrak{D} (u,\alpha), 1} (t) \Psi \, \qquad t\mapsto
    \mathbb W^\infty_{u,\alpha}(0) \Psi.
  \end{equation}
  are continuous.

  To show the mapping property for $\mathbb U_{u,\alpha}(t)$, consider
  $\Gamma(q_{u_t})$ as in the proof of Proposition
  \ref{prop:Gross-Bog:evolution} for $(u_t , \alpha_t) = \mathfrak
  s[t](u,\alpha)$. By the same argument as in
  \eqref{eq:proof:propagation:property:W} one shows that $\| \Gamma(q_{u_t})
  \mathbb U^\Lambda_{u,\alpha,0}(t)\Psi \| = 0$ for every $\Psi \in \cF
  \otimes \cF$ and all $\Lambda \in \mathbb R_+$, $t\in \mathbb R$ . The
  desired result then follows from $\norm{\Gamma(q_{u_t} ) \mathbb
    U_{u,\alpha,0}^\Lambda(t) \Psi}\rightarrow \norm{\Gamma(q_{u_t} ) \mathbb
    U_{u,\alpha} (t) \Psi}^2 $ as $\Lambda \to \infty$.

  The property that $\mathbb{U}_{u,\alpha} (t)$ is a Bogoliubov
  transformation is a direct consequence of the result that it is a
  composition of three Bogoliubov transformations.
\end{proof}

\begin{remark}[On the renormalized Nelson--Bogoliubov Hamiltonian]\label{rem:renormalized Bogoliubov}
  If the evolution $\mathbb{U}_{u,\alpha}(t)$ was a semi-group, we could
  deduce from Theorem~\ref{theorem:limit:unitary} the existence of a
  generator $\mathbb{H}_{u,\alpha}$ that would represent the renormalization
  of the Nelson--Bogoliubov Hamiltonian~\eqref{eq:reg:Nelson:Bog}.  But
  $\mathbb{U}_{u,\alpha}(t)$ is associated with a non-autonomous evolution
  equation and $\mathbb{H}_{u,\alpha}$ should depend on the time $t$. In this
  setting, only the following, weaker theory is available
  (see~\cite{nickel-1997} for details).  Consider the extension of
  $\mathbb{U}_{u,\alpha}(t)$ to a two-parameter family
  $\mathbb{U}_{u,\alpha}(t,s)$ with
  $\mathbb{U}_{u,\alpha}(t,0)=\mathbb{U}_{u,\alpha}(t)$. On the Banach space
  $C_\infty(\R, \cF\otimes \cF)$ of continuous $\cF\otimes \cF$-valued
  functions tending to zero at infinity, we can define the corresponding
  evolution semi-group of isometries by
  \begin{equation}
    (T(t)\chi)(s)=\mathbb{U}_{u,\alpha}(s,s-t)\chi(s-t).
  \end{equation}
  Then $T$ has a generator $A$, $D(A)\subset C_\infty(\R, \cF\otimes \cF)$,
  which corresponds formally to
  \begin{equation}
    (A\chi)(s)=\tfrac{d}{dt}T(t) \big|_{t=0}\chi(s) =\big(-i\mathbb{H}_{u,\alpha}(s) - \tfrac{d}{ds}\big)\chi(s).
  \end{equation}
  However, we do not have any information concerning the domain or
  self-adjointness of $\mathbb{H}_{u,\alpha}(t)$ (one could consider
  using~\cite[Thm. 2.9]{nickel-1997}, but this does not apply since the
  solutions provided by Proposition~\ref{prop:theta-Bog} need not be
  differentiable in the norm of $\cF\otimes \cF$).
\end{remark}

\subsection{Proof of the dressing identity}

\label{sec:proof:prop:commutin:Bog:evolution}

In this section, we derive the dressing identity for the Nelson--Bogoliubov
evolution, that is, we prove Proposition \ref{prop:id:Bog:ren}. To enhance
the clarity of the presentation, we will use the shorthand notation
\begin{equation}
  \begin{aligned}
    \label{eq:abbreviations:generators}
    \mathbb U^\Lambda_\theta (t) & :=  \mathbb U^\Lambda_{\mathfrak{D} [\theta](u,\alpha),\theta } (t) \quad \qquad && \mathbb W_t^\Lambda (\theta)  := \mathbb W^\Lambda_{\mathfrak{s}_0[t](u,\alpha)}(\theta) \\[1mm]
    \mathbb H^{\Lambda}_{\theta}(t) & := \mathbb H^\Lambda_{\mathfrak{D} [\theta](u,\alpha) ,\theta}(t) \qquad && \ 
    \mathbb D^\Lambda_{t}(\theta)  :=\mathbb D^\Lambda_{\mathfrak{s}_0[t](u,\alpha)}(\theta) \qquad && .
  \end{aligned}
\end{equation}
\begin{proof}[Proof of Proposition \ref{prop:id:Bog:ren}] Adopting
  \eqref{eq:abbreviations:generators}, the identity we aim to prove becomes $
  \mathbb U^\Lambda_{0} (t) e^{-i t E^\Lambda_1 } = \mathbb W_t^\Lambda
  (\theta )^* \, \mathbb U^\Lambda_{\theta } (t) \, \mathbb
  W^\Lambda_{0}(\theta)$. The idea is to prove a differential version of this
  identity after taking derivatives in $t$ and $\theta$.  To make this
  precise, first note that the coefficients of
  $\mathbb{D}_{t}^\Lambda(\theta)$ given in~\eqref{eq:Gross-Bog-generator}
  depend on $u,\alpha$ in a differentiable way. Since for $(u,\alpha)\in
  H^2(\R^3)\oplus \mathfrak{h}_1$ the flow $\mathfrak{s}_\theta[t](u,
  \alpha)$ is differentiable in $t$ in the $L^2$-sense, this implies that for
  any $\Psi \in D(\cN)$ the mapping $t \mapsto \mathbb{D}^\Lambda_{t}(\theta)
  \Psi$ is differentiable in $\cF \otimes \cF$. Since $\mathbb
  W_t^\Lambda(\theta)$ preserves $D(\cN)$, we deduce from Duhamel's
  formula~\eqref{eq:W-Duhamel} that $ \mathbb W_t^\Lambda(\theta)\Psi$ is
  differentiable in $t$ for $\Psi \in D(\cN)$ and
  \begin{equation}
    \big(i\partial_t \mathbb{W}_t^\Lambda(\theta)^*\big)  \mathbb{W}_t^\Lambda(\theta)\Psi = -\int_0^\theta d\eta\, \mathbb{W}_t^\Lambda(\eta)^* \big(\partial_t \mathbb{D}_t^\Lambda\big)(\eta) \mathbb{W}_t^\Lambda(\eta)\Psi.
  \end{equation}
  Denote by $\mathbb{V}_\theta^\Lambda(t):=\mathbb W_{t}^\Lambda(\theta )^*
  \mathbb U^{\Lambda}_{\theta}(t) \mathbb W^\Lambda_{0}(\theta)$ the right
  hand side of the identity we want to prove.  It follows from our previous
  considerations and Lemma~\ref{lem:W-reg} that for $\Psi \in D( (
  \mathcal{N} + \mathbb{T})^{1/2})$ we have
  \begin{align}
    i\partial_t \mathbb{V}_\theta^\Lambda(t) \Psi &=\big(i\partial_t\mathbb W_{t}^\Lambda(\theta )^*\big) \mathbb U^{\Lambda}_{\theta}(t) \mathbb W^\Lambda_{0}(\theta)\Psi
    + \mathbb W_{t}^\Lambda(\theta )^*  \mathbb{H}_\theta^\Lambda (t) \mathbb  U^{\Lambda}_{\theta}(t) \mathbb W^\Lambda_{0}(\theta)\Psi \notag\\
    &=: \mathbb{B}_\theta(t)\mathbb{V}_\theta^\Lambda(t)\Psi
  \end{align}
  in $D(( \mathcal{N} + \mathbb{T})^{-1/2})$.  By uniqueness of the solutions
  proved in Proposition~\ref{prop:theta-Bog} our claim will follow if we can
  show that the generators of $ \mathbb{U}_0^\Lambda(t) e^{-iE_\theta^\Lambda
    t }$ and $\mathbb{V}_\theta^\Lambda(t)$ are equal, that is for every
  $t\in \R$,
  \begin{align}
    \mathbb{H}_0^\Lambda(t)  + E_\theta^\Lambda&=\mathbb{B}_\theta(t)   \\
    &=-\int_0^\theta d\eta\, \mathbb{W}_t^\Lambda(\eta)^* \big(\partial_t \mathbb{D}_t^\Lambda\big)(\eta) \mathbb{W}_t^\Lambda(\eta) + \mathbb W_{t}^\Lambda(\theta )^*\mathbb{H}_\theta^\Lambda(t)\mathbb W_{t}^\Lambda(\theta ). \notag
  \end{align}
  Equality holds for $\theta=0$ since $E_0^\Lambda=0$ and $\mathbb
  W_t^\Lambda(0)=1$, so it is sufficient to prove that for all $\Psi, \Xi\in
  D(( \mathcal{N} + \mathbb{T})^{1/2})$
  \begin{align}
    0&=i\partial_\theta \Big\langle \Xi, \Big( \mathbb{H}_0^\Lambda(t)  + E_\theta^\Lambda - \mathbb{B}_\theta(t) \Big) \Psi \Big\rangle  \label{eq:Bog-commute-form}\\
    &=\Big\langle \mathbb W_{t}^\Lambda(\theta )\Xi, \Big(   i \partial_\theta E_\theta^\Lambda + i\partial_t \mathbb{D}_t^\Lambda(\theta) - i\partial_\theta \mathbb{H}_\theta^\Lambda(t) -  [\mathbb{H}_\theta^\Lambda(t),\mathbb{D}_t^\Lambda(\theta)]\Big) \mathbb W_{t}^\Lambda(\theta )\Psi \Big\rangle.
    \notag
  \end{align}
  where we anticipated differentiability of $\mathbb{H}_\theta^\Lambda$ which
  follows easily from the explicit calculation of its derivative below.

  The remainder of the proof is an explicit calculation of this quadratic
  form.
  For ease of presentation, we set, using that the flows commute,
  \begin{equation}\label{eq:def:ut:alphat}
    (u_t^\theta, \alpha_t^\theta):=\mathfrak{s} _\theta[t]\circ \mathfrak{D}[\theta](u,\alpha) = \mathfrak{D}[\theta]\circ \mathfrak{s}_0[t](u,\alpha) .
  \end{equation}
  Moreover, we do not make the dependence of the different objects on
  $t,\theta,\Lambda$ explicit everywhere and adopt the following shorthand
  notation
  \begin{subequations}
    \begin{align}
      h_\theta & =   h_{u_t^\theta,\alpha_t^\theta,\theta}\\
      \tau_t & =  \tau_{u_t,\alpha_t}\\
      q&=q_{u^\theta_t}=1-| u_t\rangle \langle u^\theta_t| \\
      \kappa^\Lambda _t (k,x)&=\big(q iB_{(\cdot)}^\Lambda(k) u_t^\theta\big)(x) \\
      L^\Lambda_\theta(k)&=(1-\theta)G^\Lambda_{(\cdot)}+2B^\Lambda_{(\cdot)}(k)k(-i\theta\nabla+\theta^2 F_{\alpha_t^\theta}(\cdot))\\
      M^\Lambda (k,l)&=\langle u^\theta_t, kB^\Lambda_{(\cdot)}(k)lB^\Lambda_{(\cdot)}(l)u^\theta_t\rangle\\
      K^{(1),\Lambda}_\theta &=K^{(1),\Lambda}_{\theta, u_t^\theta} \\
      K^{(2),\Lambda}_\theta &=K^{(2),\Lambda}_{\theta, u_t^\theta} \\
      V^\Lambda_\theta(x)  & =  -4 \theta \Re \scp{G^\Lambda_{x}}{B_{0}} + 2 \theta^2 \Re \scp{B_{x}^\Lambda}{\omega B_{0}}
    \end{align}
  \end{subequations}
  Note that $M^\Lambda (k,l)$ is independent of $\theta$ as the flow
  $\mathfrak{D}$ preserves the modulus of $u$.

  After commuting $\mathbb{H}_\theta$ and $\mathbb{D}_t$, the expression
  from~\eqref{eq:Bog-commute-form} takes the form
  \begin{subequations}
    \begin{align}
      & - i \partial_t \mathbb{D}_t^\Lambda(\theta)  + i \partial_\theta \mathbb H^{\Lambda}_\theta(t)  + [\mathbb H^{\Lambda}_\theta(t),\mathbb{D}_t^\Lambda(\theta)]  \notag \\[1mm]
      &  \ = \int d x \, b_x^* \big( -i\partial_t \tau_{t} +  i\partial_\theta h_{\theta}  +[h_{\theta}, \tau_{t}] \big)b_x\, \label{eq:Bog-commute-mf} \\
      &   \quad  + \int dx dy \,  D(x,y)   b_x^*b_y   + i E_b+  \bigg(  \int dx dy\, \widetilde D(x,y) \, b_x^*b_y^* - \text{h.c.} \bigg)  \label{eq:Bog-commute-bb}\\
      &   \quad +  \int dk dx \Big(  X (k,x) \, a_k^* b_x^* +   \widetilde X (k,x) \, a_k b_x^* \Big) - \text{h.c.}   \label{eq:Bog-commute-ab} \\
      &  \quad + \int dk dl\,  A(k,l) a_k^* a_{l} + i E_a  +  \bigg( \int dk dl  \,  \widetilde{A}(k,l) a_k^* a_{l}^* - \text{h.c.} \bigg)  \label{eq:Bog-commute-aa} ,
    \end{align}
  \end{subequations}
  in the sense of sesquilinear forms on $D(( \mathcal{N} +
  \mathbb{T})^{1/2})$.  We will now show that all the coefficients of
  creation and annihilation operators vanish, and that $E_a+E_b=
  \partial_\theta E^\Lambda_\theta$.\medskip

  \noindent \textbf{Mean-field part~\eqref{eq:Bog-commute-mf}}: By
  \eqref{eq:classical dressing equations}, Lemma
  \ref{lem:theta-mean-field-flow} and \eqref{eq:def:ut:alphat} we have
  $\partial_\theta u^\theta_t = -i \tau_t u^\theta_t$ and $\partial_t
  u_t^\theta = -i h_\theta u_t^\theta$, and since the derivatives commute
  \begin{align}
    0 = ( \partial_t \partial_\theta -  \partial_\theta \partial_t )  u_t^\theta  = \big( - (i \partial_t \tau_t) + i (\partial_\theta h_\theta) + [ h_\theta , \tau_t ] \big) u_t^\theta.
  \end{align}
  Since \eqref{eq:def:ut:alphat} is a bijection, this implies that
  \eqref{eq:Bog-commute-mf} vanishes.\medskip

  \noindent \textbf{Terms quadratic in $b,b^*$ and
    $E_b$~\eqref{eq:Bog-commute-bb}}: In the commutator $[\mathbb
  H^{\Lambda}_\theta(t),\mathbb{D}_t^\Lambda(\theta)]$ a quadratic term in
  $b,b^*$ can arise either by commuting two terms with $b^\#a^\#$, or terms
  with $b^\#b^\#$. In the first case, the coefficients are combinations of
  $L^\Lambda_\theta(k)$ and $\kappa^\Lambda_t(k, \cdot)$ integrated over
  $k$. In the latter case, some terms have already been taken into account
  in~\eqref{eq:Bog-commute-mf} and only the commutators of
  $K^{(1),\Lambda}_\theta$, $K^{(2),\Lambda}_\theta$ with $\tau_t$
  remain. Combining these with the derivatives of $K^{(1),\Lambda}_\theta$,
  $K_\theta^{(2),\Lambda}$ and putting them in normal order yields
  \begin{subequations}
    \begin{align}
      D(x,y) & = i \partial_\theta K^{\Lambda,(1)}_{\theta}(x,y) - \big( \tau_t(x) -  \tau_t (y) \big)  K^{\Lambda,(1)}_{\theta}(x,y) \notag \\
      & \quad - \int dk \,   \big(q(L^\Lambda_\theta(k)+L^\Lambda_\theta(-k)^*)u^\theta_t\big)(x) \overline{\kappa^\Lambda_t(k,y)} \notag \\
      &\quad + \int dk \,    \kappa^\Lambda_t(k,x)\overline{\big(q(L_\theta^\Lambda(k)+L^\Lambda_\theta(-k)^*)u^\theta_t\big)(y)} ,  \\
      E_b & =  i \int dxdk \,   \overline{\kappa^\Lambda_t (k,x)} \big(q(L^\Lambda_\theta(k)+L^\Lambda_\theta(-k)^*)\big)u^\theta_t(x),  \\[1mm]
      \widetilde D(x,y) & = \frac{i}{2} \partial_\theta K_{ \theta}^{(2),\Lambda}(x,y) - \frac12 ( \tau_t(x) + \tau_t(y))  K_{ \theta}^{(2),\Lambda}(x,y)  \notag\\
      & \quad + \int dk   \, \big( q(L^\Lambda_\theta(k)^* +L^\Lambda_\theta(-k))u^\theta_t  \big) (x)\kappa^\Lambda_t(k,y).
    \end{align}
  \end{subequations}
  We now show that $D(x,y)= 0$. First, we may observe that $i\partial_\theta
  u^\theta_t = \tau_t u^\theta_t$ and $q =1- |u^\theta_t\rangle \langle
  u^\theta_t|$ imply for any operator $T_\theta$ the identity
  \begin{equation}\label{eq:dtheta-qTu}
    i\partial_\theta ( q T_\theta u^\theta_t ) =  \tau_t  q T_\theta u^\theta_t + q [T_\theta,\tau_t] u^\theta_t+ q (i\partial_\theta T_\theta) u^\theta_t  .
  \end{equation}
  Since $[V^\Lambda_\theta , \tau_t ] = 0$, this gives us
  \begin{align}\label{eq:K1:derivative}
    & i \partial_\theta K^{(1),\Lambda}_{\theta}(x,y) - \big( \tau_t(x) -  \tau_t (y) \big)  K^{(1),\Lambda}_{\theta}(x,y) \notag\\
    & \qquad = \int dz dz' q(x,z) u^\theta_t(z) (i\partial_\theta V^\Lambda_\theta(z-z'))\bar u^\theta_t(z') q(z',y).
  \end{align}
  To evaluate the terms involving
  $L^\Lambda_\theta(k)+L^\Lambda_\theta(-k)^*$, we first calculate using
  $\overline{B^\Lambda_x}(-k)=B^\Lambda_x(k)$
  \begin{align}
    (L^\Lambda_\theta(k)+L^\Lambda_\theta(-k)^*)u^\theta_t(x)&=\big((1-\theta)2 G^\Lambda_x(k) +2\theta k^2  B^\Lambda_x(k) \big) u^\theta_t(x).
  \end{align}
  This gives
  \begin{align}
    &-\int dk \,   \big(q(L^\Lambda_\theta(k)+L^\Lambda_\theta(-k)^*)u\big)(x) \overline{\kappa^\Lambda_t(k,y)} \\
    & =  \int dz dz' q(x,z) u^\theta_t(z)\Big(2i (1-\theta) \langle B^\Lambda_{z'},G^\Lambda_z\rangle +2i\theta  \langle k^2 B^\Lambda_{z'},B^\Lambda_z\rangle \Big) \bar u^\theta_t(z')  q(z',y). \notag
  \end{align}
  Adding this and minus its complex conjugate with $x,y$ exchanged (which
  leads to an exchange of $z,z'$) gives with $(k^2+\omega(k))B_x(k)=G_x(k)$
  \begin{align}
    \int\hspace{-6pt} dz dz' q(x,z) u^\theta_t(z)\big(
    \underbrace{ 4i(1-\theta) \Re\langle B^\Lambda_{z'},G^\Lambda_z\rangle +4i\theta  \Re\langle k^2 B^\Lambda_{z'},B^\Lambda_z\rangle}_{=4i \Re\langle B^\Lambda_{z'},G^\Lambda_z\rangle -4\theta i \Re\langle \omega B^\Lambda_{z'},B^\Lambda_z\rangle=-i\partial_\theta V^\Lambda_\theta}
    \big)
    \bar u^\theta_t(z')  q(z',y).
  \end{align}
  Combined with \eqref{eq:K1:derivative}, this shows that $D\equiv 0$.

  For $\widetilde D$ this follows from the same calculation, using
  that $\langle B_{x}^\Lambda, G_y^\Lambda\rangle=\widehat{B_0^\Lambda
    G_0^\Lambda}(x-y)$ is real-valued since $G,B$ are even functions.  By the
  same reasoning, the value of the constant is
  \begin{align}
    E_b &=\begin{aligned}[t]
      - i \int dx  dz dz' q(x,z)  u^\theta_t(z)&\Big(2i (1-\theta) \langle B^\Lambda_{z'},G^\Lambda_z\rangle \\
      &  +2i\theta  \langle k^2 B^\Lambda_{z'},B^\Lambda_z\rangle \Big) \bar u^\theta_t(z')  q(z',x)
    \end{aligned} \notag\\
    &= - \frac{i}{2} \int dz dz' q(z',z) (-i \partial_\theta V^\Lambda_\theta(z-z') u^\theta_t(z) \bar u^\theta_t(z')) \notag\\[1.5mm]
    &=   2 \Re \langle G_0^\Lambda, B_0^\Lambda \rangle - 2\theta \Re\langle \omega B^\Lambda_{0},B^\Lambda_0\rangle + \partial_\theta \tfrac12 \langle u, V^\Lambda_\theta\ast |u|^2 u\rangle  ,\label{eq:Bog-commute-E_b}
  \end{align}
  where we used that $|u^\theta_t|^2$ is independent of $\theta$ and
  integrates to one.\medskip

  \noindent \textbf{Mixed terms in $a,a^*$,
    $b,b^*$~\eqref{eq:Bog-commute-ab}}. Mixed terms with, say, $a^*_kb^*_x$
  arise from the derivatives of the respective terms in $\mathbb
  H^{\Lambda}_\theta$, $\mathbb{D}_t^\Lambda$, and from the commutator
  $[\mathbb H^{\Lambda}_\theta,\mathbb{D}_t^\Lambda]$ if one commutes a term
  with one $a^\#$ and one $b^\#$ with terms with two $a^\#$s or two $b^\#$s.
  The commutator
  \begin{align}
    \bigg[ &\int dk dl \Big(2M^\Lambda(k,-l)a^*_ka^*_l+ M^\Lambda(k,l)a^*_ka^*_l +  \text{h.c.}\Big), \notag  \\ 
    &\int dx dm\, \Big(\kappa^\Lambda_t(m,x)a^*_mb^*_x - \kappa^\Lambda_t(-m,x)a_mb^*_x \Big)+ \text{h.c}\bigg]
  \end{align}
  vanishes identically, as one easily checks.

  We group the remaining terms into two parts, $X=X^{(1)}+X^{(2)}$, which
  vanish separately.  Spelling things out, we set
  \begin{align*}
    X^{(1)} (k,x) &  = (i\partial_\theta -\tau_t(x)) (q L^\Lambda_\theta(k)u^\theta_t) (x)+   (- i\partial_t +h_\theta+\omega(k)) \kappa^\Lambda_t(k,x),  \\[1mm]
    X^{(2)}(k,x) & =  \int dy  \Big(   K_{\theta}^{(1),\Lambda}(x,y)  \kappa_t ^\Lambda(k,y)  +  K_{\theta}^{(2),\Lambda}(x,y) \overline{  \kappa^\Lambda_t (-k,y) } \Big),   \\[1.5mm]
    \widetilde X^{(1)}(k,x) &  = (i \partial_\theta-\tau_t(x)) (q L^\Lambda_\theta(k)^*u^\theta_t)(x) + (i\partial_t - h_\theta  + \omega(k)) \kappa^\Lambda_t(- k,x) ,\\[2.5mm]
    \widetilde X^{(2)}(k,x) & = -X^{(2)}(-k,x).
  \end{align*}
  To see that $X^{(1)}(k,x)=0$, we use~\eqref{eq:dtheta-qTu} with
  $T_\theta=L^\Lambda_\theta(k)$ and the identity
  \begin{equation}
    \nabla \tau_t(x)=\nabla 2\Re\langle iB_x, \alpha_t^\theta\rangle =  2\Re\langle k B_x, \alpha_t^\theta\rangle= F_{\alpha_t^\theta}(x) 
  \end{equation}
  to obtain
  \begin{align}
    &(i\partial_\theta -\tau_t(x)) (q L^\Lambda_\theta(k)u^\theta_t )  =
    q [ L^\Lambda_\theta(k) , \tau_t ] u_t^\theta + q (i \partial_\theta L^\Lambda_\theta(k) ) u_t^\theta \notag\\[1mm]
    &\quad =  q \Big( 2\theta B^\Lambda_{(\cdot)}(k)k[-i\nabla, \tau_t] -iG^\Lambda_{(\cdot)}(k)+2iB^\Lambda_{(\cdot)}(k)k(-i\nabla+2\theta F_{\alpha_t^\theta}) \Big)u^\theta_t \notag\\
    &\quad=q  \Big(-iG^\Lambda_{(\cdot)}(k) + 2B^\Lambda_{(\cdot)}(k)k \nabla + 2i\theta B^\Lambda_{(\cdot)}(k)k F_{\alpha_t^\theta}\Big)  u^\theta_t.\label{eq:Bog-commute-X1}
  \end{align}
  Using
  \begin{align*}
    [-\Delta, B_x^\Lambda(k)] & = k^2 B_x^\Lambda(k) + 2 k B_x^\Lambda(k) i \nabla \\[1mm]
    [A_{\alpha}, B^\Lambda_{x}(k)] & =[2(-i\nabla)\langle kB_x
    ,\alpha\rangle+\text{h.c}, B^\Lambda_{x}(k)]=- 2 k B^\Lambda_x(k)
    \underbrace{2\Re\langle kB_x,\alpha\rangle}_{=F_{\alpha}(x)}
  \end{align*}
  we find in the same way
  \begin{align}
    (-i\partial_t +h_\theta+\omega(k))\kappa^\Lambda_t(k,x)
    &= q\Big(i\omega(k)B^\Lambda_{(\cdot)} + [ -\Delta + \theta A_{\alpha_t^\theta},  iB^\Lambda_{(\cdot)}(k) ] \Big) u^\theta_t\notag \\
    &= q \Big(iG^\Lambda_{(\cdot)} -2B^\Lambda_{(\cdot)}(k)k\nabla -2i\theta B^\Lambda_{(\cdot)}(k)k F_{\alpha_t^\theta}\Big)  u^\theta_t.
  \end{align}
  This equals the negative of~\eqref{eq:Bog-commute-X1}, so $X^{(1)}\equiv
  0$.

  The equality $X^{(2)}\equiv 0$ follows simply by expanding the expressions:
  \begin{align}
    &\int dy\, K_{\theta}^{(2),\Lambda}(x,y) \overline{  \kappa_t^\Lambda (-k,y) } \notag \\
    & = \int dy dz dz' dz''\, q(x,z)q(y,z') u^\theta_t(z)u^\theta_t(z') V^\Lambda_\theta(z-z') \overline{i B^\Lambda_{z''}(-k) u^\theta_t(z'') q(y, z'')} \notag\\
    &=\int dz dz' dz'' \, q(x,z)q(z'',z') u^\theta_t(z)u^\theta_t(z') V^\Lambda_\theta(z-z')  \bar u^\theta_t(z'')(- i B^\Lambda_{z''}(k))\notag \\
    &=-\int dz dz' \, q(x,z) u^\theta_t(z) |u^\theta_t(z')|^2 V^\Lambda_\theta(z-z')  i B^\Lambda_{z'}(k) \notag\\
    &\quad +i \int dz \, q(x,z) (V^\Lambda_\theta\ast |u^\theta_t|^2)(z) \langle u^\theta_t, B^\Lambda_\cdot(k) u^\theta_t \rangle \notag\\
    &= - \int dy \, K_{\theta}^{(1),\Lambda}(x,y)  \kappa_t^\Lambda (k,y) ,
  \end{align}
  where the last equality is obtained by performing the same calculation for
  $K^{(1)}\kappa$, which just changes the location of some complex
  conjugates.  This implies vanishing of $\widetilde X^{(2)}$ and the
  argument for $\widetilde X^{(1)}$ is completely analogous to that for
  $X^{(1)}$.  \medskip

  \noindent \textbf{Terms quadratic in
    $a,a^*$~\eqref{eq:Bog-commute-aa}}. The only way to obtain a term with
  two $a^\#$s from the commutator $[\mathbb
  H^{\Lambda}_\theta,\mathbb{D}_t^\Lambda]$ is to commute two terms with an
  $a^\#$ and one $b^\#$ each. Since the coefficient
  $M^\Lambda(k,l)$~\eqref{eq:def:M:Lambda} of the terms with two $a^\#$s in
  $\mathbb{H}^\Lambda_\theta(t)$ is independent of $\theta$, we obtain for
  the coefficients in~\eqref{eq:Bog-commute-aa}
  \begin{subequations}
    \begin{align}
      A(k,l) & = -4i\theta M^\Lambda(k,-l)  \label{eq:Bog-commute-A(k,l)} \\
      &\quad - \int dx \Big( (q L^\Lambda_\theta(k)u^\theta_t)(x) \overline{\kappa^\Lambda_t (l,x)} + \overline{ ( q L^\Lambda_\theta(k)^*u^\theta_t ) (x)}  \kappa^\Lambda_t(-l,x) \Big)  \notag \\
      & \quad + \int dx \Big(  \overline{(q L^\Lambda_\theta(l)u^\theta_t)(x)} \kappa^\Lambda_t(k,x) + (q L^\Lambda_\theta(l)^*u^\theta_t)(x) \overline{\kappa^\Lambda_t(-k,x)} \Big)  \notag   \\
      E_a & = - i \int dk dx \Big(  \overline{(q L^\Lambda_\theta(k)u^\theta_t)(x)} \kappa^\Lambda_t(k,x) + (q L^\Lambda_\theta(k)^*u^\theta_t)(x) \overline{ \kappa^\Lambda_t(-k,x) } \Big)  \\
      \widetilde A(k,l) & = \begin{aligned}[t]
        2i\theta M^\Lambda(k,l)+ \int dx\,\Big( & \overline{(q L^\Lambda_\theta(k)^*u^\theta_t)(x)} \kappa^\Lambda_t(l,x) \\
        &+ (q L^\Lambda_\theta(k)u^\theta_t)(x) \overline{ \kappa^\Lambda_t(-l,x) } \Big).
      \end{aligned}
    \end{align}
  \end{subequations}
  To see that $A(k,l)=0$, we first calculate using that
  $q^2=q=1-|u^\theta_t\rangle\langle u^\theta_t|$ and
  \begin{align}
    &\int dx \Big( (q L^\Lambda_\theta(k)u^\theta_t)(x) \overline{\kappa^\Lambda_t (l,x)} + \overline{ q L^\Lambda_\theta(k)^*u^\theta_t(x)}  \kappa^\Lambda_t(-l,x) \Big) \notag \\
    &\quad=\begin{aligned}[t]
      \int dz dz'\, q(z',z) \Big(& (L^\Lambda_\theta(k)u^\theta_t)(z) (-i B^\Lambda_{z'}(-l) ) \bar u^\theta_t(z') \\
      &+ \overline{(L^\Lambda_\theta(k)^*u^\theta_t)(z')} i B^\Lambda_{z}(-l)u^\theta_t(z)\Big)
    \end{aligned}
    \notag \\
    &\quad= \langle u^\theta_t, i[L^\Lambda_\theta(k), B_{(\cdot)}^\Lambda(-l)]u^\theta_t\rangle
    \notag\\[2mm]
    &\quad=2i\theta\langle u^\theta_t,  B_{(\cdot)}^\Lambda(-l)l k B_{(\cdot)}^\Lambda(k) u_t^\theta\rangle = -2i\theta M^\Lambda(k,-l).
  \end{align}
  The second line in~\eqref{eq:Bog-commute-A(k,l)} is the complex conjugate
  of this with $k,l$ exchanged, so it equals $2i\theta M^\Lambda(k,-l)$.
  This implies that $A\equiv 0$. The argument for $\widetilde A\equiv 0$ is
  essentially the same.

  It remains to evaluate $E_a$. We have by the calculation of $A(k,l)$
  \begin{align}
    E_a= 2 \theta \int dk  M^\Lambda(k,-k) =-2 \theta  \langle k^2 B_{0}^\Lambda, B_{0}^\Lambda\rangle . 
  \end{align}
  Consequently with~\eqref{eq:Bog-commute-E_b}
  \begin{align}
    E_a+E_b &=  2  \Re \langle G_0^\Lambda, B_0^\Lambda \rangle - 2 \theta \Re\langle (k^2+\omega) B^\Lambda_{0},B^\Lambda_0\rangle + \partial_\theta \tfrac12 \langle u, V^\Lambda_\theta\ast |u|^2 u\rangle \notag \\
    &= (2-2\theta) \Re \langle G_0^\Lambda, B_0^\Lambda \rangle+ \partial_\theta \tfrac12 \langle u, V^\Lambda_\theta\ast |u|^2 u\rangle \notag\\
    &= \partial_\theta E^\Lambda_\theta.
  \end{align}
  This completes the proof of the proposition.
\end{proof}


\section{Estimates for the
  generators \label{sec:estimates:dressed:dynamics}}\label{chap:generators}

In this section we establish the inequalities on the different generators of
the dynamics considered in the previous sections. This includes the
generators of the fluctuation dynamics for $e^{-i t H_N^\mathrm{D}}$ and its
Bogoliubov approximation used in the proofs of Theorems
\ref{thm:gross-transformed dynamics reduced density matrices} and
\ref{thm:norm:approximation:dressed:dynamics}, given in
Sections~\ref{sec:dressed:fluc:generator}, \ref{sec:estimates:Bogoliubov},
respectively.  Similar bounds for the generators associated with the dressing
flow and its Bogoliubov approximation are given in
Section~\ref{section:estimates-dressing}.


\subsection{Fock space operator bounds}\label{sec:Fock:Space:Operators}

We start by proving a general bound for operators on Fock space that will
prove very useful.

\begin{lemma}
  \label{lem:ops-Fock}
  Let $n_a, n_b, m_a, m_b\in \NNN_0$, $M=n_a+n_b+m_a+m_b$, $s,t\in \RRR$ and
  \begin{equation*}
    T:   L^2(\RRR^3)^{\otimes m_b} \otimes \mathfrak{h}_{t }^{\otimes m_a}  \to L^2(\RRR^3)^{\otimes n_b} \otimes \mathfrak{h}_{-s }^{\otimes n_a}
  \end{equation*}
  be a bounded operator of norm $\tau$ with an integral kernel
  $T((K,X),(L,Y))\in \mathscr{S}'( \RRR^{3M})$.  Set
  \begin{equation*}
    A_{n}(K)=\prod_{i=1}^n a_{k_{i}} \text{ and } B_{n}(X)=\prod_{i=1}^n b_{x_{i}}.
  \end{equation*}
  Then for all $0\leq r_b\leq n_b+m_b$, $0\leq r_a \leq n_a+m_a$,
  \begin{align*}
    &\bigg|\Big\langle \chi, \int_{\R^{3M}} T(X,K,Y,L)
    B^*_{n_b}(X) A^*_{n_a}(K)A_{m_a}(L) B_{m_b}(Y) dX dK dY dL\xi\Big\rangle \bigg| \\
    &\leq \tau \|
    (\mathcal{N}_b+M)^{\frac{n_b+m_b-r_b}{2}} (\mathcal{N}_a+M+1)^{r_a} \textnormal{d}\Gamma_a (\omega^{2s} )^{\frac{n_a}{2}}\chi\| \\
    &\qquad\times \|(\mathcal{N}_b+M)^{\frac{r_b}{2}} (\mathcal{N}_a+1)^{-r_a} \textnormal{d}\Gamma_a (\omega^{2t})^{\frac{m_a}{2}}\xi\|.
  \end{align*}
\end{lemma}
\begin{proof}
  To keep the notation manageable we give the proof in the case $n_b=m_b=0$,
  the generalization is straightforward.
  Set $n=n_a$, $m=m_a$, and let $\chi \in
  D(\textnormal{d}\Gamma(\omega^{2s})^{n/2})$, then
  \begin{align}
    \int \Big\|\prod_{i=1}^n \omega^{s}(k_i)  a_{k_i} \chi\Big\|^2 dK   &  =\int \Big(\prod_{i=1}^{n} \omega^{2s}(k_i)\Big) \Big\langle \chi,  a_{k_1}^*\cdots a_{k_n}^* a_{k_n} \cdots a_{k_1} \chi\Big\rangle dK  \notag\\
    & = \| \textnormal{d}\Gamma_a (\omega^{2s})^{n/2}\chi\|^2,
  \end{align}
  so
  \begin{equation*}
    (k_1, \dots, k_n) \mapsto \Big(\prod_{i=1}^n \omega^{s}(k_i) a_{k_i} \Big)\chi \in L^2(\RRR^{3n}, \mathcal{F}\otimes \cF) ,
  \end{equation*}
  and analogously for $\xi$. Hence
  \begin{align}
    &\bigg|\Big\langle \chi,  \int  T(K,L) A_n^*(K) A_m(L ) dK dL\,\xi\Big\rangle \bigg|\notag \\
    &=\begin{aligned}[t]
      \bigg|  \int dK dL \Big\langle (\mathcal{N}_a+1)^{r_a}&\Big(\prod_{i=1}^n \omega^{s}(k_i) a_{k_i}\Big)\chi, \\
      & T(K,L)(\mathcal{N}_a+1)^{-r_a}\Big(\prod_{i=1}^n \omega^{-s}(k_i)\Big) A_m(L)\xi\Big\rangle  \bigg|
    \end{aligned}
    \notag \\
    & \leq \|\textnormal{d}\Gamma_a (\omega^{2s})^{\frac{n}{2}}(\mathcal{N}_a+n+1)^{r_a}\chi\| \notag \\
    &\qquad \times \Big\| \int \Big(\prod_{i=1}^n\omega^{-s}(k_i) \Big)T(K,L)
    A_m(L)(\mathcal{N}_a+m+1)^{-r_a}\xi d L \Big\|_{L^2(\RRR^{3n}_K, \mathcal{F}\otimes \cF)} \notag \\
    &\leq \|T\|_{\mathfrak{h}_t^{\otimes m}\to \mathfrak{h}_{-s}^{\otimes n}}  \| \textnormal{d}\Gamma_a(\omega^{2s})^{\frac{n}{2}}(\mathcal{N}_a+M+1)^{r_a}\chi\|\|\textnormal{d}\Gamma_a(\omega^{2t})^{\frac{m}{2}}(\mathcal{N}_a+1)^{-r_a}\xi\|.
  \end{align}
  This proves the claim.
\end{proof}

Two special cases of the previous lemma we use frequently are given
separately below.
\begin{lemma} \label{lem:bounds:dGamma} For any $s\in \mathbb R$ and
  $\chi,\xi\in \mathcal F \otimes \mathcal F$, we have
  \begin{subequations}
    \begin{align}\label{eq:bound:dG:sqrt:omega}
      \big| \scp{\chi}{\textnormal{d} \Gamma_a(\omega^{s/2}) \xi} \big| & \le  \norm{ \textnormal{d}\Gamma_a(\omega^s)^{1/2} \chi} \,  \norm{\mathcal N_a^{1/2} \xi},   \\[1.5mm]
      \norm { \textnormal d\Gamma_a(\omega^{s/2})  \chi } & \le \norm{ \mathcal N_a^{1/2} \textnormal d\Gamma_a(\omega^s)^{1/2} \chi }.
      \label{eq:normbound:dGamma(omega^s/2)}
    \end{align}
  \end{subequations}
\end{lemma}
\begin{proof}
  The first inequality is a special case of Lemma~\ref{lem:ops-Fock} ($t=0$,
  $T \alpha= \omega^{s} \alpha$, $m_a=n_a=1$, $r_a=0$). The second inequality
  follows from the first by taking the supremum over $\|\xi\|=1$.
\end{proof}


\subsection{Preliminary estimates}\label{sec:preliminary:estimates}

Here we provide some bounds on the terms appearing in the mean-field
Hamiltonians as well as the kernels of the Bogoliubov Hamiltonians. These
will later be combined with the operator bounds from the previous section to
prove the estimates for the generators.

\begin{lemma} \label{lem:bounds:alpha:phi} Let $G$, $B$ and $V$ be defined as
  in \eqref{eq:definition of G-x notation section}, \eqref{eq:def:B(k)} and
  \eqref{eq:definition of V}. For every $s>0$ there exists a constant $C>0$
  such that
  \begin{subequations}
    \begin{align}
      \norm{\widehat V}_{\Lp{1+s}} & \le C  , \label{eq:bound:FT:V}\\[1mm]
      \norm{ k  B_0 }_{\h{- s }} & \le C  , \label{eq:bound:kB:hs:norm}\\[1mm]
      \forall n \in \mathbb N_0  \quad \scp{ | k|^n B_0}{|\alpha|} & \le  C  \norm{\alpha}_{\h{ (n-1)+ s }}  , \label{eq:bound:scp:kB:alpha} \\[1mm]
      \forall n\in \{1,2,3\} \quad \norm{ \langle k^n  B_{(\cdot)},\alpha \rangle u }_{\2} & \le C \norm{ \alpha }_{\h{(n-2)/2+s }} \norm{ u }_{\Sob{n/2}} \label{eq:bound:kn:alpha:u}  
    \end{align}
    for all $u \in H^n(\R^3)$ and $\alpha \in \mathfrak{h}_{n-1+s}$.
    Moreover, for every $\varepsilon>0$ there exists $C>0$ so that for $u\in
    H^1(\R^3)$, $\alpha\in \mathfrak{h}_{1/2}$
    \begin{align}
      \norm{\langle G_{(\cdot)}, \alpha \rangle u}_{L^2}&\leq \varepsilon (\norm{ u }_{\Sob{1}}^2 + \norm{ \alpha }_{\h{1/2}}^2 ) + C \norm{ u }_{L^2}^2 .\label{eq:bound:G:alpha:u}
    \end{align}

  \end{subequations}
\end{lemma}
\begin{proof}[Proof of Lemma \ref{lem:bounds:alpha:phi}]
  In view of the formula for $V$, we have $(2\pi)^{3/2}\widehat V(k) = -4
  G_0(k) B _0(k) + 2 \omega(k) B _0^2(k)$, and the first three inequalities
  then follow immediately from the integrability properties of
  $G_0(k)=\omega(k)^{-1/2}$ and
  $B_0(k)=(k^2+\omega(k))^{-1}\omega(k)^{-1/2}$.

  For~\eqref{eq:bound:kn:alpha:u}, we use the Fourier representation in $x$
  together with Parseval to write
  \begin{align}
    & \norm{\scp{k^n B_{(\cdot)}}{\alpha} u}_{\2}^2  = \frac{1}{(2\pi)^3}\Big\|\int dk\, k^n B_0(k) \alpha(k)  \int dp \, e^{i p (\cdot)} \, \hat u(p-k) \Big\|_\2^2  \\[1 mm]
    & = \Big\| \int dk\, k^n B_0(k) \alpha(k) \, \hat u(\cdot-k) \Big\|_\2^2  \notag\\[1.5mm]
    & = \int dp dk d\ell\, \frac{k^n B_0(k)\overline{\alpha (k)}}{\omega( p-k ) ^{n/2}} \frac{\ell^n B_0(\ell)\alpha (\ell)}{\omega( p-\ell )^{n/2}} \omega( p - k )^{\frac{n}{2}} \widehat u (p-k) \omega( p - \ell ) ^{\frac{n}{2}} \overline{\widehat u (p-\ell) } .\notag
  \end{align}
  Since $k \mapsto |k|^n B_0(k)^2 \omega(k)^{2-2s}$ and $k\mapsto
  \omega(k)^{-n}$ are radial and decreasing functions for $s>0$ and $n\in
  \{1,2,3\}$, it follows by symmetric rearrangement that
  \begin{align}\label{eq:symmetric:rearrangement}
    \sup_{p \in \R^3} \int dk  \frac{|k|^n  B_0(k)^2 \omega(k)^{2- 2 s }}{\omega(p-k)^n} \le \int dk \, B_0(k)^2 \omega(k)^{2-2s} \le C.
  \end{align}
  With Cauchy--Schwarz we thus find
  \begin{align}
    & \norm{\scp{k^n B_{(\cdot)}}{\alpha} u}_{\2}^2  \\
    &\le \int dp dk d\ell \, \frac{|k|^{n} B_0(k)^2  \omega(k)^{ 2 - 2 s } }{\omega(p-k)^{n}}   |\ell|^n \omega(\ell)^{2s-2 } |\alpha (\ell)|^2   \omega(p-\ell)^n |\widehat u (p-\ell) |^2 \notag\\
    &   \le  C  \int d\ell  \, \omega(\ell)^{n + 2s - 2 } |\alpha (\ell)|^2  \int dp\,   (|p|^2+1)^{\frac{n}{2}} |\widehat u (p) |^2   \le C  \norm{ \alpha}_{\h{n/2+s-1}}^2 \norm{ u }_{\Sob{n/2}}^2  .\notag
  \end{align} 

  The final inequality~\eqref{eq:bound:G:alpha:u} is proved in a similar way,
  but to obtain the small constant in front of the energy norm we start by
  introducing a cutoff $\Lambda<\infty$. We then bound
  \begin{equation}
    \norm{\langle G_{(\cdot)}, \alpha \rangle u}_{L^2} \leq \norm{\langle \id_{|\cdot|>\Lambda}G_{(\cdot)}, \alpha \rangle u}_{L^2} +  \tfrac{\varepsilon}{2}\|\alpha\|_{\mathfrak{h}_{1/2}}^2 +  \tfrac{1}{2\varepsilon}\|\id_{|\cdot|\leq \Lambda} G_0\|^2_{\mathfrak{h}_{-1/2}} \|u\|_{L^2}.
  \end{equation}
  Now the first term is treated exactly as for~\eqref{eq:bound:kn:alpha:u},
  which yields
  \begin{align}\label{eq:bound:G,alpha}
    \norm{\langle \id_{|\cdot|>\Lambda}G_{(\cdot)}, \alpha \rangle u}_{L^2}^2
    \leq   \norm{ \alpha}_{\h{1/2}}^2 \norm{ u }_{\Sob{1}}^2 \int_{|k|>\Lambda}  dk \frac{|G_0(k)|^2 }{\omega(k)^3}.
  \end{align}
  Choosing $\Lambda$ so that the final integral is less or equal to
  $\varepsilon^2$ proves the claim.
\end{proof}

The next lemma collects bounds for the different potentials that appear in
the mean-field equations introduced in Section
\ref{sec:dressed-dyn-classical}.

\begin{lemma} \label{lem:bounds:f:g} Let $V$ be defined
  by~\eqref{eq:definition of V} and $f_u$ by~\eqref{eq:def:f}.  There exists
  a constant $C>0$ such that for all $u\in H^1(\mathbb R^3)$ with
  $\norm{u}_\2=1$
  \begin{subequations}
    \begin{align}
      \norm{V\ast |u |^2}_{\Lp{\infty}} +  	\norm{V^2\ast |u |^2}_{\Lp{\infty}} + 	\norm{\nabla ( V\ast |u |^2) }_{\Lp{\infty}} &  \le C \norm{u}_{\Sob{1}}^{3/2}, \label{eq:bound:convolutions}\\[1mm]
      \norm{f_u}_{\Lp{\infty}} +    \norm{f_u}_{\h{1/2}}   &  \le C \norm{ u }_{\Sob{1}}^2.\label{eq:bound:f(t)}
    \end{align}
  \end{subequations}
  Moreover, for every $s>0$ there exists a constant $C>0$ such that for all
  $u\in H^1(\R^3)$, $\norm{u}_\2=1$ and $\alpha \in \h{1+s}$, the objects
  $F_{\alpha}$, $g_{u,\alpha}$, $\mu_{u,\alpha}$ defined in
  \eqref{eq:def:F_alpha}--\eqref{eq:def:g} satisfy
  \begin{subequations}
    \begin{align}
      \| F_{\alpha }\|_{\Lp{\infty}} & \leq C \|  \alpha  \|_{\h{s}} \label{eq:bound:Phi:infty}\\[1mm]
      \| \nabla F_{\alpha } \|_{\Lp{\infty}} & \le C \norm{\alpha }_{\h{1+s}}\\[1mm]
      \norm{g_{u,\alpha}}_{\Lp{\infty}}   + \norm{g_{u,\alpha}}_{\h{1/2}}   & \le C \norm{u }_{\Sob{1}} \norm{ \alpha }_{\h{s}} \\[1mm]
      |\mu_{u,\alpha}| & \le C \norm{u}_{\Sob{1}}^2 \norm{ \alpha }_{\h{s}}^2.
    \end{align}
  \end{subequations}
\end{lemma}

\begin{proof} 
  For the first term in \eqref{eq:bound:convolutions}, applying Young's
  inequality and Parseval,
  \begin{align}\label{eq:V:convolution}
    \norm{V\ast |u |^2}_{\Lp{\infty}} \le \norm{V}_{\Lp{2}} \norm{| u |^2}_{\Lp{2}} =  \norm{\widehat V}_{\Lp{2}}  \norm{ u }_{\Lp{4}}^2 \stackrel{\eqref{eq:bound:FT:V}}{\le} C  \norm{ u }_{\Lp{4}}^2,
  \end{align}
  and then the Cauchy-Schwarz and Sobolev inequalities yield
  \begin{align} \label{eq:bound:L4:varphi} \norm{ u }_{\Lp{4}}^2 \le \norm{u
    }_{\Lp{2}}^{1/2} \norm{ u }_{\Lp{6}}^{3/2} \le C\norm{ u }_{H^1}^{3/2}.
  \end{align}
  For the convolution involving $V^2$, we proceed similarly, and obtain
  \begin{align}\label{eq:V:2:convolution}
    \norm{V^2\ast | u |^2}_{\Lp{\infty}} \le \norm{V^2}_{\Lp{2}} \norm{ |u|^2}_{\Lp{2}} =  \norm{V}_{\Lp{4}}^{2} \norm{ u }_{\Lp{4}}^2 \le C \norm{V}_{\Lp{4}}^2 \norm{ u }_{\Sob{1}}^{3/2}.
  \end{align}
  By Hausdorff--Young, the Fourier transform is bounded from $L^{4/3}$ to
  $L^4$, so by~\eqref{eq:bound:FT:V}
  \begin{align}
    \norm{V}_{\Lp{4}} \le C\norm{\widehat V}_{\Lp{4/3}}\le C.
  \end{align}
  For the convolution involving the gradient, we use the same inequalities to
  estimate
  \begin{align}\label{eq:V:nabla:convolution}
    \norm{\nabla ( V\ast | u |^2 ) }_{\Lp{\infty}} \le 2 \norm{V}_{L^{4}} \norm{\overline{u} \nabla u }_{\Lp{4/3}} \le C  \norm{\overline{ u } \nabla u }_{\Lp{4/3}}.
  \end{align}
  With H\"older and \eqref{eq:bound:L4:varphi} we conclude that
  \begin{align}\label{eq:bound:uDu}
    \| \overline{ u }   \nabla u \|_{\Lp{4/3}} \leq C \| u \|_{\Lp{4}}  \| \nabla u \|_{\2}  \leq  C \|u \|_{\Sob{1}}^{3/2}
  \end{align}
  and thus obtain \eqref{eq:bound:convolutions}.

  In view of $f_u(k)=\scp{u}{ k B_{(\cdot)}(k) (-i \nabla ) u}$, the bound on
  $\norm{f_u}_{\Lp{\infty}}$ is obvious.  To bound the $\h{1/2}$-norm of
  $f_u$, we use that
  \begin{equation}
    \norm{f_ u \sqrt \omega }_{\2} \leq 2 \| kB_0 \sqrt \omega \|_{\Lp{4}} \|\mathscr{F} [ \overline{ u  } \nabla  u ] \|_{\Lp{4}}.
  \end{equation}
  This implies the bound by Hausdorff-Young inequality
  and~\eqref{eq:bound:uDu}.

  The bounds on $F_\alpha(x)=2 \Re \scp{k B _x}{\alpha}$ and $\nabla
  F_\alpha$ follow directly from the from the fact that $|k|B(k)\in
  \mathfrak{h}_{-s}$.
  The bound on $F_\alpha$ implies that on the $L^\infty$-norm of
  $g_{u,\alpha}=2 \scp{ u }{ k B_{(\cdot)}(k) F_{\alpha } u}$. For the
  $\h{1/2}$-norm, writing $g_{u,\alpha} = 2 kB_0 \cdot \mathscr{F}[
  F_{\alpha} | u |^2]$ gives
  \begin{align}
    \norm{g_{u ,\alpha}\sqrt \omega}_{\2} &\leq 2\|  k B_0 \sqrt{\omega} \|_{L^4} \|\mathscr{F} [ F_{\alpha} | u |^2 ] \|_{\Lp{4}} \notag\\
    &\leq C\| F_{\alpha} | u |^2\|_{\Lp{4/3}}\leq  C \| F_{\alpha}\|_{\Lp{\infty}} \| |u|^2  \|_{\Lp{4/3}}.
  \end{align}
  The claimed bound then follows from~\eqref{eq:bound:L4:varphi} by
  H\"older's inequality, since $\|u\|_{H^1}\geq 1$.
	
  The estimate for $\mu_{u ,\alpha}=\tfrac{1}{2} \scp{u }{ V \ast |u|^2 u } +
  \Re \scp{\alpha }{ f_{u} } + \Re \scp{\alpha }{g_{u,\alpha} }$ follows from
  the previous bounds.
\end{proof}

We have similar bounds on the mean-field Hamiltonian.
\begin{lemma}\label{lemma:bounds:nabla:phi} Let $h_{u,\alpha}$ be defined by
  \eqref{eq:def:dressed:meanfield:Ham}. For every $s>0$ there is a constant
  $C>0$ such that for all $(u,\alpha) \in H^3\oplus \mathfrak{h}_{1+s}$ with
  $\norm{u}_\2=1$
  \begin{align*}
    \norm{h_{u ,\alpha} u }_{\2}  
    &  \le C \big( \norm{ u }_{\Sob{2}} + \norm{ u }_{\Sob{1}}^2 \big) \big(1+ \norm{\alpha }_{\h{s}}^2  \big)  , \\[0.5mm]
    \norm{\nabla h_{u,\alpha} u  }_{\2}  
    & \le C \big( \norm{ u }_{\Sob{3}}  +  \norm{ u }_{\Sob{1}}^{5/2} \big)\big(1+ \norm{\alpha}_{\h{1+s}} \norm{\alpha}_{\h{s}}^2 \big).
  \end{align*}
\end{lemma}
\begin{proof}[Proof of Lemma \ref{lemma:bounds:nabla:phi}] The proof follows
  from Lemmas \ref{lem:bounds:alpha:phi} and \ref{lem:bounds:f:g} in
  combination with (we use $\norm{\cdot}_{\Sob{s}}\ge 1$ and
  $\norm{\cdot}_{\h{s}}\ge \norm{\cdot}_{\h{r}}$ for $s\ge r $)
  \begin{align}
    \norm{h_{u,\alpha }  u }_{\2} & \le \norm{ \Delta  u  }_{\Lp{2}} +  2  \norm{  \scp{k^2 B_{(\cdot)}}{\alpha }  u  }_{\2} +  4 \scp{| k| B_0}{|\alpha |}  \norm{ \nabla  u  }_{\2}  \notag\\[1mm]
    & \quad +  \norm{F_{\alpha }}_{\Lp{\infty}}^2  + \norm{V\ast | u  |^2}_{\Lp{\infty}} + |\mu_{ u ,\alpha} | \notag \\[1mm]
    &  \le C \big( \norm{u }_{\Sob{2}} + \norm{ u }_{\Sob{1}}^2 \big) \big(1+ \norm{\alpha }_{\h{s}}^2,  \big)
  \end{align}
  and
  \begin{align}
    \norm{\nabla h_{u ,\alpha} u  }_{\2} & \le \norm{\nabla \Delta u }_{\Lp{2}} +  2  \norm{  \scp{k^3 B_{(\cdot)}}{\alpha } u }_{\2} +  6 \norm{  \scp{k^2 B_{(\cdot)}}{\alpha}  \nabla  u  }_{\2} \notag\\[1mm]
    & \ + 4 \scp{ | k |  B_0}{|\alpha |} \norm{ \Delta  u  }_{\2}  + 2 \norm{ F_{\alpha }}_{\Lp{\infty}}  \norm{\nabla F_{\alpha }}_{\Lp{\infty}}   + \norm{F_{\alpha }}_{\Lp{\infty}}^2 \norm{\nabla  u }_{\Lp{2}}  \notag\\[1mm]
    & \ + \norm{\nabla ( V\ast|u|^2) }_{\Lp{\infty}} +  \norm{V\ast |u |^2}_{\Lp{\infty}} \norm{\nabla u  }_{\2} + |\mu_{u ,\alpha}|  \norm{\nabla u }_{\Lp{2}} \notag \\[1mm]
    & \le C \big( \norm{ u }_{\Sob{3}}  +  \norm{ u }_{\Sob{1}}^{5/2} \big)\big(1+ \norm{\alpha}_{\h{1+s}} \norm{\alpha}_{\h{s}}^2 \big).
  \end{align}
\end{proof}
Next, we state suitable bounds for the time-derivatives of $u_t$, $\alpha_t$
and $ \mu_{u_t,\alpha_t}$.  Note that the constant $C$ in the bound is
uniform in $t$ but not in $(u,\alpha)$, as it depends on the energy of the
initial condition.

\begin{lemma}\label{lem:bounds:alpha:phi:dot} Let $(u,\alpha) \in H^3(\mathbb
  R^3) \oplus \mathfrak h_{5/2}$ with $\norm{ u }_\2=1$ and $(u_t, \alpha_t)
  =\mathfrak s^{\rm D} [t](u,\alpha)$ denote the solution to \eqref{eq:SKG
    dressed}. Let $\mu_{u_t,\alpha_t}$ be defined as in
  \eqref{eq:def:dressed:mu}. There exists a constant $C>0$ such that for all
  $|t|\ge 0$
  \begin{align}
    \norm{\dot u_t}_{\Sob{1}}+ \norm{\dot \alpha_t}_{\h{1/2}} +  |\dot \mu_{u_t,\alpha_t} |   & \le C\,  \norm{ u_t}_{\Sob{3}}  \, \big(1+ \norm{\alpha_t}_{\h{3/2}} \big) . \notag 
  \end{align}
\end{lemma}
\begin{proof} We use that $\norm{\dot u_t}_{\Sob{1}}^2 = \norm{ \dot u_t
  }_\2^2 + \norm{\nabla \dot u_t}_{\2}^2$ and $i\dot u_t = h_{u_t,\alpha_t}
  u_t$ . Since $\norm{u_t}_{\Sob{1}} + \norm{\alpha_t}_{\h{1/2}}\le C$ (by
  Proposition \ref{prop:skg well posed} and
  Lemma~\ref{lem:mf-dressing-regularity}), we obtain from Lemma
  \ref{lemma:bounds:nabla:phi} for $s=\tfrac{1}{2}$
  \begin{subequations}
    \begin{align}
      \norm{h_{u_t,\alpha_t} u_t}_{\2}  
      & \le C  \norm{u_t}_{\Sob{2}},  \\[1mm]
      \norm{\nabla h_{u_t,\alpha_t} u_t }_{\2}   
      & \le C   \norm{u_t}_{\Sob{3}}  \big( 1+ \norm{\alpha_t}_{\h{3/2}}\big) .
    \end{align}
  \end{subequations}
  With the aid of Lemma \ref{lem:bounds:f:g}, one easily verifies
  \begin{align}
    \norm{\dot \alpha_t}_{\h{1/2}} \le \norm{\alpha_t}_{\h{3/2}} + 
    \norm{u_t}_{\Sob{1}}^2  \norm{ \alpha_t}_{\h{1/2}}.
  \end{align} 
  Recall that $\mu_{u ,\alpha}=\tfrac{1}{2} \scp{u }{ V \ast |u|^2 u } + \Re
  \scp{\alpha }{ f_{u} } + \Re \scp{\alpha }{g_{u,\alpha} }$. Since $V$ is an
  even function,
  \begin{align}
    \big|\tfrac{d}{dt}\scp{u }{ V  \ast |u|^2 u }\big| & =  4 |\Re \scp{\dot u_t}{V \ast |u_t|^2 u_t}| \notag \\
    &\le 4 \norm{V\ast |u_t|^2}_{\Lp{\infty}} \norm{\dot u_t}_\2 \le C \norm{\dot u_t}_\2 .
  \end{align}
  We further estimate

\begin{align}
  \norm{ \dot f_{u_t} }_{\h{-1/2}} 
  &\leq 2 \norm{  kB_{0}    }_{\h{-1/2}} \left( \norm{\dot{u}_t}_{L^2} \norm{u_t}_{H^1} + \norm{\dot{u}_t}_{H^1} \norm{u_t}_{L^2} \right)
  \leq C \norm{\dot{u}_t}_{H^1}, 
  \notag\\[1mm]
  \norm{ \dot g_{u_t,\alpha_t}  }_{\h{-1/2}} & \le 4 \norm{  \scp{\dot u_t}{ k B_{(\cdot)} \cdot F_{\alpha_t} u_t} }_{\h{-1/2}}   + 2 \norm{  \scp{ u_t}{ k B_{(\cdot)} \cdot F _{\dot \alpha_t} u_t} }_{\h{-1/2}} \notag\\[1mm]
  & \le \big( 2 \norm{\dot u_t}_\2 \norm{F_{\alpha_t}}_{\Lp{\infty}} + \norm{F_{\dot \alpha_t}}_{\Lp{\infty}}  \big) \norm{kB_0  }_{\h{-1/2}} .
\end{align}
With \eqref{eq:bound:scp:kB:alpha} we have $\norm{ F_{\dot
    \alpha_t}}_{\Lp{\infty}} \le C \norm{\dot \alpha_t}_{\h{1/2}}$, and hence
\begin{align}
  |\dot \mu_{u_t,\alpha_t}| & \le C\|\dot u_t\|_{L^2} + \norm{\dot \alpha_t}_\2 \norm{f_{u_t}+ g_{u_t,\alpha_t} }_\2  + \norm{\alpha_t}_{\h{1/2}} \norm{ \dot f_{u_t} + \dot g_{u_t,\alpha_t}}_{\h{-1/2}} \notag\\[1mm]
  &  \le C (   \norm{\dot \alpha_t}_{\h{1/2}} + \norm{\dot u_t}_{H^1}   ).
\end{align}
This completes the proof of the lemma.
\end{proof}

The next lemma summarizes estimates for the different kernels (and their
time-derivatives) that appear in the (dressed) Bogoliubov Hamiltonian
\eqref{eq:HD:Bog definition} and the fluctuation Hamiltonian introduced in
\eqref{eq:generator:restriction}, given explicitly by
\eqref{eq:def:fluct:generator}.

We introduce the integral kernels
\begin{subequations}
  \begin{align}
    N_{u}(x,k,l) & = (q_{u} k B_{(\cdot)}(k) \cdot l B_{(\cdot)}(l)u )(x), \label{eq:kernel:N(t)} \\[1mm]
    Q_{u}(x,y,k,l) & = (q_{u} k B_{(\cdot)}(k) \cdot l B_{(\cdot)}(l)q_{u})(x,y)  \label{eq:kernel:Q(t)},
  \end{align}
\end{subequations}
where we note that $k B_{x}(k) \cdot l B_{x}(l)$ acts as a multiplication
operator in $x$. Recalling the definition $\big(L_{ \alpha}(k)f\big)(x) = 2 k
B _{ x }(k) \big((-i \nabla + F_\alpha(x))f\big)(x)$
from~\eqref{eq:L(k):operator}, we set
\begin{align}
  \ell_t^{(1)} (x,k) & = (q_{u_t} L_{\alpha_t}(k) u_t)(x), \quad \ell_t^{(2)} (x,k) = (q_{u_t}  L_{\alpha_t}(k)^* u_t)(x).
\end{align}
Moreover, for $u\in L^2(\mathbb R^3)$ and $\{ u \}^\perp \subset L^2(\mathbb
R^3)$, we consider the operators
\begin{subequations}
  \begin{align}
    K_{u}^{(3)} & : \{ u \}^\perp  \rightarrow  \{ u  \}^\perp \otimes \{ u \}^\perp \notag \\[1mm]
    \psi & \mapsto (K_{u }^{(3)}\psi)(x_1,x_2) = (q_{u})_1  (q_{u})_2  W_{u}(x_1,x_2) u(x_1) (q_{u} \psi)(x_2) , \label{eq:def:of:LK3} \\[1mm]
    K_{u}^{(4)}  & : \{ u \}^\perp \otimes \{ u \}^\perp  \rightarrow  \{ u \}^\perp \otimes \{ u \}^\perp \notag \\[1mm]
    \psi & \mapsto\label{eq:def:of:LK4}
    (K_{u}^{(4)}\psi)(x_1,x_2) = (q_{u})_1 (q_{u})_2 W_{u}(x_1,x_2) (q_{u} \otimes q_{u} \psi)(x_1,x_2).
  \end{align}
\end{subequations}
with, for $V$ defined in~\eqref{eq:definition of V},
\begin{align}\label{eq:def:W}
  W_{u}(x_1,x_2) = V(x_1-x_2) - V\ast |u|^2(x_1) - V\ast |u|^2(x_2) +  \scp{u }{V \ast |u|^2 u}.
\end{align}

\allowdisplaybreaks
\begin{lemma}
  \label{lem:aux:bounds} Let $(u,\alpha) \in H^3(\mathbb R^3) \oplus
  \mathfrak h_{5/2}$ with $\norm{ u }_\2=1$, denote by $(u_t, \alpha_t)
  =\mathfrak s^{\rm D} [t](u,\alpha)$ the solution to \eqref{eq:SKG dressed},
  and $\rho(t) = \norm{ u_t}_{\Sob{3}}^2 \, (1+
  \norm{\alpha_t}_{\h{3/2}})^2$. There exists a constant $C>0$ such that for
  all $|t|\ge 0$
  \begin{subequations}
    \begin{alignat}{4}
      & \norm{ K_{u_t}^{(2)} }_{L^2(\mathbb R^6)} && \le C \hspace{2cm}
      \norm{ \dot K_{u_t}^{(2)} }_{L^2(\mathbb R^6)}  && \le C  \sqrt{\rho(t)} \label{eq:aux:bounds:d} \\[1mm]
      &\norm{ K_{u_t}^{(1)} }_{L^2 \to L^2 } && \le C \hspace{2cm}
      \norm{ \dot K_{u_t}^{(1)} }_{L^2  \to L^2 } && \le C  \sqrt{\rho(t)}  \label{eq:aux:bounds:c} \\[1mm]
      & \| \ell_t^{(1)} \|_{  L^2 \otimes \h{-1/4} } && \le C \hspace{2cm}  \| \dot \ell_t^{(1)} \|_{ L^2 \otimes \h{-1/4} } && \le C \sqrt{\rho(t)} \label{eq:aux:bounds:e}  \\[1mm]
      &\| \ell_t^{(2)} \|_{ \h{ 1/4} \to L^2 }^2 && \le C\hspace{2cm} \| \dot \ell_t^{(2)} \|_{\h{1/4} \to L^2 } && \le C  \sqrt{\rho(t)} \label{eq:aux:bounds:f}  \\[1mm]
      &\norm{M_{u_t}}_{(L^2)^{\otimes 2}} && \le C \hspace{2cm}
      \norm{\dot M_{u_t}}_{(L^2)^{\otimes 2}} && \le C   \sqrt{\rho(t)} \label{eq:aux:bounds:g}   \\[1mm]
      &\norm{ K_{u_t}^{(3)} }_{L^2 \to L^2 \otimes L^2 } && \le C
      \hspace{2cm}
      \norm{ \dot K_{u_t}^{(3)} }_{L^2 \to L^2 \otimes L^2 } && \le C  \sqrt{\rho(t)}   \label{eq:aux:bounds:2:a}\\[1.5mm]
      & \norm{N_{u_t}}_{L^2 \otimes \h{-1/8}^{\otimes 2} } && \le C
      \hspace{2cm} \norm{ \dot N_{u_t}}_{L^2\otimes \h{-1/8}^{\otimes 2}} &&
      \le C \sqrt{\rho(t)} \label{eq:aux:bounds:2:h}
      \\[0.5mm]
      & \norm{Q_{u_t}}_{L^2 \otimes \h{1/8}^{\otimes 2} \to L^2 } && \le C
      \hspace{2cm} \norm{ \dot Q_{u_t}}_{L^2\otimes \h{1/8 }^{\otimes 2} \to
        L^2} && \le C \sqrt{\rho(t)} \label{eq:aux:bounds:2:b}
    \end{alignat}
  \end{subequations}
  where $L^2$ stands for $L^2(\mathbb R^3)$.
\end{lemma}

Let us note the evident fact that the bounds \eqref{eq:aux:bounds:d} --
\eqref{eq:aux:bounds:g} hold uniformly also for the $\Lambda$-dependent
kernels introduced in \eqref{eq:theta:HG:Bog}. For instance,
$\norm{K^{(2),\Lambda}_{1,u_t}}_{L^2(\mathbb{R}^6)} \le C$ for all
$\Lambda\in \mathbb R_+ \cup \{ \infty \}$ with $K^{(2),\Lambda}_{1,u_t}$
defined by \eqref{eq:def:K2:theta}. While we do not state them explicitly
here, such uniform bounds will be used in the proofs of Lemmas
\ref{lem:bounds:dressed:Bog:theta} and \ref{lem:H:Bog:Lambda:difference}.

\begin{proof} Recall that $\norm{u_t}_\2=1$ and $ \norm{u_t}_{H_1} +
  \norm{\alpha_t}_{\h{1/2}} \le C$ for all $|t|\ge 0$ by
  Proposition~\ref{prop:skg well posed} and
  Lemma~\ref{lem:mf-dressing-regularity}.  We go through the claimed bounds
  line-by-line.  \medskip

  \noindent \textit{Line \eqref{eq:aux:bounds:d}}. We use that
  $\|q_{u_t}\|=1$, so $\norm{ K_{u_t}^{(2)} }^2_{ L^2(\mathbb R^6)} \le
  \norm{\widetilde K_{u_t}^{(2)}}_{L^2(\mathbb R^6)}$ and
  \begin{align}
    \norm{ \widetilde K_{u_t}^{(2)} }^2_{L^2(\mathbb R^6)} & = \int dx dy\, |u_t(x)|^2  V^2(x-y) |u_t(y)|^2 \notag \\
    &\le \norm{V^2 \ast |u_t|^2}_{\Lp{\infty}} \norm{u_t}_\2^2 \le  C   \label{eq:K1:K2:norm:2}
  \end{align}
  by Lemma \ref{lem:bounds:f:g}. Invoking
  \begin{align}
    \dot K_{ u_t}^{(2)}
    &= \dot{q}_{u_t}  \otimes q_{u_t} \widetilde{K}_{u_t}^{(2)}
    + q_{u_t} \otimes \dot{q}_{u_t}  \widetilde{K}_{u_t}^{(2)}
    + q_{u_t} \otimes q_{u_t} \tfrac{d}{dt}{\widetilde{K}}_{u_t}^{(2)}
  \end{align}
  together with $\dot {q}_{u_t} = -\dot{p}_{u_t}$, $\norm{\dot{p}_{u_t}}_{L^2
    \rightarrow L^2} \leq 2 \norm{\dot{u}_t}_{\2}$ and
  $\norm{\tfrac{d}{dt}\widetilde{K}_{u_t}^{(2)}}_{L^2(\mathbb R^6)} \le C \
  \norm{\dot u_t}_{\2} $ as in \eqref{eq:K1:K2:norm:2}, we obtain
  \begin{align}
    \norm{\dot K_{u_t}^{(2)}}_{L^2(\mathbb R^6)} \le C \norm{\dot u_t}_{\2} \le C  \sqrt{\rho(t)},
  \end{align}
  where we employed Lemma \ref{lem:bounds:alpha:phi:dot} in the last
  step.\medskip

  \noindent \textit{Line \eqref{eq:aux:bounds:c}}. Since
  $|K_{u_t}^{(1)}(x,y)| = | K^{(2)}_{u_t} (x,y)|$ we can use that the
  operator norm is bounded by the Hilbert--Schmidt norm,
  $\norm{K_{u_t}^{(1)}}_{L^2 \to L^2 } \le \norm{K_{u_t}^{(2)}}_{L^2(\mathbb
    R^6)}$, so that we can apply the previous bounds. The time-derivative is
  bounded analogously.\medskip

  \noindent \textit{Line \eqref{eq:aux:bounds:e}}. We recall $\ell_t^{(1)}
  (x,k) = (q_{u_t} kB_{(\cdot)}(k) \cdot (-i\nabla+ F_{\alpha_t}) u_t)(x) $
  and estimate
  \begin{align}
    \| \ell_t^{(1)} \|_{ L^2 \otimes \h{-1/4} } & \le 2 \norm{ kB_0 }_{\h{-1/4}}  \norm{  (-i\nabla+ F_{\alpha_t}) u_t }_{\2}\notag \\
    &\le C (\norm{u_t}_{\Sob{1}} + \norm{ F_{\alpha_t}}_{\Lp{\infty}}) \le C,
  \end{align}
  and similarly for
  \begin{align}
    \|\dot  \ell_t^{(1)} \|_{ L^2 \otimes  \h{-1/4}} & \le 2 \norm{  k B_0}_{\h{-1/4}} \norm{\tfrac{d}{dt}q_{u_t}e^{-ik(\cdot)}(-i\nabla + F_{\alpha_t}) u_t}_\2 \\
    & \le C (  \norm{\dot u_t}_{\Sob{1}} + \norm{ F_{\dot \alpha_t}}_{\Lp{\infty}} ) \le C (  \norm{\dot u_t}_{\Sob{1}} + \norm{ \dot \alpha_t }_{\h{1/2}} ) \le C  \sqrt{\rho(t)}, \notag
  \end{align}
  where we used $ \norm{\dot q_{u_t}}_{L^2 \to L^2 } \le 2 \norm{\dot
    u_t}_{L^2 } \le 2 \norm{\dot u_t}_{H^1 } $ and $ \norm{F_{\dot
      \alpha_t}}_{\Lp{\infty}} \le C \norm{\dot \alpha_t}_{\h{1/2}}$.\medskip

  \noindent \noindent \textit{Line \eqref{eq:aux:bounds:f}}. We have
  \begin{equation}
    \ell_t^{(2)} (x,k) = \ell^{(1)}_t(x,k) + 2 ( q_{u_t}  k^2 B _{(\cdot)}(k) u_t \big)(x),
  \end{equation}
  and thus
  \begin{align}
    \norm{ \ell^{(2)}_t}_{\mathfrak h_{1/4} \to L^2}
    &\le \|\ell^{(1)}\|_{\mathfrak h_{1/4} \to L^2} + 2\sup_{\norm{\eta}_{\h{1/4}}=1} \Big\| \int dk\, k^2 B_{(\cdot)}(k) \eta(k) u_t \Big\|_{L^2} \notag\\
    &\le \|\ell^{(1)}\|_{L^2 \otimes \mathfrak h_{-1/4}} + C \|u_t\|_{H^1}
  \end{align}
  by~\eqref{eq:bound:kn:alpha:u} (with $\alpha=\eta$, $u=u_t$, $n=2$ and
  $s=1/4$). Similarly, for the time-derivative
  \begin{align}
    \| \dot \ell_t^{(2)} \|_{\h{1/4} \to L^2 }    &\le  \|\dot \ell^{(1)}\|_{L^2 \otimes \mathfrak h_{-1/4}} + C \norm{\dot u_t}_{\Sob{1}} \leq C \sqrt{\rho(t)}.
  \end{align}
  \noindent \textit{Line \eqref{eq:aux:bounds:g}}. The estimate
  \begin{align}
    \abs{M_{u_t}(k,l)} &= \abs{\frac{k B_0(k) \cdot l B_0(l) (2 \pi)^{3/2} (k + l) \cdot \mathcal{F}[i \nabla \abs{u_t}^2](k+l)}{\abs{k+l}^2}} \notag\\
    &\leq C \norm{u_t\nabla u_t}_{L^1} \frac{\abs{k} B_0(k) \abs{l} B_0(l)}{\abs{k+l}}
  \end{align}
  and the Hardy--Littlewood--Sobolev inequality imply
  \begin{align}
    \norm{M_{u_t}}_{L^2(\mathbb{R}^6)} &\leq C \norm{u_t}_{H^1}^2 \norm{\abs{\cdot} B_0}_{L^{3}}^2 \leq C \norm{u_t}_{H^1}^2 .
  \end{align}
  For $\dot M_{u_t}$ we obtain by the same argument
  \begin{equation}
    \norm{\dot M_{u_t}}_{L^2(\mathbb{R}^6)} \leq C \norm{u_t}_{H^1} \norm{\dot u_t}_{H^1},
  \end{equation}
  which implies the claim by Lemma~\ref{lem:bounds:alpha:phi:dot}.  \medskip

  \noindent \textit{Line \eqref{eq:aux:bounds:2:a}}. Recall the definition of
  $W_u$ in \eqref{eq:def:W}, and denote $q_i = (q_{u_t})_i$.  Using
  $\norm{u_t}_{H^1} \leq C$ we get
  \begin{align}\label{eq:bound:K3}
    \norm{K^{(3)}_{ u _t} }_{L^2 \to L^2 \otimes L^2} &    =  \sup_{\norm{\psi }=1} \norm{ q_1 q_2 W_{u_t}(x_1,x_2) u_t(x_1) (q_{u_t} \psi )(x_2) }_{L^2 \otimes L^2} \notag\\[1mm]
    &  \le C \norm{q_{u_t}}_{L^2 \to L^2}^3 \big(  \norm{V^2  \ast |u_t|^2}_{\Lp{\infty}}^{1/2} + \norm{V  \ast |u_t|^2}_{\Lp{\infty}} \big)
    \le C
  \end{align}
  by Lemma \ref{lem:bounds:f:g}. For the norm of the time-derivative, one
  computes
  \begin{align}
    & (\dot K^{(3)}_{u_t}  \psi )(x_1,x_2) =  q_1  q_2  \dot W_{u_t}(x_1,x_2)  u_t(x_1) (q \psi)(x_2) \notag \\[1mm]
    &\qquad \qquad  + \dot q_1 q_2  W_{u_t}(x_1,x_2) u_t(x_1) (q \psi)(x_2)
    + q_1  \dot q_2  W_{u_t}(x_1,x_2) u_t(x_1) (q \psi)(x_2)  \notag\\[1mm]
    &\qquad \qquad + q_1 q_2  W_{u_t}(x_1,x_2) \dot u_t(x_1) (q \psi)(x_2) + q_1  q_2   W_{ u_t}(x_1,x_2) u_t(x_1) (\dot q \psi)(x_2) 
  \end{align}
  where each term can be estimated similarly as in \eqref{eq:bound:K3}. Using
  Lemma \ref{lem:bounds:f:g} in
  \begin{align}\label{eq:convolution:time-derivative}
    \norm{\tfrac{d}{dt}  V \ast \abs{u_t}^2 }_{\Lp{\infty}} \le 2 \norm{\dot u_t}_\2\norm{V^2 \ast | u_t |^2}_{\Lp{\infty}}^{1/2} \le C \norm{\dot u_t}_{\2},
  \end{align}
  one obtains $\norm{ \dot W_{u_t} }_{L^\infty(\mathbb R^6)} \le C \norm{\dot
    u_t}_{\2}$.  Together with $\norm{\dot q_{u_t} }_{L^2 \to L^2} \le
  2\norm{\dot u_t}_{L^2}$, this leads to $ \norm{ \dot K^{(3)}_{u_t} }_{L^2
    \to L^2 \otimes L^2} \le C \sqrt{\rho(t)} $.\medskip

  \noindent \textit{Lines \eqref{eq:aux:bounds:2:h} and
    \eqref{eq:aux:bounds:2:b}}. Recalling the definitions of $N_t, Q_t$ in
  \eqref{eq:kernel:N(t)}, \eqref{eq:kernel:Q(t)} it follows readily that
  \begin{subequations}
    \begin{align}
      \norm{N_{u_t}}_{L^2 \otimes \h{-1/8}^{\otimes 2} } + \norm{Q_{u_t}}_{L^2 \otimes \h{1/8}^{\otimes 2} \to L^2}   & \le C \norm{kB_0}_{\h{-1/8}}^2  \le C
      \\[1mm]
      \norm{\dot N_{u_t}}_{L^2 \otimes \h{-1/8}^{\otimes 2} } + 
      \norm{\dot Q_{u_t}}_{L^2 \otimes \h{1/8}^{\otimes 2} \to L^2}  & \le C \norm{kB_0}_{\h{-1/8}}^2 \norm{\dot u_t}_\2 \le C  \sqrt{\rho(t)}.
    \end{align}
  \end{subequations}
  This completes the proof of the lemma.
\end{proof}

\subsection{Estimates for the Bogoliubov
  Hamiltonians}\label{sec:estimates:Bogoliubov}

The first Lemma provides bounds on the dressed Bogoliubov Hamiltonian
\eqref{eq:HD:Bog definition} and its time-derivative, its difference to the
operator $\mathbb T = \text{d} \Gamma_b(-\Delta) + \text{d}
\Gamma_a(\omega)$, and its commutator with the total number operator
$\mathcal N = \mathcal N_b+\mathcal N_a$. These imply existence and
uniqueness of the associated dynamics, as explained in
Proposition~\ref{prop:theta-Bog}.  Similar bounds also hold for the family of
interpolating Bogoliubov Hamiltonians \eqref{eq:theta:HG:Bog} (recall that $
\mathbb{H}_{u,\alpha}^{\rm D}(t) = \mathbb {H}_{u,\alpha,1}^\infty (t)$).
Note that for $\theta=1$ the bounds in part (b) are uniform in $\Lambda$.

\begin{lemma} \label{lem:bounds:dressed:Bog:theta} \textbf{\textnormal{(a)}}
  Let $(u,\alpha) \in H^3(\mathbb R^3) \oplus \mathfrak h_{5/2}$ with
  $\norm{u}_\2=1 $ and $(u_t,\alpha_t) = \mathfrak s^{\rm D} [t](u,\alpha)$
  denote the solution to \eqref{eq:SKG dressed}. There exists a constant
  $C>0$ such that for all $t\in \mathbb R$
  \begin{subequations}
    \begin{align}
      \pm \big( \mathbb H_{u,\alpha}^{\rm D}(t)- \mathbb T \big) & \le \tfrac{1}{2}\mathbb T + C ( \mathcal N + 1 )  \label{eq:key:bound:1:Bog:Lambda}\\[1mm]
      \pm i [\mathcal N ,   \mathbb H_{u,\alpha}^{\rm D} (t) ]  & \le  \tfrac{1}{2} \mathbb T + C (\mathcal N + 1)  \label{eq:key:bound:2:Bog:Lambda}\\[1mm]
      \pm \tfrac{d}{dt} \mathbb  H_{u,\alpha}^{\rm D}(t) & \le  \tfrac{1}{2} \mathbb T  +  C  \rho(t) (\mathcal N + 1) \label{eq:key:bound:Bog:3:Lambda}
    \end{align}
  \end{subequations}
  as quadratic forms on $\Fock \otimes \Fock $, where $\rho(t) =
  \norm{u_t}_{\Sob{3}}^2 (1+ \norm{\alpha_t}_{\h{3/2}})^2 $.\medskip

  \noindent \textbf{\textnormal{(b)}} Let $(u,\alpha) \in H^3(\mathbb R^3)
  \oplus \mathfrak h_{5/2}$ with $\norm{u}_\2 =1$ and $(u_t,\alpha_t) =
  \mathfrak s_\theta [t](u,\alpha)$ as defined in
  \eqref{eq:s_theta-definition}. There exists a constant $C>0$ such that for
  all $t\in \mathbb R$, $|\theta| \le 1$ and $\Lambda \in \mathbb R_+$
  \begin{subequations}
    \begin{align}
      \pm \big( \mathbb H_{u,\alpha,\theta}^\Lambda(t)- \mathbb T \big) & \le \tfrac{1}{2}\mathbb T + C (1   +  | 1-\theta | \Lambda ) ( \mathcal N + 1 )  \label{eq:key:bound:1:Bog:Lambda:x}\\[1.5mm]
      \pm i [\mathcal N , \mathbb  H_{u,\alpha,\theta}^\Lambda(t) ]  & \le  \tfrac{1}{2} \mathbb T + C (1   +  | 1-\theta | \Lambda )  (\mathcal N + 1)  \label{eq:key:bound:2:Bog:Lambda:x}\\[1.5mm]
      \pm \tfrac{d}{dt} \mathbb  H_{u,\alpha,\theta}^\Lambda(t) & \le  \tfrac{1}{2} \mathbb T  +  C (1   +  | 1-\theta | \Lambda )   \rho(t) (\mathcal N + 1) \label{eq:key:bound:Bog:3:Lambda:x}
    \end{align}
  \end{subequations}
  as quadratic forms on $\Fock \otimes \Fock $, where $\rho(t) =
  \norm{u_t}_{\Sob{3}}^2 (1+ \norm{\alpha_t}_{\h{3/2}})^2 $.
\end{lemma}

\begin{proof} The proof follows essentially by combining the operator bounds
  of Lemma~\ref{lem:ops-Fock} with the kernel bounds of
  Lemma~\ref{lem:aux:bounds}. We give the details below.\medskip

  \noindent\textit{Proof of \eqref{eq:key:bound:1:Bog:Lambda}}. Recall that
  \begin{align}
    \mathbb H^{\rm D}_{u,\alpha}(t) - \mathbb T &= \textnormal{d}\Gamma_b\big(  A_{\alpha_t} + F^2_{\alpha_t} - \mu_{u_t,\alpha_t} \big) + \mathbb K_{u_t}^{(1)} + \big( \mathbb K_{u_t}^{(2)} + \text{h.c.} \big) \notag \\
    & \quad + \bigg( \int dk dx \Big( \ell_t^{(1)}(x,k) b_x^* a_k^*   + \ell_t^{(2)}(x,k)  b_x^* a_k  \Big) + \text{h.c.} \bigg) \notag\\
    &\quad +  \int dk dl \, M_{u_t}(k,l) \, \mathcal A_{kl} .\label{eq:diff:Bog:A}
  \end{align}
  Since $A_{ \alpha_t } = 2 (- i\nabla_x) \cdot \scp{kB_x}{ \alpha_t} +
  \text{h.c.}$ the first term in the first line is bounded, using the
  Cauchy--Schwarz inequality, by
  \begin{align}
    & \pm   \textnormal{d}\Gamma_b\big(  A_{\alpha_t} + F^2_{\alpha_t} - \mu_{u_t,\alpha_t} \big)  \notag\\
    & \hspace{0.5cm} \le \varepsilon \text{d} \Gamma_b(-\Delta)  + \Big( \frac{4}{\varepsilon} \scp{|k |  B_0}{|\alpha_t|}^2 + \norm{ F_{\alpha_t} }_{\Lp{\infty}}^2 + | \mu_{u_t,\alpha_t} | \Big) \mathcal N_b \notag\\
    &\hspace{0.5cm} \le  \varepsilon \text{d} \Gamma_b(-\Delta)  + \frac{C}{\varepsilon}  \mathcal N_b,
  \end{align}
  where the last bound follows from Lemmas \ref{lem:bounds:alpha:phi} and
  \ref{lem:bounds:f:g}, and $\norm{u_t}_{\Sob{1}} + \norm{\alpha_t}_{\h{1/2}}
  \le C$.

  For the second and third term in the first line, we apply Lemmas
  \ref{lem:ops-Fock} and \ref{lem:aux:bounds} to get
  \begin{subequations}
    \begin{align}
      \label{eq:bound:mathbb:K:1}
      \pm \mathbb K_{u_t}^{(1)} &   \le \norm{ K_{u_t}^{(1)} }_{L^2  \to L^2 } (\mathcal N_b+1) \le C (\mathcal N_b+1), \\[1mm]
      \label{eq:bound:mathbb:K:2}
      \pm ( \mathbb K_{u_t}^{(2)} + \text{h.c.} ) &  \le 2 \norm{ K_{u_t}^{(2)} }_{L^2(\mathbb R^6)} (\mathcal N_b+1) \le  C (\mathcal N_b+1).
    \end{align}
  \end{subequations}

  For the mixed quadratic terms, we use the Lemma \ref{lem:ops-Fock} choosing
  $m_a=m_b=0$, $n_a=n_b=1$, $s=1/2$ and $r_a=0$, $r_b=1$ for the term
  involving $b_x^*a_k^* $, and $m_a= 1 = n_b $, $n_a=m_b=0$, $t=1/2$ and
  $r_a= r_b=0$ for the term involving $b_x^* a_k$. Thus gives
  \begin{align}\label{eq:bound:ell:1}
    & \pm \bigg( \int dk dx  \Big( \ell^{(1)}_t(x,k) b_x^*a_k^* +  \ell^{(2)}_t(x,k) b_x^* a_k \Big) + \text{h.c.} \bigg) \notag\\
    & \quad\quad\quad  \le  \varepsilon \text{d} \Gamma_a(\omega) + C\varepsilon^{-1}  \Big( \|  \ell_t^{(1)}\|^2_{   L^2 \otimes \h{-1/2}} + \|  \ell_t^{(2)}\|^2_{\h{1/2} \to L^2 }\Big)  (\mathcal N_b+1) \notag\\[1mm]
    & \quad\quad\quad  \le  \varepsilon  \text{d} \Gamma_a(\omega) + C\varepsilon^{-1} (\mathcal N_b+1),
  \end{align}
  where we used again Lemma \ref{lem:aux:bounds} and monotonicity of the
  $\mathfrak{h}_s$-norms.

  In the last term in \eqref{eq:diff:Bog:A} we have $\mathcal A_{kl} = -2
  a_k^* a_{-l} + a_k^* a_l^* + a_{-l} a_{-k}$, so using that the operator
  norm is bounded by the Hilbert-Schmidt norm, we have with
  Lemma~\ref{lem:ops-Fock} and \eqref{eq:aux:bounds:g}
  \begin{align}\label{eq:bound:M}
    \pm \int dk dl \ M_{u_t}(k,l)  \mathcal A_{k,l} & \leq C  ( \mathcal N_a + 1) .
  \end{align}
  \noindent \textit{Proof of \eqref{eq:key:bound:2:Bog:Lambda}}. The
  commutator is easily found to be
  \begin{multline} [\mathcal N , \mathbb H^{\rm D}_{u,\alpha}(t) ]
    =  2 \bigg(  \int dk dx\, \ell_t^{(1)}(x,k) b_x^* a_k^* - \text{h.c.} \bigg)\\  +  2  \int dk dl \, M_{u_t}(k,l) \, ( a^*_k a_l^* - a_{-k} a_{-l} ),
  \end{multline}
  which can be estimated exactly as in \eqref{eq:bound:ell:1} and
  \eqref{eq:bound:M}.\medskip

  \noindent \textit{Proof of \eqref{eq:key:bound:Bog:3:Lambda}}. We compute
  \begin{align}
    \label{eq:time derivative of Bogoliubov Hamiltonian}
    \frac{d}{dt} \mathbb H^{\rm D}_{u,\alpha}(t)& = \textnormal{d}\Gamma_b ( \dot h_{u_t,\alpha_t} ) +  \dot{ \mathbb K}_{u_t}^{(1)} + (  \dot {\mathbb K}_{u_t}^{(2)} + \text{h.c.}) \notag \\
    & \quad + \bigg( \int dk dx \Big( \dot{\ell}^{(1)}_t (x,k) b_x^* a_k^*   + \dot{\ell}^{(2)}_t(x,k)  b_x^* a_k  \Big) + \text{h.c.} \bigg) \notag \\
    &\quad +  \int dk dl \, \dot{M}_{u_t}(k,l) \, \mathcal A_{kl}  ,
  \end{align}
  with
  \begin{align}
    \dot h_{u_t,\alpha_t} = A_{\dot{\alpha}_t}
    + 2 F_{\alpha_t} \cdot F_{\dot{\alpha}_t}  + \tfrac{d}{dt} V \ast   \abs{u_t}^2 - \dot{\mu}_{u_t,\alpha_t} .
  \end{align}

  Since $A_{\dot \alpha_t } = 2 (- i\nabla_x) \cdot \scp{kB_x}{\dot \alpha_t}
  + \text{h.c.}$ we can use Cauchy--Schwarz and Lemma
  \ref{lem:bounds:alpha:phi} to obtain
  \begin{align}
    \pm \textnormal{d}\Gamma_b( A_{\dot{\alpha}_t } ) \leq \varepsilon \text{d} \Gamma_b (- \Delta)
    + \frac{C}{\varepsilon} \norm{ \dot{\alpha}_t}_{\h{1/2}}^2 \, \mathcal{N}_b .
  \end{align}
  Recalling $F_{\dot \alpha_t} (x) = 2 \Re \scp{k B_x}{ \dot \alpha_t}$ we
  estimate
  \begin{align}
    \pm \textnormal{d}\Gamma_b(  F_{\alpha_t}  \cdot F_{\dot{\alpha}_t}  )
    & \leq   \norm{F_{\alpha_t}}_{\Lp{\infty}} \norm{F_{\dot{\alpha}_t}}_{\Lp{\infty}} \, \mathcal N_b \le C \norm{\dot \alpha_t}_{\h{1/2}}\, \mathcal{N}_b  .
  \end{align}
  The time-derivative of the convolution is estimated in
  \eqref{eq:convolution:time-derivative}, and thus
  \begin{align}
    \pm \textnormal{d}\Gamma_b(  \tfrac{d}{dt} V\ast | u_t |^2  ) \le C \norm{\dot u_t }_\2 \, \mathcal N_b.
  \end{align}
  Combining the above estimates, we arrive at
  \begin{align}
    \pm  \textnormal{d}\Gamma_b( \dot h_{u_t,\alpha_t} )
    &\leq  \varepsilon \text{d} \Gamma_b (- \Delta)
    +  \frac{C}{\varepsilon} \Big( \norm{ \dot{\alpha}_t}_{\h{1/2}}^2 + \norm{\dot u_t}_{\2}  \Big) \mathcal{N}_b .
  \end{align}

  Similarly as in \eqref{eq:bound:mathbb:K:1} and \eqref{eq:bound:mathbb:K:2}
  we bound the remaining terms in the first line in \eqref{eq:time derivative
    of Bogoliubov Hamiltonian}
  \begin{subequations}
    \begin{align}
      \pm \dot{\mathbb K}_{u_t}^{(1)}
      &\leq \norm{ \dot K_{u_t}^{(1)} }_{L^2  \rightarrow L^2 } \left( \mathcal{N}_b + 1 \right) \le C  \sqrt{\rho(t)} \left( \mathcal{N}_b + 1 \right),\\[1mm]
      \pm ( \dot{\mathbb K}_{u_t}^{(2)} + \text{h.c.} )
      &\leq 2 \norm{ \dot K_{u_t}^{(2)} }_{L^{2}(\mathbb{R}^6)}
      \left( \mathcal{N}_b + 1 \right) \le C  \sqrt{\rho(t)} \left( \mathcal{N}_b + 1 \right).
    \end{align}
  \end{subequations}
  By the reasoning of~\eqref{eq:bound:ell:1} we obtain
  \begin{multline}
    \pm \bigg( \int dk dx \Big( {\dot{\ell}_t^{(1)}} (x,k) b_x^* a_k^*   + {{\dot{\ell}_t^{(2)}} (x,k)} b_x^* a_k  \Big) + \text{h.c.} \bigg)  \\
    \leq  \varepsilon  \text{d} \Gamma_a ( \omega)+ \frac{C \rho(t) }{\varepsilon} \left( \mathcal{N}_b + 1 \right),
  \end{multline}
  and analogously to \eqref{eq:bound:M} one also verifies that
  \begin{align}
    \pm \int dk dl \, \dot{M}_{u_t}(k,l) \, \mathcal A_{kl}
    & \leq  C \rho(t)\,  \left( \mathcal{N}_a + 1 \right).
  \end{align}
  The above estimates prove \eqref{eq:key:bound:Bog:3:Lambda} and thus
  complete the proof of part (a) of the lemma.\medskip

  Turning to part (b), first note that for $\theta=1$ the proof is verbatim
  the same as the one for part (a). To see this, recall that all kernels that
  appear in $\mathbb H_{u,\alpha,1}^\Lambda(t)$ satisfy the same bounds as in
  Lemma \ref{lem:aux:bounds}, as explained thereafter (in particular, they
  are uniformly bounded in $\Lambda$).

  The crucial difference for $\theta \neq 1$, apart from the trivial
  $\theta$-dependence of the kernels, is the appearance of the term
  \begin{align}\label{eq:theta:G:term:proof}
    (1-\theta) \int dk dx\, \bigg( (q_{u_t} G_x^\Lambda(k)   u_t) a_k^* b_x^*  + (q_{u_t} \overline{G_x^\Lambda(k)}  u_t) a_k^* b_x  \bigg) +   \text{h.c.}
  \end{align}
  in $\mathbb H_{u,\alpha,\theta}^\Lambda$ which is not present when
  $\theta=1$ (hence this term did not appear in the proof of part (a)). Since
  $\norm{q_{u_t} G_{(\cdot)}^\Lambda u_t }_{L^2\otimes L^2} \le
  \norm{G_0^\Lambda}_\2 \le C \Lambda$, we can bound the above term by $C
  \Lambda (\mathcal N+1)$. All other contributions in $\mathbb
  H_{u,\alpha,\theta}^\Lambda(t)$ are estimated as for $\theta=1$, i.e. they
  are uniformly bounded in $\Lambda$. This explains the $\Lambda$-dependent
  upper bound in \eqref{eq:key:bound:1:Bog:Lambda}.
  The bounds for the commutator and the time-derivative are obtained in the
  same way. 
\end{proof}

The next lemma was used in the proof of Proposition \ref{prop:theta-Bog} to
show that $\mathbb U_1^\Lambda(t) \to \mathbb U_1^\infty(t)$ strongly as
$\Lambda \to \infty$.

\begin{lemma}\label{lem:H:Bog:Lambda:difference}
  Let $(u,\alpha) \in H^3(\mathbb R^3) \oplus \mathfrak h_{5/2}$ with
  $\norm{u}_\2 =1$, $(u_t,\alpha_t) = \mathfrak s_1[t](u,\alpha)$ as defined
  in \eqref{eq:s_theta-definition} and $ \rho(t) = \norm{u_t}_{\Sob{3}}^2 (1+
  \norm{\alpha_t}_{\h{3/2}})^2$. There is a family $\varepsilon_\Lambda >0$
  with $\varepsilon_\Lambda \xrightarrow{\Lambda \to \infty} 0 $ such that
  \begin{multline*}
    \big| \scp{\chi }{ ( \mathbb H_{u,\alpha,1}^\infty(t) - \mathbb H_{u,\alpha,1}^\Lambda (t) ) \phi } \big|\\
    \le\varepsilon_\Lambda e^{C \int_0^{|t|} \rho(s) ds} \norm{ ( \mathbb  T + \mathcal N +1 )^{1/2} \chi} \norm{ ( \mathbb  T + \mathcal N +1 )^{1/2} \phi}
  \end{multline*}
  for all $t\in \mathbb R$ and $\chi, \phi \in D((\mathbb T + \mathcal
  N)^{1/2})$.
\end{lemma}
\begin{proof}
  The difference $\mathbb H_{u,\alpha,1}^\infty(t) - \mathbb
  H_{u,\alpha,1}^\Lambda (t)$ is up to the term
  $\textnormal{d}\Gamma_b(h_{u_t,\alpha_t}) + \textnormal{d}\Gamma_a(\omega)$
  precisely of the same form as \eqref{eq:theta:HG:Bog} with all kernels
  replaced by kernels of the form $K_{u_t,1}^{ (1), \infty } - K_{ u_t , 1
  }^{ (1) , \Lambda}$ and $M^\infty_{u_t , 1}(k,l) - M_{ u_t ,
    1}^\Lambda(k,l)$, and analogously for the other terms. The claimed bound
  is now obtained following the same steps as in the proof of
  \eqref{eq:key:bound:1:Bog:Lambda} and taking into account that, by Lemma
  \ref{lem:aux:bounds} and continuity in $\Lambda$, the norms of the kernel
  differences all vanish as $\Lambda \to \infty$.
\end{proof}

\subsection{Fluctuation generator for the dressed dynamics}
\label{sec:dressed:fluc:generator}

We start by stating the precise form of the fluctuation generator $H^{\rm
  D}_{u,\alpha}(t)$ introduced in~\eqref{eq:generator:restriction}.  For
$[a]_+ = \max\{0,a\}$ and $\mathfrak s^{\rm D}[t](u,\alpha)$ set
\begin{align}
  H_0(t)  &= \text{d}\Gamma_a(\omega) +  \text{d}\Gamma_b(h_{u_t,\alpha_t})  +  \big[1 - \tfrac{\mathcal N_b}{N} \big]_+   \mathbb K^{(1)}_{u_t}  \notag\\
  &\quad+ \Big( \mathbb K^{(2)}_{u_t} \frac{\sqrt{[(N-\mathcal N_b)(N-\mathcal N_b-1)]_+}}{ N } + \text{h.c.}\Big) \notag\\
  &\quad  +   \int dk dx   (q_{u_t} L_{\alpha_t}(k) u_t)(x)  \, b_x^*a_k^* \sqrt{[1 - \tfrac{\mathcal N_b}{N}]_+}  + \textnormal{h.c.} \notag\\
  &\quad+ \int dk dx (q_{u_t}  L_{\alpha_t}(k)^* u_t)(x)    b_x^* a_k   \sqrt{[1 - \tfrac{\mathcal N_b}{N}]_+}  + \textnormal{h.c.} \notag  \\[1.5mm]
  &\quad  +   \int dk dl \, M_{u_t}(k,l)    \big[1 - \tfrac{\mathcal N_b}{N} \big]_+    \,  \big( -2 a^*_k a_{-l} + a_k^* a_l^* + a_{-k} a_{-l} \big) \label{eq:def:H:0:b}
\end{align}
with $L_{\alpha_t}(k)$, $M_{u_t}(k,l)$ and $\mathbb K^{(j)}_{u_t} $ defined
in \eqref{eq:L(k):operator}, \eqref{eq:def:dressed:meanfield:Ham} and
\eqref{eq:def:K(2)}, respectively. Note that up to the $N$-dependent factors
$H_0(t)$ coincides with $\mathbb H_{u,\alpha}^{\rm D}(t)$. We further
introduce the operators
\begin{subequations}
  \begin{align}
    \mathbb K^{(3)}_{u_t} &  = \int dx_1 dx_2 dx_3\, K_{u_t}^{(3)}(x_1,x_2,x_3) b^*_{x_1} b^*_{x_2} b_{x_3}\\
    \mathbb K^{(4)}_{u_t} &  = \frac{1}{2}\int dx_1 dx_2 dx_3 dx_4 \, K_{u_t}^{(4)}(x_1,x_2,x_3,x_4) b^*_{x_1} b^*_{x_2} b_{x_3}b_{x_4},
  \end{align}
\end{subequations}
where for $u\in L^2(\mathbb R^3)$ and $\{ u \}^\perp \subset L^2(\mathbb
R^3)$ we used the kernels of the operators $K_{u}^{(3)}$, $K_{u}^{(4)}$
introduced in \eqref{eq:def:of:LK3}, \eqref{eq:def:of:LK4}. Lastly, recall
\eqref{eq:kernel:N(t)} and \eqref{eq:kernel:Q(t)} and let
\begin{align}
  J_{u,\alpha}(x,k,y) & = 2  (q_{u} kB_{(\cdot)}(k) \cdot (-i\nabla+ F_{\alpha}) q_{u})(x,y) \label{eq:kernel:J(t)} .
\end{align}

The proof of the next lemma follows from a straightforward computation, which
is postponed to Appendix \ref{section:Generator of the fluctuation dynamics
  of the Gross-transformed Nelson dynamics}.  \allowdisplaybreaks
\begin{lemma}
  \label{lem:fluctuation:generator} For $ (u,\alpha)\in H^3(\mathbb R^3)
  \oplus \mathfrak h_{5/2}$ with $\norm{u}_\2 =1$ let $(u_t,\alpha_t) =
  \mathfrak s^{\mathrm{D}}[t](u,\alpha)$ be the solution to \eqref{eq:SKG
    dressed}. The operator $H_{u,\alpha}^{ {\rm D}, \le N }(t) : \cF^{\leq
    N}_{\perp u_t} \otimes \cF \rightarrow \cF \otimes \cF$ defined by
  \eqref{eq:fluctuation:generator:2} satisfies the identity $H^{ {\rm D}, \le
    N }_{u,\alpha}(t) = H_{u,\alpha}^{\rm D}(t) \restriction \cF^{\leq
    N}_{\perp u_t} \otimes \cF $ where $H^{\rm D}_{u,\alpha}(t) : \cF \otimes
  \cF \rightarrow \cF \otimes \cF $ is given by
  \begin{align}\label{eq:def:fluct:generator}
    H_{u,\alpha}^{\rm D}(t)  & =  \sum_{j=0}^5 H_j (t)
  \end{align}
  with $H_0(t)$ defined by \eqref{eq:def:H:0:b}, and
  \begin{subequations}
    \begin{align}
      H_1(t) & = \begin{aligned}[t]
        &- \frac{1}{N} \mathrm{d} \Gamma(V\ast |u_t|^2-\mu_{u_t,\alpha_t})  \\
        &\quad + \Big( \mathbb K^{(3)}_{u_t} \frac{\sqrt{[N-\mathcal N_b]_+}}{ N} + \textnormal{h.c.} \Big) + \frac{1}{N} \mathbb K^{(4)}_{u_t} \label{eq:def:H:1:b} ,
      \end{aligned}
      \\
      H_2(t) & =  -  \frac{1}{\sqrt N}  \mathcal N_b \hat \Phi(f_{u_t}+g_{u_t,\alpha_t}) , \\[1mm]
      H_3(t) & =  \frac{1}{\sqrt N} \int dk dx dy\, J_{u_t,\alpha_t}(x,k,y)\, b_x^* b_y \, a_k^* + \textnormal{h.c.} , \\[1.5mm]
      H_4(t) & =\begin{aligned}[t]
        \frac{1}{\sqrt N} \int dk dl dx \Big(& N_{u_t}(x,k,l) \, b_x^* \big[ 1  - \tfrac{\mathcal N_b}{N} \big]_+^{1/2}    \\
        &+ \overline{ N_{u_t}(x, k, l) } \big[ 1  - \tfrac{\mathcal N_b}{N}\big]_+^{1/2} b_x  \Big) \mathcal A_{kl} ,
      \end{aligned}
      \\
      H_5(t) & = \frac{1}{N} \int dk dl  dx dy\, Q_{u_t}(x,y,k,l) \, b_x^* b_y \, \mathcal A_{kl},
    \end{align}
  \end{subequations}
  where $\mathcal A_{kl} = -2 a^*_k a_{-l} + a_k^* a_l^* + a_{-k} a_{-l}$.
\end{lemma}

The next lemma provides estimates for the fluctuation Hamiltonian that are an
important ingredient of the proof of Theorem \ref{thm:gross-transformed
  dynamics reduced density matrices}.

\begin{lemma}\label{lem:bounds:fluc:generator} Let $(u,\alpha) \in
  H^3(\mathbb R^3) \oplus \mathfrak h_{5/2}$ with $\norm{u}_\2=1$,
  $(u_t,\alpha_t)= \mathfrak s^{\rm D}[ t](u,\alpha)$ denote the solution to
  \eqref{eq:SKG dressed}. There exists a constant $C>0$ such that for all
  $t\in \mathbb R$
  \begin{subequations}
    \begin{align}
      \pm \big( H^{\rm D}_{u,\alpha}(t) - \mathbb T \big) & \le \tfrac{1}{2}\mathbb T + C  (\mathcal N+ 1) (1 + \tfrac{1}{N}\mathcal N_b )^2    \label{eq:key:bound:1:L(t)}\\[1.5mm]
      \pm i [\mathcal N ,  H^{\rm D}_{u,\alpha}(t)  ]  & \le  \tfrac{1}{2} \mathbb T + C  (\mathcal N+ 1) (1 + \tfrac{1}{N}\mathcal N_b )^2    \label{eq:key:bound:3:L(t)}\\[1.5mm]
      \pm \tfrac{d}{dt}  H^{\rm D}_{u,\alpha}(t)  & \le  \tfrac{1}{2} \mathbb T  + C \rho(t)  (\mathcal N+ 1) (1 + \tfrac{1}{N}\mathcal N_b )^2 .   \label{eq:key:bound:4:L(t)}
    \end{align}
  \end{subequations}
  as quadratic forms on $\Fock \otimes \Fock$, where $\rho(t) =
  \norm{u_t}^2_{\Sob{3}} (1+\norm{\alpha_t}_{\h{3/2}} )^2 $.
\end{lemma}

\begin{proof}[Proof of \eqref{eq:key:bound:1:L(t)}] Comparing
  \eqref{eq:def:H:0:b} with \eqref{eq:diff:Bog:A} we see that $H_0(t)$
  differs from $\mathbb H^{\rm D}_{u,\alpha}(t)$ only by the factors
  \begin{align}
    0\le \big[1 - \tfrac{\mathcal N_b}{N} \big ]_+ \le 1 ,\quad 0\le \tfrac{\sqrt{[(N-\mathcal N_b)(N-\mathcal N_b-1)]_+}}{ N } \le 1
  \end{align}
  and is thus estimated in analogy to the proof of Lemma
  \ref{lem:bounds:dressed:Bog:theta}.

  For later purpose we consider $ \scp{\chi}{H_1(t) \phi}$ for $\chi,\phi \in
  \mathcal F \otimes \mathcal F$. The first two terms in $H_{1}(t)$ are
  estimated by
  \begin{align}
    \label{eq:estimate second quantisation of mean-field potential of V}
    \big| \scp{ \chi}{ \frac{1}{N}  \textnormal{d}\Gamma_b \big( V\ast |u_t|^2 - \mu_{u_t,\alpha_t} \big) \phi } \big| \le \frac{C}{N} \| \mathcal  N_b^{1/2}  \chi \| \, \| \mathcal N_b^{1/2} \phi \|,
  \end{align}
  where we used Lemma \ref{lem:bounds:f:g}, and
  \begin{align}
    &\Big| \scp{\chi}{\Big( \mathbb K^{(3)}_{u_t} \tfrac{\sqrt{[N-\mathcal N_b]_+}}{ N}  + \text{h.c.} \Big) \phi } \Big| \notag\\
    &\quad \le  \tfrac{1}{\sqrt N} \norm{ K^{(3)}_{ u_t } }_{L^2 \to L^2\otimes L^2} \| ( \mathcal N_b + 1 )^{1/2} \chi \|\, \| (\mathcal N_b+1) \phi\|  \notag \\
    &  \quad \le  \tfrac{C}{\sqrt N} \| ( \mathcal N_b + 1 )^{1/2} \chi \|\, \| (\mathcal N_b+1) \phi\|,\label{eq:estimate K3-term}
  \end{align}
  by Lemmas \ref{lem:ops-Fock} and \ref{lem:aux:bounds}.

  To estimate the term involving $\mathbb K_{u_t}^{(4)}$ we recall
  \eqref{eq:definition of V} and write $ W_{u_t}(x,y) = W^{{\rm
      b}}_{u_t}(x,y) - 4 \Re\scp{G_x}{B_y} $, where $ W_{u_t}^{{\rm b}}(x,y)$
  is point-wise bounded.  Using the symmetry of $\chi^{(n)}, \phi^{(n)}\in
  \mathcal F^{(n)} \otimes \mathcal F$ in the $n$ particle coordinates, we
  find
  \begin{multline}
    \scp{\chi }{\mathbb{K}_{u_t}^{(4)} \phi}
    \\
    = \sum_{n=2}^{\infty} \frac{n(n-1)}{2} \scp{q_1 q_2\chi^{(n)} }{\big( W_{u_t}^{\rm{b}}(x_1,x_2) - 4 \Re \scp{G_{x_1}}{B_{x_2}} \big) q_1 q_2 \phi^{(n)}}
  \end{multline}
  with
  \begin{align}
    \Big|\sum_{n=2}^{\infty} \frac{n(n-1)}{2} \scp{q_1 q_2 \chi^{(n)}}{ W_{u_t}^{\rm{b}}(x_1,x_2)  q_1 q_2 \phi^{(n)} }\Big|
    &\leq C \norm{\mathcal{N}_b^{\frac12} \chi}
    \norm{\mathcal{N}_b^{\frac32} \phi} .
  \end{align}
  We use the Cauchy--Schwarz inequality and the fact that
  \begin{equation}\label{eq:identity:p-k}
    e^{ikx}= (1+ (-i\nabla_1 - k )^2)^{-1/2} e^{ikx} (1-\Delta_1)^{1/2}
  \end{equation}
  to bound the remaining term by
  \begin{align}
    & \Big|\sum_{n=2}^{\infty} \frac{n(n-1)}{2} \scp{q_1 q_2 \chi^{(n)} }{\Re \scp{G_{x_1}}{B_{x_2}}  q_1 q_2 \phi^{(n)}}\Big|
    \\
    &  \leq
    C
    \Big( \sum_{n=0}^{\infty} n^3 \norm{ \phi^{(n)}}^2 \Big)^{1/2}
    \Big( \sum_{n=0}^{\infty} n
    \norm{\Re \scp{G_{x_1}}{B_{x_2}}  q_1 q_2 \chi^{(n)} }^2 \Big)^{1/2}
    \nonumber \\
    &  \leq  C\norm{\mathcal{N}_b^{3/2} \phi} \notag
    \\
    &\ \times \Big( \sum_{n=0}^{\infty} n
    \Big\|\int dk   ( 1 + (- i \nabla_1 - k)^2 )^{-\frac12}  \overline{G_{x_1}(k)} B_{x_2}(k) ( 1 - \Delta_1 )^{\frac12}  q_1 q_2 \chi^{(n)}\Big\|^2 \Big)^{\frac12} .\notag
  \end{align}
  The last factor is bounded by
  \begin{multline}
    \sup_{p \in \mathbb{R}^3}
    \left\{   \int dk \, \frac{G_0(k) B_0(k)}{(1 + (p - k)^2)^{1/2}}  \right\}
    \left( \sum_{n=0}^{\infty} n
      \norm{ \left( 1 - \Delta_1 \right)^{1/2} q_1 \phi^{(n)}}^2 \right)^{1/2}\\
    \leq  C
    \norm{\textnormal{d}\Gamma_b(1 - \Delta)^{1/2} \chi },
  \end{multline}
  where the supremum over $p \in \mathbb R^3$ is finite by the same argument
  as in \eqref{eq:symmetric:rearrangement}, and where we further used
  \begin{align}\label{eq:nabla:q:Psi}
    \norm{\left( 1 - \Delta_1 \right)^{1/2} q_1 \chi ^{(n)}(t)} \leq C \norm{u_t}_{H^1} \norm{\left( 1 - \Delta_1 \right)^{1/2} \chi^{(n)}(t)}.
  \end{align}
  Adding up the relevant terms, we arrive at the desired estimate
  \begin{multline}\label{eq:chi:phi:bound:H1}
    \big| \scp{\chi } { H_1(t)  \phi } \big|\\
    \leq C  \norm{\left( \mathcal{N} +\mathbb T + 1 \right)^{1/2} \chi}\Big( \tfrac{1}{\sqrt N}
    \norm{(\mathcal{N}_b+1) \phi} + \tfrac{1}{N}
    \norm{(\mathcal{N}_b+1) \mathcal N_b^{1/2} \phi} \Big).
  \end{multline}

  The bound for $H_2(t)$ is straightforward,
  \begin{align}\label{eq:chi:phi:bound:H2}
    \big| \scp{\chi}{H_2(t) \phi} \big|   \le  C \tfrac{1}{\sqrt N}  \| (\mathcal N+1)^{1/2} \chi \| \, \| (\mathcal N+1) \phi \| ,
  \end{align}
  where we used that $\norm{f_{u_t}+g_{u_t,\alpha_t}}_\2 \le C $ by Lemma
  \ref{lem:aux:bounds}.

  With the definition of the kernel of $H_3(t)$ in \eqref{eq:kernel:J(t)}, we
  write
  \begin{align}
    & \scp{\chi}{ H_3(t) \chi} \\
    & \ = \sum_{n=1}^\infty  \frac{2n}{\sqrt N}   \int dk \, kB_0(k) \cdot 2\Re \scp{q_1  \chi^{(n)}}{  e^{ - i k x_1 } (-i\nabla_{1} + F_{\alpha_t}(x_1) ) q_1  \, a_k^*\chi^{(n)}  } \notag
  \end{align}
  with $\chi^{(n)}\in \mathcal F^{(n)} \otimes \mathcal F $. With $\norm{
    F_{\alpha_t} }_{\Lp{\infty}} \le C$ and $ \norm{kB_0}_{\h{-1/4}} \le C$,
  we obtain by Cauchy--Schwarz
  \begin{align}
    &  \big| \scp{\chi}{H_3 (t) \chi}\big| \leq  \sum_{n=1}^\infty \frac{4n}{\sqrt N}  \norm{  (-i\nabla_{1} + F_{\alpha_t}(x_1) ) q_1  \chi^{(n)}  }   \int dk \, |kB_0(k)|\, \norm{a_k \chi^{(n)}}  \notag\\
    & \quad\stackrel{\eqref{eq:nabla:q:Psi}}{\le} C  \sum_{n=1}^\infty \frac{4n}{\sqrt N} \norm{(1 -\Delta_1)^{1/2}   \chi^{(n)}  }   \norm{ \text{d}  \Gamma_a (\sqrt\omega)^{1/2} \chi^{(n)} } \notag\\
    &\quad \stackrel{\hphantom{\eqref{eq:bound:dG:sqrt:omega}}}{\le} \frac{C}{\sqrt N} \| \textnormal{d}\Gamma_b(1-\Delta)^{1/2} \chi \| \, \| \mathcal N_b^{1/2} \textnormal{d}\Gamma_a(\sqrt  \omega)^{1/2} \chi\| \notag\\
    &\quad\stackrel{\eqref{eq:bound:dG:sqrt:omega}}{\leq} \frac{C}{\sqrt N} \| \textnormal{d}\Gamma_b(1-\Delta)^{1/2} \chi \| \, \| d\Gamma_a(\omega)^{1/2} \chi \|^{1/2} \,  \| \mathcal N_a^{1/2} \mathcal N_b \chi\|^{1/2},
    \label{eq:jt:bound}
  \end{align}
  which implies the desired bound.

  To bound $H_4(t)$, we apply Lemma \ref{lem:ops-Fock} (for instance for the
  term involving $b_x^* a_k^* a_{-l}$ we choose $m_b=0,n_b = m_a=n_a=1$,
  $s=t=1/4$ and $r_a=r_b=0$) and use from Lemma \ref{lem:aux:bounds} that $
  \norm{N_t}_{\h{1/4} \to L^2 \otimes \h{-1/4}} + \norm{ N_t }_{L^2\otimes
    \h{1/4}^{\otimes 2}\to \mathbb C} \le 2 \norm{ N_t}_{ L^2 \to
    \h{-1/4}^{\otimes 2} } \le 2 \norm{N_t}_{L^2 \otimes \h{-1/8}^{\otimes
      2}} \le C $. This implies
  \begin{align}
    \big| \scp{\chi}{H_4(t) \chi} \big| &  \le C \Big[  \norm{ (\mathcal N_b+1)^{1/2}  \text{d} \Gamma_a(\sqrt \omega )^{1/2} \chi } \norm{  \text{d} \Gamma_a(\sqrt \omega )^{1/2} \chi } \\[1mm]
    & \qquad + \norm{ (1+\mathcal N_b)^{\tfrac{1}{2}}  (1+\mathcal N_a)^{\tfrac{1}{2}} \chi}  \norm{(1+\mathcal N_a)^{-\tfrac{1}{2}}   \text{d} \Gamma_a(\sqrt \omega ) \chi} \notag\\[2mm]
    & \qquad +   \norm{  (\mathcal N_b+1)^{1/2}  (\mathcal N_a+1)^{1/2} \chi}  \norm{ (\mathcal N_a+1 )^{-1/2} \text{d} \Gamma_a(\sqrt \omega )  \chi} \Big]. \notag
  \end{align}
  The desired bound now follows easily from \eqref{eq:bound:dG:sqrt:omega}.
  For $H_5(t)$, we use that $\| Q_t\|_{L^2 \otimes \mathfrak h_{1/4} \to L^2
    \otimes \mathfrak h_{-1/4} } + \norm{Q_t}_{L^2 \otimes \h{1/4}^{\otimes
      2} \to L^2 } \le C\norm{Q_t}_{L^2 \otimes \h{1/8}^{\otimes 2} \to L^2 }
  \le C$ by Lemma \ref{lem:aux:bounds}, and Lemma \ref{lem:ops-Fock} imply
  \begin{align}
    \big| \scp{\chi}{H_5(t) \chi } \big| & \le \frac{C}{N}  \Big[  \norm{ (\mathcal N_b + 1 )^{1/2} \text{d} \Gamma_a (\sqrt \omega)^{1/2} \chi }^2 \\[1mm]
    & \qquad + \norm{(\mathcal N_b+1)  (\mathcal N_a+1)^{1/2} \chi }  \norm{ (\mathcal N_a+1)^{-1/2} \text{d}\Gamma_a(\sqrt{\omega})     \chi} \Big].\notag
  \end{align}
  The desired bound now follows from Lemma \ref{lem:bounds:dGamma}.\medskip

  \noindent \textit{Proof of \eqref{eq:key:bound:3:L(t)}}. The commutation
  relations imply, e.g., $[\cN, a_k]=-a_k$, so the non-zero terms in the
  commutator $[\cN, H_{u,\alpha}^\mathrm{D}(t)]$ have the same kernels, up to
  signs, as those in $H_{u,\alpha}^\mathrm{D}(t)-\mathbb T$. They can thus be
  estimated as in the proof of \eqref{eq:key:bound:1:L(t)}, and we omit the
  details.
  \medskip

  \noindent \textit{Proof of \eqref{eq:key:bound:4:L(t)}}. This inequality is
  obtained following similar steps as in the proof of
  \eqref{eq:key:bound:1:L(t)} with some obvious modifications, like the use
  of the bounds for the time-derivatives in Lemma \ref{lem:aux:bounds} and
  the use of $\| \dot q_{u_t}\|_{H^1} \le C \sqrt{\rho(t)} $,
  cf. Lemma. \ref{lem:bounds:alpha:phi:dot}.
\end{proof}

The next lemma shows that the fluctuation generator can be approximated by
the Bogoliubov Hamiltonian $\mathbb H^{\rm D}_{u,\alpha}(t)$ for large $N$,
when tested on suitable states.

\begin{lemma}\label{lem:fluctuationgen:BogHam:difference}
  Let $(u, \alpha) \in H^3(\mathbb{R}^3) \oplus \mathfrak{h}_{5/2}$ with
  $\norm{u}_{L^2} =1$ and $(u_t,\alpha_t)= \mathfrak s^{\rm D}[ t](u,\alpha)$
  denote the solution to \eqref{eq:SKG dressed}. There exists a constant
  $C>0$ such that
  \begin{multline}
    \big| \scp{\chi }{\left( H^{\rm D}_{u,\alpha}(t) - \mathbb{H}_{u,\alpha}^{\rm D}(t) \right)  \phi} \big| \\
    \le C \rho(t)   N^{-1/2} \ln N \norm{\left( \mathcal{N} + \mathbb{T} + 1 \right)^{1/2} \chi }
    \norm{\left( \mathcal{N}^3 + \mathbb{T} + 1 \right)^{1/2} \phi}\notag
  \end{multline}
  for all $\chi \in \mathcal F^{\le N} \otimes \mathcal F$ and $\phi \in
  \mathcal F \otimes \mathcal F$, where $\rho(t) = \norm{u_t}^2_{\Sob{3}}
  (1+\norm{\alpha_t}_{\h{3/2}} )^2$.
\end{lemma}

\begin{proof} Recalling the definitions of the fluctuation generator
  \eqref{eq:def:fluct:generator} and the Bogoliubov Hamiltonian
  \eqref{eq:diff:Bog:A}, we write
  \begin{subequations}
    \begin{align}
      & \scp{\chi}{ ( H^{\rm D}_{u,\alpha}(t) - \mathbb{H}_{u,\alpha}^{\rm D}(t) ) \phi} \notag\\[1mm]
      &  =  \label{eq:Norm approximation Gross transformed dynamics 2b}
      \scp{\chi}{ \mathbb K^{(1)}_{u_t} \left( \big[1 - \tfrac{\mathcal N_b}{N} \big]_+ - 1 \right)   \phi}  \\[1mm]
      & \ + \scp{\chi}{  \left( \mathbb K^{(2)}_{u_t} \big( N^{-1} \sqrt{[(N-\mathcal N_b)(N-\mathcal N_b-1)]_+} - 1 \big) + \text{h.c.} \right) \phi} \label{eq:Norm approximation Gross transformed dynamics 2bb}
      \\
      \label{eq:Norm approximation Gross transformed dynamics 2c}
      &\ + \scp{\chi}{\Big( \int dx  dk\, \big(q_{u_t} L_{ \alpha_t}(k) u_t\big)(x) a^{*}_k b_x^* \Big( \big[1 - \tfrac{\mathcal N_b}{N} \big]_+^{\frac12} - 1 \Big) + \text{h.c.}   \Big)  \phi}
      \\
      \label{eq:Norm approximation Gross transformed dynamics 2d}
      &\ +  \scp{\chi}{ \Big( \int dx  dk\, \big(q_{u_t} L_{ \alpha_t}(k)^* u_t\big)(x) a_{k} b_x^* \Big( \big[1 - \tfrac{\mathcal N_b}{N} \big]_+^{\frac12} - 1 \Big) + \text{h.c.}   \Big)  \phi}
      \\
      \label{eq:Norm approximation Gross transformed dynamics 2e}
      &\ +  \scp{\chi}{\int dk dl \, M_{u_t}(k,l) \mathcal{A}_{kl} \Big( \big[1 - \tfrac{\mathcal N_b}{N} \big]_+ - 1 \Big)  \phi} \\[1mm]
      \label{eq:Norm approximation Gross transformed dynamics 2a}
      & \ + \scp{\chi}{\left( H_1(t) + H_2(t) + H_3(t) + H_4(t) + H_5(t) \right) \phi}
    \end{align}
  \end{subequations}
  When bounding these terms, we have to take care to put any powers of the
  number operator exceeding one-half to the right, i.e., on $\phi$. At the
  same time, the power of $\cN + \mathbb{T}$ acting on $\phi$ cannot exceed
  one-half, either.  For the terms in~\eqref{eq:Norm approximation Gross
    transformed dynamics 2b}--\eqref{eq:Norm approximation Gross transformed
    dynamics 2e} the estimates are rather straightforward and given at the
  end of the proof.  The most difficult estimate is that for the term coming
  from $H_3(t)$, which is also responsible for the presence of the factor
  $\ln N$ in the statement. \medskip

\noindent\textit{Term $H_3(t)$}. Recall the expression for the kernel $J_{u,\alpha} = 2  (q_{u} kB_{(\cdot)}(k) \cdot (-i\nabla+ F_{\alpha}) q_{u})$, which multiplies $b_x^*b_y a_k^*$ in $H_3(t)$. This term is problematic, for when the gradient acts on $\phi$ we cannot put further powers of $\cN$ on $\phi$ while keeping control by $\norm{\left( \mathcal{N}^3 + \mathbb{T} \right)^{1/2} \phi}$.
We deal with this problem by using the identity
\begin{align}
  k \cdot B_{x_1}(k) \left( - i \nabla_1 + F_{\alpha_t}(x_1) \right)a^*_k
  &= \left( - i \nabla_1 + F_{\alpha_t}(x_1) \right)\cdot k B_{x_1}(k)
  + a^*_k k^2 B_{x_1}(k)
\end{align}
and splitting the momentum integration into $| k | \le \Lambda$ and
$|k|>\Lambda$. Together with the adjoint expression, which is less of a
problem, this gives
\begin{subequations}
  \begin{align}
    \label{eq:Norm approximation Gross transformed dynamics 2f1}
    &\tfrac{1}{2} \scp{\chi}{H_3(t) \phi}   \notag \\
    &= \frac{1}{\sqrt{N}} \sum_{n=0}^{\infty} n  \scp{\left( - i \nabla_1 + F_{\alpha_t}(x_1) \right)  q_1  \chi^{(n)} }{ \int dk \, k \overline{B_{x_1}(k)} a_k q_1 \phi^{(n)} }
    \\
    \label{eq:Norm approximation Gross transformed dynamics 2f2}
    &\ +  \frac{1}{\sqrt{N}} \sum_{n=0}^{\infty} n  \scp{\chi^{(n)}  }{q_1 \int\limits_{\abs{k} \geq \Lambda} dk \, k B_{x_1}(k) a_k^* \left( - i \nabla_1 + F_{\alpha_t}(x_1) \right)  q_1 \phi^{(n)}  }
    \\
    \label{eq:Norm approximation Gross transformed dynamics 2f3}
    &\ +  \frac{1}{\sqrt{N}} \sum_{n=0}^{\infty} n  \scp{\left( - i \nabla_1 + F_{\alpha_t}(x_1) \right)  q_1 \chi^{(n)} }{ \int\limits_{\abs{k} \leq \Lambda} dk \, k B_{x_1}(k) a_k^*   q_1 \phi^{(n)} }
    \\
    \label{eq:Norm approximation Gross transformed dynamics 2f4}
    &\ +  \frac{1}{\sqrt{N}}\sum_{n=0}^{\infty} n  \scp{ \int_{\abs{k} \leq \Lambda} dk \, k^2 \overline{B_{x_1}(k)} a_k   q_1 \chi^{(n)}  }{  q_1 \phi^{(n)} }  .
  \end{align}
\end{subequations}
In the first line, the gradient acts on $\chi$, so we can simply bound it as
in \eqref{eq:jt:bound},
\begin{align}
  | \eqref{eq:Norm approximation Gross transformed dynamics 2f1} |
  &\leq C N^{-1/2}
  \norm{\textnormal{d} \Gamma_b(1 - \Delta)^{1/2} \chi } \norm{\textnormal{d} \Gamma_a (\omega)^{1/2}\phi\|^{1/2} }  (\mathcal{N}+1)^{3/2} \phi \|^{1/2}.
\end{align}
In the second line, we can use that $\chi^{(n)} = 0$ for $n > N$ to remove a
factor of $(n/N)^{1/2}$, since the lower cutoff $\Lambda$ will give us a
small pre-factor. With the Cauchy-Schwarz inequality and Lemma
\ref{lem:ops-Fock} this gives
\begin{align}
  &| \eqref{eq:Norm approximation Gross transformed dynamics 2f2} |
  \leq \sum_{n=0}^{\infty} n^{1/2} \Big| \scp{  \chi^{(n)} }{q_1 \int_{\abs{k} \geq \Lambda} dk \, k B_{x_1}(k) a_k^* \left( - i \nabla_1 + F_{\alpha_t}(x_1) \right)  q_1 \phi^{(n)} } \Big|
  \nonumber \\
  &\ \leq C
  \Big( \sum_{n=0}^{\infty} \Big\|\int_{\abs{k} \geq \Lambda} dk \, k \overline{B_{x_1}(k)} a_k q_1 \chi^{(n)} \Big\|^2 \Big)^{\frac12}
  \Big(  \sum_{n=0}^{\infty} n \norm{\left( 1 - \Delta_1 \right)^{\frac12} \phi^{(n)} }^2 \Big)^{\frac12}
  \nonumber \\
  &\ \leq C \norm{\id_{\abs{\cdot} \geq \Lambda} \omega^{-1/2} k B_0}_{L^2}
  \norm{ \textnormal{d} \Gamma_a(\omega)^{1/2} \chi }
  \norm{\textnormal{d} \Gamma_b(1-\Delta)^{1/2}\phi}
  \nonumber \\[2mm]
  &\ \leq C  \Lambda^{-1/2} \norm{\mathbb{T}^{1/2} \chi} \norm{\left( \mathcal{N} + \mathbb{T} \right)^{1/2} \phi },
\end{align}
where we used $\norm{\id_{\abs{\cdot} \geq \Lambda} \omega^{-1/2} k
  B_0}_{L^2} \leq \sqrt{4 \pi / \Lambda}$ in the last
step. Lemma~\ref{lem:ops-Fock} together with $\norm{\id_{\abs{\cdot} \leq
    \Lambda} k B_0}_2 \leq \sqrt{4 \pi \ln \Lambda }$ yields for the third
line
\begin{align}
  | \eqref{eq:Norm approximation Gross transformed dynamics 2f3} |
  &\leq C N^{-1/2} C N^{-1/2} \sqrt{\ln \Lambda }
  \norm{\textnormal{d} \Gamma_b (1 - \Delta)^{1/2} \chi } \norm{\left( \mathcal{N} + 1 \right) \phi } .
\end{align}
It remains to bound the last line, \eqref{eq:Norm approximation Gross
  transformed dynamics 2f4}. Here, we will need to use the regularity of
$\chi$ in $x$ to improve the integrability of $k^2 B_x(k)$.  Using the
identity~\eqref{eq:identity:p-k} to this end, we obtain
\begin{align}
  &\Big\|\int\limits_{\abs{k} \leq \Lambda} dk \, k^2 \overline{B_{x_1}(k)} a_k \chi\Big\|^2
  = \hspace{-1pt} \int\limits_{\abs{k} \leq \Lambda}\hspace{-1pt} dk  \hspace{-1pt} \int\limits_{\abs{l} \leq \Lambda} \hspace{-1pt} dl \, k^2 B_{0}(k)  l^2 B_{0}(l)
  \scp{e^{i k x_1} a_k \chi}{ e^{i l x_1}  a_l \chi}
  \nonumber \\
  &=  \int\limits_{\abs{k} \leq \Lambda} dk  \int\limits_{\abs{l} \leq \Lambda} dl \,
  \begin{aligned}[t]
    k^2 B_{0}(k)   l^2 B_{0}(l) \Big\langle & ( (- i \nabla_1 - l)^2 + 1 )^{-\frac12}   e^{ikx_1} (1 - \Delta_1)^{\frac12}  a_k \chi , \\
    &( (- i \nabla_1 - k)^2 + 1 )^{-\frac12} e^{ilx_1}  (1 - \Delta_1)^{\frac12} a_l \chi \Big\rangle
  \end{aligned}
  \nonumber \\
  &\ \leq
  \int_{\abs{k} \leq \Lambda} dk \,  \int_{\abs{l} \leq \Lambda} dl \, k^4 B_{0}(k)^2
  \norm{( (- i \nabla_1 - k)^2 + 1 )^{-1/2}
    e^{ilx_1}  (1 - \Delta_1)^{1/2} a_l \chi}^2
  \nonumber \\
  &\ \leq
  \sup_{p \in \mathbb{R}^3}
  \left\{  \int_{\abs{k} \leq \Lambda} dk \,   \frac{k^4 B_{0}(k)^2}{1 + (p-k)^2}  \right\} \norm{\mathcal{N}_a^{1/2} \left( 1 - \Delta_1 \right)^{1/2} \chi}^2.
\end{align}
By symmetric rearrangement (similarly as in
\eqref{eq:symmetric:rearrangement}) the supremum over $p\in \mathbb R^3$ is
bounded by a constant times $\ln \Lambda$.  With this inequality, we can
estimate the remaining term by
\begin{align}
  &| \eqref{eq:Norm approximation Gross transformed dynamics 2f4} |
  = \frac{1}{\sqrt N} \sum_{n=1}^{\infty} n \Big| \scp{ \int\limits_{\abs{k} \leq \Lambda} dk \, k^2 \overline{B_{x_1}(k)} a_k   \mathcal{N}_a^{-\frac12}  q_1 \chi^{(n)} }{  q_1 \left( \mathcal{N}_a + 1 \right)^{\frac12} \phi^{(n)} } \Big|
  \nonumber \\
  &\quad\leq C \frac{\sqrt{\ln \Lambda }}{\sqrt{N}}
  \Big( \sum_{n=0}^{\infty} n \norm{\left( 1 - \Delta_1 \right)^{\frac12} q_1 \chi^{(n)}}^2 \Big)^{\frac12}
  \Big( \sum_{n=0}^{\infty} n \norm{\left( \mathcal{N}_a + 1 \right)^{\frac12} \phi^{(n)} }^2 \Big)^{\frac12}
  \nonumber \\[1mm]
  &\quad\leq C \frac{\sqrt{\ln \Lambda }}{\sqrt{N}}
  \norm{\left( \mathcal{N} + \mathbb{T} \right)^{\frac12} \chi }
  \norm{\left( \mathcal{N} + 1 \right) \phi } .
\end{align}
If we choose the cutoff parameter $\Lambda = N$ we thus arrive at
\begin{multline}
  \big| \scp{\chi}{H_3(t) \phi} \big| \notag\\
  \leq CN^{-\frac12} \norm{\left( \mathcal{N} + \mathbb{T} + 1 \right)^{\frac12} \chi }
  \big(
  \norm{\left( \mathcal{N}^3 + \mathbb{T} + 1 \right)^{\frac12} \chi } + \sqrt{\ln  N} \norm{\left( \mathcal{N} + 1 \right) \phi }
  \big).
\end{multline}

The terms $H_j(t)$ for $j\neq3$ are somewhat easier to treat, since if there
are any gradients (as in $H_2(t)$, via $f_{u}$) they act on $u_t$ and not
$\chi, \phi$. The powers of $\omega$ needed to render the kernels integrable
are strictly less than one, so the possibility of distributing factors of
$\cN$ given by Lemma~\ref{lem:ops-Fock} is sufficient to treat theses terms,
as we now show.  \medskip

\noindent\textit{Term} $H_1(t)+H_2(t)$. For this contribution we can use the
already established bounds from the proof of Lemma
\ref{lem:bounds:fluc:generator}, that is, \eqref{eq:chi:phi:bound:H1} and
\eqref{eq:chi:phi:bound:H2}, respectively.\medskip

\noindent \textit{Term $H_4(t)$}. We use $\norm{N_t}_{\mathfrak{h}_{1/8}
  \rightarrow L^2 \otimes \mathfrak{h}_{- 1/8}} +
\norm{N_t}_{\mathfrak{h}_{1/8}^{\otimes 2} \rightarrow L^2} \leq
\norm{N_t}_{L^2 \otimes \mathfrak{h}_{-1/8}^{\otimes 2}} \le C $ (see Lemma
\ref{lem:aux:bounds}) and Lemma \ref{lem:ops-Fock} to get
\begin{align}
  \big| \scp{\chi}{ H_4(t) \phi } \big|
  & \leq  \frac{C}{\sqrt N} \begin{aligned}[t]
    \Big[&
    \norm{\left( \mathcal{N} + 1 \right)^{\frac12} \chi } \norm{\textnormal{d} \Gamma_a (\omega^{\frac14}) \phi}
    \\
    &+ \norm{\left( \mathcal{N} + 1 \right)^{-\frac12} \textnormal{d} \Gamma_a(\sqrt{\omega}) \chi } \norm{\left( \mathcal{N} + 1 \right) \phi }
    \\
    &+ \norm{ \textnormal{d} \Gamma_a(\sqrt \omega)^{\frac12} \chi } \norm{ \textnormal{d} \Gamma_a(\sqrt \omega)^{\frac12} \left( \mathcal{N} + 1 \right)^{\frac12} \phi}
    \Big] .
  \end{aligned}
\end{align}
By means of \eqref{eq:normbound:dGamma(omega^s/2)} we then obtain
\begin{align}
  \big| \scp{\chi}{H_4(t) \phi } \big|  & \leq C N^{-1/2} 
  \norm{\left( \mathcal{N} + \mathbb{T} + 1 \right)^{1/2} \chi } 
  \norm{\left( \mathcal{N}^3 + \mathbb{T} + 1 \right)^{1/2} \phi } .
\end{align}
\noindent \textit{Term $H_5(t)$}. Recalling that $| k| B_0(k) \in
\mathfrak{h}_{-s}$ for $s >0$, Lemma~\ref{lem:ops-Fock} gives
\begin{align}
  \big| \scp{\chi}{ H_5(t) \phi } \big|
  &\leq C N^{-1} \Big[ \norm{\cN_b^{1/2} \textnormal{d} \Gamma_a (\sqrt{\omega})^{1/2} \chi }
  \norm{ \textnormal{d} \Gamma_a (\sqrt{\omega})^{1/2} \mathcal{N}_b ^{1/2} \phi }
  \nonumber \\[2mm]
  &\hspace{2cm} + 
  \norm{\cN_b^{1/2}\left(  \mathcal{N}_a + 1 \right)^{-1/2} \textnormal{d}\Gamma_a (\sqrt{\omega}) \chi }
  \norm{\left( \mathcal{N} + 1 \right) \phi}
  \nonumber \\[2mm]
  &\hspace{2cm} + 
  \norm{ \mathcal{N}_b   \chi }
  \norm{ \textnormal{d}\Gamma_a(\omega^{1/4})  \phi } \Big].
\end{align} 
Since $ \chi=\id_{\cN_b\leq N}\chi$,
Equation~\eqref{eq:normbound:dGamma(omega^s/2)} leads to
\begin{align}
  \big| \scp{\chi }{H_5(t) \phi } \big| & \leq C N^{-1/2}  \norm{\left( \mathcal{N} + \mathbb{T} + 1 \right)^{1/2} \chi } \norm{\left( \mathcal{N}^3 + \mathbb{T} + 1 \right)^{1/2} \phi } .
\end{align}

We conclude by estimating the terms from \eqref{eq:Norm approximation Gross
  transformed dynamics 2b}--\eqref{eq:Norm approximation Gross transformed
  dynamics 2e} by using that
\begin{subequations}
  \begin{align}
    \label{eq:expansion of the number operators}
    \pm \big(\big[1 - \tfrac{\mathcal N_b}{N} \big]_+ - 1 \big) & \leq N^{-1} \mathcal{N}_b , \\
    \pm \big( N^{-1} \sqrt{[(N-\mathcal N_b)(N-\mathcal N_b-1)]_+} - 1 \big)
    & \leq C  N^{-1} \mathcal{N}_b.
  \end{align}
\end{subequations}

\noindent\textit{Terms \eqref{eq:Norm approximation Gross transformed dynamics 2b} and  \eqref{eq:Norm approximation Gross transformed dynamics 2bb}}. Using \eqref{eq:estimate second quantisation of mean-field potential of V}, \eqref{eq:estimate K3-term} and $\id_{\mathcal N_b \le N} \chi = \chi$,
we can estimate the first two lines by
\begin{align}
  | \eqref{eq:Norm approximation Gross transformed dynamics 2b} | + | \eqref{eq:Norm approximation Gross transformed dynamics 2bb} |
  &\leq C N^{-1/2} \norm{\left( \mathcal{N} + 1 \right)^{1/2} \chi }
  \norm{\left( \mathcal{N} + 1 \right) \phi }.
\end{align}

\noindent \textit{Term \eqref{eq:Norm approximation Gross transformed
    dynamics 2c}}.  Using Lemmas \ref{lem:aux:bounds} and \ref{lem:ops-Fock}
we arrive at 
\begin{align}
  &|\eqref{eq:Norm approximation Gross transformed dynamics 2c}| =  2 \Big| \mathrm{Re} \scp{\chi }{  \int dx  dk\, \ell^{(1)}(x,k) a^{*}_k b_x^* \Big( \big[1 - \tfrac{\mathcal N_b}{N} \big]_{+}^{1/2} - 1 \Big)  \phi}\Big| \notag \\
  & \stackrel{\hphantom{\eqref{eq:expansion of the number operators}}}{\leq}
  \|\ell^{(1)}\|_{L^2\otimes \mathfrak{h}_{-1/2}} \norm{ \text{d} \Gamma_a(\omega)^{1/2} \chi } \norm{\left( \mathcal{N} + 1 \right)^{1/2}  ([1 - \tfrac{\mathcal N_b}{N} ]_{+}^{1/2} - 1 ) \phi } \notag \\
  &\qquad \qquad + \|\ell^{(1)}\|_{L^2\otimes \mathfrak{h}_{-1/4}}  \norm{ ([1 - \tfrac{\mathcal N_b}{N} ]_{+}^{1/2} - 1 ) \chi }
  \norm{ \text{d} \Gamma_a(\sqrt{\omega})^{1/2}  \mathcal{N}_b ^{1/2} \phi } \notag\\
  &\stackrel{\eqref{eq:expansion of the number operators}}{\leq}
  CN^{-1}   \norm{ \text{d} \Gamma_a(\omega)^{\frac12} \chi } \norm{\left( \mathcal{N} + 1 \right)^{\frac32}   \phi }
  + CN^{-1}\norm{\mathcal{N}_b   \chi }\norm{ \text{d} \Gamma_a(\sqrt{\omega})^{\frac12} \mathcal{N}_b^{\frac12} \phi }
  \notag   \\
  &\stackrel{\eqref{eq:normbound:dGamma(omega^s/2)}}{\leq} C N^{-\frac12} (\norm{ \text{d} \Gamma_a(\omega)^{\frac12} \chi } +\norm{ N^{-\frac12}\mathcal{N}_b   \chi }) (\norm{\left( \mathcal{N} + 1 \right)^{3/2}   \phi } + \norm{\mathbb{T}^{\frac12}\phi} ).
\end{align}
This implies the claimed bound since $\id_{\cN_b\leq N}\chi=\chi$.\medskip

\noindent \textit{Term \eqref{eq:Norm approximation Gross transformed
    dynamics 2d}}. This term is treated in close analogy to the previous one,
leading to
\begin{align}
  | \eqref{eq:Norm approximation Gross transformed dynamics 2d}|
  &\leq N^{-1/2} \norm{\left( \mathcal{N} + \mathbb T + 1 \right)^{1/2} \chi }
  \norm{\left( \mathcal{N}^3 + \mathbb T + 1 \right)^{1/2} \phi } .
\end{align}
\textit{Term \eqref{eq:Norm approximation Gross transformed dynamics 2e}}.
By means of Lemmas \ref{lem:aux:bounds} and \ref{lem:ops-Fock}, and
\eqref{eq:expansion of the number operators} we get
\begin{align}
  | \eqref{eq:Norm approximation Gross transformed dynamics 2e} |
  &\leq C N^{-1} \norm{\left( \mathcal{N} + 1 \right)^{1/2} \chi }
  \norm{\left( \mathcal{N} + 1 \right)^{3/2} \phi} .
\end{align}

This completes the proof of the lemma.
\end{proof}

\subsection{Estimates for the dressing
  transformation}\label{section:estimates-dressing}

In this section, we will derive estimates for the fluctuation generator
associated with the dressing transformation $D_{u,\alpha}$ defined in
\eqref{eq:def:fluc:gen:dressing:trafo}, and its quadratic approximation
$\mathbb{D}^\Lambda_{u,\alpha}$ defined in \eqref{eq:Gross-Bog-generator}.
We also give the proof of Lemma~\ref{lem:beta:gamma:relation}.

\begin{lemma} \label{lem:bounds:Bog:generator:Gross} Let $\mathbb
  D^\Lambda_{u,\alpha}(\theta)$ and $D_{u,\alpha}(\theta)$ be defined by
  \eqref{eq:Gross-Bog-generator} and \eqref{eq:def:fluc:gen:dressing:trafo}.
  There exists a constant $C>0$, such that for all $(u,\alpha) \in
  H^1(\mathbb R^3) \oplus \mathfrak h_0$ with $\| u\|_{L^2} = 1 $, $|\theta|
  \le 1 $ and $\Lambda \in \R_+\cup\{\infty\}$
  \begin{align*}
    \pm   \mathbb D^{\Lambda}_{u,\alpha} (\theta)   & \le   C ( \| u\|_{H^1}^2 + \| \alpha \|_{\mathfrak h_0} )    ( \mathcal N+ 1 ) \\
    \pm i [\mathcal N , \mathbb D_{u,\alpha}^{\Lambda}(\theta) ]  & \le   C ( \| u\|_{H^1}^2 + \| \alpha \|_{\mathfrak h_0} )  (\mathcal N + 1) \\
    \pm \tfrac{d}{d\theta} \mathbb D_{u,\alpha}^{\Lambda}(\theta) & \le   C (\| u\|_{H^1}^{3/2} + \| \alpha \|_{\mathfrak h_0} ) (\mathcal N + 1) \\
    \pm i [\mathcal N ,  D_{u,\alpha}(\theta) ]  & \le  C( \| u\|_{H^1}^2 + \| \alpha \|_{\mathfrak h_0} )    (\mathcal N+ 1) (1 + (\tfrac{1}{N}\mathcal N_b)^{1/2})   
  \end{align*}
  in the sense of quadratic forms on $\mathcal F\otimes \mathcal F$ and
  \begin{align*}
    \big| \langle \phi , ( D_{u,\alpha}(\theta) - \mathbb D_{u,\alpha}^\infty(\theta) ) \chi \rangle \big| \le C  \| \phi \| \, \| (\mathcal N+1)^{3/2} \chi \| N^{-1/2}
  \end{align*}
  for all $\phi,\chi \in \cF \otimes \cF$.
\end{lemma}
\begin{proof} Recall that $\mathbb D_{u,\alpha}^\Lambda(\theta)$ and
  $D_{u,\alpha}(\theta)$ are defined w.r.t. the mean-field flow
  $(u^\theta,\alpha^\theta) = \mathfrak D[\theta](u,\alpha)$ and that $|
  u^\theta | =| u | $ by \eqref{eq:explicit:solution:dressing}. One readily
  shows that $\| \tau _{u,\alpha} \|_{L^\infty} \le 3 \| B_0\|_\2 \| \alpha
  \|_\2 $, and $\| \kappa^\Lambda_{u^\theta} \|_{L^2 (\mathbb R^6)} \le \|
  B_0 \|_\2$. By means of \eqref{eq:explicit:solution:dressing} and
  \eqref{eq:bound:L4:varphi} one further obtains
  \begin{align}
    \| \partial_\theta \kappa^{\Lambda}_{u^\theta} \|_{L^2(\mathbb R^6)} = \| \tau_{u,\alpha} \kappa_{u^\theta}^\Lambda \|_{L^2(\mathbb R^6)}  \le 3 \| B_0\|^2_\2 \| \alpha\|_\2 .
  \end{align}

  Using this, the estimates involving $\mathbb D^\Lambda_{u,\alpha}(\theta)$
  follow standard bounds for creation and annihilation operators that are
  special cases of Lemma~\ref{lem:ops-Fock}.

  Since $0\le [ 1- \tfrac{\mathcal N_b}{N} ]_+ \le 1 $ the commutator of
  $\mathcal N$ with the first two terms in
  \eqref{eq:def:fluc:gen:dressing:trafo} can be bounded as before.  Now let
  us denote the last term in \eqref{eq:def:fluc:gen:dressing:trafo}, which is
  cubic in the creation/annihilation operators, by
  $D_{u,\alpha}^{(3)}(\theta)$. Using the canonical commutation relations and
  again standard estimates for creation and annihilation operators, one
  obtains
  \begin{align}
    \label{eq:commutator:bound:D(3)}
    \pm i [\mathcal N,  D_{u,\alpha}^{(3)} (\theta)] &  \le 4 \| B_0 \|_{L^2} N^{-1/2} \mathcal N_b\, \mathcal N_a^{1/2}.
  \end{align}
  This proves the bound on the commutator $[\cN,D_{u,\alpha}^{(3)}]$.

  To show the last inequality, write
  \begin{align}
    & D_{u,\alpha}(\theta) - \mathbb D^\infty_{u,\alpha}(\theta) \\
    &=  \int dx dk  \Big( \kappa^\infty_{u^\theta}( k,x )  a_k^*  - \kappa^\infty_{u^\theta}( - k , x ) a_{ k} \Big) b_x^* \big( \big[ 1  - \tfrac{\mathcal N_b}{N}]_+^{1/2} -1 \big)   \notag \\
    & \quad  +   \int dx dk \Big( \overline{\kappa^\infty_{u^\theta}( k,x )}  a_k  - \overline{\kappa^\infty_{u^\theta}( - k , x )} a_{ k}^* \Big) b_x \big( \big[ 1  - \tfrac{\mathcal N_b-1}{N}]_+^{1/2} -1 \big)    + D_{u,\alpha}^{(3)}(\theta).\notag
  \end{align}
  It is straightforward to show that
  \begin{equation}
    |\langle \phi , D_{u,\alpha}^{(3)} (\theta) \chi \rangle | \le  4 \| B_0 \|_{L^2} \| \phi\| \|  \mathcal N_b\, \mathcal N_a^{1/2} \chi\| N^{-1/2}.
  \end{equation}
  For the first two terms, the bounds $\| \kappa^\infty_{u^\theta}
  \|_{L^2(\mathbb R^6) } \le \| B_0 \|_{\2}$ and $( [ 1 - \tfrac{\mathcal
    N_b-j}{N} ]_+^{1/2} -1 )^2 \le C N^{-1} (\mathcal N_b+1)$ for $j\in
  \{0,1\}$ imply that
  \begin{multline}
    \Big| \langle \phi, \int dx dk \, \Big( \kappa^\infty_{u^\theta}( k,x )  a_k^*  - \kappa^\infty_{u^\theta}( - k , x ) a_{ k} \Big) b_x^\bullet \big( \big[ 1  - \tfrac{\mathcal N_b-j}{N}]_+^{1/2} -1 \big)    \chi \rangle \Big| \\
    \le C \| \phi \|   \| (\mathcal N_b+1)^{1/2} \mathcal N_{b}^{1/2}\mathcal N_a^{1/2}  \chi\| N^{-1/2}
  \end{multline}
  where $\bullet\in \{\varnothing,*\}$. This completes the proof of the
  lemma.
\end{proof}

We now turn to the proof of Lemma \ref{lem:beta:gamma:relation}, which
relates the energy of excitations in $W^\mathrm{D} \Psi_N$ to the difference
of the energy of $\Psi$ to its mean-field energy. As explained below the
statement of the lemma, it is not possible to apply the strategy of the proof
of Theorem \ref{thm:gross-transformed dynamics reduced density matrices}
because $\mathbb T$ is not dominated by the generator $D_{u,\alpha}(\theta)$
in \eqref{eq:def:fluc:gen:dressing:trafo}. Instead, our proof relies on
comparing $X_{\mathfrak D}^* \mathbb T X_{\mathfrak D}$ directly with the
difference between the many-body energy per particle and the dressed
mean-field energy $\mathcal E_1$, evaluated at $\mathfrak
D(u,\alpha)$. Energy estimates of this kind were previously used in a
different context in \cite{KP2010}.

\begin{proof}[Proof of Lemma \ref{lem:beta:gamma:relation}]

  We recall Lemma \ref{lem:H:D:representation} and the fact that
  $\mathcal{E}= \mathcal E_1 \circ \mathfrak D$ as shown by Equation
  \eqref{eq:proof:conversation:mf:energy} for $\mathcal E_0 = \mathcal
  E$. With this at hand, we write the difference between the many-body energy
  per particle and the mean-field energy as
  \begin{align}
    \label{eq:relation energy difference}
    &  N^{-1} \scp{\Psi_N}{H_N \Psi_N} - \mathcal{E} (u , \alpha)
    =  N^{-1} \scp{ W^{\rm D}  \Psi_N}{ H_N^{\rm D} W^{\rm D}  \Psi_N} - \mathcal{E}_1 \circ \mathfrak D(u,\alpha)  .
  \end{align}
  Moreover, for $\zeta = {X}_{\mathfrak D(u,\alpha) } W^{\rm D}
  \Psi$ 
  , we can use \eqref{eq:gamma:second:quant} to write the relevant $\gamma$
  functional as
  \begin{align}
    \label{eq:relation between kinetic energy on excitation space und corresponding functional on N-particle space}
    \gamma\big[ W^{\rm D} \Psi_N, \mathfrak D(u ,\alpha) \big] = N^{-1} \scp{\zeta}{  \mathbb T  \zeta}.
  \end{align}

  To relate the expressions on the right-hand side of \eqref{eq:relation
    energy difference} and \eqref{eq:relation between kinetic energy on
    excitation space und corresponding functional on N-particle space}, we
  make use of the excitation map $X_{\mathfrak D(u,\alpha)}$. To do so, we
  rewrite $H_N^{\rm D}$ in terms of the fluctuation generator $H_{\mathfrak
    D(u,\alpha)}^{\rm D}(0)$ from \eqref{eq:def:fluct:generator}. This will
  allow us to employ previously established estimates.

  To ease up the notation, we set from now on
  $(u^\mathrm{D},\alpha^\mathrm{D}) = \mathfrak D(u,\alpha)$ and the
  shorthand $q=q_{u^\mathrm{D}}$, $h=h_{\mathfrak{D}(u,\alpha)}$,
  $f=f_{u^\mathrm{D}}$, $g=g_{u^\mathrm{D},\alpha^\mathrm{D}}$ (see
  \eqref{eq:def:dressed:meanfield:Ham}--\eqref{eq:def:g} for the definitions
  of these objects).  We can employ the results from Appendix
  \ref{section:Generator of the fluctuation dynamics of the Gross-transformed
    Nelson dynamics} to obtain
  \begin{align}
    X_{\mathfrak D(u,\alpha)} H_N^{\rm D} X_{\mathfrak D(u,\alpha)}^* & = H^{\rm D}_{\mathfrak D(u,\alpha)} (0)
    + \scp{u^\mathrm{D}}{h u^\mathrm{D}} \big( N - \mathcal{N}_b \big)+ N \scp{\alpha^\mathrm{D}}{\omega \alpha^\mathrm{D}}
    \nonumber \\
    &\quad
    + \sqrt{N} \Phi \big( \omega \alpha^\mathrm{D} + f + g \big) - b^* (u^\mathrm{D} ) b (q h u^\mathrm{D})
    \nonumber \\
    &\quad
    + \sqrt{N - \mathcal{N}_b}  b (q h u^\mathrm{D})
    + b^* (q h u^\mathrm{D}) \sqrt{N - \mathcal{N}_b}.
  \end{align}
  From the first inequality of Lemma \ref{lem:bounds:fluc:generator}, the
  fact that $\id_{\cN_b\leq N}\zeta=\zeta$, and \eqref{eq:expansion of the
    number operators}, we get
  \begin{align}
    \langle \zeta, \mathbb T \zeta\rangle
    &\leq 2 \langle \zeta, H_{ \mathfrak D(u , \alpha) }^{\rm D}(0) \zeta\rangle  + C \langle \zeta, \mathcal{N}  \zeta\rangle.
  \end{align}
  With the formula \eqref{eq:theta:mf-energy} for the dressed mean-field
  energy $\mathcal{E}_1$, we arrive at
  \begin{align}
    & N^{-1} \scp{ \zeta }{ \mathbb T  \zeta} -\abs{N^{-1} \scp{\Psi_N}{H_N \Psi_N} - \mathcal{E}( u, \alpha) } \label{eq:energy difference rest}\\
    & \leq C N^{-1}\langle \zeta, \mathcal{N}  \zeta\rangle + \Big|-N^{-1} \scp{ u^\mathrm{D}}{h u^\mathrm{D}} \scp{ \zeta }{\mathcal{N}_b \zeta}
    -N^{-1} \scp{ \zeta}{b^* (u^\mathrm{D} ) b (q h u^\mathrm{D}) \zeta}
    \notag\\
    & \qquad +N^{-\frac12} \scp{ \zeta }{\Phi \big( \omega \alpha^\mathrm{D} + f +g \big) \zeta }
    + 2N^{-\frac12}  \Re \scp{ \zeta }{\sqrt{1 - \mathcal{N}_b/N} b (q h u^\mathrm{D}) \zeta}\Big|. \notag
  \end{align}
  We bound the terms on the right hand side of in~\eqref{eq:energy difference
    rest} by (using that $\|\zeta\|=\|\Psi_N\|=1$)
  \begin{multline}
    |\eqref{eq:energy difference rest}| \leq C \scp{\zeta}{\mathcal{N} \zeta} N^{-1}\big(1 + \norm{h u^\mathrm{D}}_{L^2}
    \big)\\
    +CN^{-1/2} \big( \norm{\omega \alpha^\mathrm{D}  + f +g}_{L^2}   \norm{\mathcal{N}_a^{1/2} \zeta}    + C\norm{\mathcal{N}_b^{1/2} \zeta} \big).
  \end{multline}
  By Lemmas \ref{lem:bounds:f:g} and \ref{lemma:bounds:nabla:phi}, the norms
  of $f$, $g$, $hu^\mathrm{D}$ are bounded in terms of the
  $H^2\oplus\mathfrak{h}_{1/2}$-norm of $(u^\mathrm{D}, \alpha^\mathrm{D})$,
  which by Lemma \ref{lem:mf-dressing-regularity} is controlled by the norm
  of $(u,\alpha)\in H^2\oplus \mathfrak{h}_{3/2}$.  Thus there exists a
  constant $C$, depending on this norm, so that
  \begin{align}
    |\eqref{eq:energy difference rest}| \leq C \big(N^{-1}\scp{\zeta}{\mathcal{N} \zeta} +  (N^{-1} \scp{\zeta}{\mathcal{N} \zeta})^{1/2}\big). \label{eq:gamma-zeta-bound}
  \end{align}
  Moreover by \eqref{eq:beta:second:quant} and Lemma~\ref{lemma:reduced
    densities dressing flow} for $\theta=1$, we have
  \begin{align}
    \label{eq: relation between number operator and functional on N-particle space}
    N^{-1} \scp{\zeta}{ \mathcal{N}  \zeta}
    & = \beta \big[ W^{\rm D} \Psi_N, \mathfrak D( u , \alpha ) \big]  \le C \big(  \beta \big[\Psi_N, ( u , \alpha) \big] + N^{-1} \big).
  \end{align}
  Combined with~\eqref{eq:gamma-zeta-bound} this proves the statement of the
  Lemma.
\end{proof}


\appendix

\section{Initial states}
\label{sect:initial states prop}

\begin{proof}[Proof of Proposition \ref{proposition:class of suitable initial states}]
  Let $W^{\rm D}_{\geq K} = W(N^{-1/2} \sum_{j=1}^N B_{K,x_j})$ and $W^{\rm
    D}_{\geq K, x_j} = W(N^{-1/2} B_{K,x_j})$. The first inequality of the
  Proposition can be obtained similarly as \cite[Prop.~II.2]{LMS2021}. More
  explicitly, we use \eqref{eq: Weyl operators product} and \eqref{eq: Weyl
    operators shift property} to estimate
  \begin{align}
    N^{-1} \norm{\mathcal{N}_a^{1/2} W^*(\sqrt{N} \alpha) \Psi_{N,K}}
    &= N^{-1} \int d^3 k \,
    \norm{a_k (W^{\rm D}_{\geq K})^* (u^{\otimes N} \otimes \Omega)}^2 \notag \\
    & \leq \norm{B_{K,0}}_{L^2 }^2   \leq C K^{-2} .
  \end{align}
  By means of
  \begin{align}
    W^{\rm D}_{\geq K} (q_u)_1  (W^{\rm D}_{\geq K})^*
    &= (q_u)_1 + \ket{u} \bra{u}_1  - W^{\rm D}_{\geq K, x_1} \ket{u} \bra{u}_1 (W^{\rm D}_{\geq K,x_1})^*
  \end{align}
  and
  \begin{align}
    \big( 1 - (W^{\rm D}_{\geq K,x_1}) \big) \big( 1 - (W^{\rm D}_{\geq K,x_1})^* \big) \leq N^{-1} \hat{\Phi}(i B_{K,x_1})^2
  \end{align}
  we get
  \begin{align}
    \abs{\scp{\Psi_{N,K}}{(q_u)_1 \Psi_{N,K}}}
    &\leq 2 \norm{\left( 1 - (W^{\rm D}_{\geq K,x_1})^* \right)  u^{\otimes N}\otimes W(\sqrt{N}\alpha)\Omega}  
    \nonumber \\
    &\leq 2 N^{-1/2} \norm{B_{K,0}}_{L^2}  \norm{\big( \mathcal{N}_a + 1 \big)^{1/2}  u^{\otimes N}\otimes W(\sqrt{N}\alpha)\Omega} \notag \\
    & \leq 
    C K^{-1} 
    \big( N^{-1/2}  +  \norm{\alpha}_{L^2}  \big), 
  \end{align}
  and thus $\beta[\Psi_{N,K}, u, \alpha] \leq C K^{-1} \left( 1 +
    \norm{\alpha}_{L^2} \right)$.
  Since $W^{\rm D}_{\geq K} = \prod_{j=1}^N W^{\rm D}_{\geq K,x_j}$, the
  transformation relations of the dressing transformation from \cite[Section
  II]{GW2018} lead to
  \begin{align}
    W^{\rm D}_{\geq K} H_N (W^{\rm D}_{\geq K})^* 
    &= \sum_{j=1}^N
    \Big[ - \Delta_{j}
    + N^{-1/2} \hat{\Phi}  (\id(\abs{\cdot} \leq K) G_{x_j}) 
    \nonumber \\
    &\quad + N^{-1} \big( a( k B_{K, x_j})^2 + \text{h.c.} + 2a^*( k B_{K, x_j})a( k B_{K, x_j})\big) \notag \\
    &\quad 
    - 2  N^{-1/2} \big( 
    i \nabla_{x_j} \cdot a (k B_{K,x_j}) + a^* (k B_{K,x_j}) \cdot i \nabla_{x_j} \big)
    \Big]
    \nonumber\\
    &\quad 
    + N^{-1} \sum_{i < j}  V_{K}(x_i - x_j)
    + \textnormal{d} \Gamma_a (\omega) +  E_K
  \end{align}
  with
  \begin{align}
    V_{K}(x_i - x_j)
    &= 2 \Re \scp{B_{K,x_i}}{\omega B_{K,x_j}} - 4 \Re \scp{G_{x_i}}{B_{K,x_j}}
  \end{align}
  and
  \begin{align}
    E_{K} = \int_{\abs{k} \leq K} \frac{dk}{\omega(k) \left( k^2 + \omega(k) \right)} .
  \end{align}
  The shifting property of the Weyl operator \eqref{eq: Weyl operators shift
    property} then lets write the expectation value of the energy per
  particle as
  \begin{multline}
    N^{-1} \scp{\Psi_{N,K}}{H_N \Psi_{N,K}} =  N^{-1} E_K   \mathcal{E}( u,\alpha) \\
    +\Big\langle u,\big( 2 \Re \scp{ G_{(\cdot)}}{\alpha_{\geq K}}
    + A_{\alpha_{\geq K},(\cdot)} + F_{\alpha_{\geq K}}^2 
    + \frac{1}{2} V_{K} * \abs{u}^2 \big) u\Big\rangle
  \end{multline}
  with $\alpha_{\geq K} = \id_{\abs{\cdot} \geq K} \alpha$ and $A_\alpha$,
  $F_\alpha$ as defined in \eqref{eq:def:A_alpha},
  \eqref{eq:def:F_alpha}. Note that $\abs{E_K} \leq C \left( 1 + \ln K
  \right)$. By means of $\sup_{x \in \mathbb{R}^3} \abs{\scp{k
      B_x}{\alpha_{\geq K}}} \leq C K^{-1} \norm{\alpha}_{\mathfrak{h}_1}$,
  $\norm{V_K} \leq C K^{-3/2}$ and $\norm{u}_{L^2} = 1$ we get
  \begin{align}
    \abs{\scp{u}{A_{\alpha_{\geq K},(\cdot)} u}} &\leq C K^{-1} \norm{\alpha}_{\mathfrak{h}_1} \norm{u}_{H^1},
    \quad 
    \norm{F_{\alpha_{\geq K}}}_{L^{\infty}} \leq C K^{-1}  \norm{\alpha}_{\mathfrak{h}_1} 
  \end{align}
  and
  \begin{align}
    \norm{V_K * \abs{u}^2}_{L^{\infty}} 
    &\leq \norm{V_K}_{L^2} \norm{\abs{u}^2}_{L^2} 
    \leq C K^{-3/2} \norm{u}_{H^1}^2  .
  \end{align}
  Inequality \eqref{eq:bound:G,alpha} with $G_{(\cdot)} \mapsto \omega^{-1/2}
  G_{(\cdot)}$ and $\alpha \mapsto \omega^{1/2} \alpha$ leads to
  \begin{equation}
    \norm{\Re \scp{ G_{(\cdot)}}{\alpha_{\geq K}} u}_{L^2} \leq C K^{-1} \norm{\alpha}_{\mathfrak{h}_1} \norm{u}_{H^1}.
  \end{equation}
  In total, we obtain
  \begin{multline}
    \abs{N^{-1} \scp{\Psi_{N,K}}{H_N \Psi_{N,K}} - \mathcal{E} (u, \alpha) }\\
    \leq C  \left(   K^{-1} +  N^{-1} (1+ \ln K )  \right) \left( \norm{u}_{H^1}^2 + \norm{\alpha}_{\mathfrak{h}_1}^2 \right) .
  \end{multline}
\end{proof}

\section{Bogoliubov transformations}
\label{app:Bogoliubov:transformations}
For a linear map $T$ on a complex Hilbert space, we denote by $\bar T
f=\overline{T \bar f}$ its complex conjugate.

\begin{lemma}\label{lem:Bogoliubov-convergence}
  Let $\mathscr{H}$ be a Hilbert space and $\mathbb{U}_n$ with $n\in
  \mathbb{N}$ be a family of unitary Bogoliubov transformations on the Fock
  space over $\mathscr{H}$, that is, there exist $\mathfrak{u}_n$ linear,
  bounded and $\mathfrak{v}_n$ Hilbert Schmidt, so that
  \begin{equation*}
    \mathbb{U}_n^* a^*(f) \mathbb{U}_n =a^*(\mathfrak{u}_n f) + a(\mathfrak{v}_n \bar f), 
    \qquad \mathbb{U}_n^* a(f) \mathbb{U}_n = a(\overline{\mathfrak{u}}_n f) + a^*( \overline{ \mathfrak{v}_n  f}).
  \end{equation*}
  Assume that
  \begin{equation*}
    \mathbb{U}_\infty:=\slim_{n\to \infty} \mathbb{U}_n
  \end{equation*}
  exists, and moreover there exists a self-adjoint $A$, $D(A)$ with $A\geq 1$
  and $C>0$ so that for all $n\in \mathbb{N}\cup\{\infty\}$ and $\Psi\in
  D(\textnormal{d} \Gamma(A)^{1/2})$ it holds
  \begin{equation*}
    \langle \mathbb{U}_n\Psi, \cN \mathbb{U}_n\Psi\rangle \leq C \langle  \Psi, (1+\textnormal{d}\Gamma(A))\Psi\rangle.
  \end{equation*}
  Then $\mathbb{U}_\infty$ is a Bogoliubov transformation and the
  corresponding maps $\mathfrak{u},\mathfrak{v}$ satisfy $\|
  \mathfrak{u}\|_{\mathscr{H}\to \mathscr{H}} \leq C+1$, $\| \mathfrak{v}
  \|_{\mathfrak{S}_2(\mathscr{H})}\leq C$.
\end{lemma}
\begin{proof}
  We start by showing that for $\Psi, \Phi \in D(\mathrm{d}\Gamma(A)^{1/2})$,
  $f\in \mathscr{H}$, and $\bullet\in \{\varnothing, *\}$
  \begin{equation}
    \lim_{n \to \infty} \langle \Phi, \mathbb{U}_n^* a^\bullet(f) \mathbb{U}_n \Psi \rangle = \langle \Phi, \mathbb{U}_\infty^* a^\bullet(f) \mathbb{U}_\infty \Psi \rangle.
  \end{equation}
  To see this, note that
  \begin{align}
    &\Big| \langle \Phi, \mathbb{U}_n^* a^\bullet(f) \mathbb{U}_n \Psi \rangle - \langle \Phi, \mathbb{U}_\infty^* a^\bullet(f) \mathbb{U}_\infty \Psi \rangle\Big|\\
    &\quad \leq \begin{aligned}[t]
      C \|f\|_\mathscr{H} \Big(& \|(\mathbb{U}_n - \mathbb{U}_\infty) \Phi\|_\cF \|(1+\mathrm{d}\Gamma(A))^{1/2} \Psi\|_\cF \\
      & +\|(\mathbb{U}_n - \mathbb{U}_\infty) \Psi\|_\cF \|(1+\mathrm{d}\Gamma(A))^{1/2} \Phi\|_\cF\Big),
    \end{aligned} \notag
  \end{align}
  which tends to zero since $\mathbb{U}_n$ converges strongly to
  $\mathbb{U}_\infty$.

  Now let $f\in D(A)$, so $a^*(f)\Omega\in
  D(\mathrm{d}\Gamma(A)^{1/2})$. Then we have, using that $\mathbb{U}_n$ is a
  Bogoliubov transformation for $n\in \mathbb{N}$,
  \begin{subequations}
    \begin{align}
      \langle a^*(f)\Omega, \mathbb{U}_\infty a^*(g)\mathbb{U}_\infty^* \Omega\rangle
      &= \lim_{n \to \infty}\langle a^*(f) \Omega, (a^*(\mathfrak{u}_n g)+a(\mathfrak{v}_n \bar g)) \Omega\rangle \notag \\
      &= \lim_{n \to \infty} \langle  f, \mathfrak{u}_n g\rangle,\\
      \langle a^*(f)\Omega, \mathbb{U}_\infty a(g)\mathbb{U}_\infty^* \Omega\rangle
      &= \lim_{n \to \infty}\langle a^*(f) \Omega, (a(\overline{\mathfrak{u}}_n g)+a^*(\overline{\mathfrak{v}_n   g})) \Omega\rangle \notag\\
      &= \lim_{n \to \infty} \langle  f, \overline{\mathfrak{v}_n    g}\rangle.
    \end{align}
  \end{subequations}
  Since moreover
  \begin{equation}
    | \langle a^*(f)\Omega, \mathbb{U}_\infty a^\bullet(g)\mathbb{U}_\infty^* \Omega\rangle | \leq \|f\|_\mathscr{H} \|g\|_\mathscr{H}  \overbrace{\|(1+\cN^{1/2})\mathbb{U}_\infty\Omega\|_\cF}^{\leq C+1},
  \end{equation}
  the operators $\mathfrak{u}_n, \mathfrak{v}_n$ converge weakly to operators
  $\mathfrak{u}, \mathfrak{v}$ with norm less than $C+1$.  Weak convergence
  of $\mathfrak{u}_n$, $\mathfrak{v}_n$, implies that for $\Phi\in
  D(\cN^{1/2})$, $\Psi\in \Fock$, and $f\in \mathscr{H}$
  \begin{equation}
    \lim_{n\to \infty} \langle \Phi, a(\mathfrak{u}_n f)\Psi\rangle = \langle \Phi, a( \mathfrak{u} f)\Psi\rangle, \qquad \lim_{\Lambda\to \infty} \langle \Phi, a( \mathfrak{v}_n f)\Psi\rangle = \langle \Phi, a( \mathfrak{v}   f)\Psi\rangle,
  \end{equation}
  and thus for $\Phi, \Psi\in D(\cN^{1/2})$
  \begin{equation}
    \langle \Phi, \mathbb{U}_\infty^* a^*(f) \mathbb{U}_\infty \Psi \rangle 
    = \langle a(\mathfrak{u} f) \Phi,  \Psi \rangle +\langle  \Phi, a(\mathfrak{v} \bar f) \Psi \rangle
    = \langle  \Phi, (a^*(\mathfrak{u}f) + a(\mathfrak{v} \bar f)) \Psi \rangle,
  \end{equation}
  and similarly for $a(f)$. Moreover, we have
  \begin{equation}
    \|\mathfrak{v}_n\|_{\mathfrak{S}_2}^2 = \|\cN^{1/2}\mathbb{U}_n\Omega\|^2 \leq C,
  \end{equation}
  so the sequence $\mathfrak{v}_n$ is bounded in $\mathfrak{S}_2$, whence it
  has a subsequence that converges weakly in $\mathfrak{S}_2$. Since $\langle
  f, \mathfrak{v}_n g\rangle=\mathrm{Tr}( |f\rangle \langle g|
  \mathfrak{v}_n)$, the limit must be $\mathfrak{v}$, so $\mathfrak{v}\in
  \mathfrak{S}_2$ with norm less than $C$.  This proves the claim.
\end{proof}

\begin{lemma}
  The unitaries $\mathbb W^\Lambda_{u,\alpha}(\theta)$ and $\mathbb
  U_{u,\alpha,1}^\Lambda(t)$ defined in Propositions
  \ref{eq:Gross-Bog-generator} and \ref{prop:theta-Bog} are Bogoliubov
  transformations under the hypothesis given there.
\end{lemma}
\begin{proof}
  In view of the strong convergence of $\mathbb{U}^\Lambda_{u,\alpha,1}(t)$
  to $\mathbb{U}^\infty_{u,\alpha,1}(t)$ and the bound from
  Proposition~\ref{prop:theta-Bog}(i), it is sufficient it prove the claim on
  $\mathbb{U}_{u,\alpha,1}$ for $\Lambda<\infty$. For $\mathbb
  W^\Lambda_{u,\alpha}$ such a distinction is not necessary.  Note that the
  terms in the generators $\mathbb{D}^\infty_{u,\alpha}$ and
  $\mathbb{H}_{u,\alpha,1}^\Lambda$, $\Lambda<\infty$ of these unitaries with
  two creation operators ($b^*$ and/or $a^*$) have coefficients that are
  square integrable functions of their arguments (compare
  Lemma~\ref{lem:aux:bounds}). Since this is the only relevant property, we
  give an exemplary proof in the case of $\mathbb W_{u,\alpha}^\infty$.

  As a first step, consider the (tentative) equations for the operators
  $\mathfrak{u}(t)$, $\mathfrak{v}(t)$ associated with a Bogoliubov
  transformation. These are usually expressed in terms of the matrix
  \begin{equation}
    \mathcal{V}=\begin{pmatrix}
      \mathfrak{u} & \mathfrak{v} \\ \overline{\mathfrak{v}} & \overline{\mathfrak{u}}
    \end{pmatrix}\,
  \end{equation}
  where, in our case, $\mathfrak{u}, \mathfrak{v}:L^2(\R^3)\oplus
  L^2(\R^3)\to L^2(\R^3)\oplus L^2(\R^3)$ are linear and bounded.  For the
  case of $\mathbb{W}_{u,\alpha}^\infty$, asking that
  \begin{equation}
    i\partial_\theta \big(c^*(\mathfrak{u}(\theta) (f\oplus g)) +c(\mathfrak{v}(\theta) (\bar f\oplus\bar g)) \big)=[ \mathbb{D}_{u,\alpha}^\infty(\theta) , c^*(\mathfrak{u}(\theta) (f\oplus g)) +c(\mathfrak{v}(\theta)(\bar f\oplus \bar g))]
  \end{equation}
  yields the equations for $\mathcal{V}$
  \begin{equation}
    i\partial_\theta \mathcal{V}(\theta)=\mathcal{A}(\theta)\mathcal{V}(\theta),
  \end{equation}
  with
  \begin{align}
    \mathcal{A}(\theta) &=\begin{pmatrix} A_\mathrm{d} & -A_\mathrm{o} \\ \bar A_\mathrm{o} & -\bar A_\mathrm{d} \end{pmatrix}, \\
    A_\mathrm{o}=\begin{pmatrix} 0 & \kappa_{u^\theta}^\infty (\cdot,x) \\  \kappa_{u^\theta}^\infty (k,\cdot)& 0
    \end{pmatrix},& \quad
    A_{\mathrm{d}}=\begin{pmatrix} \tau_{u, \alpha} & -\kappa_{u^\theta}^\infty (-(\cdot),x) \\  -\overline{\kappa}_{u^\theta}^\infty (-k,\cdot)& 0
    \end{pmatrix}, \notag
  \end{align}
  where $\kappa$ acts as an integral operator by integrating in the variable
  denoted by $(\cdot)$.  These equations admit a unique solution
  $\mathcal{V}(\theta)$ with $\mathcal{V}(0)=1$ ($\mathcal{A}$ is a bounded
  perturbation of the diagonal terms, which are generators).

  Since $\kappa_{u^\theta}^\infty\in L^2(\R^6)$, the off-diagonal part
  $A_\mathrm{o}$ is a Hilbert-Schmidt operator, and consequently the
  off-diagonal part $\mathfrak{v}(\theta)$ of $\mathcal{V}(\theta)$ is also
  Hilbert-Schmidt~\cite[Lem.4.9]{Bossmann2019}.  By Shale-Stinespring
  criterion (see, e.g.,~\cite[Lem. 4.2]{Bossmann2019}) there thus exists a
  Bogoliubov transformation $\mathbb{V}(\theta)$ associated to
  $\mathfrak{u}(\theta)$, $\mathfrak{v}(\theta)$,
  i.e. $\mathcal{V}(\theta)$. This transformation is determined up to a
  $\theta$-dependent phase.  We can fix this phase by asking that
  $\mathbb{V}(\theta)\Omega = \mathbb{W}_0^\infty(\theta)\Omega$, as we now
  show.

  The vacuum vector $\Omega$ spans the one-dimensional space on which
  $c(f\oplus g)=b(f)+a(g)$ vanish for all $f,g\in L^2(\R^3)$. Then
  $\mathbb{V}(\theta)\Omega$ spans the joint kernel of $\mathbb{V}(\theta)
  c(f\oplus g) \mathbb{V}(\theta)^*$.  Using the equation satisfied by
  $\mathcal{V}(\theta)^{-1}$, we find (since we do not know a priori that
  $\mathbb{V}$ is differentiable, $i\partial_\theta \mathbb{V}(\theta)$
  denotes the weak derivative)
  \begin{align}
    0&= i\partial_\theta\big(\mathbb{V}(\theta) c(f\oplus g) \mathbb{V}(\theta)^*\big)\mathbb{V}(\theta) \Omega \notag \\
    &= \big(\mathbb{V}(\theta) c(f\oplus g) \mathbb{V}(\theta)^*\big) i\partial_\theta \mathbb{V}(\theta)\Omega
    + [\mathbb{V}(\theta) c(f\oplus g) \mathbb{V}(\theta)^*, \mathbb{D}_0^\infty(\theta)]\mathbb{V}(\theta)\Omega .
  \end{align}
  It follows that
  \begin{equation}
    i\partial_\theta \mathbb{V}(\theta)\Omega  - \mathbb{D}_0^\infty(\theta)\mathbb{V}(\theta)\Omega
  \end{equation}
  is in the kernel of $\mathbb{V}(\theta) c(f\oplus g) \mathbb{V}(\theta)^*$,
  and thus proportional to $\mathbb{V}(\theta)\Omega$. The constant of
  proportionality must be real, since $\mathbb{V}(\theta)$ is unitary, and
  thus we can set it to zero by adjusting the phase, i.e., setting
  $\widetilde{\mathbb{V}}(\theta)=\mathbb{V}(\theta) e^{i \int_0^\theta
    \nu(s)ds}$, which is also a Bogoliubov transformation associated with
  $\mathcal{V}(\theta)$. By uniqueness of the solution proved in
  Proposition~\ref{prop:Gross-Bog:evolution}, we thus have
  $\widetilde{\mathbb{V}}(\theta)\Omega =\mathbb{W}^\infty_0(\theta)\Omega$.
  Using the explicit action of $\mathbb{V}(\theta)$ on the creation and
  annihilation operators, one shows by induction that
  $\widetilde{\mathbb{V}}(\theta)\Psi=\mathbb{W}_0^\infty(\theta)$ for any
  state $\Psi$ obtained by application of a finite number of creation and
  annihilation operators (see the proof of~\cite[Lem.4.8]{Bossmann2019}).
  Since the span of such $\Psi$ is dense, this proves equality and thus that
  $\mathbb{W}_0^\infty(\theta)$ is a Bogoliubov transformation.
\end{proof}

\section{Fluctuation generator of the dressed dynamics}
\label{section:Generator of the fluctuation dynamics of the Gross-transformed
  Nelson dynamics}

In this section we provide the derivation of the fluctuation generator of the
dressed Nelson dynamics.

In order to disentangle the calculation, we write the excitation map as
\begin{equation}
  X_{u,\alpha} =X_u \otimes W^*(\sqrt{N} \alpha),
\end{equation}
where $X_u: \bigotimes_{\rm sym}^N L^2(\R^3) \to \cF_{\perp u} $ acts as
$\Psi_N \mapsto ( (X_u\Psi_N)^{(k)})_{k=0}^N$ with
\begin{equation}
  (X_u \Psi_N )^{(k)} =   \binom{N}{k}^{1/2} \prod_{i=1}^k (q_{u})_i
  \scp{u^{\otimes (N-k)}}{\Psi_N}_{L^2(\mathbb{R}^{3(N-k)})} \in \cF^{(k)}_{\perp u}
\end{equation}

We first calculate the result of applying only the Weyl operator.
\begin{lemma}
  For $ (u,\alpha)\in H^3(\mathbb R^3) \oplus \mathfrak h_{5/2}$ with
  $\norm{u}_\2 =1$ let $(u_t,\alpha_t) = \mathfrak
  s^{\mathrm{D}}[t](u,\alpha)$ be the solution to \eqref{eq:SKG dressed}.
  Let
  \begin{equation}
    H_{\alpha}^{\mathrm{D}, \leq N}(t) = i \dot{W}(\sqrt{N} \alpha_t)^* W(\sqrt{N} \alpha_t) + W(\sqrt{N} \alpha_t)^*  H_N^{\rm D} W(\sqrt{N} \alpha_t)
  \end{equation}
  then
  \begin{subequations}
    \begin{align}
      H_{\alpha}^{\mathrm{D}, \leq N}(t) & = \sum_{j=1}^N \Big(-\Delta_{x_j} + A_{\alpha_t,x_j} +  F^2_{\alpha_t} (x_j) -  \Re\scp{\alpha_t}{f_{u_t}+g_{u_t,\alpha_t}}  \Big)   \label{eq:Lambda_1_line1}\\[2mm]
      & \quad + \frac{1}{N} \sum_{i<j} V(x_i-x_j) + \textnormal{d}\Gamma_a(\omega) - \sqrt N \hat \Phi(f_{u_t} + g_{u_t,\alpha_t} )  \label{eq:Lambda_1_line2} \\[1.5mm]
      & \quad + \frac{2}{\sqrt N} \sum_{j=1}^N  \int d k \, kB_{x_j}(k) \cdot \big(-i\nabla_j+ F_{\alpha_t}(x_j) \big)  a_k^* + \mathrm{h.c.}  \label{eq:Lambda_1_line3}\\
      & \quad +  \frac{1}{N} \sum_{j=1}^N \big(a(k B_{x_j})^2 + \textnormal{h.c.} +2a^*(k B_{x_j})a(k B_{x_j}) \big)  \label{eq:Lambda_1_line4}
    \end{align}
  \end{subequations}
\end{lemma}
\begin{proof}
  We recall the form of $H^\mathrm{D}_N$ given in
  Lemma~\ref{lem:H:D:representation}. By the shift property of the Weyl
  operator \eqref{eq: Weyl operators shift property}, we have
  \begin{subequations}
    \begin{align}
      W(\sqrt{N} \alpha_t)^*  \textnormal{d}\Gamma_a(\omega)  W(\sqrt{N} \alpha_t) & = \textnormal{d} \Gamma_a(\omega) +  \hat \Phi(\sqrt N \omega \alpha_t ) + N \scp{\alpha_t}{\omega \alpha_t }, \label{eq:Weyl-dGamma}\\[1,5mm]
      W(\sqrt{N} \alpha_t)^*   \hat{A}_{x_j}   W(\sqrt{N} \alpha_t) & =   \hat{A}_{x_j}  +  \sqrt N A_{\alpha_t,x_j}.
    \end{align}
    with $ A_{\alpha_t,x} = -2 i\nabla_x \cdot \scp{kB_{x}}{\alpha_t} +
    \text{h.c.}$, and
    \begin{align}
      &W(\sqrt{N} \alpha_t)^* \big(a(k B_{x_j})^2 + \textnormal{h.c.} +2a^*(k B_{x_j})a(k B_{x_j})\big)W(\sqrt{N} \alpha_t) \notag \\
      &\ =  \big(a(k B_{x_j})^2 + \textnormal{h.c.} +2a^*(k B_{x_j})a(k B_{x_j})\big) \notag \\
      &\qquad  +  2 \sqrt N F_{\alpha_t}(x_j) \hat{\Phi}( k B_{x_j} ) +  N F^2_{\alpha_t}(x_j),
    \end{align}
  \end{subequations}
  where we inserted $F_{\alpha_t}(x)= 2 \Re \scp{k B_x}{\alpha_t}$ and used
  \begin{align}
    2 \big| \scp{\alpha_t}{kB_{x_j}} \big|^2 + \scp{\alpha_t}{kB_{x_j}}^2 + \scp{kB_{x_j}}{\alpha_t}^2 = \big( 2 \Re \scp{k B_{x_j}}{\alpha_t} \big)^2.
  \end{align}
  The operators $-\Delta_{x_j}$ and $V(x_j-x_i)$ are left unchanged by
  $W(\sqrt{N} \alpha)$, so it remains to the term with the time-derivative.
  We use the formula $ W(\sqrt{N} \alpha_t)^* = e^{a(\sqrt N \alpha_t)} e^{-
    a^*(\sqrt N \alpha_t)} e^{\frac{N}{2}\norm{\alpha_t}^2}$ to compute
  \begin{align}
    \tfrac{d}{dt}  W(\sqrt{N} \alpha_t)^*
    & = e^{a(\sqrt N \alpha_t)} a(\sqrt N \dot \alpha_t) e^{- a^*(\sqrt N \alpha_t)} e^{\frac{N}{2}\norm{\alpha_t}^2}\notag\\
    & \quad - e^{a(\sqrt N \alpha_t)} e^{-a^*(\sqrt N \alpha_t)} e^{ \frac{N}{2}\norm{\alpha_t}^2}  \big( a^*(\sqrt N \dot \alpha_t) - N \Re \scp{\dot \alpha_t}{\alpha_t} \big) \notag\\
    & =   W(\sqrt{N} \alpha_t)^* \Big( a(\sqrt N \dot \alpha_t) - a^*(\sqrt N \dot \alpha_t) - N\Im \scp{\dot \alpha_t}{ \alpha_t} \Big)
  \end{align}
  where the last step follows from
  \begin{align}
    e^{a^*(\sqrt N \alpha_t)} a(\sqrt N \dot \alpha_t) e^{-a^*(\sqrt N \alpha_t)} =   a(\sqrt N \dot \alpha_t) - N \scp{ \dot \alpha_t}{\alpha_t}  .
  \end{align}
  Inserting the equation of motion \eqref{eq:classical:transformed:equations}
  for $\alpha_t$, we find
  \begin{align}
    & i\tfrac{d}{dt}  W(\sqrt{N} \alpha_t)^*  =  W(\sqrt{N} \alpha_t)^* \big(-  \sqrt N \hat \Phi ( i\dot \alpha_t) + N \Re \scp{ \alpha_t}{ i \dot \alpha_t} \big) \\[1mm]
    & =  W(\sqrt{N} \alpha_t)^* \big( - \sqrt N \hat \Phi (\omega\alpha_t +  f_{u_t} + g_{u_t,\alpha_t} ) + N \Re \scp{ \alpha_t}{ \omega \alpha_t +  f_{u_t} + g_{u_t,\alpha_t} } \big), \notag
  \end{align}
  and with the shift property of the Weyl operator \eqref{eq: Weyl operators
    shift property}, we get
  \begin{multline}\label{eq:Weyl-derivative}
    (i \tfrac{d}{dt}  W(\sqrt{N} \alpha_t)^*)  W(\sqrt{N} \alpha_t) \\
    = - \sqrt N \hat \Phi (\omega\alpha_t + f_{u_t} + g_{u_t,\alpha_t}) - N \Re \scp{ \alpha_t}{ \omega \alpha_t + f_{u_t} + g_{u_t,\alpha_t}} .
  \end{multline}
  Noting that the terms involving $\omega \alpha_t$
  in~\eqref{eq:Weyl-derivative} and \eqref{eq:Weyl-dGamma} cancel gives the
  claim.
\end{proof}

The excitation map for the particles $X_u$ satisfies a general transformation
property.

\begin{lemma}\label{lem:trafo:dGamma(B)} For any densely defined operator $B
  : D(B) \subseteq L^2(\mathbb R^3) \to L^2(\mathbb R^3)$ and $u\in D(B)$ we
  have
  \begin{align}
    & X_u \bigg( \sum_{j=1}^N B_j\bigg) X_u^* = \scp{u }{B u} (N-\mathcal N_b) \notag \\
    & \quad +  \int dx \bigg( (q_{u} B u )(x) b_x^* \sqrt{[N - \mathcal N_b]_+}  + \overline{(q_{u}  B^{*} u)}(x)  \sqrt{[N - \mathcal N_b]_+}  b_x \bigg) \notag\\
    & \quad + \int dx dy\, b_x^* B b_y
  \end{align}
  as an operator identity on $\Fock^{(\le N)}_{\perp u}$.
\end{lemma}
\begin{proof}
  This follows by writing $\sum_j B_j$ as the restriction to
  $L^2(\R^3)^{\otimes N}$ of $\mathrm{d}\Gamma(B)$ and using the identities
  of Lemma~\ref{lemma:properties U}.
\end{proof}

With this, we can now give the proof of the formula for the fluctuation
Hamiltonian.

\begin{proof}[Proof of Lemma~\ref{lem:fluctuation:generator}]
  We first calculate
  \begin{equation}
    H_{u,\alpha}^{\mathrm{D}, \leq N} (t)= X_{u_t} H_{\alpha}^{\mathrm{D}, \leq N} (t) X_{u_t}^*  + i \dot X_{u_t} X_{u_t}^*
  \end{equation}
  and then add convenient terms that vanish on $\cF_{\perp u_t}^{\leq N}
  \otimes \cF$ to obtain the symmetric expression for
  $H_{u,\alpha}^\mathrm{D}(t)$.

  The terms in $H_{\alpha}^{\mathrm{D}, \leq N} (t)$ that are not either
  invariant under $X_u$ or of the well-known form arising for many-boson
  systems are those of lines \eqref{eq:Lambda_1_line3} and
  \eqref{eq:Lambda_1_line4}.  Using $\overline{B_{x_j}(k)} = B_{x_j}(-k)$ we
  write
  \begin{multline}
    \big(a(k B_{x_j})^2 + \textnormal{h.c.} +2a^*(k B_{x_j})a(k B_{x_j})\big)
    \\= \int dk dl\, k B_{x_j}(k) \cdot l B_{x_j}(l) \LaTeXunderbrace{\big[- 2 a_k^* a_{-l} + a_k^* a_l^* +a_{-k} a_{-l} \big]}_{=: \mathcal A_{kl}}
  \end{multline}
  and by Lemma \ref{lem:trafo:dGamma(B)} \allowdisplaybreaks
  \begin{align}
%
    & \frac{1}{N} X_{u_t} \Big( \int dk dl \, \sum_{j=1}^N  \big( kB_{x_j}(k)\cdot l B_{x_j}(l) \big)  \, \mathcal A_{kl}\Big)X_{u_t}^*\notag\\
    & \quad = \frac{1}{N}\int dk dl \underbrace{\scp{u_t}{k B_{(\cdot)}(k) \cdot l B_{(\cdot)}(l) u_t}}_{=M_{u_t}(k,l)} ( N - \mathcal N_b ) \, \mathcal A_{kl} \notag\\
    & \quad  \quad + \frac{1}{N} \int dk dl dx \bigg( \underbrace{(q_{u_t} k B_{(\cdot)}(k) \cdot l B_{(\cdot)}(l) u_t)(x)}_{=N_{u_t}(x,k,l)} b_x^*  \sqrt{[ N - \mathcal N_b]_+ } +\text{h.c} \notag\\
    & \quad \quad + \frac{1}{N} \int dk dl  dx dy \underbrace{(q_{u_t} k B_{(\cdot)}(k) \cdot l B_{(\cdot)}(l) q_{u_t} )(x,y)}_{=Q_{u_t}(x,y,k,l)} b_x^* b_y \, \mathcal A_{kl}.
  \end{align}
  Recalling the formulas \eqref{eq:def:f}, \eqref{eq:def:g}, and
  \eqref{eq:L(k):operator} for $f_{u_t}(k)$, $g_{u_t,\alpha_t} (k)$, and
  $L_\alpha(k)$, we also find
  \begin{align}
    & \frac{2}{\sqrt N} X_{u_t} \bigg( \sum_{j=1}^N  \int d k \, kB_{x_j}(k) \cdot \big(-i\nabla_j+ F_{\alpha_t}(x_j) \big)  a_k^* + \text{h.c.} \bigg) X_{u_t}^* \\
    &\quad  = \frac{2}{\sqrt N}  \int d k \, X_{u_t} \bigg( \sum_{j=1}^N  kB_{x_j}(k) \cdot \big(-i\nabla_j+ F_{\alpha_t}(x_j) \big) \bigg) X_{u_t}^*  a_k^* + \text{h.c.} \notag\\
    & \quad  = \frac{1}{\sqrt N}  \int dk \big( f_{u_t}(k) + g_{u_t,\alpha_t} (k)\big) ( N - \mathcal N_b  ) \, a_k^* + \text{h.c.} \notag\\
    & \qquad  + \frac{2}{\sqrt N}  \int dk dx  (q_{u_t} L_{\alpha_t}(k)u_t)(x) b_x^* a_k^* \sqrt{[N - \mathcal N_b]_+ }+ \text{h.c.} \notag\\
    &\qquad  + \frac{2}{\sqrt N}  \int dk dx ( \overline{ q_{u_t} L_{\alpha_t}(k)^* u_t)} (x)  \sqrt{[N - \mathcal N_b]_+ }  b_x  a_k^* + \text{h.c.} \notag\\
    & \quad \quad + \frac{2}{\sqrt N} \int dk  dx dy \underbrace{(q_{u_t} kB_{(\cdot)}(k) \cdot (-i\nabla + F_{\alpha_t}) q_{u_t})(x,y)}_{J_{u_t, \alpha_t}(x,k,y)} b_x^* b_y \, a_k^* + \text{h.c.} \notag
  \end{align}
  The term $\sqrt N \hat \Phi(f_{u_t} +g_{u_t,\alpha_t})$ from the first line
  cancels with the corresponding term in $H_{\alpha}^{\mathrm{D},\leq N}(t)$,
  and the remaining term $- N^{-1/2} \mathcal N_b \hat
  \Phi(f_{u_t}+g_{u_t,\alpha_t})$ equals $H_2(t)$ from the formula for
  $H^{\rm D}_{u,\alpha}(t)$.  We thus have
  \begin{multline}
    H_{u,\alpha}^{\mathrm{D},\leq N}(t) = i\dot{X}_{u_t} X_{u_t}^* \\+ X_{u_t} \Big( \sum_{j=1}^N h_{u_t,\alpha_t} + \frac{1}{N} \sum_{i<j} V(x_i-x_j) \Big) X_{u_t}^*  +\sum_{j=2}^5 H_j(t) . \label{eq:excitation:part:boson}
  \end{multline}
  From~\cite[Eq.(40)]{Lewin:2015a} and the equation satisfied by $u_t$, we
  deduce
  \begin{multline}\label{eq:X_u-derivative}
    i\dot{X}_{u_t} X_{u_t}^* = b^*(u_t)b(q_{u_t}h_{u_t,\alpha_t} u_t) -\langle ih_{u_t,\alpha_t}, u_t\rangle (N-\cN_b)
    \\ - \sqrt{N-\cN_b}b(q_{u_t}h_{u_t,\alpha_t} u_t) - \text{h.c}.
  \end{multline}
  This combines with the second term above to yield $H_0(t)+H_1(t)$ as in the
  analogous computations in \cite{Bossmann2019,Lewin:2015a}. Taking into
  account the obvious modifications from replacing $v \mapsto \tfrac{N-1}{N}
  V$, one can use for instance \cite[Eq. (2.20)]{Bossmann2019}. .
\end{proof}

\section*{Acknowledgements}
The authors would like to thank S\"oren Petrat for many fruitful discussions
throughout the project. M.F.\ acknowledges the support of the MUR grant
``Dipartimento di Eccellenza 2023--2027'', and of the ``Centro Nazionale di
ricerca in HPC, Big Data and Quantum Computing''. J.L. was supported by the
Agence Nationale de la Recherche (ANR) through the projects DYRAQ
(ANR-17-CE40-0016) and QUACO (ANR-17-CE40-0007), and the ICB received
additional support through the EUR-EIPHI Graduate School
(ANR-17-EURE-0002). N.L. gratefully acknowledges support from the Swiss
National Science Foundation through the NCCR SwissMap and funding from the
European Union's Horizon 2020 research and innovation programme under the
Marie Sk\l
odowska-Curie grant agreement N\textsuperscript{o} 101024712.\\

{}


\begin{thebibliography}{11}


\bibitem{AF2014} Z.~Ammari and M.~Falconi.  Wigner measures approach to the
  classical limit of the Nelson model: convergence of dynamics and ground
  state energy.  \emph{J. Stat. Phys.} 157(2), 330--362 (2014).

\bibitem{AF2017} Z.~Ammari and M.~Falconi.  Bohr's correspondence principle
  for the renormalized Nelson model.  \emph{SIAM J. Math. Anal.} 49(6),
  5031--5095 (2017).

\bibitem{bachelot} A.~Bachelot.  Probl\`eme de {C}auchy pour des syst\`emes
  hyperboliques semi-lin\'eaires.  \emph{Ann. Inst. H. Poincar\'e Anal. Non
    Lin\'eaire} 1(6), 453--478 (1984).
  
\bibitem{Bardos00}~C.~Bardos, F.~Golse~and~N\,J.~Mauser.~Weak coupling limit
  of the n-particle Schr\"odinger equation. \emph{Methods
    Appl. Anal. 7(2):275--294} (2000).

\bibitem{BOS2015} N. Benedikter, G. de Oliveira and B.~Schlein.  Quantitative
  derivation of the Gross--Pitaevskii equation.  \emph{Comm. Pure
    Appl. Math.} 68(8), 1399--1482 (2015).

\bibitem{BPS16} N. Benedikter, M. Porta and B. Schlein. Effective evolution
  equations from quantum dynamics. \emph{SpringerBriefs in Mathematical
    Physics} (2016).

\bibitem{BBCS19} C. Boccato, C. Brennecke, S. Cenatiempo and
  B. Schlein. Bogoliubov theory in the Gross--Pitaevskii limit. \emph{Acta
    Math.} 222(2): 219--335 (2019).

\bibitem{Boccato:2020}C. Boccato, C. Brennecke, S. Cenatiempo and B. Schlein.
  The excitation spectrum of Bose gases interacting through singular
  potentials \emph{J. Eur. Math. Soc.} 22,(7) 2331--2403 (2020).

\bibitem{Boccato:2016} C.~Boccato, S.~Cenatiempo and B.~Schlein. Quantum
  many-body fluctuations around nonlinear Schr{\"o}dinger
  dynamics. \emph{Ann.\ H. Poincar{\'e}} 18(1):113--191 (2016).

\bibitem{Bogoliubov1947} N.\,N.~Bogoliubov. On the theory of
  superfluidity. \emph{J. Phys.} 11(1):23 (1947).

\bibitem{BPPA20} L. Bo\ss mann, N. Pavlovi\'c, P. Pickl and A. Soffer. Higher
  order corrections to the mean- field description of the dynamics of
  interacting bosons. \emph{J. Stat. Phys.} 178(6):1362--1396 (2020).

\bibitem{Bossmann2019} L. Bo\ss mann, S. Petrat, P. Pickl and
  A. Soffer. Beyond Bogoliubov Dynamics. \emph{Pure Appl. Anal.} 3(4)
  677--726 (2021).

\bibitem{bourgain} J.~Bourgain. Periodic nonlinear Schr\"{o}dinger equation
  and invariant measures. \emph{Comm. Math. Phys.} 166(1) 1--26 (1994).
  
\bibitem{brennecke:2017} C.~Brennecke, P.\,T.~Nam, M.~Napi{\'o}rkowski and
  B.~Schlein. Fluctuations of N-particle quantum dynamics around the
  nonlinear Schr\"odinger equation. \emph{Ann.\ Inst.\ H.\ Poincar{\'e}
    Anal.\ Non Lin{\'e}aire} 36(5), 1201--1235 (2019).


\bibitem{BS2019} C.~Brennecke and B.~Schlein.~Gross-Pitaevskii dynamics for
  Bose--Einstein condensates.  \emph{Anal. PDE} 12(6), 1513--1596 (2019).

\bibitem{BSS21} C.~Brennecke, B.~Schlein and S.~Schraven.~Bogoliubov theory
  for trapped bosons in the Gross--Pitaevskii
  regime. \emph{Ann. H. Poincar\'e} 23, 1583--1658 (2022).

\bibitem{burq} N.~Burq and N.~Tzvetkov. Random data Cauchy theory for
  supercritical wave equations. II. A global existence
  result. \emph{Invent. Math.} 173(3), 477--496 (2008).

\bibitem{carusotto13} I.~Carusotto and C.~Ciuti. Quantum fluids of
  light. \emph{Rev. Mod. Phys.}  85, 300--368 (2013).
  
\bibitem{Chen2011} L. Chen, J.\,O. Lee and B. Schlein. Rate of convergence
  towards Hartree dynamics. \emph{J. Stat. Phys.} 144, 872--903 (2011).

\bibitem{Chong2016} J.~Chong. Dynamics of large boson systems with attractive
  interaction and a derivation of the cubic focusing NLS equation in
  $\mathbb{R}^3$. \emph{J. Math. Phys.} 62(4), 042106 (2021).

\bibitem{CHT2018} J.~Colliander, J.~Holmer and N.~Tzirakis.  Low regularity
  global well-posedness for the {Z}akharov and
  {K}lein--{G}ordon--{S}chr\"odinger systems.  \emph{Trans. Amer. Math. Soc.}
  360(9), 4619--4638 (2008).

\bibitem{CCFO2021} R.~Carlone, M.~Correggi, M.~Falconi and
  M.~Olivieri. Microscopic derivation of time-dependent point
  interactions. \emph{SIAM J. Math. Anal.} 53(4), 4657--4691 (2021).

\bibitem{CF2018} M.~Correggi and M.~Falconi. Effective potentials generated
  by field interaction in the quasi-classical limit.
  \emph{Ann. H. Poincar\'e} 19(1), 189-235 (2018).
  
\bibitem{CFO2019} M.~Correggi, M.~Falconi and M.~Olivieri.  Quasi-classical
  dynamics.  \emph{J. Eur. Math. Soc.} 25(2), 731--783 (2023).

\bibitem{davies} E.\,B.~Davies.  Particle-boson interactions and the weak
  coupling limit.  \emph{J. Math. Phys.} 20, 345--351 (1979).

\bibitem{DBGRN} S.~De Bievre, F.~Genoud and S.~Rota Nodari. Orbital
  Stability: Analysis Meets Geometry. In: C.~Besse, J.C.~Garreau (eds),
  Nonlinear Optical and Atomic Systems. Lecture Notes in Mathematics
  2146. Springer (2015).

\bibitem{DFPP} D-A.~Deckert, J.~Fr\"ohlich, P.~Pickl and A.~Pizzo. Dynamics
  of sound waves in an interacting Bose gas. \emph{Adv. Math. 293, 275--323}
  (2016).

\bibitem{derezinski14} J. Derezi\'{n}ski and M. Napi\'{o}rkowski. Excitation
  spectrum of interacting bosons in the mean-field infinite-volume
  limit. \emph{Ann. Henri Poincare} 15, 2409--2439 (2014).

\bibitem{ESY07}L. Erd\"os, B. Schlein and H.-T. Yau. Derivation of the cubic
  non-linear Schr\"odinger equation from quantum dynamics of many-body
  systems. \emph{Invent. Math.} 167(3):515--614 (2007).

\bibitem{ESY10}L. Erd\"os, B. Schlein and H.-T. Yau. Derivation of the
  Gross--Pitaevskii equation for the dynamics of Bose--Einstein
  condensate. \emph{Ann. Math.} 172(1):291--370 (2010).

\bibitem{ES01}L. Erd\"os and H.-T. Yau. Derivation of the nonlinear
  Schr\"odinger equation from a many-body Coulomb
  system. \emph{Adv. Theor. and Math. Phys.} 5(6):1169--1205 (2001).

\bibitem{falconi} M.~Falconi.  Classical limit of the Nelson model with
  cutoff.  \emph{J. Math. Phys.} 54(1), 012303 (2013).

\bibitem{Falconi21} M.~Falconi, N. Leopold, D. Mitrouskas and
  S. Petrat. Bogoliubov dynamics and higher-order corrections for the
  regularized Nelson model. \emph{Rev. Math. Phys}. 33, 2350006 (2023).

\bibitem{fontaine18} Q.~Fontaine, T.~Bienaim\'e, S.~Pigeon, E.~Giacobino,
  A.~Bramati and Q.~Glorieux. Observation of the Bogoliubov Dispersion in a
  Fluid of Light. \emph{Phys. Rev. Lett.} 121, 183604 (2018).
  
\bibitem{FG2017} R.\,L.~Frank and Z.~Gang.  Derivation of an effective
  evolution equation for a strongly coupled polaron.  \emph{Anal. PDE} 10(2),
  379--422 (2017).

\bibitem{FS2014} R.\,L.~Frank and B.~Schlein.  Dynamics of a strongly coupled
  polaron.  \emph{Lett. Math. Phys.} 104, 911--929 (2014).

\bibitem{frerot23} I.~Fr\'erot, A.~Vashist, M.~Morassi, A.~Lema\^itre,
  S.~Ravets, J.~Bloch, A.~Minguzzi and M.~Richard. Bogoliubov excitations
  driven by thermal lattice phonons in a quantum fluid of
  light. \emph{Preprint},
  \href{https://arxiv.org/abs/2304.08677}{arXiv:2304.08677} (2023).

\bibitem{GNV06} J.~Ginibre, F.~Nironi and G.~Velo. Partially classical limit
  of the Nelson model.  \emph{Ann. H. Poincar{\'e}} 7, 21--43 (2006).

\bibitem{GV1979a} J.~Ginibre and G.~Velo. The classical field limit of
  scattering theory for non-relativistic many-boson systems. I..
  \emph{Commun. Math. Phys.} 66(1), 37--76 (1979).

\bibitem{GV1979b} J.~Ginibre and G.~Velo. The classical field limit of
  scattering theory for non-relativistic many-boson systems. II..
  \emph{Commun. Math. Phys.} 68(1), 45--68 (1979).

\bibitem{Ginbre_Velo_expansion2} J.~Ginibre and G.~Velo. The classical field
  limit of non-relativistic bosons. II. Asymptotic expansions for general
  potentials.  \emph{Ann. Inst. H. Poincar{\'e} Physique th{\'e}orique}
  33(4), 363--394 (1980).

\bibitem{Ginbre_Velo_expansion1} J.~Ginibre and G.~Velo. The classical field
  limit of nonrelativistic bosons. I. Borel summability for bounded
  potentials.  \emph{Ann. Phys.} 128(2), 243--285 (1980).

\bibitem{Golse16} F. Golse.~On the dynamics of large particle systems in the
  mean field limit. In Macroscopic and Large Scale Phenomena: Coarse
  Graining, Mean Field Limits and Ergodicity, pages 1--144. Springer (2016).

\bibitem{GrechS13} P. Grech and R. Seiringer. The excitation spectrum for
  weakly interacting bosons in a trap. \emph{Comm. Math. Phys.} 322, 559--591
  (2013).

\bibitem{G2017} M.~Griesemer.  On the dynamics of polarons in the
  strong-coupling limit.  \emph{Rev. Math. Phys.} 29(10), 1750030 (2017).

\bibitem{GW2018} M.~Griesemer and A.~W\"{u}nsch.  On the domain of the Nelson
  Hamiltonian.  \emph{J. Math. Phys.} 59(4), 042111 (2018).

\bibitem{GrillakisMachedon:2013} M.~Grillakis and M.~Machedon. Pair
  excitations and the mean field approximation of interacting
  Bosons. I. \emph{Commun.\ Math.\ Phys.} 324(2):601--636 (2013).

\bibitem{GM2017} M.~Grillakis and M.~Machedon. Pair excitations and the mean
  field approximation of interacting Bosons. II. \emph{Commun. PDE} 42(2),
  24--67 (2017).

\bibitem{GMM2010} M.~Grillakis, M.~Machedon and D.~Margetis. Second-order
  corrections to mean field evolution of weakly interacting
  bosons. I. \emph{Commun. Math. Phys.} 294(1), 273 (2010).

\bibitem{GMM2011} M.~Grillakis, M.~Machedon and D.~Margetis. Second-order
  corrections to mean field evolution of weakly interacting
  bosons. II. \emph{Adv. Math.} 228(3), 1778--1815 (2011).

\bibitem{gruenrock} A.~Gr\"{u}nrock and H.~Pecher. Bounds in time for the
  Klein-Gordon-Schr\"{o}dinger and Zakharov system. \emph{Hokkaido Math. J.}
  35(1), 139--153 (2006).

\bibitem{HST22} C.~Hainzl,~B.~Schlein,~A.~Triay.~Bogoliubov Theory in the
  Gross-Pitaevskii Limit: a Simplified
  Approach. \emph{Preprint}. \href{https://arxiv.org/abs/2203.03440}{arXiv:2203.03440}
  (2022).

\bibitem{hairer-regularity} M.~Hairer.  A theory of regularity structures.
  \emph{Invent. Math.} 198(2), 269--504 (2014).

\bibitem{Hepp74} K.~Hepp. The classical limit for quantum mechanical
  correlation functions. \emph{Commun. Math. Phys.} 35(4):265--277 (1974).

\bibitem{hiroshima1998} F.~Hiroshima. Weak coupling limit with a removal of
  an ultraviolet cutoff for a Hamiltonian of particles interacting with a
  massive scalar field.  \emph{Infin. Dimens. Anal. Qu.} 1, 407--423 (1998).

\bibitem{KP2010} A.~Knowles and P.~Pickl. Mean-field dynamics: singular
  potentials and rate of convergence. \emph{Commun. Math. Phys.} 298(1),
  101--138 (2010).

\bibitem{Kuz2017} E.~Kuz. Exact evolution versus mean field second-order
  correction for bosons interacting via short-range two-body
  potential. \emph{Differ. Integral Equ.} 30(7/8), 587--630 (2017).


\bibitem{LaSch19} J.~Lampart and J.~Schmidt. On {N}elson-type {H}amiltonians
  and abstract boundary conditions.\ \emph{Commun. Math. Phys.}
  367(2):629--663 (2019).

 

\bibitem{L2022} N.~Leopold.  Norm approximation for the Fr\"ohlich dynamics
  in the mean-field regime. \emph{J. Funct. Anal.} 285(4), 109979 (2023).

\bibitem{LMRSS2021} N.~Leopold, D.~Mitrouskas, S.~Rademacher, B.~Schlein and
  R.~Seiringer. Landau--Pekar equations and quantum fluctuations for the
  dynamics of a strongly coupled polaron. \emph{Pure Appl. Anal.}  3(4),
  653--676 (2021).


\bibitem{LMS2021} N.~Leopold, D.~Mitrouskas and R.~Seiringer. Derivation of
  the Landau--Pekar equations in a many-body mean-field
  limit. \emph{Arch. Ration. Mech. Anal.} 240, 383--417 (2021).


\bibitem{LP2019} N.~Leopold and S.~Petrat. Mean-field dynamics for the Nelson
  model with fermions. \emph{Ann. H. Poincar{\'e}} 20(10), 3471--3508 (2019).

\bibitem{LP2018} N.~Leopold and P.~Pickl.  Mean-field limits of particles in
  interaction with quantized radiation fields.  In: D.~Cadamuro, M.~Duell,
  W.~Dybalski and S.~Simonella (eds). Macroscopic Limits of Quantum
  Systems. Vol 270 of Springer Proceedings in Mathematics \& Statistics,
  pages 185--214 (2018).

\bibitem{LP2020} N.~Leopold and P.~Pickl.  Derivation of the
  Maxwell--Schr\"odinger equations from the Pauli--Fierz Hamiltonian.
  \emph{SIAM J. Math. Anal.} 52(5), 4900--4936 (2020).

\bibitem{LRSS2019} N.~Leopold, S.~Rademacher, B.~Schlein and R.~Seiringer.
  The Landau--Pekar equations: adiabatic theorem and accuracy.
  \emph{Anal. PDE} 14, 2079--2100 (2021) .

\bibitem{Lewin:2015a} M.~Lewin, P.\,T.~Nam and B.~Schlein. Fluctuations
  around Hartree states in the mean-field regime. \emph{Am.\ J.\ Math.}
  137(6):1613--1650 (2015).

\bibitem{LNSS2015} M.~Lewin, P.\,T.~Nam, S.~Serfaty and
  J.\,P.~Solovej. Bogoliubov spectrum of interacting Bose
  gases. \emph{Comm. Pure Appl. Math.} 68(3), 413--471 (2015).

 

\bibitem{M2018} O.~Matte and J.\,S.~M\o ller. Feynman--Kac formulas for the
  ultra-violet renormalized Nelson model. \emph{Ast\'erisque} 404 (2018).

\bibitem{mcmullen1994} C.\,T.~McMullen. Complex dynamics and
  renormalization. \emph{Annals of Mathematics Studies} 135. Princeton
  University Press, Princeton, NJ (1994).
  
\bibitem{M2021} D.~Mitrouskas. A note on the Fr\"ohlich dynamics in the
  strong coupling limit. \emph{Lett. Math. Phys.} 111, 45 (2021).

\bibitem{mpp} D.~Mitrouskas, S.~Petrat and P.~Pickl. Bogoliubov corrections
  and trace norm convergence for the Hartree dynamics. \emph{Rev.\ Math.\
    Phys.} 31(8) (2019).

\bibitem{nam:2016} P.\,T.~Nam and M.~Napi{\'o}rkowski. A note on the validity
  of Bogoliubov correction to mean-field dynamics. \emph{J.\ Math.\ Pures
    Appl.} 108(5):662--688 (2017).

\bibitem{nam:2015} P.\,T.~Nam and M.~Napi{\'o}rkowski. Bogoliubov correction
  to the mean-field dynamics of interacting bosons. \emph{Adv.\ Theor.\
    Math.\ Phys.} 21(3):683--738 (2017).

\bibitem{namnap_review} P.\,T.~Nam and M.~Napi{\'o}rkowski. Norm
  approximation for many-body quantum dynamics and Bogoliubov theory. In: A.\
  Michelangeli and G. Dell'Antonio, (eds). Advances in Quantum Mechanics:
  contemporary trends and open problems. Springer-INdAM Series, 18 (2017).

\bibitem{namnap_low_dim} P.\,T.~Nam and M.~Napi{\'o}rkowski. Norm
  approximation for many-body quantum dynamics: focusing case in low
  dimensions. \emph{Adv. Math.} 350, 547--587 (2019).

\bibitem{NS2020} P.T. Nam and R. Salzmann. Derivation of 3D Energy-Critical
  Nonlinear Schr\"odinger Equation and Bogoliubov Excitations for Bose Gases.
  \emph{Comm. Math. Phys.} 375, 495--571 (2020).

\bibitem{Nam-Seiringer} P.\,T.~Nam and R.~Seiringer. Collective excitations
  of Bose gases in the mean-field regime. \emph{Arch. Rational Mech. Anal.}
  215, 381--417 (2015).

\bibitem{Nam-Triay} P.\,T. Nam and A. Triay. Bogoliubov excitation spectrum
  of trapped Bose gases in the Gross--Pitaevskii
  regime. \emph{Preprint},
  \href{https://arxiv.org/abs/2106.11949}{arXiv:2106.11949} (2021).

\bibitem{Narpiokowski23} M. Napi\'{o}rkowski.~Dynamics of interacting bosons:
  a compact review. Lecture notes series: Density functionals for
  many-particle systems, pages 117--154. National University of Singapore
  (2023).

\bibitem{nelson} E.~Nelson. Interaction of nonrelativistic particles with a
  quantized scalar field.  \emph{J. Math. Phys.} 5, 1190 (1964).

\bibitem{nickel-1997} G.~Nickel. Evolution semigroups for nonautonomous
  Cauchy problems. \emph{Abstr. Appl. Anal.} 2 (1997).

\bibitem{pecher} H.~Pecher.  Some new well-posedness results for the
  {K}lein--{G}ordon--{S}chr\"odinger system.  \emph{Diff. Int. Equations}
  25(1/2), 117--142 (2012).

\bibitem{soffer} S.~Petrat, P.~Pickl and A.~Soffer. Derivation of the
  Bogoliubov time evolution for a large volume mean-field
  limit. \emph{Ann. H. Poincar{\'e}} 21(2), 461--498 (2020).

\bibitem{pickl} P.~Pickl.  A simple derivation of mean field limits for
  quantum systems.  \emph{Lett. Math. Phys.} 97, 151--164 (2011).

\bibitem{Pickl2015} P.~Pickl. Derivation of the time dependent
  Gross--Pitaevskii equation with external fields. \emph{Rev. Math. Phys.}
  27(1):1550003 (2015).

\bibitem{Pizzo:iii} A. Pizzo. Bose particles in a box III. A convergent
  expansion of the ground state of the Hamiltonian in the mean field limiting
  regime.
  \emph{Preprint}, \href{https://arxiv.org/abs/1511.07026}{arXiv:1511.07026}
  (2015).

\bibitem{RS2009} I.~Rodnianski and B.~Schlein.  Quantum fluctuations and rate
  of convergence towards mean field dynamics. \emph{Commun. Math. Phys.}
  291(1), 31--61 (2009).

\bibitem{Seiringer11} R. Seiringer. The excitation spectrum for weakly
  interacting bosons. \emph{Commun. Math. Phys.} 306, 565--578 (2011).

\bibitem{JPS2007} J.\,P.\ Solovej.  Many Body Quantum
  Mechanics. \emph{Lecture
    notes},
  \url{http://web.math.ku.dk/~solovej/MANYBODY/mbnotes-ptn-5-3-14.pdf} (2014)
  .


\bibitem{Spohn80} H.~Spohn. Kinetic equations from Hamiltonian dynamics:
  Markovian limits. \emph{Rev. Modern Phys.} 52(3):569--615 (1980).

\bibitem{teufel} S.~Teufel. Effective N-body dynamics for the massless Nelson
  model and adiabatic decoupling without spectral
  gap. \emph{Ann. H. Poincar\'e} 3, 939--965 (2002).

 
%
%
%
%
%
%
%
%
%
%
%
%
%
%
%
%
%
%
%
%
%
%
%
%
%
%
%
%
%
%
%
%
%
%
%
%
%
%
%
    
\end{thebibliography}
\end{document}